\documentclass[12pt]{article}
\usepackage[T1]{fontenc}
\usepackage[utf8]{inputenc}
\usepackage{amssymb,amsmath}
\usepackage{amsfonts}
\usepackage{mathtools}
\usepackage{amsthm}
\usepackage{upgreek}
\usepackage{enumitem}
\usepackage{amsfonts}
\usepackage{chngcntr}

\usepackage{graphicx,psfrag}
\usepackage{amsfonts}
\usepackage{mathtools}

\usepackage{geometry}
\usepackage[colorlinks=true,linkcolor=blue, citecolor=blue, urlcolor = blue]{hyperref}
\usepackage{lmodern}
\usepackage{avant}
\usepackage[round,sort,authoryear]{natbib}
\usepackage{comment}
\usepackage{setspace}
\usepackage{mathptmx}

\counterwithin*{equation}{section}

\theoremstyle{definition}
\newtheorem{theorem}{Theorem}
\newtheorem{definition}{Definition}
\newtheorem{assumption}{Assumption}
\newtheorem{lemma}{Lemma}
\newtheorem{example}{Example}

\newtheorem{proposition}{Proposition}
\newtheorem{remark}{Remark}
\newtheorem{corollary}{Corollary}

\newcommand\norm[1]{\left\lVert#1\right\rVert}

\DeclarePairedDelimiter{\floor}{\lfloor}{\rfloor}

\usepackage{amsmath,amssymb}
\newcommand{\vertiii}[1]{{\left\vert\kern-0.25ex\left\vert\kern-0.25ex\left\vert #1 
    \right\vert\kern-0.25ex\right\vert\kern-0.25ex\right\vert}}

\newif\ifshow % toggle true or false based on if want to hide section
\showfalse % hide the sections

\ifshow
  \includecomment{wrap}
\else
  \excludecomment{wrap} 
\fi

\newcommand{\nc}{\newcommand}
\nc{\mbb}{\mathbb}\nc{\bb}{\mathbb}
\nc{\mbf}{\mathbf}\nc{\mb}{\mathbf}
\nc{\mc}{\mathcal}
\nc{\msf}{\mathsf}\nc{\ms}{\mathsf}
\nc{\acc}{\ms{acc}}
\nc{\ack}{\ms{ack}}
\nc{\alp}{\alpha}\nc{\al}{\alpha}\nc{\gka}{\alpha}
\nc{\ap}{\ms{ap}}
\nc{\apd}{\ms{apd}}
\nc{\base}{\ms{base}}\nc{\ba}{\ms{base}}
\nc{\bet}{\beta}\nc{\gkb}{\beta}
\nc{\boucle}{\ms{loop}}\nc{\Loop}{\ms{loop}}\nc{\lo}{\ms{loop}}
\nc{\bu}{\bullet}
\nc*{\cc}{\raisebox{-3pt}{\scalebox{2}{$\cdot$}}}
\nc{\centre}{\ms{center}}\nc{\Center}{\ms{center}}\nc{\cen}{\ms{center}}\nc{\ce}{\ms{center}}
\nc{\ci}{\circ}
\nc{\code}{\ms{code}}\nc{\cod}{\ms{code}}\nc{\decode}{\ms{decode}}\nc{\encode}{\ms{encode}}
\nc{\de}{:\equiv}
\nc{\dr}{\right}\nc{\ga}{\left}
\nc{\ds}{\displaystyle}
\nc{\ep}{\varepsilon}
\nc{\eq}{\equiv}
\nc{\ev}{\ms{ev}}
\nc{\fib}{\ms{fib}}
\nc{\funext}{\ms{funext}}\nc{\fu}{\ms{funext}}
\nc{\gam}{\gamma}
\nc{\glue}{\ms{glue}}\nc{\gl}{\ms{glue}}
\nc{\happly}{\ms{happly}}\nc{\ha}{\ms{happly}}
\nc{\id}{\ms{id}}
\nc{\ima}{\ms{im}}%\nc{\im}{\ms{im}}
\nc{\inc}{\subseteq}
\nc{\ind}{\ms{ind}}
\nc{\inl}{\ms{inl}}
\nc{\inr}{\ms{inr}}
\nc{\isContr}{\ms{isContr}}\nc{\co}{\ms{isContr}}\nc{\iC}{\ms{isContr}}\nc{\ic}{\ms{isContr}}
\nc{\isequiv}{\ms{isequiv}}\nc{\iseq}{\ms{isequiv}}\nc{\ieq}{\ms{isequiv}}
\nc{\ishae}{\ms{ishae}}\nc{\ish}{\ms{ishae}}\nc{\ih}{\ms{ishae}}
\nc{\isProp}{\ms{isProp}}\nc{\prop}{\ms{isProp}}\nc{\iP}{\ms{isProp}}\nc{\ip}{\ms{isProp}}
\nc{\isSet}{\ms{isSet}}\nc{\isS}{\ms{isSet}}\nc{\iss}{\ms{isSet}}\nc{\iS}{\ms{isSet}}\nc{\is}{\ms{isSet}}
\nc{\lam}{\lambda}%\nc{\la}{\lambda}
\nc{\LEM}{\ms{LEM}}\nc{\lem}{\ms{LEM}}\nc{\LE}{\ms{LEM}}
\nc{\lv}{\lvert}\nc{\rv}{\rvert}\nc{\lV}{\lVert}\nc{\rV}{\rVert}
\nc{\Map}{\ms{Map}}
\nc{\merid}{\ms{merid}}\nc{\meri}{\ms{merid}}\nc{\mer}{\ms{merid}}\nc{\me}{\ms{merid}}
\nc{\N}{\bb N}
\nc{\na}{\ms{nat}}
\nc{\nn}{\noindent}
\nc{\one}{\mb1}
\nc{\oo}{\operatorname}
\nc{\pd}{\prod}
\nc{\ps}{\mc P}
\nc{\pa}{\ms{pair}^=}
\nc{\ph}{\varphi}
\nc{\ppmap}{\ms{ppmap}}
\nc{\pr}{\ms{pr}}
\nc{\Prop}{\ms{Prop}}
\nc{\qinv}{\ms{qinv}}\nc{\qin}{\ms{qinv}}\nc{\qi}{\ms{qinv}}
\nc{\rec}{\ms{rec}}
\nc{\refl}{\ms{refl}}%\nc{\re}{\ms{refl}}
\nc{\seg}{\ms{seg}}
\nc{\Set}{\ms{Set}}
\nc{\sm}{\scriptstyle}
\nc{\sms}{\ms s}
\nc{\sq}{\square}
\nc{\suc}{\ms{succ}}\nc{\su}{\ms{succ}}
\nc{\tb}{\textbf}
\nc{\then}{\Rightarrow}
\nc{\tms}{\ms t}
\nc{\tx}{\text}
\nc{\transport}{\ms{transport}}\nc{\tr}{\ms{transport}}
\nc{\two}{\mb2}
\nc{\Type}{\text-\ms{Type}}\nc{\type}{\text-\ms{Type}}\nc{\ty}{\text-\ms{Type}}
\nc{\U}{\mc U}%\nc{\V}{\mc V}
\nc{\ua}{\ms{ua}}
\nc{\uniq}{\ms{uniq}}
\nc{\univalence}{\ms{univalence}}
\nc{\vide}{\varnothing}
\nc{\ws}{\ms{sup}}
\nc{\zero}{\mb0}

\title{\textbf{Optimal Estimation Methodologies for Panel Data Regression Models}\footnote{Part of these lecture notes were prepared during my Ph.D. studies at the Department of Economics, University of Southampton and further additions were incorporated during the academic year 2022/2023 when I was a Visiting Lecturer in Economics at the Department of Economics, University of Exeter Business School. I am grateful to Jose Olmo, Jean-Yves Pitarakis and  Tassos Magdalinos  from the School of Economic, Social and Political Sciences as well as Zudi Lu and Chao Zheng from the School of Mathematical Sciences, University of Southampton;  for helpful and stimulating discussions on the study of related econometric and statistical techniques. I am also grateful to  Giuseppe Cavaliere and Sebastian Kripfganz from the University of Exeter Business School  and Pietro Spini from the University of Bristol  as well as Markku Lanne and Mika Meitz from the Faculty of Social Sciences, University of Helsinki for helpful conversations. Financial support from the Research Council of Finland (grant 347986) is gratefully acknowledged. Address correspondence to Christis Katsouris, Faculty of Social Sciences, University of Helsinki, P.O. Box 17, FI-00014, Finland; email: christis.katsouris@finland.fi.
 } }

\author{\textbf{Christis Katsouris}\thanks{Dr. Christis Katsouris, is currently a Postdoctoral Researcher at the Faculty of Social Sciences, University of Helsinki. } \\ Department of Economics, University of Southampton\\ University of Exeter Business School\\ Faculty of Social Sciences, University of Helsinki}

\date{\today}

\begin{document}

\maketitle

\begin{abstract}
\vspace*{-0.8 em}
This survey study discusses main aspects to optimal estimation methodologies for panel data regression models. In particular, we present current methodological developments for modeling stationary panel data as well as robust methods for estimation and inference in nonstationary panel data regression models. Some applications from the network econometrics and high dimensional statistics literature are also discussed within a stationary time series environment.  
\end{abstract}

% $$\text{mathbf }\mb1,\text{ pmb }\pmb1,
%\text{ boldsymbol }\boldsymbol1$$ \tiny, \scriptsize, \footnotesize, \small, \normalsize, \large, \Large, \LARGE, \huge, \Huge

\maketitle

%%-------------------------------------------------------------------------%%
\newpage 
   
\begin{small}
\begin{spacing}{0.9}
\tableofcontents
\end{spacing}
\end{small}

%%-------------------------------------------------------------------------%%
\newpage

\section{Introduction}

Optimal estimation methodologies (e.g., see \cite{hilborn1969optimal} and \cite{dreze1976bayesian}) in the case of  multivariate data series  has been a research area of interest both in the econometrics as well as in the statistics literature the past 7 decades (see \cite{haavelmo1943statistical},  \cite{marschak1944random}, \cite{anderson1949estimation},  \cite{koopmans1950identification}, \cite{james1954normal},  \cite{zellner1974time},  \cite{sargan1976econometric}, \cite{espasa1977spectral} and  \cite{forchini2003conditional}). Moving into an era of ultra-high dimensional data structures where tools such as machine learning and deep learning techniques are employed for statistical learning purposes; it is of paramount importance to have a deep understanding of the classical optimal estimation methodologies in various econometric environments (see, \cite{white1996estimation, white2014asymptotic}). In previous lecture series, \cite{katsouris2023high, katsouris2023limit, katsouris2023quantile} discusses recent developments and several open problems in the time series and network econometrics literature are mentioned, with a special interest in nonstationary regression models and quantile regressions. The purpose of this lecture series is to present relevant issues on optimal estimation methodologies for panel data regression models (see also \cite{chamberlain1982multivariate, chamberlain1984chapter} and \cite{baltagi2008econometric}).

We shall begin by discussing the so-called \textit{identification problem} in linear and nonlinear econometric models (e.g., see \cite{rothenberg1971identification} and \cite{dreze1976bayesian}) which are commonly presented in the literature based on suitable distributional conditions such as the Gaussianity assumption of structural disturbances (e.g., see  \cite{phillips1976iterated}). However, in order to narrow down the relevant literature we only focus on estimation and inference methodologies for panel data regression models which is considered as a statistical learning mechanism for various applications found in economics, finance, biostatistics and climate sciences among many other fields.  Specifically, in the structural econometrics literature, \textit{weak identification} is commonly discussed as a model specification issue. Some relevant questions of interest:
\begin{itemize}

\item What is the connection between \textit{weak identification} and \textit{nearly singularity}? How is the \textit{nearly singularity problem} tackled in the econometrics literature in relation to \textit{weak identification}? 

%\item What about \textit{weak instrumentation} and \textit{weak identification}? 

%\item The presence of \textit{nuisance parameters} in the asymptotic distribution of test statistics: A well-known fact is that any size adjustments may not be feasible due to the presence of nuisance parameters in the asymptotic distribution. Obviously, this underscores the fundamental difference between a finite-sample theory and an asymptotic approximation.  

\end{itemize}

%%-------------------------------------------------------------------------%%
%\newpage 

\subsection{Identification of Non-Linear Dynamic Systems}

Overall, system identification plays an important role in revealing the unknown mechanisms of underlying complex phenomena. System identification includes detection of the model structure and estimation of the associated parameters. Moreover, a system identification problem can be thought of as an optimization problem where the optimal model is searched from a large predefined candidate model, given a criterion. The criterion is used to evaluate the performance of each model by measuring the discrepancy between the observed data and the model predictions (see, \cite{guo2016ultra}). Good criteria result to not only better parameter estimation but also a good search path along which the search process converges quickly to the optimal solution. Different criteria have been used in system identification such as the $L^2$ norm in least squares regression and the $L^1$ norm in least absolute value regression (\cite{guo2016ultra}).

%%-------------------------------------------------------------------------%%
\newpage 

Among these criteria, the least squares criterion is the most used because of its excellent properties, for example, least squares estimation can be configured to give estimates which are unbiased and efficient when the noise satisfies some basic assumptions. Thus, the least squares problem has analytic solutions and can be easily solved using the QR decomposition technique. However, the least squares technique  does not allow to capture some special characteristics in system identification such as the interconnectedness in data points; so-called \textit{network dependence} (see, \cite{katsouris2023limit}).  

A more accurate estimation methodology, which particularly overcomes the over-parametrization problem commonly found when using the least squares algorithm, is to employ an alternative criterion so-called ultra-least squares. The ULS criterion enables to characterise the model fitness more accurately. To put things into perspective, the ULS criterion considers model fitting in a smaller space, more specifically, the Sobolev space $H^m ( [0,T] )$. In other words, the ULS criterion takes into consideration not only the residuals but also the associate weak derivaties to measure the model fitness. Notice that the $L^2$ norm only emphasizes the similarity of two functions as a whole but disregards the closeness or shape (e.g., shape restrictions). Therefore, system identification can be interpreted as discovering unknown rules form a set of observations.  

A useful norm to develop metric properties in the Sobolev space is defined as below: 
\begin{align}
\norm{ x }_{ H^m } = \sqrt{  \sum_{j=0}^{m} \norm{ D^j x }_2^2 }    
\end{align}
where $D^j$ represents the $j-$th differentation operator. Based on the $\norm{ \cdot }_{ H^m }$ norm, a new criterion can then be defined as below:
\begin{align}
\mathcal{J}_H := \norm{ y - \sum_{i=1}^k \beta_i x_i }_2^2 + \sum_{j=0}^{m} \norm{ D^j \left( y - \sum_{i=1}^k \beta_i x_i \right) }_2^2.      
\end{align}
Due to the fact differentiation is a linear operator, the above criterion can be written as below: 
\begin{align}
\mathcal{J}_H := \norm{ y - \sum_{i=1}^k \beta_i x_i }_2^2 + \sum_{j=0}^{m} \norm{ \left( D^j y - \sum_{i=1}^k \beta_i D^j x_i \right) }_2^2.      
\end{align}
Thus, the $\mathcal{J}_H$ criterion consists of two parts: the first term corresponds to the standard least squares criterion which is based on evaluating the statistical distance (agreements) over the sample points; while the second term corresponds to the agreement of the weak derivatives which essentially emphases the agreement in shape (shape restrictions). Lastly, an interesting aspect worth investigating further is an ULS criterion obtained by adapting the $\mathcal{J}_H$ criterion to the nonlinear system identification problem\footnote{Nonlinear system identification involves both the estimation of the parameters and more importantly the problem of how to detect the structure of the unknown model. Model structure detection for linear systems is relatively easy and usually involves determining the order and time delay in a linear model. On the other hand, model detection can be complicated when the system is nonlinear due to the presence of many potential model terms and complex dynamics.}.

%%-------------------------------------------------------------------------%%
\newpage

\subsubsection{The Ultra-LS Problem}

\begin{definition}[\cite{guo2016ultra}]
Under the $H^m$ norm, the classical least squares problem is equivalent to the ultra-LS problem defined below: 
\begin{align}
\begin{bmatrix}
y
\\
D^1 y
\\
\vdots
\\
D^m y
\end{bmatrix} 
= 
\sum_{i=1}^k \beta_i 
\begin{bmatrix}
x_i
\\
D^1 x_i
\\
\vdots
\\
D^m x_i
\end{bmatrix} 
\end{align}
(which is a $k-$regressors problem). 
\end{definition}

In order to evaluate the contribution of the unknown weak derivatives in the $\mathcal{J}_H$ criterion, we consider distributional assumptions. Consider the signal process $y(t)$, then the associated distribition $T_y$ is defined as a functional $T_y: C_0^{\infty} ( [0,T] ) \to \mathbb{R}$ such that 
\begin{align}
\langle T_y, \varphi \rangle = \int_{[0,T] } y(t) \varphi (t) dt,     
\end{align}
for all $\varphi \in C_0^{\infty} ( [0,T] )$. Therefore, the distribution $T_y$ has weak derivatives which are defined: 
\begin{align}
\langle D^j T_y, \varphi \rangle =  (-1)^j \int_{[0,T] } y(t) \varphi^{(j)} (t) dt,  
\end{align}
Similarly , the distributions that correspond to $x_i$ are defined as:
\begin{align}
\langle T_{x_i}, \varphi \rangle = \int_{[0,T] } x_i(t) \varphi (t) dt,     
\end{align}
Thus, the regression problem is now solved in terms of conditional distribution projections. In other words, the system identification problem involves fitting the distribution $T_y$ by the combination of a set of distributions $T_{x_i}$. Therefore, the ultra-least squares problem becomes:
\begin{align}
\begin{bmatrix}
y
\\
\langle D^1 T_y, \varphi \rangle
\\
\vdots
\\
\langle D^m T_y, \varphi \rangle
\end{bmatrix} 
= 
\sum_{i=1}^k \beta_i 
\begin{bmatrix}
x_i
\\
\langle D^j T_x, \varphi \rangle
\\
\vdots
\\
\langle D^j T_x, \varphi \rangle
\end{bmatrix} 
\end{align}
Notice that above we have used interchangeable the terms ultra-least squares and conditional density projections to distinguish to the case where conditional density projections refer explicitly to the least squares problem. For the remainder of this study we focus on $L^2$ estimation techniques and leave the extensions to the Sobolev space for future discussions (see,  \cite{mahan2021nonclosedness}, \cite{abdeljawad2022approximations} and \cite{duzgun2023clustering}).

%%-------------------------------------------------------------------------%%
\newpage

\subsection{Consistent Estimation for Parametric Models}

When estimating the parameters of a correctly specified model, identification of the true parameters is a necessary condition for the consistent estimation and robust inference. However, identification is not a sufficient condition for consistency because the estimator may be constructed in such a way as not to be consistent for the true parameters, but for some other parameters which are nevertheless identifiable. Therefore, a model may fail to be identified, but estimation of a misspecified version of the model may yield identifiable parameters (see, \cite{bates1985unified}).

Suppose that a correctly specified model has the following functional form
\begin{align}
Y_t = X_t \theta + \epsilon_t, 
\end{align}
where $\mathbb{E} \left( X_t^{\prime} \epsilon_t \right) \neq 0$, and no instrumental variables are available for $X_t$. In general, we consider related conditions for existence and unique characterization of optimal estimation methodologies in panel data regression models (e.g., see \cite{andrews2001consistent}). 

\begin{corollary}(Existence)
Suppose there exists a measurable function $\hat{\theta}_n ( \omega )$ such that 
\begin{align}
Q_n \left( \omega, \hat{\theta}_n ( \omega ) \right) = \underset{ \theta \in \Theta }{ \mathsf{inf} } Q_n \big( \omega, \theta_n ( \omega ) \big), \ \ \ \text{for all} \ \omega \in \Omega.  
\end{align}
\end{corollary}

\paragraph{Structure} The remainder of this study is organized as follows. In Section \ref{Section2} we discuss commonly used GMM estimation techniques and their asymptotic properties. In Section \ref{Section3} we present key aspects related to estimation and inference for panel data regression models under time series stationarity while in Section \ref{Section4} we consider the case of nonstationary panel data regressions. Section \ref{Section5} presents recent developments in relation to panel data model estimation from the network econometrics literature. Section \ref{Section6} and Section \ref{Section7} present some further applications.  Network connectivity is a driving force for the risk transmission of various economic, financial and social phenomena. Related statistical problems include the modelling of financial contagion in stock markets, cointegration dynamics and market exuberance as well as the spread of epidemic diseases and the monitoring of climate change and biodiversity loss (see, \cite{gove2023coral}). We present relevant econometric methodologies which can be employed when considering the empirical and theoretical implications of such topics.

%%-------------------------------------------------------------------------%%
\newpage

\section{GMM Estimation Techniques and Properties}
\label{Section2}

\subsection{Large Sample Theory and Inference in GMM Estimation}

\subsubsection{Parameter Estimation using GMM}

We follow the framework proposed by \cite{Martinez2020asymptotic} who consider the econometric estimation using the GMM methodology, as briefly described below. Let $\theta \in \Theta$ denote a $p-$dimensional vector of parameters partitioned into $\theta = \left( \vartheta^{\prime}, \psi^{\prime} \right)^{\prime}$ of dimensions of $p_{\vartheta}$ and $p_{\psi}$, respectively. 

Denote with 
\begin{align}
F_T( \theta ) = \frac{1}{T} \sum_{t=1}^T f_t( \theta ),   
\end{align}
to represent the sample moments, where $f_t( \theta )$ is a $k-$dimensional vector-valued function of data and parameters with $k \geq p$ and $\mathbb{E} \big[ f_t( \theta ) \big] = 0$ as the true value of $\theta$. Moreover, we let with $r( \theta )$ to be a known function of the parameters such that $r : \Theta \to \mathbb{R}^q, q \leq p_{\psi}$.  

Suppose that $f_t ( \vartheta, .   )$ and $r ( \vartheta, .  )$ are continuously differentiable with respect to $\psi$, and let 
\begin{align}
J_T ( \theta ) =  \frac{ \partial F_T( \theta ) }{ \partial \psi^{\prime} }  \ \ \ \text{and} \ \ \   R(\theta) = \frac{ \partial r ( \theta ) }{ \partial \psi^{\prime} } 
\end{align}
Moreover, we denote with $\hat{V}_f( \theta )$ the $(k \times k)$ matrix that is positive definite almost surely, and define the GMM objective function as below
\begin{align}
S_T \left( \vartheta, \psi \right) = F_T ( \vartheta, \psi )^{\prime} \hat{V}_f \left( \vartheta, \tilde{\psi} \right)^{-1}  F_T ( \vartheta, \psi ),   
\end{align}
Furthermore, suppose that the constrained GMM estimator of $\psi$ given $\vartheta$ exists and is given by the following expression
\begin{align}
\hat{\psi} ( \theta ) = \underset{ \psi }{ \mathsf{arg min} } \ F_T ( \vartheta, \psi )^{\prime} \hat{V}_f \left( \vartheta, \tilde{\psi} \right)^{-1}  F_T ( \vartheta, \psi ).  
\end{align}

We also simplify the notation as below
\begin{align}
\hat{\psi} \equiv \hat{\psi} ( \vartheta ), \ \ \ \hat{r} ( \vartheta ) = r \left( \vartheta, \hat{\psi} \right)    
\end{align}
In addition we consider $\hat{C} ( \vartheta )$ be an almost surely full-rank $k \times ( k - p_{\psi}  )$ matrix that spans the null-space of $\tilde{V}_f ( \vartheta )^{-1/2} \hat{J}_T ( \vartheta )$ such that 
\begin{align}
\hat{C} ( \vartheta ) \hat{C} ( \vartheta )^{\prime} &= M_{ \tilde{V}_f ( \vartheta )^{-1/2} \hat{J}_t ( \vartheta )  } 
\\
M_x &= \left( I - P_X \right), \ \ \ P_X = X \left( X^{\prime} X \right) X^{\prime}.     
\end{align}

%%-------------------------------------------------------------------------%%
\newpage

\subsubsection{Weak Identification Aspects}

In particular, \cite{Martinez2020asymptotic} develops an asymptotic theory framework based on fixed-smoothing asymptotics for the test statistics in order to account for the estimation uncertainty in the underlying LRV estimators. Consider the following long-run variance estimator
\begin{align}
V_{ff} (\theta) = \underset{ T \to \infty }{ \mathsf{lim} } \ \mathsf{Var} \left( \frac{1}{ \sqrt{T}} \sum_{t=1}^T f( Y_t, \theta ) \right). 
\end{align}
Therefore, a non-parametric estimator of the LRV takes the quadratic form below
\begin{align}
\hat{V}_{ff} (\theta) &= \frac{1}{T} \sum_{t=1}^T \sum_{s=1}^T \omega_h \left( \frac{t}{T},  \frac{t}{T} \right) \big[ f( Y_t, \theta) -  \bar{f}( Y_t, \theta) \big] \big[ f( Y_t, \theta) -  \bar{f}( Y_t, \theta) \big]^{\prime}  
\\
\bar{f}( Y_t, \theta) &= \frac{1}{T} \sum_{s=1}^T f( Y_t, \theta ) ,    
\end{align}
such that $\omega(.,.)$ is a weighting function, and $h$ is the smoothing parameter indicating the amount of nonparametric smoothing. For example, we can estimate the kernel density the following way 
\begin{align}
 \omega_h \left( \frac{t}{T},  \frac{t}{T} \right) = k \left( \frac{(t - s)}{hT}    \right)    
\end{align}
for some kernel function $k(.)$, leading to the usual kernel LRV estimator. Thus, by substituting the smoothing estimator of the particular kernel function, we obtain the following test statistic
\begin{align}
Q_T ( \theta ) = \frac{1}{2} \left[ \frac{1}{\sqrt{T}} \sum_{s=1}^T f( Y_t, \theta ) \right]^{\prime} \hat{V}_{ff}^{-1}  \left[ \frac{1}{\sqrt{T}} \sum_{s=1}^T f( Y_t, \theta ) \right].     
\end{align}
Then, the K statistic is based on the first-order derivative of $Q_T (\theta)$. Define as below the gradients
\begin{align}
\mathsf{g}_j (Y_t, \theta) &= \frac{ \partial f ( Y_t, \theta ) }{ \partial \theta_j } \in \mathbb{R}^{ m \times 1 }, \ \ j \in \left\{ 1,..., d \right\}, 
\\
\mathsf{g} (Y_t, \theta) &= \frac{ \partial f ( Y_t, \theta ) }{ \partial \theta^{\prime} } = \big( \mathsf{g}_1 ( Y_t, \theta ),...., \mathsf{g}_d ( Y_t, \theta ) \big) \in \mathbb{R}^{ m \times d }, 
\\
\bar{\mathsf{g}} (Y_t, \theta) &= \frac{1}{T} \sum_{t=1}^T \sum_{t=1}^T \frac{ \partial f( Y_t, \theta ) }{ \partial \theta^{\prime} }  \in \mathbb{R}^{ m \times d }. 
\end{align}
Taking the first-order and second-order derivatives of $\hat{V}_{ff}(\theta)$ with respect to $\theta_j$, we obtain 
\begin{align}
\hat{V}_{ \mathsf{g}_j f } (\theta) &= \frac{1}{T} \sum_{t=1}^T \sum_{s=1}^T \omega_h \left( \frac{t}{T},  \frac{t}{T} \right) \big[ \mathsf{g}_j ( Y_t, \theta) -  \bar{\mathsf{g}}_j ( Y_t, \theta) \big] \big[ f( Y_t, \theta) -  \bar{f}( Y_t, \theta) \big]^{\prime}  
\\
\hat{V}_{ \mathsf{g}_j  \mathsf{g}_j } (\theta) &= \frac{1}{T} \sum_{t=1}^T \sum_{s=1}^T \omega_h \left( \frac{t}{T},  \frac{t}{T} \right) \big[ \mathsf{g}_j ( Y_t, \theta) -  \bar{\mathsf{g}}_j ( Y_t, \theta) \big] \big[ \mathsf{g}_j ( Y_t, \theta) -  \bar{\mathsf{g}}_j ( Y_t, \theta) \big]^{\prime}  
\end{align}

%%-------------------------------------------------------------------------%%
\newpage

Then, it follows that 
\begin{align}
\frac{ \partial Q_T(\theta) }{\partial \theta} = D_T(\theta) V_{ff}^{-1}(\theta) \left[ \frac{1}{ \sqrt{T}} \sum_{t=1}^T f(Y_t, \theta) \right],     
\end{align}
Denote with  $D_T (\theta) = \big[ D_{T,1}(\theta),...., D_{T,d}(\theta) \big] \in \mathbb{R}^{ m \times d}$, such that 
\begin{align}
D_{T,j} (\theta) = \left[ \frac{1}{\sqrt{T}} \sum_{t=1}^T \mathsf{g}_j (Y_t, \theta) \right] - \hat{V}_{\mathsf{g},f} (\theta) \hat{V}_{ff}^{-1}(\theta) \left[ \frac{1}{\sqrt{T}} \sum_{t=1}^T f_j (Y_t, \theta) \right]  \in \mathbb{R}^{m \times 1}.  
\end{align}
Then, the $K$ statistic for testing the null hypothesis $H_0: \theta = \theta_0$ against the alternative hypothesis given by $H_1 : \theta \neq \theta_0$ is given by 
\begin{align}
\mathcal{K}_T (\theta_0) = \left( \frac{ \partial Q_T (\theta_0) }{\partial \theta} \right)^{\prime} \big[ D_T(\theta_0)^{\prime} \bar{V}^{-1}_{ff}(\theta_0) D_T(\theta_0)  \big]^{-1} \left( \frac{ \partial Q_T (\theta_0) }{\partial \theta} \right).    
\end{align}
where for any concave function $\phi(\theta)$, $\partial \phi(\theta_0) / \partial \theta$ is defined to be 
\begin{align}
\frac{ \partial \phi (\theta_0) }{ \partial \theta} = \frac{ \partial \phi (\theta) }{ \partial \theta} \bigg|_{\theta = \theta_0} \in \mathbb{R}^{ d \times 1 }. 
\end{align}
Thus, to consider fixed-smoothing asymptotics, we employ the orthonormal series LRV estimator 
\begin{align}
\omega_h \left( \frac{t}{T}, \frac{s}{T} \right) = \frac{1}{G} \sum_{\ell = 1}^G \Phi_{\ell} \left( \frac{t}{T} \right) \Phi_{\ell} \left( \frac{s}{T} \right),     
\end{align}
where $G$ is a smoothing parameter for this estimator and $\Phi_{\ell} (.)$ is a set of a basis functions on $L^2 [0,1]$. The weighting function is expressed with respect to a set of basis functions on the space of $L^2 [0,1]$. Therefore, the LRV estimator takes the following form 
\begin{align*}
\bar{V}_{ff}(\theta) = \frac{1}{G} \sum_{\ell = 1}^G \left\{ \frac{1}{ \sqrt{T} } \sum_{t=1}^T \Phi_{\ell} \left( \frac{t}{T} \right) \big[ f( Y_t, \theta ) -  \bar{f}( Y_t, \theta )   \big]   \right\} \left\{ \frac{1}{ \sqrt{T} } \sum_{t=1}^T \Phi_{\ell} \left( \frac{t}{T} \right) \big[ f( Y_t, \theta ) -  \bar{f}( Y_t, \theta )   \big]   \right\}^{\prime}.     
\end{align*}
Then, an updated estimator for $\mathcal{J}_T (\theta_0)$ needs to be obtained from the sample such that  
\begin{align}
\mathcal{J}_T(\theta_0) = \left[ \frac{1}{\sqrt{T}} \sum_{t=1}^T \tilde{f}( Y_t, \theta_0 )     \right]^{\prime} \hat{V}_{ff}^{-1}(\theta_0 )  \left[ \frac{1}{\sqrt{T}} \sum_{t=1}^T \tilde{f}( Y_t, \theta_0 )  \right]
\end{align}

\begin{remark}
The modified statistic is not the same as the original statistic due to the projection of the function into the space which is induced by the transformation of the column vector space. This property allows us to obtain a consistent estimator $\hat{\theta}$ of $\theta_0$. However, to obtain an unbiased estimator for the variance of the estimator, we also need to obtain unbiased estimators for each partial derivative that the covariance matrix is composed to.

%%-------------------------------------------------------------------------%%
\newpage

Specifically, for the variance estimator we obtain the following
\begin{align}
V(\theta_0) := V 
=
\begin{bmatrix}
V_{ff}(\theta_0) & V_{fg}(\theta_0)
\\
V_{gf}(\theta_0) & V_{gg}(\theta_0) 
\end{bmatrix}.
\end{align}
\end{remark}
Therefore, the CLT to hold the following asymptotic distribution to hold
\begin{align}
\begin{pmatrix}
\displaystyle \frac{1}{ \sqrt{T} } \sum_{t=1}^T \big\{ f(Y_t, \theta_0 ) - \mathbb{E} \big[ f(Y_t, \theta_0) \big] \big\} 
\\
\\
\displaystyle \frac{1}{ \sqrt{T} } \sum_{t=1}^T \mathsf{vec} \big\{  \mathsf{g}( Y_t, \theta_0 ) - \mathbb{E} \big[ \mathsf{g}( Y_t, \theta_0) \big] \big\} 
\end{pmatrix}    
\Rightarrow 
\begin{pmatrix}
\psi_f
\\
\psi_{ \mathsf{g} }
\end{pmatrix},
\end{align}
where $\psi_f \in \mathbb{R}^{ m \times 1 }$ and $\psi_{\mathsf{g}}  \in \mathbb{R}^{ m  d \times 1 }$. Therefore, it holds that we have a sequence of matrices 
\begin{align}
T^{\kappa} \times \mathbb{E} \left[ \mathsf{g} ( Y_t, \theta_0 ) \right] \to \Pi = \left( \Pi_1,..., \Pi_d \right) \in \mathbb{R}^{ m \times d }     
\end{align}
In other words, we consider the convergence rate of the $m-$system equations. By multiplying with $\textcolor{red}{T^ {\kappa}  }$ we ensure that we take into account the different convergence rate depending on the type of functional form specification for the model under investigation. Specifically, when $\kappa = 0$, then the $m-$moment conditions doesn't include the correct rate of convergence which implies that the matrix $\Pi$ has a full column rank and therefore the parameter under the null hypothesis $\theta_0$, can be estimated at the usual parametric $\sqrt{T}$-rate. In other words, the case which corresponds to the weak identification of the model specification occurs when $\kappa = 1/2$, since it asymptotically converge into the null matrix, such that,  $\Pi = 0$ and therefore, $\theta_0$ cannot be consistently estimated (see,  \cite{Martinez2020asymptotic}).  Therefore, the estimation procedure for the case of fixed autocorrelation is given as below
\begin{align*}
D_{T, j}(\theta_0) - \sqrt{T} \mathbb{E} \left[  \mathsf{g} (Y_t, \theta_0 ) \right]
=
\frac{1}{ \sqrt{T} } \sum_{t=1}^T \bigg\{ \mathsf{g} (Y_t, \theta_0 ) - \mathbb{E} \big[ \mathsf{g} (Y_t, \theta_0 ) \big] \bigg\} 
- \hat{V}^{-1}_{ \mathsf{g}_j, f  } (\theta_0) \hat{V}^{-1}_{ ff} (\theta_0) \times \left[  \frac{1}{\sqrt{T}} \sum_{t=1}^T f( Y_t, \theta_0 )  \right]
\end{align*}
Now, to estimate the above moment conditions the important component of the estimation procedure is to obtain unbiased estimators for the LRV covariance matrices which are computed based on a set of basis functions. Let $\ell \in \left\{ 1,..., G \right\}$, then the basis functions shall satisfy: \textit{(i)}. $\Phi_{\ell} (.)$ are piesewise monotonic, continuously differentiable, and \textit{(ii)}. $\Phi_{\ell} (.)$ are orthonormal in the space of $L^2 [ 0,1]$ functions and satisfy $\displaystyle \int_0^1 \Phi_{\ell} (x) dx = 0$.  Therefore, the corresponding estimators are obtained as below
\begin{align}
\frac{1}{\sqrt{T}} \sum_{t=1}^T \Phi_{\ell} \left( \frac{t}{T} \right) \left[ f( Y_t, \theta_0 ) - \bar{f} ( Y, \theta_0 ) \right] \Rightarrow \int_0^1 \Phi_{\ell}(r) dB_f (r)
\\
\frac{1}{\sqrt{T}} \sum_{t=1}^T \Phi_{\ell} \left( \frac{t}{T} \right) \left[ \mathsf{g}( Y_t, \theta_0 ) - \bar{\mathsf{g}} ( Y, \theta_0 ) \right] \Rightarrow \int_0^1 \Phi_{\ell}(r) dB_f (r)
\end{align}

%%-------------------------------------------------------------------------%%
\newpage

\subsection{LGMM Estimation of Time Series via Conditional Moment Restrictions}

\subsubsection{Statistical Problem Formulation}

Following \cite{gospodinov2012local}, consider that a univariate process is \textit{strictly stationary} and \textit{geometically ergodic} and denote the conditional moment restrictions imposed by economic theory
\begin{align}
\label{CGMM}
\mathbb{E} \big[ u \big( r_{t+1}, r_t,..., r_{t-p+1} , \theta_0 \big) | r_t,..., r_{t-p+1} \big] = 0.    
\end{align}
where $u : \mathbb{R}^{p+1} \times \Theta$ is a known function up to a vector of unknown parameters $\theta_0 \in \Theta$.

Consider the AR(1) model with martingale difference errors
\begin{align}
r_{t+1} = \gamma_0 + \gamma_1 r_t + u_{t+1}, \ \ \ \ \mathbb{E} [ u_{t+1} | r_t ] = 0.     
\end{align}
In this case, the moment function is specified as $u( r_{t+1}, r_t, \theta_0 ) = r_{t+1} - \gamma_0 - \gamma_1 r_t$, with $\theta_0 = \left( \gamma_0, \gamma_1 \right)$. Let $x_t = ( r_t,..., r_{t-p+1} )^{\prime}$ and $y_{t+1} = ( r_{t+1}, x_t^{\prime} )^{\prime}$. Then, the conditional moment restriction model \eqref{CGMM} is estimated by the GMM estimator based on the unconditional moment restrictions 
\begin{align}
\mathbb{E} \big[ \mathsf{g} ( y_{t+1}, \theta_0 ) \big] = \mathbb{E} \big[ \mathcal{A} ( x_t, \theta_0 ) u ( y_{t+1}, \theta_0 ) \big] = 0.
\end{align}
with a matrix of instruments $\mathcal{A} ( x_t, \theta_0 )$, which is implied by the original model. The continuously updated GMM estimator is defined as 
\begin{align}
\hat{\theta}_{GMM} &:= \underset{ \theta \in \Theta }{ \mathsf{arg min} } \ \left( \frac{1}{T-p}  \sum_{t=p}^{n-1} \mathsf{g} ( y_{t+1}, \theta ) \right)^{\prime} W_n( \theta )^{-1} \left( \frac{1}{T-p}  \sum_{t=p}^{n-1} \mathsf{g} ( y_{t+1}, \theta ) \right)
\\
W_n ( \theta ) &= \frac{1}{ n-p } \sum_{t=p}^{n-1} \mathsf{g} ( y_{t+1}, \theta ) \mathsf{g} ( y_{t+1}, \theta )^{\prime}    
\end{align}
is an optimal weight matrix to estimate the parameters from the unconditional moment restrictions $\mathbb{E} \big[ \mathsf{g} ( y_{t+1}, \theta_0 \big] = 0$. In particular,  \cite{gospodinov2012local} pursue an alternative approach and use a localized version of the GMM estimator that operates directly on the conditional moment restriction.
\begin{align}
w_{tj} = \frac{  \displaystyle  \mathbb{K} \left( \frac{ x_j - x_t }{  h }  \right)  }{  \displaystyle  \sum_{s=p}^{n-1} \mathbb{K} \left( \frac{ x_s - x_t }{  h }  \right) }    
\end{align}
Therefore, the kernel estimator of the conditional moment condition $\mathbb{E} \big[ u ( y_{t+1}, \theta ) | x_t \big]$ is defined as
\begin{align}
\widetilde{u}_n (x_t, \theta) = \sum_{j=p}^{n-1} w_{tj} u \big( y_{j+1}, \theta \big).    
\end{align}

%%-------------------------------------------------------------------------%%
\newpage

Then, the LGMM estimator minimizes its quadratic form as below
\begin{align}
\hat{\theta}_{LGMM} &=   \underset{ \theta \in \Theta }{ \mathsf{arg min} } \ \sum_{t=p}^{n-1} \left\{ \widetilde{u}_n (x_t, \theta)^{\prime} \big[ V_n( x_t, \theta ) \big]^{-1}  \widetilde{u}_n (x_t, \theta) \right\} \boldsymbol{1} \left\{ | x_t | \leq c_n \right\} 
\\
V_n( x_t, \theta ) &= \sum_{j=p}^{n-1} w_{tj} u \big( y_{j+1}, \theta \big) u \big( y_{j+1}, \theta \big)^{\prime}  
\end{align}

\begin{definition}[$\mathcal{D}-$Bounded Function] A function $a: \mathbb{R}^{p+1} \times A \to \mathbb{R}^q$ is called $\mathcal{D}-$bounded on A with order $s$ if the following conditions hold:

\begin{itemize}
    \item[(i)] $a(y, \theta)$ is almost surely differentiable at each $\theta \in A$,

    \item[(ii)] For each $t = p,..., T-1$, it holds that
    \begin{align*}
        \underset{ \theta \in A }{ \mathsf{sup} } \ \mathbb{E} \big[ \big| a( y_{t+1}, \theta ) \big|^2 \big] < \infty \ \ \ \text{and} \ \ \  \underset{ \theta \in A }{ \mathsf{sup} } \ \mathbb{E} \left[ \left| \frac{ \partial a( y_{t+1}, \theta ) }{ \partial \theta^{\prime} } \right|^s \right] < \infty.
    \end{align*}

     \item[(iii)] For each $t = p,..., T-1$, there exist constants $C_1, C_2 \in (0, \infty)$ such that
     \begin{align*}
        \underset{ x \in \mathbb{R}^{p+1} }{ \mathsf{sup} } \ \underset{ \theta \in A }{ \mathsf{sup} } \  \mathbb{E} \big[ \big| a( y_{t+1}, \theta ) \big| \big| x_t = x \big] f(x) < C_1, \ \ \  \underset{ x \in \mathbb{R}^{p+1} }{ \mathsf{sup} } \ \underset{ \theta \in A }{ \mathsf{sup} } \  \mathbb{E} \left[ \left| \frac{ \partial a( y_{t+1}, \theta ) }{ \partial \theta^{\prime} } \right| \right| x_t = x \big] f(x) < C_2 
    \end{align*}
    
\end{itemize}

\end{definition}

\begin{remark}
Notice that the $\mathcal{D}-$boundedness assumes boundedness of the conditional and unconditional (higher-order) moments of the function $a$ and its derivative.   
\end{remark}

\subsubsection{Derivations and Mathematical Proofs}

\begin{lemma}
Suppose Assumptions hold and the function $a : \mathbb{R}^{p+1} \times A$ is $\mathcal{D}-$bounded on $A$ with order $s > 2$. If $\frac{ \mathsf{log} (n ) }{ n h^p } \to 0$ as $n \to \infty$, then for any $0 < \xi < \infty$, 
\begin{align*}
\underset{ |x| \leq c_n  }{ \mathsf{sup} } \ \underset{ \theta \in A }{ \mathsf{sup} }  \ \left| \frac{1}{(T-p) h^p} \sum_{j=p}^{n-1} \mathbb{K} \left( \frac{x_j - x }{h}  \right) a \big( y_{j+1}, \theta \big) - \mathbb{E} \left[ \frac{1}{h^p} \mathbb{K} \left( \frac{x_j - x }{h}  \right) a \big( y_{j+1}, \theta \big)  \right] \right|  = \mathcal{O}_p \left( \sqrt{ \frac{\mathsf{log} (n) }{ nh^p } } \right).
\end{align*}

\end{lemma}

Moreover, the objective function of the LGMM estimator and its population counterpart are written respectively as below
\begin{align*}
\mathcal{Q}_n ( \theta ) 
&= 
\sum_{t=p}^{n-1} \left\{ \widetilde{u}_n (x_t, \theta)^{\prime} \big[ V_n( x_t, \theta ) \big]^{-1}  \widetilde{u}_n (x_t, \theta) \right\} \boldsymbol{1} \left\{ | x_t | \leq c_n \right\} 
\\
\mathcal{Q} ( \theta ) 
&= 
\mathbb{E} \left[ \mathbb{E} \big[ u \big( y_{t+1}, \theta \big) \big| x_t \big]^{\prime} V( x_t, \theta )^{-1} \mathbb{E} \big[ u \big( y_{t+1}, \theta \big) \big| x_t \big]  \right],
\end{align*}
where $V( x, \theta )^{-1}$ exists for each $x \in \mathbb{R}^p$ and $\theta \in \Theta$. 

%%-------------------------------------------------------------------------%%
\newpage

Define with 
\begin{align}
\hat{f}(x) &= 
\frac{1}{ (n-p) h^p } \sum_{j=p}^{n-1} \mathbb{K} \left( \frac{ x_j - x }{ h } \right)   
\\
\hat{u} ( x, \theta ) &= 
\frac{1}{ (n-p) h^p } \sum_{j=p}^{n-1} \mathbb{K} \left( \frac{ x_j - x }{ h } \right) u \big( y_{j+1}, \theta \big).   
\end{align}
Then, it holds that
\begin{align}
\underset{ |x| \leq c_n  }{ \mathsf{sup} } \ \left| \hat{f}(x) -  \mathbb{E} \left[ \hat{f}(x) \right] \right|  &= \mathcal{O}_p \left( \sqrt{ \frac{\mathsf{log} (n) }{ nh^p } } \right).  
\\
\underset{ |x| \leq c_n  }{ \mathsf{sup} } \ \underset{ \theta \in \Theta }{ \mathsf{sup} }  \ \big|  \hat{u}( x, \theta)   -  \mathbb{E} \left[ \hat{u}( x, \theta) \right] \big|  &= \mathcal{O}_p \left( \sqrt{ \frac{\mathsf{log} (n) }{ nh^p } } \right).  
\end{align}
Then, by a change of variables $a = \frac{  x_j - x }{h}$ and an expansion around $a = 0$, whe get that 
\begin{align}
\underset{ |x| \leq c_n  }{ \mathsf{sup} } \ \left|  \mathbb{E} \left[ \hat{f}(x) \right] - f(x) \right|  = h^p \underset{ |x| \leq c_n  }{ \mathsf{sup} } \  \left| \int \mathbb{K}(a) \frac{  df( x + \bar{a} h ) }{ da } da \right| = O( h^p).      
\end{align}
Combining these results we obtain that 
\begin{align}
\underset{ |x| \leq c_n  }{ \mathsf{sup} } \ \left| \hat{f}(x)  - f(x) \right|  &= \mathcal{O}_p \left( \sqrt{ \frac{\mathsf{log} (n) }{ nh^p } } \right) +  O( h^p),    
\\
\underset{ |x| \leq c_n  }{ \mathsf{sup} } \ \underset{ \theta \in \Theta }{ \mathsf{sup} }  \ \big|  \hat{u}( x, \theta)   -  \mathbb{E} \left[ u \big( y_{t+1}, \theta \big) \right] \big|  &= \mathcal{O}_p \left( \sqrt{ \frac{\mathsf{log} (n) }{ nh^p } } \right) +  O( h^p). 
\end{align}

\medskip

\begin{example}
Suppose that the data are generated by a zero-mean AR(1) model such that
\begin{align}
r_{t+1} = \theta_0 r_t  u_{t+1},    
\end{align}
for each $t = 1,...,T$, where $u_t \sim \textit{i.i.d} (0,1)$, and the conditional moment condition restriction can be obtained by defining $u \big( y_{t+1}, \theta_0 \big) = r_{t+1} - \theta_0 r_t$, with $y_{t+1} = ( r_{t+1}, r_t )^{\prime}$ and $x_t = r_t$. Then, to highlight the effect of smoothing on the moment functions, we compare the OLS estimator 
\begin{align}
\hat{\theta}_{OLS} = \mathsf{arg min}_{\theta \in \Theta} \sum_{t=1}^{n-1} u \big( y_{t+1}, \theta_0 \big)^2
\end{align}
and the LGMM estimator with a common weight matrix 
\begin{align}
\hat{\theta}_{LGMM} = \underset{ \theta \in \Theta }{ \mathsf{argmin} } \ \sum_{t=1}^{n-1} \widetilde{u}_n ( x_t, \theta )^2 \boldsymbol{1} \left\{ | x_t | \leq c_n \right\}  
\end{align}
\end{example}

%%-------------------------------------------------------------------------%%
\newpage

\paragraph{Proof of Part (b).} By expanding the first-order condition $\partial \mathcal{Q}_n / \partial \theta = 0$ around $\theta_0$ we obtain
\begin{align}
\frac{ \partial \mathcal{Q}_n }{ \partial \theta } +  \frac{ \partial \mathcal{Q}^2_n ( \bar{\theta} ) }{ \partial \theta \partial \theta^{\prime} }  &\overset{ d }{ \to } \mathcal{N} \big( 0, \ell ( \theta_0 ) \big)
\\
\frac{1}{2} \frac{ \partial \mathcal{Q}^2_n ( \bar{\theta} ) }{ \partial \theta \partial \theta^{\prime} } 
&\overset{ p }{ \to } \ell ( \theta_0 ).
\end{align}
Moreover, consider the score function such that 
\begin{align}
\mathcal{S}_n ( x_t, \theta_0 ) = \sum_{j=p}^{n-1} w_{tj} u \big( y_{j+1}, \theta_0 \big) \left[ \frac{ \partial u \big( y_{j+1}, \theta_0 \big) }{  \partial \theta_{\ell} } \right]^{\prime} + \sum_{j=p}^{n-1} w_{tj} \frac{ \partial u \big( y_{j+1}, \theta_0 \big) }{  \partial \theta_{\ell} } u \big( y_{j+1}, \theta_0 \big)^{\prime},      
\end{align}
for $\ell = 1,...,k$. 

%We study the first and higher order asymptotic properties of the LGMM estimator for strictly stationary and geometrically ergodic Markov processes. 

\subsection{Efficient Method of Moments}

Another relevant methodology to GMM estimation especially for irregular data structures with nonlinear dynamics is the estimation approach of Efficient Method of Moments (EMM), proposed by \cite{newey1987hypothesis} (see, also \cite{ortelli2005robust}). Specifically, the simulation-based EMM technique provides a systematic way for generating moment conditions for simulated method of moments (SMM) estimation. This approach is also related to the \textit{indirect inference} estimator. Both EMM and II use a first-stage auxiliary statistical model to generate moment conditions. However, the EMM mimics the first-order conditions for estimation of the auxiliary model, which is computationally much more tractable than mimicking the optimization problem itself, as is done in the II approach (see, \cite{li2009simulation, li2010indirect}).  Furthermore, EMM is useful in situations in which analytical characterization and evaluation of the likelihood function is infeasible. Thus, EMM selects moments based on the score function of an auxiliary model, called the score generator, to define a criterion function for SMM estimation (see, \cite{chung2001testing}).  In the case of panel data estimation, asymptotically EMM estimators for dynamic panel data regressions provide bias corrections when the number of time periods is fixed or tends to infinity with the number of panel units (see, \cite{breitung2022bias}).   

\begin{example}
Consider the pure autoregressive panel data model as below: 
\begin{align}
y_{it} = \mu_i + \rho y_{it-1} + u_{it}, \ \ \ t = 1,..., T    
\end{align}
where $\bar{y}_{ -1, i } = \frac{1}{T} \sum_{ t=1 }^T y_{i, t-1}$. Notice that the bias-corrected profile likelihood estimator $\tilde{\rho}$ results from solving the equation $\tilde{\rho} - \mathbb{E} \left[ \mathcal{S} \left( \tilde{\rho} \right) \right] = 0$. The bias term $\mathbb{E} \left[ \mathcal{S} \left( \tilde{\rho} \right) \right]$ is a complicated function of $\rho$ as it involves an expectation of a ratio of two random variables both depending on $\rho$. To simplify the derivation of the bias function, we first assume that the variance $\sigma^2$ is know, resulting to the profile score function below: 
\begin{align}
\mathcal{S} ( \rho ) = \frac{1}{ \sigma^2 } \sum_{ i = 1}^N \sum_{t = 1}^T \big( y_{i, t-1 } - \bar{y}_{-1,i} \big) \big( y_{i, t-1 } - \bar{y}_{-1,i} \big) .  
\end{align}
\end{example}

%%-------------------------------------------------------------------------%%
\newpage

\section{Panel Data Model Estimation}
\label{Section3}

In particular, the \textit{weak instrument problem} of the system GMM estimator in dynamic panel data models is studied by \cite{bhargava1991identification} and \cite{bun2010weak}.  Moreover estimation and inference in panel data with cross-sectional dependence is discussed in 
\cite{bai2004estimating}.

\subsection{Illustrative Examples}

\begin{example}
Consider the following simple panel data regression model
\begin{align}
y_{it} = \rho_i y_{i,t-1} + \varepsilon_{it}    
\end{align}
Notice that testing for a unit root in the simple panel AR$(1)$ model is identical on testing the hypothesis $\rho = 1$ in the panel AR$(1)$ model with covariates such that
\begin{align}
y_{it} = \rho_i y_{i,t-1} +  \beta_i^{\prime} x_{it} +  \varepsilon_{it}    
\end{align}

\end{example}

\medskip

\begin{example}
The dynamic error components regression is characterized by the presence of lagged dependent variable among the regressors such that 
\begin{align}
y_{it} = \delta y_{i,t-1} + x_{it}^{\prime} \beta + \mu_i + v_{it}, \ \ \ \ i = 1,...,N \ \ \ \ t = 1,...,T 
\end{align}    
\end{example}

\begin{example}
Consider the following panel data model as below
\begin{align}
y_{i,t} &= \lambda_i^{\prime} D_t + u_{i,t}
\\
u_{i,t} &= \rho u_{i,t-1} + \varepsilon_{i,t} 
\end{align}
where $i = 1,...,N$ and $t = 1,...,T$ corresponds to the cross-sectional units and the time periods. 
\end{example}

\begin{remark}
The above example demonstrates the particular dependence structure which implies that the panel data have persistence captured by the innovation equation imposed in the first stage equation. Additionally our interest is in capturing and modelling network dependence which is considered to be a different type of dependence to the usual cross-sectional dependence. According to \cite{moon2000estimation}, it holds that when there is a common time series local to unity parameter across independent individuals in a panel, it is apparent that the cross-section data carry additional information that can be used to in estimating the common localizing parameter $c$. In other words, the main purpose here is to propose a consistent local to unity modelling approach for panel data with cross-sectional dependence.  
\end{remark}

%%-------------------------------------------------------------------------%%
\newpage

\subsection{GMM Estimation for Panel Data Regression Models}

Consider again the following dynamic panel data regression model
\begin{align}
y_{it} = \gamma y_{i,t-1} + x_{it}^{\prime} \beta + \alpha_i + u_{it},    
\end{align}
where $\alpha_i$ represents the time-invariant unobserved heterogeneity (relevant references include among others the studies of \cite{huang2020identifying} and \cite{bonhomme2015grouped}. 

\medskip

\begin{example}
\cite{wintoki2012endogeneity} use the dynamic GMM estimator to estimate the effect of board structure on firm performance. Moreover, correctly modelling the presence \textit{unobserved heterogeneity} is crucial as it captures, among other aspects, managerial quality, which is likely to correlate with both firm performance and board structure. The econometric specification of interest is
\begin{align}
y_{it} = x_{it}^{\prime} \beta + \gamma_1 y_{i,t-1} +  \gamma_2 y_{i,t-2}    + \alpha_i + u_{it},
\end{align}
where $y_{it}$ is a measure of fund performance, such as either return on assets (ROA) or return on sales (ROS), and $x_{it}$ includes three board structure variables: board size, board composition, and board leadership. Thus, an important reason to include two lags of firm performance in the model is to make the equation dynamically complete, in the sense that any residual serial correlation in $u_{it}$ is controlled for. Furthermore, control variables include the firm's market-to-book ratio, firm age, and the standard deviation of its stock returns (over previous 12 months).  

\begin{itemize}

\item For the first-differenced equation, lagged values $y_{i,t-p}$ and $x_{i,t-p}$ are used as instruments, where $p > 2$. Therefore, for these instruments to be valid they must be relevant, that is, capture variation in current governance, as well as exogenous. This means that exogenous regressors should be uncorrelated to $u_{it}$. Based on economic theory a reasonable explanation to this fact: If the board structure today is one that trades off the expected costs and benefits of alternative board structures, then current shocks to performance must have been unanticipated when the boards were chosen. 

\item Thus the optimal estimation methodology in the given setting is the system GMM estimator which employs lagged levels as instruments for the first-differenced equation and using lagged differences as instruments for the levels equation. Moreover, the maintained assumption is that there is no serial correlation in $u_{it}$, and thus no second order serial correlation in $\Delta u_{it}$. In particular, the SGMM estimator captures the dynamic relationship between current government and past firm performance while provides statistical consistency and the unbiasedness property is not violated. 
    
\end{itemize}
    
\end{example}

%%-------------------------------------------------------------------------%%
\newpage

\subsubsection{Moment Conditions}

Consider the model of interest as below:
\begin{align}
y_{it} = x_{it}^{\prime} \beta + \varepsilon_{it},     
\end{align}
where it is assumed that $\mathbb{E} \big( z_{it} \varepsilon_{it} \big) = 0$ for a given vector of instruments $z_{it}$ of dimension $R \geq K$, where $K$ is the number of elements of $\beta$ (that is, the number of regressors in the model). Then, the set of population moment conditions can be written as below: 
\begin{align}
\mathbb{E} \big[ z_{it} \big( y_{it} -  x_{it}^{\prime} \beta \big) \big] = 0.  
\end{align}
In other words, these $R$ conditions can help to estimate the $K$ unknown parameters in $\beta$. Thus, the identification assumption is satisfied only for the true parameter values and is nonzero otherwise. On the other hand, since these expectations are unobservable in practice we rely on sample moments for statistical inference purposes which is given by 
\begin{align}
\frac{1}{NT} \sum_{i=1}^N \sum_{t=1}^T z_{it} \big( y_{it} -  x_{it}^{\prime} \beta \big).  
\end{align}
Moreover, the GMM estimator for $\beta$ is obtained by minimizing a quadratic form in the sample averages  
\begin{align}
\underset{ \beta }{ \mathsf{min} } \ \left( \frac{1}{NT} \sum_{i=1}^N \sum_{t=1}^T  z_{it} \big( y_{it} -  x_{it}^{\prime} \beta \big) \right)^{\prime} W_{NT} \left( \frac{1}{NT} \sum_{i=1}^N \sum_{t=1}^T  z_{it} \big( y_{it} -  x_{it}^{\prime} \beta \big) \right)
\end{align}
where $W_{NT}$ is an $( R \times R )$ positive definite weighting matrix, which might depend upon the observed sample and thus needs to be estimated.  When implementing the GMM estimation approach, we usually adjust the weighting matrix to obtain an asymptotically more efficient estimator. In particular, if the error term is heteroscedastic, but there is no correlation between different error terms, an empirical weighting matrix is given by the following expression 
\begin{align}
W_{NT}^{\mathsf{opt} } = \left( \frac{1}{NT} \sum_{i=1}^N \sum_{t=1}^T   \hat{\varepsilon}^2_{it}  z_{it} z_{it}^{\prime} \right),    
\end{align}
where $\hat{\varepsilon}$ is the residual given by $y_{it} - x_{it}^{\prime} \hat{\beta}_1$ such that $\hat{\beta}_1$ denotes an initial consistent estimator for $\beta$. Therefore, this makes the optimal GMM estimator a two-step estimator. During the first step, a consistent estimator for $\beta$ is obtained, which is used to calculate residuals and construct the estimated optimal weighting matrix. During the second step, an asymptotically efficient estimator is obtained. Thus, the optimal estimator can be obtained as below: 
\begin{align}
\hat{\beta}_{\mathsf{GMM}} = \left[ \left( \frac{1}{NT} \sum_{i=1}^N \sum_{t=1}^T  x_{it} z_{it}^{\prime} \right)     \right] .   
\end{align}

%%-------------------------------------------------------------------------%%
\newpage

\subsection{Panel Data with Cross Sectional Dependence}

In this section, we discuss the framework proposed by \cite{gonccalves2015bootstrap} that corresponds to modeling panel data with cross-sectional dependence (see, also \cite{phillips2003dynamic},  \cite{bond2002projection} and \cite{pesaran2021general}). A relevant issue for identification and estimation is examined by a large stream of literature which develops econometric methodologies for capturing cross sectional dependence and heterogeneity via the use of dynamic panel models, based on the seminal contributions of \cite{pesaran2006estimation}. Moreover, \cite{kapetanios2014nonlinear} present a framework for nonlinear panel models with cross-sectional dependence. Recently, in the spatial econometrics literature various methodologies have been proposed to model both spatial dependence and cross-sectional effects such as \cite{li2020spatial}. Moreover, \cite{Olmo2023} propose a network regression model with an estimated interaction matrix which incorporates both the cross-sectional as well as the network dependence in the form of a metric distance between the set of regressors (see, also \cite{kapar2022dynamic}). 

We shall denote with $\mathsf{cum} ( w_0 ) = \mathbb{E} ( w_0 )$ and $\mathsf{cum} ( w_0, w_{t_1} ) = \mathsf{Cov} ( w_0, w_{t_1} )$. Notice that for a given time series $\left\{ w_t \right\}$ and for $j \in \mathbb{N}$, we let $\mathsf{cum} \left( w_0, w_{t_1},..., w_{t_j - 1} \right)$ to denote the $j-$th order joint cumulant of $\left( w_0, w_{t_1},..., w_{t_j - 1} \right)$, where $t_1,...,t_{j-1}$ are integers. In particular, \cite{gonccalves2015bootstrap} impose the assumption of a martingale difference sequence restriction on $\left\{ \epsilon_{it}, t = 1,2,...  \right\}$ for each $i \in \left\{ 1,..., n \right\}$. Therefore, the \textit{m.d.s} assumption implies that the model for the conditional mean of $y_{it}$ given $\mathcal{F}_{i,t-1}$ is correctly specified. 

Denote with $Z_{nt}^{*}$ be a sequence of bootstrap statistics. These convergence modes hold
\begin{itemize}
\item $Z_{nt}^{*} = o_{ P^{*} }(1)$ in probability, or $Z_{nt}^{*} \overset{ P^{*} }{ \to } 0$ in probability, if for any 
\begin{align}
\epsilon > 0, \delta > 0 \ \ \underset{ n, T \to \infty }{ \mathsf{lim} } \mathbb{P} \big[ \mathbb{P}^{*} \big( \left| Z_{nT}^{*} \right| > \delta \big) > \epsilon \big] = 0.
\end{align}

\item $Z_{nt}^{*} = \mathcal{O}_{ P^{*} }(1)$ in probability, if for all $\epsilon > 0$ there exists a $M_{\epsilon} < \infty$ such that 
\begin{align}
\underset{ n, T \to \infty }{ \mathsf{lim} } \mathbb{P} \big[ \mathbb{P}^{*} \big( \left| Z_{nT}^{*} \right| > M_{\epsilon} \big) > \epsilon \big] = 0.
\end{align} 

\item $Z_{nT}^{*} \overset{ d^{*} }{ \to } Z$ in probability if, conditional on the sample, $Z_{nT}^{*}$, weakly converges to $Z$ under $P^{*}$, for all samples contained in a set with probability converging to one. Specifically, we write $Z_{nT}^{*} \overset{ d^{*} }{ \to } Z$ in probability if and only if $\mathbb{E}^{*} \big( f \left( Z_{nT}^{*} \right) \big) \to \mathbb{E} \left( f(Z) \right)$ in probability for any bounded and uniformly continuous function $f$.  
\end{itemize}
More precisely, the particular assumption of a correctly specified model for the conditional mean, allow us to obtain results for the recursive-design bootstrap based on the wild bootstrap.  We use the following notation for the bootstrap asymptotics.

%%-------------------------------------------------------------------------%%
\newpage 

\subsubsection{Asymptotic theory for the fixed effects estimator when $N,T \to \infty$}

Consider the stationary linear dynamic panel model with fixed effects 
\begin{align}
y_{it} = \alpha_i + \theta_0 y_{it-1} + \varepsilon_{it}
\end{align}
where $| \theta_0 | < 1$ and $\alpha_i$ are individual specific fixed effects that capture the unobserved individual heterogeneity. The standard fixed effects OLS estimator of $\theta_0$ is given by (see, \cite{gonccalves2015bootstrap})
\begin{align}
\hat{\theta} = \left( \frac{1}{NT} \sum_{i=1}^N \sum_{t=1}^T  \big( y_{it-1} - \bar{y}_{i(t-1)} \big)^2 \right)^{-1} \left( \frac{1}{NT} \sum_{i=1}^N \sum_{t=1}^T \big( y_{it-1} - \bar{y}_{i(t-1)} \big) \big( y_{it} - \bar{y}_{i (t)} \big) \right)
\end{align}
where
\begin{align}
\bar{y}_{i (t)} = \frac{1}{T} \sum_{t=1}^T y_{it} \ \ \ \ \text{and} \ \ \ \ \bar{y}_{i(t-1)} = \frac{1}{T} \sum_{t=1}^T y_{it-1}
\end{align}
Therefore, the main goal of this section is to provide a set of assumptions under which we can prove the bootstrap results that will follow and at the same present the asymptotic theory of the fixed effects estimator under these assumptions. 
\begin{theorem}[\cite{gonccalves2015bootstrap}]
Let $\left\{ y_{it} \right\}$ be generated as above. Then, we have that 
\begin{align}
\sqrt{NT} \left( \hat{\theta} - \theta_0 \right) \overset{ d }{ \to } \mathcal{N} \big( D, C \big). 
\end{align}
\end{theorem}
Thus, we consider the joint asymptotic theory of $\hat{\theta}$ as $N, T \to \infty$. Then, the fixed effects OLS estimator can be represented as below
\begin{align}
\sqrt{NT} \left( \hat{\theta} - \theta_0 \right) 
&= 
A_{NT}^{-1} \frac{1}{  \sqrt{NT} } \sum_{i=1}^N \sum_{t=1}^T \big( y_{it-1} - \bar{y}_{i(t-1)} \big) \big(  \varepsilon_{it} - \bar{\varepsilon}_{i} \big), 
\\
A_{NT} 
&= 
\frac{1}{NT} \sum_{i=1}^N \sum_{t=1}^T \big( y_{it-1} -   \bar{y}_{i(t-1)} \big)^2. 
\end{align}
and $\mu_i = \mathbb{E} \left( y_{it-1} \right) = \alpha_i / (1 - \theta_0)$. Therefore, we obtain that 
\begin{align}
\sqrt{NT} \left( \hat{\theta} - \theta_0 \right) 
= 
A^{-1} \frac{1}{  \sqrt{NT} } \sum_{i=1}^N \sum_{t=1}^T \big( y_{it-1} - \mu_i \big) \big(  \varepsilon_{it} - \bar{\varepsilon}_{i} \big) + o_p(1),  
\end{align}
since it can be shown that $A_{NT} \overset{ d }{ \to } A$.  

Furthermore, the following decomposition holds for the normalized score, 
\begin{align*}
\frac{1}{  \sqrt{NT} } \sum_{i=1}^N \sum_{t=1}^T \big( y_{it-1} - \mu_i \big) \big(  \varepsilon_{it} - \bar{\varepsilon}_{i} \big)
= 
\frac{1}{  \sqrt{NT} } \sum_{i=1}^N \sum_{t=1}^T \big( y_{it-1} - \mu_i \big) \varepsilon_{it} - \frac{1}{  \sqrt{NT} } \sum_{i=1}^N \sum_{t=1}^T  \big( y_{it-1} - \mu_i \big) \bar{\varepsilon}_{i},  
\end{align*}

%%-------------------------------------------------------------------------%%
\newpage 

Notice that we can investigate the stochastic behaviour of the two terms above separately. The above result has two implications for the validity of the proposed bootstrap procedure. First, the bootstrap needs to mimic the asymptotic variance of $\hat{\theta}$ by $C = A^{-1} B A^{-1}$. More precisely, the variance has the usual sandwich form under conditional heteroscedasticity. In particular, it depends on the long run variance of the score process which is defined as below
\begin{align}
B := \underset{ N,T \to \infty }{ \mathsf{lim} } \mathsf{\text{Var}} \left(\frac{1}{\sqrt{NT}} \sum_{i=1}^N \sum_{t=1}^T \big( y_{it-1} - \mu_i \big) \varepsilon_{it} \right). 
\end{align}
In other words, the bootstrap validity depends on replicating the properties of the cross sectional average of the fourth order cumulants of $\varepsilon_{it}$. Second, the bootstrap needs to capture the asymptotic bias term $D$ created by the estimation of the fixed effects. More specifically, as the decomposition above shows, this noncentrality parameter results from the correlation between the averaged error terms $\bar{\epsilon}_i$ and the demeaned regressors $\left( y_{it-1} - \mu_i \right)$ and is non zero when $\rho = \mathsf{lim} \frac{N}{T} \neq 0$.  

\paragraph{Recursive-design wild bootstrap}

[\cite{gonccalves2015bootstrap}]  The recursive-design bootstrap can generate a panel of pseudo observations $\left\{ y_{it}^{*}, i = 1,..., n ; t = 1,..., T \right\}$ recursively from the panel AR(1) model with estimated parameters, 
\begin{align}
y_{it}^{*} = \hat{\alpha}_i + \hat{\theta} y_{it-1}^{*} + \epsilon_{it}^{*}, \ \ i \in \left\{ 1,..., n \right\}, t \in \left\{ 1,..., T \right\}, 
\end{align}
where $\hat{\alpha}_i = \frac{1}{T} \sum_{t=1}^T \left( y_{it} - \hat{\theta} y_{it-1} \right)$ and $\hat{\theta}$ is a fixed effects OLS consistent estimator. Moreover, the initial condition  is given by $y_{i0}^{*} = \frac{ \hat{\alpha_i } }{ 1 - \hat{\theta} }$, which is equivalent to setting $y_{it-1}^{*}$ to the stationary mean in the bootstrap world. In particular, the bootstrap residuals are obtained with the wild bootstrap $\epsilon_{it}^{*} = \hat{\epsilon}_{it} \eta_{it}$, where $\eta_{it} \overset{ \textit{i.i.d} }{  \sim  } (0,1)$ over $(i,t)$ such that $\hat{\epsilon}_{it} = y_{it} - \hat{\alpha}_i - \hat{\theta} y_{it-1}$ are the estimated residuals.  

Based on the above definitions,  \cite{gonccalves2015bootstrap} consider the bootstrap analogue of $\hat{\theta}$ for the recursive-design wild bootstrap OLS estimator, denoted by $\hat{\theta}^{*}_{rd}$ as below
\begin{align*}
\hat{\theta}^{*}_{rd} = \left( \frac{1}{NT} \sum_{i=1}^N \sum_{t=1}^T  \big( y^{*}_{it-1} - \bar{y}^{*}_{i(t-1)} \big)^2 \right)^{-1} \left( \frac{1}{NT} \sum_{i=1}^N \sum_{t=1}^T \big( y^{*}_{it-1} - \bar{y}^{*}_{i(t-1)} \big) \big( y^{*}_{it} - \bar{y}^{*}_{i (t)} \big) \right),
\end{align*}
where $\bar{y}^{*}_{i(t-1)}$ and $\bar{y}^{*}_{i (t)}$ are defined analogously to $\bar{y}_{i(t-1)}$ and $\bar{y}_{i (t)}$. Then, the following theorem, provides a result related to the asymptotic bootstrap validity of the recursive design estimator.

\begin{theorem}[\cite{gonccalves2015bootstrap}]
Under Assumption 1 above, it follows that 
\begin{align}
\underset{ x \in \mathbb{R} }{ \mathsf{sup} } \left| \ \mathbb{P}^{*} \left( \sqrt{NT} \left( \hat{\theta}_{rd}^{*} - \hat{\theta} \right) \leq x \right) -  \mathbb{P} \left( \sqrt{NT} \left( \hat{\theta} - \theta_0 \right) \leq x \right) \ \right| \overset{ d }{ \to } 0.
\end{align}
\end{theorem}
Notice that the proof for the above theorem needs to account for the incidental parameter bias generated by the estimation of the fixed effects due to the particular panel data structure.

%%-------------------------------------------------------------------------%%
\newpage 

\paragraph{Pairs Bootstrap}[\cite{gonccalves2015bootstrap}] 

An alternative bootstrap approach which is found to be robust to conditional heteroscedasticity of unknown form in the error term of a pure time series autoregressive model is the pairs bootstrap, where one resamples with replacement the vector that collects the dependent variable and its lagged values.

\subsubsection{Bootstrapping the bias-corrected estimator}

\paragraph{Proof of Lemma B1}[\cite{gonccalves2015bootstrap}] 

We can write the following 
\begin{align}
\frac{1}{NT} \sum_{i=1}^N \sum_{t=1}^T \varepsilon_{it}^{*2} - \sigma^2 = \left[ \frac{1}{NT} \sum_{i=1}^N \sum_{t=1}^T \hat{\varepsilon}^2_{it} \left( \eta_{it}^2 - 1 \right) \right] + \left[ \frac{1}{NT} \sum_{i=1}^N \sum_{t=1}^T \hat{\varepsilon}_{it}^2 - \sigma^2 \right] := F_1 + F_2
\end{align}
since $\varepsilon_{it}^{*2} = \hat{\varepsilon}^2_{it} \otimes \eta_{it}^2$. Moreover, the residual term can be expressed as below 
\begin{align}
\hat{\varepsilon}_{it} = \varepsilon_{it} + \big( \alpha_i - \hat{\alpha}_i \big) + \big( \theta_0 - \hat{\theta} \big) y_{it-1} 
\end{align}
We can also write $\big( \alpha_i - \hat{\alpha}_i \big) = \big( \hat{\varepsilon}_{it} - \varepsilon_{it} \big) + \big( \theta_0 - \hat{\theta} \big) y_{it-1}$. Therefore to show that $F_2 = o_p(1)$ which implies that $\frac{1}{NT} \sum_{i=1}^N \sum_{t=1}^T \hat{\varepsilon}_{it}^2 \overset{ p }{ \to } \sigma^2$, and thus we need to show that $\underset{ 1 \leq i \leq n }{ \mathsf{sup} } \left| \hat{\alpha}_i - \alpha_i \right| = o_p(1)$ under the assumptions above (see, \cite{gonccalves2015bootstrap} for further details). Moreover, it holds that
\begin{align}
\mathbb{E} \left[ \left( \sum_{t=1}^T \varepsilon_{it} \right)^2 \right] = \sum_{t=1}^T \mathbb{E} \left( \varepsilon_{it}^2 \right) = \mathcal{O}(T), 
\end{align}
which implies that $\sum_{t=1}^T \varepsilon_{it} = \mathcal{O}_p \left( \sqrt{T} \right)$, uniformly in $i$, and thus $\frac{1}{ \sqrt{T} } \sum_{t=1}^T \varepsilon_{it} = \mathcal{O}_p \left( 1 \right)$ uniformly in $i$. Furthermore, given that $\frac{1}{T} \sum_{t=1}^T y_{it-1} = \mathcal{O}_p(1)$, uniformly in $i$ and $\left( \hat{\theta} - \theta_0 \right) = o_p(1)$, we have that $\underset{ 1 \leq i \leq n }{ \mathsf{sup} } \left| \hat{\alpha}_i - \alpha_i \right| = o_p(1)$, that is, the fixed effect estimator is bounded in probability almost surely. 

\begin{proof}
\begin{align*}
\underset{ 1 \leq i \leq n }{ \mathsf{sup} } \left| \hat{\alpha}_i - \alpha_i \right|  
&= 
\underset{ 1 \leq i \leq n }{ \mathsf{sup} } \left| \frac{1}{\sqrt{T}} \left( \frac{1}{\sqrt{T}} \sum_{t=1}^T \varepsilon_{it} \right) - \big( \hat{\theta} - \theta_0 \big) \frac{1}{T} \sum_{t=1}^T y_{it-1}  \right|
\\
&\leq 
\frac{1}{\sqrt{T}} \ \underset{ 1 \leq i \leq n }{ \mathsf{sup} } \ \left| \left( \frac{1}{\sqrt{T}} \sum_{t=1}^T \varepsilon_{it} \right)   \right| + \left| \hat{\theta} - \theta_0 \right| \underset{ 1 \leq i \leq n }{ \mathsf{sup} } \left| \frac{1}{T} \sum_{t=1}^T y_{it-1}       \right| 
\\
&= 
\frac{1}{\sqrt{T}} \mathcal{O}_p(1) + o_p(1)  \mathcal{O}_p(1) = o_p(1). 
\end{align*}
Therefore, notice that parts of the proof rely on the uniform convergence over $i$ of $\hat{\alpha}_i$ towards $\alpha_i$, in addition to the convergence of $\hat{\theta}$ towards $\theta_0$. Recall: $\mathcal{O}_p(1) o_p(1) = o_p(1)$ and $\mathcal{O}_p(1) \mathcal{O}_p(1) = \mathcal{O}_p(1)$.
\end{proof}

%%-------------------------------------------------------------------------%%
\newpage

\subsection{CCE Estimation in Panel Data Models}

The common correlated effects estimation approach proposed by \cite{pesaran2006estimation}, provides a sufficiently general setting for panel data models with cross-sectional dependence and thus renders a variety of panel model specifications as special cases. In the panel data literature with $T$ small and $n$ being large, the primary parameters of interest are the means of the individual specific slope coefficients, $\boldsymbol{\beta}_i$. Let $\boldsymbol{M}_{ \widehat{\boldsymbol{F}}_x  } = \boldsymbol{I}_T - \bar{\boldsymbol{X}} \left(  \bar{\boldsymbol{X}}^{\prime} \bar{\boldsymbol{X}} \right) \bar{\boldsymbol{X}}^{\prime}$ denote a projection matrix. Then, the modified estimator is given by
\begin{align}
\widehat{\boldsymbol{\beta}}_{x} = \left( \sum_{i=1}^N \boldsymbol{X}_i^{\prime} \boldsymbol{M}_{\widehat{F} } \boldsymbol{X}_i \right)^{-1} \left( \sum_{i=1}^N \boldsymbol{X}_i^{\prime} \boldsymbol{M}_{\widehat{F} } \boldsymbol{y}_i  \right).      
\end{align}
where $N$ are the number of cross-sectional units and $T$ are the number of time-series observations. The quantity of interest here is the asymptotic distribution of the above estimator. It can be proved that the asymptotic variance of $\widehat{\boldsymbol{\beta}}_{x}$ is identical to that of $\widehat{\boldsymbol{\beta}}$, so no asymptotic efficiency is lost by omitting $\bar{\boldsymbol{y}}$, although the bias term of the particular expression still remains computational intractable.     
\begin{corollary}
As $(N,T) \to \infty$ such that $T / N \to \tau < \infty$ it holds that
\begin{align}
\sqrt{NT} \left( \widehat{\boldsymbol{\beta}}_{x}  - \boldsymbol{\beta} \right) \to_d \mathcal{N} \left( \boldsymbol{0}_{ k \times 1}, \boldsymbol{\Sigma}^{-1} \boldsymbol{\Psi}     \boldsymbol{\Sigma}^{-1} \right)    
\end{align}
\end{corollary}
Due to the reasons explained above in order to correct the bias term that appears in the asymptotic expression, we need to employ a bootstrap approximation which can ensure a coordinate-wise convergence in probability to the true value of the population parameter. The asymptotic validity of this expression is obtained with the use of uniform coordinatewise convergence as shown below 
\begin{align}
\underset{ x \in \mathbb{R}^{ k \times 1 } }{ \mathsf{sup} } \left| \mathbb{P}^{*} \left[ \sqrt{NT} \left( \widehat{\boldsymbol{\beta}}_x^{*} - \widehat{\boldsymbol{\beta}}_x \right) \leq x  \right] - \mathbb{P} \left[ \sqrt{NT} \left( \widehat{\boldsymbol{\beta}}_x - \boldsymbol{\beta}_x \right) \leq x  \right]  \right| \to_p 0.   
\end{align}
Therefore, the above expression establishes the consistency of the bootstrap for the distribution of the estimator $\widehat{\boldsymbol{\beta}}_x$ for general $m \leq k$, and hence validates the construction of bootstrap confidence intervals. Therefore, to establish the asymptotic validity of the bootstrap $t-$intervals, define with 
$\boldsymbol{\Theta} = \boldsymbol{\Sigma}^{-1}  \boldsymbol{\Psi} \boldsymbol{\Sigma}^{-1} $, and let $\widehat{\boldsymbol{\Theta}}^{*}$ be the bootstrap world equivalent of the corresponding variance estimator as below 
\begin{align}
\boldsymbol{\Psi}_i = \frac{1}{N-1} \sum_{i=1}^N \widehat{\boldsymbol{Q}}_i \left( \widehat{\boldsymbol{\beta}}_i - \widehat{\boldsymbol{\beta}}_{mg} \right) \left( \widehat{\boldsymbol{\beta}}_i - \widehat{\boldsymbol{\beta}}_{mg} \right)^{\prime} \widehat{\boldsymbol{Q}}_i 
\end{align}
Next we concentrate on the following sample variance estimator
\begin{align}
\widehat{\boldsymbol{\Omega}}_v = \frac{1}{N(N-1)} \sum_{i=1}^N  \left( \widehat{\boldsymbol{\beta}}_i - \widehat{\boldsymbol{\beta}}_{mg} \right) \left( \widehat{\boldsymbol{\beta}}_i - \widehat{\boldsymbol{\beta}}_{mg} \right)^{\prime}   
\end{align}
Therefore, under the assumption of homogeneous slopes $\boldsymbol{\beta}_i = \boldsymbol{\beta}$, we establish its asymptotic distribution as $(N, T) \to \infty$ such that $T / N \to \tau < \infty$ in the case of general $m \leq (k+1)$ (see,  \cite{harding2020common}).

%%-------------------------------------------------------------------------%%
\newpage 

Statistical inference techniques to panel data with or without cross-sectional dependence include slope homogeneity testing. Several studies have extended these methods to nonstationary panel data models as in \cite{kapetanios2011panels} and \cite{huang2021nonstationary}. However, no statistical testing methodology exists for slope homogeneity that covers these cases, which is currently a topic worth investigating further.  

Consider the formulation of the CCE estimator using vector notation such that 
\begin{align}
\boldsymbol{Y}_i &= \mathbf{1} \boldsymbol{\alpha}_i + \boldsymbol{X}_i \boldsymbol{\beta}_i + \boldsymbol{U}_i, 
\\
\boldsymbol{U}_i &= \boldsymbol{F} \boldsymbol{\lambda}_i + \boldsymbol{\varepsilon}_i, 
\\
\boldsymbol{X}_i &= \mathbf{1} \boldsymbol{\mu}_i^{\prime} + \boldsymbol{F} \boldsymbol{\Gamma}_i + \boldsymbol{V}_i, 
\end{align}
where we have that $\boldsymbol{Y}_i = ( y_{i1}, ..., y_{iT} )^{\prime}$, $\boldsymbol{X}_i = ( x_{i1},...,  x_{iT} )$, $\boldsymbol{U}_i = ( u_{i1},..., u_{in} )$ and $\boldsymbol{F} = ( \boldsymbol{f}_1,..., \boldsymbol{f}_T )^{\prime}$. Define the projection matrices $\boldsymbol{P}_{A} = \boldsymbol{A} ( \boldsymbol{A}^{\prime} \boldsymbol{A} )^{-1} \boldsymbol{A}^{\prime}$ and $\boldsymbol{M}_A = ( \boldsymbol{I} - \boldsymbol{P}_A  )$. Then, the transformed equation can be written as $\boldsymbol{M} \boldsymbol{Y}_i  = \boldsymbol{M} \boldsymbol{X}_i \boldsymbol{\beta}_i +  \boldsymbol{M} \boldsymbol{U}_i$. Therefore, the CCE pool estimator is defined as below
\begin{align}
\widehat{\boldsymbol{\beta}}_{CCE} = \left(  \sum_{i=1}^n \boldsymbol{X}_i^{\prime} \boldsymbol{M} \boldsymbol{X}_i \right)^{-1} \left( \sum_{i=1}^n \boldsymbol{X}_i^{\prime} \boldsymbol{M} \boldsymbol{Y}_i \right).   
\end{align}
Under the alternative hypothesis, we have that the CCE estimator deviates from the true parameters at least for a non-zero fraction of individual units. Therefore, the particular model parametrization can be employed to construct tests statistics for slope homogeneity in panel data models with multifactor error structure.  Define the weighted average CCE estimator as below
\begin{align}
\tilde{ \boldsymbol{\beta} }_{WCCE} = \left(  \sum_{i=1}^n \frac{1}{ \breve{\sigma}_i^2 } \boldsymbol{X}_i^{\prime} \boldsymbol{M} \boldsymbol{X}_i \right)^{-1} \left( \sum_{i=1}^n \frac{1}{ \breve{\sigma}_i^2 } \boldsymbol{X}_i^{\prime} \boldsymbol{M} \boldsymbol{Y}_i \right)      
\end{align}
Then, the proposed test statistic is constructed as below
\begin{align}
\tilde{\Delta}^{ cce  }_{ \mathsf{adj} } &= \left(  \frac{1}{ n^{1/2} } \tilde{S}_{cce}  - \sqrt{n} k \right),    
\\
\tilde{S}_{cce} &= \sum_{ i = 1 }^n \frac{1}{ \breve{\sigma}_i^2 }  \left( \hat{\boldsymbol{\beta}}_{i, cce} - \tilde{\boldsymbol{\beta}}_{wcce}  \right)^{\prime} \left[ \boldsymbol{X}_i^{\prime} \boldsymbol{M} \boldsymbol{X}_i \right] \left( \hat{\boldsymbol{\beta}}_{i, cce} - \tilde{\boldsymbol{\beta}}_{wcce}  \right).     
\end{align}
The modified PY test is interpreted as the weighted average distance between $\hat{\boldsymbol{\beta}}_{i, cce}$ and $\tilde{\boldsymbol{\beta}}_{i, wcce}$.
\begin{align}
y_{it} &= \alpha_{i0} +  \boldsymbol{x}_{it}^{\prime} \boldsymbol{\beta}_{i0}  + u_{it}, \ \  
u_{it} = \lambda_i^{\prime} \boldsymbol{f}_t  + \boldsymbol{\varepsilon}_{it} 
\\
\boldsymbol{x}_{it} &= \boldsymbol{\mu}_i + \boldsymbol{\Gamma}_i^{\prime} \boldsymbol{f}_t + \boldsymbol{v}_{it} 
\end{align}
Combining the two equations we obtain 
\begin{align}
\boldsymbol{w}_{it} := 
\begin{pmatrix}
y_{it}
\\
\boldsymbol{x}_{it}
\end{pmatrix}
=
\boldsymbol{C}^{\prime} \tilde{\boldsymbol{f}}_t + \boldsymbol{\eta}_{it}
\end{align}

%%-------------------------------------------------------------------------%%
\newpage 

where
\begin{align}
\boldsymbol{\eta}_{it} = 
\begin{pmatrix}
u_{it} + \boldsymbol{\beta}_i^{\prime} \boldsymbol{v}_{it}  
\\
\boldsymbol{v}_{it}  
\end{pmatrix}, 
\ \ \ 
\boldsymbol{C}_i = 
\begin{pmatrix}
\alpha_i + \boldsymbol{\beta}_i^{\prime} \boldsymbol{\mu}_i  &  \boldsymbol{\lambda}_i^{\prime} + \boldsymbol{\beta}_i^{\prime} \boldsymbol{\Gamma}_i^{\prime} 
\\
\boldsymbol{\mu}_i   &   \boldsymbol{\Gamma}_i^{\prime} 
\end{pmatrix}. 
\end{align}
Therefore, based on the above reparametrizations we have that
\begin{align}
\bar{\boldsymbol{M}} = \boldsymbol{I}_T - \boldsymbol{G} \bar{\boldsymbol{P}} \left( \bar{\boldsymbol{P}}^{\prime} \boldsymbol{G}^{\prime} \boldsymbol{G} \bar{\boldsymbol{P}} \right)^{+}  \bar{\boldsymbol{P}}^{\prime} \boldsymbol{G}^{\prime}    
\end{align}
where $\bar{\boldsymbol{P} } = \frac{1}{n} \sum_{ i = 1 }^n \boldsymbol{P}_i \bar{\boldsymbol{M}}$ is the residual maker of $\boldsymbol{G} \bar{\boldsymbol{P}}$.

\begin{remark}
Notice that according to \cite{pesaran2006estimation}, the above set-up is sufficiently general and renders a variety of panel models as special cases. Specifically, in the panel data literature with $T$ small and $n$ being large, the primary parameters of interest are the means of the individual specific slope coefficients, $\boldsymbol{\beta}_i$, for $i \in \left\{ 1,..., n \right\}$. Moreover, the common factor loadings, $\boldsymbol{\mu}_i$ and $\boldsymbol{\gamma}_i$, are generally treated as nuisance parameters. Notice that in this study we do not consider the case of unobserved common factors (i.e., latent group structure) as incorporating such features in our econometric specification will require different methodologies for estimation and inference.  
\end{remark}

Consider the following specification 
\begin{align}
\boldsymbol{X}_i = \boldsymbol{G} \boldsymbol{\Pi}_i + \boldsymbol{V}_i,     
\end{align}
where $\boldsymbol{G} = ( \boldsymbol{D}, \boldsymbol{F} )$ is the $T \times m + n$ matrix of integrated factors and $\boldsymbol{V}_i$ is a stationary error matrix. Moreover, we denote the OLS residuals of the multiple regression as $\hat{\boldsymbol{V}}_i = \boldsymbol{X}_i - \boldsymbol{G} \hat{\boldsymbol{\Pi}}_i$, where $\hat{\boldsymbol{\Pi}}_i = \left( \boldsymbol{G}^{\prime} \boldsymbol{G} \right)^{-1}\boldsymbol{G}^{\prime} \boldsymbol{X}_i$. Moreover, observe that $\hat{\boldsymbol{V} }_i = \boldsymbol{M}_{ \mathsf{g} } \boldsymbol{X}_i$. Then, we can write
\begin{align*}
\frac{ \hat{\boldsymbol{V} }_i^{\prime} \hat{ \boldsymbol{V} }_i }{ T } -  \frac{ \boldsymbol{V}_i^{\prime} \boldsymbol{V}_i }{ T } 
&= 
\frac{1}{T} \hat{\boldsymbol{V}}_i^{\prime} \left(  \hat{\boldsymbol{V}}_i - \boldsymbol{V}_i \right) +  \frac{1}{T} \left(  \hat{\boldsymbol{V}}_i - \boldsymbol{V}_i \right)^{\prime} \boldsymbol{V}_i
\\
&= 
- \boldsymbol{X}_i^{\prime} \boldsymbol{M}_{ \mathsf{g} } \boldsymbol{G} \frac{1}{T} \left( \hat{\boldsymbol{\Pi}}_i - \boldsymbol{\Pi}_i \right) - \frac{1}{T} \left( \hat{\boldsymbol{\Pi}}_i - \boldsymbol{\Pi}_i \right)^{\prime} \boldsymbol{G}^{\prime} \boldsymbol{V}_i, 
\\
&= 
\left( \hat{\boldsymbol{\Pi}}_i - \boldsymbol{\Pi}_i \right)^{\prime} \frac{1}{T} \boldsymbol{G}^{\prime} \boldsymbol{V}_i  
\end{align*}
since it holds that $\boldsymbol{M}_{ \mathsf{g} } \boldsymbol{G} = \boldsymbol{0}$.

\begin{remark}
Further studies related to testing for slope homogeneity include among others \cite{de2019cce} and \cite{de2021bootstrap} while the case of panel data models with interactive effects are considered by \cite{su2013testing}, \cite{chudik2015common} and \cite{westerlund2019estimation}. In order to establish asymptotic theory results for relevant estimators and test statistics in the homogeneous slope setting, one can impose a common slope condition and then derive an analytical expression of the adjusted CCEP estimator (see, \cite{de2021bootstrap}).   
\end{remark}

%%-------------------------------------------------------------------------%%
\newpage

\subsection{IV Estimation of Dynamic Linear Panel Data Models}

Consider the following autoregressive distributed lag, ARDL(1,0), panel data model with homogeneous slopes and a multifactor error structure (see, \cite{norkute2021instrumental}) such that 
\begin{align}
y_{it} = \rho y_{i, t- 1} + \boldsymbol{\beta}^{\prime} \boldsymbol{x}_{it} + u_{it}, \ \ \ i \in \left\{ 1,..., N \right\} \ t \in \left\{ 1,..., T \right\},   
\end{align}
where the multifactor error structure is captured with the following equations
\begin{align}
u_{it} &= \boldsymbol{\gamma}^{0 \prime}_{y_i} \boldsymbol{f}_{y,t}^0 + \varepsilon_{it},    
\\
\boldsymbol{x}_{it} &= \Gamma_{x_i}^{0 \prime} \boldsymbol{f}_{y,t}^0 + \boldsymbol{v}_{it}
\end{align}
where $| \rho | < 1$ and $\boldsymbol{\beta} = ( \beta_1, \beta_2,..., \beta_k )^{ \prime }$ such that at least one of $\left\{ \beta_{\ell} \right\}_{ \ell = 1 }^k$ is non-zero and $\boldsymbol{x}_{it} = ( x_{1it},..., x_{kit} )^{\prime}$ is a $( k \times 1 )$ vector of regressors and $\boldsymbol{f}_{y,t}^0 = ( f_{x,1t}^0, f_{x,2t}^0,..., f_{x, m_x t}^0 )$ denoters an $( m_x \times 1)$ vector of true factors, and $\boldsymbol{v}_{it} = ( v_{1it}, v_{2it},..., v_{kit} )^{\prime}$ 

\begin{remark}
Notice that incorporating features that capture unobserved individual effects, ensures that the error terms remain uncorrelated. In particular, the presence of serial correlation can potentially lead to incorrect estimates for the standard errors of model parameters. Thus, the approach proposed by  \cite{norkute2021instrumental} permits correlations between and within  $\Gamma_{x_i}^{0}$ and $\boldsymbol{\gamma}^{0}_{y_i}$. This specification allows to control for endogeneity of $\boldsymbol{x}_{it}$ that steams from the common components, but assumes that $\boldsymbol{x}_{it}$ is strongly exogenous with respect to $\varepsilon_{it}$. Lastly, note that a dynamic panel data model specification is not the same as a time varying model specification.
\end{remark}

\begin{example}
Consider the following linear dynamic panel data model:
\begin{align}
y_{i,t} 
= 
\alpha_1 y_{i,t-1} + \beta_1 x_{i,t} + \beta_2 x_{i,t-1} + \mu_{i,t}, \ \ \
\mu_{i,t}
= 
\eta_i + \varepsilon_{i,t}
\end{align}
A dynamic panel data process is one that includes one or more lags of the dependent variable in the functional form of the model, that is, $\alpha_1 y_{i,t-1}$. In particular, this feature reflects the fact that $y_{i,t}$ is autoregressive. Moreover, the effect of a temporary change in the covariate (observed or unobserved) on $y_{i,t}$ does not completely dissipate for the next observation. 
\end{example}

\begin{example}
Consider the following static panel data model 
\begin{align}
y_{i,t} = \beta_1 x_{i,t} + \eta_i + \epsilon_{i,t}    
\end{align}
when the true DGP is given by
\begin{align}
y_{i,t} = \alpha_1 y_{i,t-1} + \beta_1 x_{i,t} + \eta_i + \varepsilon_{i,t}
\end{align}

%%-------------------------------------------------------------------------%%
\newpage

\begin{itemize}

\item Common estimators such as pooled OLS, OLS, fixed effects, generalized least squares, random effects; assume that $\mathbb{E} [ \epsilon_{i,t_1} | x_{i, t_2} ] = 0 \ \forall \ t_1$ and $t_2$.

\item Specifically when $T$ is small, it is assumed that this holds for any past, current or future values of $x_{i,t}$ - strict exogeneity assumption. This also implies that the errors will be serially correlated such that $\mathbb{E} ( \epsilon_{i,t}  \epsilon_{i,t-1} ) \neq 0$.

\item In a dynamic panel data process, the long-run effect (LRE) of a covariate differs from its short-run effect (SRE). For example, in the dynamic model specification the SRE of $x_{i,t}$ is $\beta_1$ and the LRE is $\beta_1 / ( 1 - \alpha_1 )$. The two most common estimators used to account for unobserved individual effects are: OLS-FE and GLS-RE. Thus, the GLS-RE method asssumes that $\mathbb{E} [ \eta_i | x_{i,t} ] = \mathbb{E} [ \eta_i | z_{i,t,m} ] = 0$. However, with a dynamic model specification this assumption cannot be met. GMM estimation methods can be employed for dynamic models with fixed-effects. 

\end{itemize}
Consider the following model specification:
\begin{align}
y_{i,t} &= \alpha_1 y_{i,t-1} + \beta_1 x_{i,t} + \beta_2 x_{i,t-1} + \eta_i + \epsilon_{i,t}    
\\
y_{i,t-1} &= \alpha_1 y_{i,t-2} + \beta_1 x_{i,t-1} + \beta_2 x_{i,t-2} + \eta_i + \epsilon_{i,t-1} 
\end{align}
Subtracting by sides the above two specifications we obtain that 
\begin{align}
\Delta y_{i,t} &= \alpha_1 \Delta y_{i,t-1} + \beta_1 \Delta x_{i,t} + \beta_2 \Delta x_{i,t-1} + \Delta \epsilon_{i,t}.
\end{align}
However, by definition of the DGP $\mathbb{E} [ \Delta \epsilon_{i,t} | \Delta y_{i,t-1} ] \neq 0$ which is a requirement for an unbiased estimation. A solution to this problem is to use the first difference or level of the second lag of the dependent variable $\Delta y_{i,t-2}$ as an instrument for $\Delta y_{i,t-1}$. A better solution is the GMM estimation approach where the instruments define moment conditions. Then the GMM proceeds by selecting the values for the parameters in the model that minimizes a weighted sum of the squared moment conditions. On the other hand, the GMM can underperform when the variance between and within cases is large or when the autoregressive coefficient is near to unity. The solution to this is the System GMM which implies that there are additional moment conditions to be estimated. However, one drawback is that the optimal weighting matrix can be difficult to estimate with limited information. In particular, this occurs with moments based on weak instruments or when the number of moment conditions is large relative to $N$. The many instruments and weak instruments problem can result in bias in the direction of the OLS-FE estimator. In the case of dynamic panel data models with fixed effects, the system GMM is found to be insufficient. Consifder the following simulation design.
\begin{align}
y_{i,t} = \alpha_1 y_{i,t-1} + \beta_1 x_{i,t} + \mu_{i,t},
\ \ \ 
\mu_{i,t} = \eta_i + u_{i,t},
\ \ \
x_{i,t} = 0.75 \eta_i + v_{i,t}.
\end{align}
where $u_{i,t} \sim N(0,1)$ and $v_{i,t} \sim N(0,16)$. In terms of the estimation methodology, using the MLE for a dynamic panel data model with a fixed (small) $T$ leads to an incidental parameters problem. In particular, with fixed $T$, consistent MLE requires $N$ to increase faster than the number of parameters estimated. 
\end{example}

%%-------------------------------------------------------------------------%%
\newpage

Although, under the presence of fixed effects, the number of fixed-effects $(\eta_i)$ approaches infinity at the same rate as $N \to \infty$. In other words, for each case we add the MLE estimation adds a parameter to be estimated. Therefore, we cannot rely on asymptotics as $N \to \infty$ since the application of maximum likelihood leads to inconsistent estimates and thus an alternative estimation or transformation approach is required for robust statistical estimation and inference purposes. The key is to consider the orthogonal reparametrization\footnote{The orthogonal reparametrization proposed by \cite{lancaster2002orthogonal} it changes the meaning of the parameters representing the individual effects but not the meaning of the other parameters. This approach is particularly useful when due to the functional form of the model orthogonality cannot be achieved but information orthogonality can.} (OPM) approach such that we are not actually interested in estimates of the $( \eta_i )$ (as these are incidental parameters). In particular, we are interested in estimates of the common parameters such as $\left\{ \beta_1, \beta_2, \alpha_1, \sigma^2 \right\}$. Then the OPM approach implies a reparametrization of the incidental parameters so that the incidental and common parameters are \textit{information orthogonal}. The \cite{lancaster2002orthogonal} reparametrization approach allows us to write the likelihood in which the incidental parameters are informationally orthogonal from the other parameters.     

Some important terms which we will need to obtain relevant results for their asymptotic behaviour can be obtained as below
\begin{align}
\sum_{t=1}^T 
\begin{pmatrix}
\boldsymbol{y}_{t-1} . \boldsymbol{y}_{t-1}^{\prime} \ & \ \boldsymbol{y}_{t-1} . \bar{\boldsymbol{\varepsilon}}_{t-1}^{\prime}  
\\
\bar{\boldsymbol{\varepsilon}}_{t-1}^{\prime} . \boldsymbol{y}_{t-1}^{\prime}  \ & \  \bar{\boldsymbol{\varepsilon}}_{t-1} \bar{\boldsymbol{\varepsilon}}_{t-1}^{\prime}
\end{pmatrix}
\end{align}
Next, we consider expanding the following sample moments 
\begin{align}
\frac{1}{T} \sum_{t=1}^T \boldsymbol{y}_{t-1}. \boldsymbol{y}_{t-1}^{\prime} = \frac{1}{T} \sum_{t=1}^{T-1} \bar{\boldsymbol{\varepsilon}}_{t} \bar{\boldsymbol{\varepsilon}}_{t}^{\prime} + \frac{1}{T} \sum_{t=1}^{T-1} \left\{  \bar{\boldsymbol{\varepsilon}}_t^{\prime} . \boldsymbol{d}^{\prime} \left( - \boldsymbol{C}^{\prime} \right)^t + \left( - \boldsymbol{C} \right)^t \boldsymbol{d} \bar{\boldsymbol{\varepsilon}}_t^{\prime}  \right\}         
\end{align}
Next, we can consider the kernel density estimates of the distribution of the Mahalanobis distance given by the following expression 
\begin{align}
Q_T = \left( \hat{\boldsymbol{\beta}}_T - \boldsymbol{\beta} \right)^{\prime} \left( \sum_{t=1}^T \boldsymbol{W}_t^{\prime} \boldsymbol{\Sigma}_{\varepsilon}^{-1} \boldsymbol{W}_t \right) \left( \hat{\boldsymbol{\beta}}_T - \boldsymbol{\beta} \right)   
\end{align}
compared to the theoretical $\chi^2$ densities. 

Moreover notice that the above process is asymptotically stationary, and the specification $\boldsymbol{A} = \boldsymbol{M}$ means that it will also be asymptotically unidentified. Therefore, such a lack of identification is well known to manifest itself in $\sum_{t=1}^T \boldsymbol{W}_t^{\prime} \boldsymbol{\Sigma}_{\varepsilon}^{-1} \boldsymbol{W}_t$ having less than full rank and $Q_T$ having fewer degree of freedom that might be anticipated on the basis of conventional asymptotic theory.

%%-------------------------------------------------------------------------%%
\newpage

\subsection{Quantile Censored Panel Data Regression}

Given a quantile $\tau \in (0,1)$, consider the following QR model defined by \cite{galvao2013estimation} such that 
\begin{align}
y_{it}^{*} = \alpha_{i0} (\tau) + \boldsymbol{x}_{it}^{\top} \boldsymbol{\beta}_0(\tau) + u_{it}, \ \ \ i = 1,...,N \ \ t = 1,...,T,     
\end{align}
where $\boldsymbol{x}_{it}$ is a $( p \times 1 )$ vector of regressors, $\boldsymbol{\beta}_0 (\tau)$ is a $(p \times 1)$ vector of parameters and $\alpha_{i0}(\tau)$ is a scalar individual effect for each $i$, and $u_{it}$ is the innovation term whose $\tau-$th conditional quantile is zero. Notice that the quantile-specific individual effect, $\alpha_{i0}(\tau)$, is intended to capture individual specific sources of variability, or unobserved heterogeneity that was not adequately controlled by other covariates. In general, each $\alpha_{i0}(\tau)$ and $\boldsymbol{\beta}_0(\tau)$ can depend on $\tau$, but we assume $\tau$ to be fixed throughout the framework here. Moreover, the model is semiparametric in the sence that the functional form of the conditional distribution of $y_{it}^{*}$ given $\left( \boldsymbol{x}_{it}, \alpha_{i0} \right)$ is left unspecified and no parametric assumption is made on the relation between $\boldsymbol{x}_{it}$ and $\alpha_{i0}$. Thus, the QR model can be written as below
\begin{align}
\mathcal{Q}_{ y_{it}^{*} } \left( \tau | \boldsymbol{x}_{it}, \alpha_{i0} \right) = \alpha_{i0} + \boldsymbol{x}_{it}^{\top} \boldsymbol{\beta}_0.     
\end{align}

\begin{remark}
Equivariance to monotone transformation is an important property of QR models. Specifically, for a given monotone transformation $\mathcal{P}_c(y)$ of variable $y^{*}$, it holds that 
\begin{align}
\mathcal{Q}_{\mathcal{P}_c(y^{*})} \left( \tau | \boldsymbol{x}_{it}, \alpha_{i0} \right) \equiv \mathcal{P}_x \big( \mathcal{Q}_{\mathcal{P}_c(y^{*})} \left( \tau | \boldsymbol{x}_{it}, \alpha_{i0} \right) \big)   
\end{align}
Thus the parameter of intestest $\boldsymbol{\beta}_0$, can be interpreted as representing the effect of $\boldsymbol{x}_{it}$ on the $\tau$th conditional quantile function of the dependent variable while controlling for heterogeneity, which represented by $\alpha_i$. Thus, this model can be considered as a conditional model. 
In order to control for fixed effects we could define the estimator $\left( \hat{\boldsymbol{\alpha}}, \hat{\boldsymbol{\beta}} \right)$ solving the following minimization problem: 
\begin{align}
\mathcal{Q}_{1,N} \left(  \boldsymbol{\alpha}, \boldsymbol{\beta} \right) = \frac{1}{NT} \sum_{i=1}^N \sum_{t=1}^T \rho_{\tau} \big( y_{it} - \mathsf{max} \left( C_{it}, \alpha_i + \boldsymbol{x}_{it}^{\top} \boldsymbol{\beta} \right) \big)   
\end{align}
where $\boldsymbol{\alpha} := \left( \alpha_1,..., \alpha_N \right)$ and $\rho_{\tau} (u) := u \left( \tau - \boldsymbol{1} \left( u < 0 \right) \right)$. We assume that the number of individuals is denoted by $N$ and the number of time periods is denoted by $T = T_N$ that depends on N. 
\end{remark}
The main problem of the above estimator is caused by its low frequency of convergence. Furthermore, additional regressors, large proportions of censored observations, and large samples only worsen the problem. Due to censored effects we consider the equivalent minimizer (see, \cite{galvao2013estimation})
\begin{align}
\mathcal{Q}_{2,N} \left( \boldsymbol{\alpha}, \boldsymbol{\beta} \right) = \frac{1}{NT} \sum_{i=1}^N \sum_{t=1}^T \rho_{\tau} \left(  y_{it} - \alpha_i - \boldsymbol{x}_{it}^{\top} \boldsymbol{\beta} \right) \times \boldsymbol{1} \left\{  \alpha_{i0} + \boldsymbol{x}_{it}^{\top} \boldsymbol{\beta}_0 > C_{it}  \right\}.   
\end{align}
Therefore, we denote with $\delta_{it} = \boldsymbol{1} \left( y_{it}^{*} > C_{it}    \right)$ to indicate uncensored observations.

%%-------------------------------------------------------------------------%%
\newpage

We define with 
\begin{align}
u_{it} := y_{it}^{*} - \alpha_{i0} - \boldsymbol{x}_{it}^{\top} \boldsymbol{\beta}_0,  
\end{align}
whose $\tau-$th conditional quantile given $( \boldsymbol{x}_{it}, \alpha_i, C_{it} )$ equals zero. Furthermore, it holds that 
\begin{align}
\pi_0 \left( \alpha_i, \boldsymbol{x}_{it}, C_{it} \right) 
&:= 
\mathbb{P} \big( \delta_{it} = 1 | \boldsymbol{x}_{it}, \alpha_i, C_{it} \big) =  \mathbb{P} \big( u_{it} > - \alpha_{i0} -  \boldsymbol{x}_{it}^{\top} \boldsymbol{\beta}_0 + C_{it} | \boldsymbol{x}_{it}, \alpha_i, C_{it} \big)
\\
\pi_1 \left( \alpha_i, \boldsymbol{x}_{it}, C_{it} \right) &:= 
\mathbb{P} \big( \delta_{it} = 0 | \boldsymbol{x}_{it}, \alpha_i, C_{it} \big) =  \mathbb{P} \big( u_{it} > 0 | \boldsymbol{x}_{it}, \alpha_i, C_{it} \big) = 1 - \tau.
\end{align}
In other words, the restriction set selects those observations $(i,t)$ where the conditional quantile line is above the censoring point $C_{it}$. Then, the objective function is equivalent to the following 
\begin{align}
\mathcal{Q}_{3,N} \left( \boldsymbol{\alpha}, \boldsymbol{\beta} \right) = \frac{1}{NT} \sum_{i=1}^N \sum_{t=1}^T \rho_{\tau} \left(  y_{it} - \alpha_i - \boldsymbol{x}_{it}^{\top} \boldsymbol{\beta} \right) \times \boldsymbol{1} \big\{  \pi_0 \left( \alpha_{i0}, \boldsymbol{x}_{it}, C_{it} \right) > 1 - \tau \big\}.     
\end{align}

\subsubsection{Large Sample Properties}

We investigate the asymptotic properties of the proposed two-step estimator. A particular issue we impose is that the individual fixed effects parameter $\boldsymbol{\alpha}$ whose dimension tends to infinity. However, it has been noted in the literature that leaving the individual heterogeneity unrestricted in a nonlinear or dynamic panel model generally results in inconsistent estimators of the common parameters due to the incidental parameters problem. In other words, noise in the estimation of the individual specific effects leads to inconsistent estimates of the common parameters due to the nonlinearity of the problem. Therefore, to overcome this problem it has become standard in the panel QR literature to employ a large $N$ and $T$ asymptotics (as joint limits).  Denote with $\norm{ \pi - \pi_0 }_{\infty} = \underset{ w }{ \mathsf{sup} } \left|  \pi(\boldsymbol{w}) -  \pi_0(\boldsymbol{w}) \right|$ for a given $\pi(.)$ and a generic vector $\boldsymbol{w}$ (see, \cite{galvao2013estimation}).

\begin{assumption}[\cite{galvao2013estimation}]

\item[$\boldsymbol{A_1}:$] Let $\left\{ \left( \boldsymbol{x}_{it}, y_{it}^{*} \right) \right\}$ are independent across subjects and independently and identically distributed (\textit{i.i.d}) for each $i$ and all $t \geq 1$. 

\item[$\boldsymbol{A_2}:$] $\underset{ i \geq 1 }{ \mathsf{sup} } \ \mathbb{E} \big[ \norm{\boldsymbol{x}_{i1} }^{2s} \big] < \infty$ and some real $s \geq 1$. 

\item[$\boldsymbol{A_3}:$] Let $u_{it} = y_{it}^{*} 0 \alpha_{i0} - \boldsymbol{x}_{it}^{\top} \boldsymbol{\beta}_0$ and $\pi_{i0} ( \boldsymbol{x}_{it} ) := \pi_0 \left( \alpha_{i}, \boldsymbol{x}_{it} \right)$. Then, $F \left( u | \boldsymbol{x} \right)$ is defined as the conditional distribution function of $u_{it}$ given $\boldsymbol{x}_{it} := \boldsymbol{x}$. Assume that $F_i ( u | \boldsymbol{u} )$ has density given by $f_i ( u | \boldsymbol{x} )$. Let $f_i(u)$ denote the marginal density of $u_{it}$. 

\item[$\boldsymbol{A_4}:$] For each $\delta > 0$, it holds that
\begin{align}
\epsilon_{\delta} := \underset{ i \geq 1 }{ \mathsf{inf} } \  \underset{ | \alpha | + \norm{\boldsymbol{\beta} }_1 = \delta }{ \mathsf{inf} }  \times \mathbb{E} \left[ \int_0^{  \alpha + \boldsymbol{x}_{i1}^{\top} \boldsymbol{\beta} }  \big( f_i( s | \boldsymbol{x}_{i1} ) - \tau \big) ds \ \boldsymbol{1} \big\{ \pi_{i0} (\boldsymbol{x}_{i1} ) > 1 - \tau \big\} \right]
\end{align}

\end{assumption}

%%-------------------------------------------------------------------------%%
\newpage

\subsection{Semiparametric Approach}

\subsubsection{Bootstrap Algorithms for Cluster-Robust Inference}

In this section, we focus on the asymptotic validity of statistical procedures for cluster-robust bootstrap inference and cluster-robust confidence intervals in quantile regression models (see, \cite{galvao2011quantile}, \cite{hagemann2017cluster}, \cite{galvao2020unbiased}, \cite{galvao2023bootstrap} and \cite{galvao2023hac} among others). We consider the recentered population objective function given by the following expression 
\begin{align}
\beta \mapsto M_n ( \beta, \tau ) := \mathbb{E} \big[ \mathbb{M}_n ( \beta, \tau ) -  \mathbb{M}_n \big( \beta (\tau), \tau \big)  \big]    
\end{align}
Notice that the map $\beta \mapsto M_n ( \beta, \tau )$ is differentiable with derivative given by $M^{\prime}_n ( \beta, \tau ) := \partial M_n ( \beta, \tau ) \big/ \partial \beta^{\top}$. Specifically, the first-order condition of the QR objective function can be written as
\begin{align}
\sqrt{n} M^{\prime}_n \big( \beta (\tau), \tau \big) := - \frac{1}{ \sqrt{n} } \sum_{i=1}^n \sum_{k=1}^{c_1} \mathbb{E} \big[ \psi_{\tau} \left( Y_{ik} - X_{ik}^{\top} \beta ( \tau ) \right) X_{ik} \big] = 0,     
\end{align}
where $\psi_{\tau} ( \mathsf{z} ) = \big( \tau - \boldsymbol{1} \left\{ \mathsf{z} < 0 \right\} \big)$. Then, the sample analogue of this condition is
\begin{align}
\frac{1}{ \sqrt{n} } \sum_{i=1}^n \sum_{i=1}^{ c_i } \psi_{\tau} \left( Y_{ik} - X_{ik}^{\top} \beta ( \tau ) \right) X_{ik}  = 0.   
\end{align}
can be thought of as nearly solved by the QR estimate $\beta = \hat{\beta}_n ( \tau )$. Notice that to ensure that the bootstrap counterparts of the above quantities, that correspond to the QR estimate, accurately reflect the within-cluster dependence, the resampling scheme perturbs the gradient condition at the cluster level. In particular, the bootstrap resampling is approximated using the bootstrap gradient process $\mathbb{W}_n ( \tau ) := \mathbb{W}_n \big(  \hat{\beta}_n( \tau ) , \tau \big)$ evaluated at the original QR estimate to construct the new objective function 
\begin{align*}
\beta \mapsto \mathbb{M}^{*}_n ( \beta, \tau ) \ \vline_{ \beta = \hat{\beta}_n( \tau )  }
&\equiv \mathbb{M}_n ( \beta, \tau )  + \mathbb{W}_n ( \tau)^{\top} \beta / \sqrt{n} 
\\
&=
\left\{ \frac{1}{n} \sum_{i=1}^n \sum_{k=1}^{c_1} \rho_{\tau} \big( Y_{ik} - X_{ik}^{\top} \beta \big) + \frac{1}{ n } \sum_{i=1}^n W_i \sum_{k=1}^{c_i} \psi_{\tau} \left( Y_{ik} - X_{ik}^{\top} \beta ( \tau ) \right) X_{ik}  \beta \right\} \vline_{ \beta = \hat{\beta}_n( \tau ) }
\end{align*}
and define the process $\tau \hat{\beta}^{*}_n ( \tau )$ as any solution to $\mathsf{min}_{ \beta \in B } \mathbb{M}^{*}_n ( \beta, \tau )$. Then, $\hat{\beta}^{*}_n ( \tau )$ can be interpreted as the $\beta$ that nearly solves the corresponding first-order solution based on the proposed bootstrap resampling procedure. Then, the distributional convergence occurs both in the standard sense and with probability approaching one, conditional on the sample data $D_n := \left\{ ( Y_{ik}, X_{ik}^{\top} )^{\top}: 1 \leq k \leq c_i, 1 \leq i \leq n \right\}$.

%%-------------------------------------------------------------------------%%
\newpage

\subsubsection{Weighted Bootstrap for Semiparametric M-estimators}

\begin{theorem}[see, \cite{ma2005robust}]
Suppose that the $M-$estimator $\hat{\theta}_n$ and the weighted $M-$estimator $\hat{\theta}^{*}_n$ satisfy the following approximation:
\begin{align}
\sqrt{n} \left( \hat{\theta} - \theta_0 \right) &= \tilde{I}_0^{-1} \sqrt{n} \mathbb{P}_n \tilde{m} + o_p(1) 
\\
\sqrt{n} \left( \hat{\theta}^* - \theta_0 \right) &= \tilde{I}_0^{-1} \sqrt{n} \mathbb{P}^*_n \tilde{m} + o_p(1) 
\end{align}
Then, we have that $\sqrt{n} \left( \hat{\theta} - \theta_0 \right) = \tilde{I}_0^{-1} \big( \mathbb{P}^*_n - \mathbb{P}_n  \big)$ and using stochastic equicontinuity properties relevant asymptotic theory results can be established. 
\end{theorem}

\subsection{Moving Block Bootstrap for Analyzing Longitudinal Data}

A block bootstrap algorithm in a longitudinal model is proposed by \cite{ju2015moving}. In particular, assume that the data generating process is based on a longitudinal data model specification. 
\begin{itemize}
    
    \item[Step 1.] Let $\hat{e}_{ij}$, for $i = 1,..., n_0, j = 1,..., m$, be the residuals form the model fit such that 
    \begin{align}
        \hat{e}_{ij} = y_{ij} - x_{ij} \hat{\beta},
    \end{align}
    where $\hat{\beta}$ is the ordinary least square estimate. 
     
    \item[Step 2.] Assuming that $m = bk$ with $b$ and $k$ integers: Let $B_1^{*},..,B_k^{*}$ denotes $k$ uniform draws with replacement from the integers $\left\{ 0,..., m-b \right\}$. These represent the starting point for each block of length $b$. A block bootstrap resample of residuals, $\left( \hat{e}_{i1}^{*},..., \hat{e}_{i1}^{*} \right)$, is defined by:
    \begin{align}
        \hat{e}^{*}_{i,(j-1)b + s} = \hat{e}_{i, B_{j}^{*} + s}, \ \ 1 \leq j \leq k, 1 \leq s \leq b, \ \text{for each} \ i. 
    \end{align}
    
    \item[Step 3.] The bootstrapped response, $y_{ij}^{*}$, are then generated from the estimated model with residuals $\hat{e}_{ij}$ and the original covariates:
    \begin{align}
        y_{ij}^{*} = x_{ij} \hat{\beta} + \hat{e}^{*}_{ij}.
    \end{align}
    
    \item[Step 4.] From the resampled responses, $y_{ij}^{*}$, and original covariates, we fit the model and obtain new parameter estimates. 
    
    \item[Step 5.] Repeating steps (2) through (4) a large number, $R$, of times one obtains $R$ bootstrap replicates from which features of the distribution of the parameter estimates can be estimated. In particular, the bootstrap variance estimates are simply variance of the $B$ computed values for each parameter.

\end{itemize}

%%-------------------------------------------------------------------------%%
\newpage

\begin{proposition}[Within block bootstrap, see \cite{ju2015moving}] For each $i$ subject, we construct overlapping blocks $( m - b + 1 )$ blocks and block size $b$, such that $B_1,..., B_{m-b+1}$. 

\begin{itemize}
   
    \item Let us define $m / b = k$ which is assumed to be an integer for simplicity, in general $k = \floor{m / b}$.
    
    \item We can add the $k$ blocks with replacement amomg $B_1,..., B_{m-b+1}$. We get the $B_1^{*},..., B_{k}^{*}$ with $kb = m$, and create $\left\{ \hat{e}^{*}_{i1},..., \hat{e}^{*}_{im} \right\}$ from $\left\{ \hat{e}_{i1},..., \hat{e}_{im} \right\}$, where $\hat{e}_{ij} = y_{ij} - \hat{\beta}_0 - \hat{\beta}_1 x_{ij}$. 
    
\end{itemize}

\end{proposition}
Thus, we can add up to $n_0$ individuals and plug this into the model and the results is a pseudo sample series $y_{11}^{*},..., y_{n m}^{*}$. Then, from the model $y_{ij}^{*} = \hat{\beta}_0 + \hat{\beta}_1 x_{ij} + \hat{e}^{*}_{ij}$, we fit the regression model and produce the new parameters $\hat{\beta}_0^{*}$ and $\hat{\beta}_1^{*}$. As a result, the asymptotic validity and justification of Moving Block Bootstrap in Longitudinal Data can be established by carefully considering analytical expressions using the robust regression M-estimator which solves the following optimization problem
\begin{align}
\sum_{i=1}^n \sum_{j=1}^{m} x_{ij}^{\prime} \psi \left( y_{ij} - x_{ij} \beta \right) = 0,   
\end{align}
in relation to the mixing properties of innovation and the bootstrapping scheme. Non-asymptotic theory and related probability bound results such as the Hoeffding's inequality can be found to be useful for these derivations (see, \cite{praestgaard1993exchangeably} and \cite{bentkus2004hoeffding}). 
\begin{theorem}[Hoeffding's inequality] Let $( c_1,.., c_N )$ be elements of a vector space $\boldsymbol{V}$, and let $( U_1,..., U_n )$ and $( V_1,..., V_n )$ denote, respectively, a sample without and with replacement of size $n \leq N$ from $( c_1,.., c_N )$. Let $\varphi: \boldsymbol{V} \to \mathbb{R}$ be a convex function. Then, it holds that 
\begin{align}
\mathbb{E} \left[ \varphi \left( \sum_{j=1}^n U_j \right) \right] \leq  \sum_{j=1}^n \left( \mathbb{E} \left[  \varphi \left(  U_j \right) \right] \right).    
\end{align}

\end{theorem}

\subsection{Specification Testing in Panel Data Models}

In this section we consider relevant aspects to specification testing in panel data regression models. Relevant studies include \cite{metcalf1996specification}, \cite{su2013nonparametric} and \cite{su2015specification} among others. Thus, in order to correctly define the estimation and inference procedure we first need related regularity conditions regarding the dependence structure across the panel data. Specifically, based on existing results in the literature we can assume that we may have independence across cross-sectional units and strong mixing over time. In other words, we may assume that the innovation sequences in the given setting have bounded higher order moments using results such that Berneisten's inequalities for strong mixing processes. This allow us to study the asymptotic properties of related test statistics and estimators without worrying about the existence of cross-sectional dependence since we decompose that effect into conditional independence within a small neighborhood of values, similar to the meaning of near-epoch dependence in related econometric models.

%%-------------------------------------------------------------------------%%
\newpage

We consider the example below which represents a panel data model where cross-sectional dependence is captured by the presence of common factor loadings. In other words, using individual fixed effects in the panel, facilitates the presence of heterogeneity of shocks across the cross-sectional units. As a matter of fact, this is a more realistic assumption since shocks such as technology shocks, oil price shocks and financial crises are more likely to have unequal effect across the cross-section. A small economy for example, tends to be more vulnerable to such shocks than a large economy. 
\begin{example}
\begin{align}
Y_{it} = m( X_{it} ) + F_t^{0 \prime} \lambda^0 + \varepsilon_{it},    
\end{align}
\end{example}

\begin{remark}
Relevant research questions of interest in relation to the econometric specification and panel data structure, is whether the data support the use of network dependence, and especially what would be the estimation and inference benefits in comparison to a panel data modelling approach with interactive fixed effects as in the framework proposed by \cite{su2015specification} (see also Section \ref{Section5} and \ref{Section7}).     
\end{remark}

\subsubsection{The hypotheses and test statistic}

Therefore our main objective is to construct a test for linearity for the proposed specification form. In other words, we are interested in testing the null hypothesis
\begin{align}
H_0: \mathbb{P} \left( m(X_{it}) = X_{it}^{\prime} \beta^0 \right) = 1 \ \ \ \text{for some} \ \beta^0 \in \mathbb{R}^p.    
\end{align}
Under the alternative hypothesis we have that
\begin{align}
H_0: \mathbb{P} \left( m(X_{it}) = X_{it}^{\prime} \beta^0 \right) < 1 \ \ \ \text{for some} \ \beta^0 \in \mathbb{R}^p.       
\end{align}
Furthermore, to facilitate the local power analysis, we define a sequence of Pitman local alternatives
\begin{align}
H_1: \gamma_{NT} : m(X_{it}) = X_{it}^{\prime} \beta^0 + \gamma_{NT} \Delta (X_{it}) \ \ \ \text{a.s for some} \ \beta^0 \in \mathbb{R}^p.       
\end{align}
where the function $\Delta (.) \equiv \Delta_{NT} (.)$ is a measurable nonlinear function, $\gamma_{Nt} \to 0$, as $( N,T ) \to \infty$. To do this we use the following notation. Define with $e_{it} \equiv Y_{it} - X_{it}^{\prime} \beta^0 - F_t^{0 \prime} \lambda_i^0$. Define the probability density function of the covariates with $f_{it}(.)$ which satisfies related regularity conditions that ensure its validity. Moreover, since we have that $e_{it} = \varepsilon_{it}$ and it holds that $\mathbb{E} \left( e_{it} | X_{it} \right) = 0$ under $H_0$, such that
\begin{align}
J \equiv \mathbb{E} \big[ e_{it} \mathbb{E} \left( e_{it} | X_{it} \right) f_i(X_{it}) \big]  =  \mathbb{E} \big[ \left\{ \mathbb{E} \left( e_{it} | X_{it} \right) \right\}^2 f_i(X_{it}) \big] = 0.
\end{align}
Under the null hypothesis we have that the following relation holds: $e_{it} = \varepsilon_{it} + m(X_{it}) - X_{it}^{\prime} \beta^0$.
\begin{align}
\mathbb{E} \left( e_{it} | X_{it} \right) = m(X_{it}) - X_{it}^{\prime} \beta^0, \ \ \ \text{is not equal to} \ 0 \ \textit{almost surely}    
\end{align}
implying that $\mathbb{E} \big[  e_{it} \mathbb{E} \left( e_{it} | X_{it} \right) f_i (X_{it}) \big] > 0$, under $H_1$ (see, \cite{su2015specification}).

%%-------------------------------------------------------------------------%%
\newpage

Based on the above notation we can proceed with the introduction of the consistent test for the correct specification of the linear panel data model based on this observation. Specifically, in order to construct the test statistic, we need to estimate the model under the null hypothesis and obtain the restricted residuals $\hat{\epsilon}_i = \left( \hat{\epsilon}_{i1},..., \hat{\epsilon}_{iT} \right)^{\prime}$ for $i \in \left\{ 1,..., N \right\}$. Then, we can obtain the sample analog of $J$ such that 
\begin{align*}
J_{NT} 
= \frac{1}{(NT)^2} \sum_{i=1}^N \sum_{j=1}^N \sum_{t=1}^T \sum_{s=1}^T \hat{\varepsilon}_{it} \hat{\varepsilon}_{it} K_h \left( X_{it} - X_{js} \right)    
=
\frac{1}{(NT)^2} \sum_{i=1}^N \sum_{j=1}^N  \hat{\varepsilon}_i \mathcal{K}_{ij} \hat{\varepsilon}_j 
\end{align*}
where $K_h (x) = \prod_{\ell=1}^p h_{\ell}^{-1} k \left( \frac{x_{\ell}}{ h_{\ell} } \right)$ is a univariate kernel function with the vector $h = \left( h_1,..., h_p \right)$ is a bandwidth parameter, and $\mathcal{K}_{ij}$ is an $\left( T \times T \right)$ matrix whose $(t,s)-$th element is given by definition $\mathcal{K}_{ij} = K_h \left( X_{it} - X_{js} \right)$.

\begin{assumption}
Suppose that the following conditions hold:
\begin{itemize}
    \item[\textit{(i)}.] $\mathbb{E} \left( \varepsilon_{it} | \mathcal{F}_{t-1} \right) = 0$, almost surely, for each $i$, where $\mathcal{F}_{t-1}$, where $\mathcal{F}_{t-1}$ is the $\sigma-$field generated by
    \begin{align}
      \left\{  \left\{ \varepsilon_{i,t-1} \right\}_{i=1}^N, \left\{ \varepsilon_{i,t-2} \right\}_{i=1}^N, \left\{ \varepsilon_{i,t-3} \right\}_{i=1}^N,...  \right\}
    \end{align}
    
    \item[\textit{(ii)}.] There exist possibly time-varying moments such that
    \begin{align*}
        \mathbb{E} \left( \varepsilon_{it} \varepsilon_{jt} | \mathcal{F}_{t-1} \right) &= \mathbb{E} \left( \varepsilon_{it} \varepsilon_{jt} \right) =: \omega_{ij} ( \tau_t )
        \\
        \mathbb{E} \left( \varepsilon_{it} \varepsilon_{jt} \varepsilon_{kt} \varepsilon_{\ell t}   | \mathcal{F}_{t-1} \right) &= \mathbb{E} \left( \varepsilon_{it} \varepsilon_{jt} \varepsilon_{kt} \varepsilon_{\ell t} \right) =: \xi_{ijk \ell} ( \tau_t )
    \end{align*}
    where $\omega_{ij} (.)$ and $\xi_{ijk \ell} (.)$ satisfying
    \begin{align}
        \sum_{i,j = 1}^N \underset{ 1 \leq t \leq T }{ \mathsf{sup} } \left| \omega_{ij} (\tau_t) \right| = \mathcal{O}(N) \ \ \ \text{and} \ \ \ \sum_{i,j, k, \ell = 1}^N \underset{ 1 \leq t \leq T }{ \mathsf{sup} } \left| \xi_{ijk \ell}  (\tau_t) \right| = \mathcal{O}(N)
    \end{align}

\end{itemize}

\end{assumption}

Notice that Assumption 1 rules out conditional heteroscedasticity that depends on the past information at time $t-1$. However, it does allow for unconditional heteroscedasticity that depends on cross-sectional units and the scaled time index $\tau_t$ and is therefore, less restrictive than conditional homoscedasticity. 

\begin{remark}
Any stationary and invertible ARMA process can be expressed as an AR$(\infty)$ process. Therefore, in order to accommodate the ARMA process, we can extend $u_{it}$ to be an AR$(\infty)$ process $u_{it} = \sum_{j=1}^{\infty} \rho_j u_{it-j} + \varepsilon_{it}$, where $\left\{ \rho_j \right\}_{j=1}^{\infty}$ satisfies the stationarity condition. Notice that in $J_{Nt}$ the kernel function $k(.)$ depends on the nonstochastic term $\tau_t$. As a result, to obtain the asymptotic distribution of $J_{NT}$, we have to rely on the martingale central limit theorem (CLT) (see, Theorem 2 in \cite{brown1971general}) instead of applying \cite{hall2014martingale} CLT for a second-order degenerate U-statistic.  Furthermore, a bias correction might be needed especially when the data generating process under consideration has underline trend dynamics. For example, under the presence of a trend, it has been proved that it will affect the asymptotic distribution of unit root testing procedures. 
\end{remark}

%%-------------------------------------------------------------------------%%
\newpage

\subsubsection{A Boostrap implementation of the test statistic}

\begin{remark}
Notice that due to the nonparametric form of the proposed test statistic, which includes kernel based estimators, this results to slow convergence rates and therefore the asymptotic normal distribution may not serve as a good approximation. Specifically, this kernel-based test obtaining critical values from the normal distribution can be sensitive to the choice of bandwidths and suffer substantial finite sample size distortions (see, \cite{su2015specification}). 
\end{remark}

\begin{itemize}
    
\item[\textbf{(a).}] Obtain the restricted residuals $\hat{\varepsilon}_{it} = Y_{it} - X_{it}^{\prime} \hat{\beta} - \hat{F}^{\prime}_t \hat{\lambda}_i$, where the parameters $\hat{\beta}$, $\hat{F}_t$ and $\hat{\lambda}_i$ are estimates under the null hypothesis of linearity and correct model specification. Calculate the test statistic $\hat{\Gamma}_{NT}$ based on $\left\{ \hat{\varepsilon}_{it}, X_{it} \right\}$.   
    
\item[\textbf{(b).}] For $i \in \left\{ 1,..., N \right\}$ and $t \in \left\{ 1,..., T \right\}$, obtain the bootstrap error $\varepsilon_{it}^{*} = \hat{\varepsilon}_{it} \eta_{it}$, where $\eta_{it}$ are independently and identically distributed $\mathcal{N}(0,1)$ across $i$ and $t$. Next, we generate analog $Y_{it}^{*}$ of $Y_{it}$ by holding the estimated parameters from the previous step fixed, i.e., $\left( X_{it}, \hat{F}_t, \hat{\lambda}_i \right)$ such that: 
\begin{align}
Y_{it}^{*} = \hat{\beta}^{\prime} X_{it} + \hat{\lambda}^{\prime}_i \hat{F}_t + \varepsilon_{it}^{*},  
\end{align}
          
\item[\textbf{(c).}] Next, given the estimated bootstrap resample which keeps the set of covariates $X_{it}$ fixed such that $\left\{ Y_{it}^{*}, X_{it} \right\}$, we obtain the corresponding QMLEs $\hat{\beta}^*$, $\hat{F}_t^*$ and the corresponding bootstrapped factor loadings $\hat{\lambda}_i^*$. Next, we estimate the corresponding residuals given by 
\begin{align}
\hat{\varepsilon}_{it}^* = Y_{it}^* - X_{it} \hat{\beta}^{*} - \hat{F}_t^{* \prime} \hat{\lambda}_i^{*} 
\end{align}
and calculate the bootstrap test statistic $\hat{\Gamma}^{*}$ based on $\left\{ \hat{\varepsilon}^{*} _{it}, X_{it} \right\}$.   
    
\item[\textbf{(d).}] We then repeat steps 2-3 for $B$ times and denote the sequence of bootstrapped test statistics as $\left\{ \hat{\Gamma}_{NT,b}^{*} \right\}_{b=1}^B$. The bootstap $p-$value is calculated as $p^{*} \equiv B^{-1} \sum_{b=1}^B \mathbf{1} \left\{ \hat{\Gamma}^{*}_{NT,b} \geq \hat{\Gamma}_{NT} \right\}$.  
     
\end{itemize}

\medskip

\begin{remark}
Notice that if $H_0$ holds, for the original sample, $\hat{\Gamma}_{NT}$ also converges in distribution to $\mathcal{N}(0,1)$ so that a test based on the bootstrap $p-$value will have the right asymptotic level. On the other hand, if $H_1$ holds for the original sample, $\hat{\Gamma}_{NT}$ diverges at rate $NT (h!)^{1/2}$ whereas $\hat{\Gamma}_{NT}^{*}$ is asymptotically normal $\mathcal{N}(0,1)$, which implies the consistency of the bootstrap-based test. 
\end{remark}

\medskip

\begin{definition}
Let $\left( \Omega, \mathcal{F}, \mathbb{P} \right)$ be a probability space. Let $\left\{ \xi_t, t \geq 1   \right\}$ be a sequence of random variables defined on $\left( \Omega, \mathcal{F}, \mathbb{P} \right)$. Then, the sequence $\left\{ \xi_t, t \geq 1   \right\}$ is said to be conditionally strong mixing given $\mathcal{G}$ the sub$-\sigma-$algebra of $\mathcal{F}$. 
\end{definition}

%%-------------------------------------------------------------------------%%
\newpage

\subsubsection{Asymptotic distribution of the test statistic}

\begin{assumption}
We assume that the following regularity conditions hold:

\begin{itemize}
    \item[\textit{(i)}.] For each $i \in \left\{ 1,..., N \right\}, \left\{ ( X_{it}, \varepsilon_{it}): t = 1,2,...  \right\}$ is conditionally strong mixing given $\mathcal{D}$ with mixing coefficients such that
    \begin{align}
        \left\{ \alpha^{\mathcal{D}}_{NT,i}(t), 1 \leq t \leq T - 1 \right\}
    \end{align}
    and 
    \begin{align}
       \alpha_{\mathcal{D}}(.) \equiv  \alpha^{\mathcal{D}}_{NT}(.) \equiv \underset{1 \leq i \leq N }{\mathsf{max}} \ \alpha^{\mathcal{D}}_{NT,i}(.) 
    \end{align}
    satisfies $\sum_{s=1}^{\infty} \alpha_{\mathcal{D}}(s)^{1 / \text{q}_3 } \leq C_{\alpha} < \infty$, \ \textit{almost surely} for some $\tilde{\eta} \in (0,1/3)$.  
    
    \item[\textit{(ii)}.] $( \varepsilon_i, X_i )$ for $i \in \left\{ 1,...,N \right\}$ are mutually independent of each other conditional on the neighborhood $\mathcal{D}$. 
    
    \item[\textit{(iii)}.] For each $i = 1,...,N$ we have that $\mathbb{E} \left( \varepsilon_{it} | \mathcal{F}_{NT,t-1} \right) = 0$ \textit{almost surely} where 
    \begin{align}
        \mathcal{F}_{NT,t-1} \equiv \sigma \left( \left\{ F^0, \lambda^0, X_{it}, X_{it-1}, \varepsilon_{i,t-1}, X_{i,t-2}, \varepsilon_{i,t-2},... \right\}_{i=1}^N \right)
    \end{align}
    
    \item[\textit{(iv)}.] For each $i = 1,...,N$, let $f_{i,t}(x)$ denote the marginal PDF of $X_{it}$ given $\mathcal{D}$, and $f_{i,ts}(x, \bar{x})$ the joint PDF of $X_{it}$ and $X_{is}$ given $\mathcal{D}$. Furthermore, we assume that $f_{i,t}(.)$ and $f_{i,ts}(.,.)$ are continuous in their arguments and uniformly bounded by $C_f < \infty$. 
    
\end{itemize}

\end{assumption}

Based on the above regularity conditions, \cite{su2015specification}  presents the exact estimation procedure to construct a consistent specification testing procedure. 

\subsection{Testing for Trend Specifications in Panel Data Models}

The framework proposed by \cite{wu2023testing} considers testing for trend specification in panel data regression models. In particular, the asymptotic distributions of the proposed test statistic are established under the assumption of cross-sectional dependence, although by restricting to the case that the error components to follow a \textit{martingale difference sequence} (MDS) and thus, rule out serial dependence. Therefore, the panel data trend model and the hypotheses of interest are presented below.  Suppose that we observe the panel data of $\left\{ y_{it}, i = 1,...,N, t = 1,..., T \right\}$, where $y_{it}$ is a scalar dependent variable of interest, $N$ the number of panel individuals and $T$ the number of periods. 

Thus, the model becomes as below: 
\begin{align}
y_{it} = \alpha_i + \mathsf{g}_t + u_{it}, 1 \leq i \leq N \ \ \ \text{and} \ \ \  1 \leq t \leq T,     
\end{align}
where $\alpha_i$ represents the unobserved individual-specific effect that satisfies $\sum_{i=1}^N \alpha_i = 0$ and $u_{it}$ is the error component. More flexible error structure can be also allowed using a suitable specification for heteroscedasticity, cross-sectional and serial dependence in $u_{it}$ (see,  \cite{wu2023testing}). 

An example of an application, is when $y_{it}$ represents the total rainfall or temperature across the United Kingdom, $\alpha_i$ is the unobserved region-specific effect and $\mathsf{g}_t$ represents the common climate change trend, and $u_{it}$ is the region specific error. Notice that the classical panel models often assume \textit{i.i.d} disturbances. This assumption is likely to be violated as the dynamic effect of exogenous shocks to the dependent variable is often distributed over several time periods. Additionally, we assume that spillover effects, competition and global shocks can all induce disturbances that display cross-sectional dependence. Therefore, in order to allow for $y_{it}$ to be general enough to accommodate both cross-sectional and serial dependence, we assume $u_{it}$ to follow an AR$(p)$ process such that 
\begin{align}
A(L) u_{it} = \varepsilon_{it}     
\end{align}
where $A(L) = \left( 1 - \sum_{j=1}^p \rho_j L^j \right)$ with $p \geq 1$ a fixed integer, such that the polynomial operator has all roots strictly outside the unit circle. Furthermore, we assume that the dynamic structure of $u_{it}$ is homogeneous across units. In particular, the homogenous panel autoregressive models are widely used to capture the dynamics of macroeconomic and financial variables. However, most of the studies in the literature consider the case in which innovations are \textit{i.i.d} over time and and across individuals. Here, we assume that that the innovation $\varepsilon_{it}$ is assumed to follow an MDS such that $\mathbb{E} \left( \varepsilon_{it} | \mathcal{F}_{t-1} \right) = 0$ almost surely for each $i$ and allow for cross-sectional dependence and heteroscedasticity, where $\mathcal{F}_{t-1}$ is the information set available at time $t-1$.

\subsection{Hausman Type Specification Test for Nonlinearity}
\label{appn} 

Several studies consider specification testing in panel data regressions (e.g., see \cite{lee2012hahn}). Let $\chi_t = ( y_t, x_t ) \in \mathbb{R}^2$ be a strictly stationary $\beta-$mixing process and define with $\mathsf{g}(x) = \mathbb{E} \big[ y_t | x_t \big]$. 

Consider testing the hypothesis that $\mathsf{g} (x) = \beta_0 + \beta_1 x$ against the alternative that $\mathsf{g}(x)$ is non-linear function of $x$.  Let $\theta = \left( \theta_{l}, \theta_{nl}  \right)$ where $\theta_l$ is the average partial effect under the linear specification and $\theta_{nl} = \mathbb{E} \left[ \frac{  \partial \mathsf{g} (x_t) }{ \partial x } \right]$ is the average partial effect under the non-linear specification. An estimator for $\theta$ based on a $Z-$estimator using a plug in non-parametric estimate $\hat{\mathsf{g}}_k = \hat{\mathsf{g}}_k(x)$. For this purpose we define the moment function below
\begin{align}
\hat{m} \left( \chi_t, \theta,  \hat{\mathsf{g}}_k \right) 
=
\begin{bmatrix}
\big( ( y_t - \bar{y} ) - \theta_{l} ( x_t - \bar{x} ) \big) ( x_t - \bar{x} ) 
\\
\frac{ \partial P^{\kappa} (x_t)^{\prime} }{ \partial x } \bar{\beta}_k - \theta_{nl}
\end{bmatrix}
\end{align}
and let $m_n ( \theta ) = \frac{1}{n} \sum_{t=1}^n \hat{m} \left( \chi_t, \theta,  \hat{\mathsf{g}}_k \right)$. 

%%-------------------------------------------------------------------------%%
\newpage

The limiting distribution of the test statistic is analyzed for the following data-generating mechanism under local alternatives $\mathsf{g}_h (x)$,
\begin{align}
y_t = \beta_0 + \beta_1 x + \frac{ h(x_t) }{ \sqrt{n} } + u_t,    
\end{align}
where $u_t = y_t - \mathbb{E} \left[ y_t | x_t \right]$ is such that $\mathbb{E} [ u_t | x_t ] = 0$. Let $\theta_0 = ( \psi_1, \theta_{nl} )^{\prime}$ be the value of $\theta$ for the true data generating process under local alternatives. 

Under regularity conditions it follows (from \cite{newey1994asymptotic}), that for $h$ fixed
\begin{align}
\sqrt{n} \left( \hat{\theta}_{\kappa} - \theta_0 \right) = Q^{-1} \left( \frac{1}{ \sqrt{n} } \sum_{t=1}^n \big[  m( \chi_t, \theta_0, \mathsf{g}_h ) + \gamma ( \chi_t ) \big] \right) + o_p(1).  
\end{align}
The correction term $\gamma ( \chi_t )$ accounts for non-parametric estimation of the nuisance parameter $\mathsf{g}_h$ and can be derived using the methods developed in Newey (1994). It is given by
\begin{align}
\gamma ( \chi_t ) = 
\begin{bmatrix}
0
\\
\delta_{nl} (x_t)
\end{bmatrix}
\left( y_t - \beta_0 - \beta_1 x - \frac{ h(x_t) }{ \sqrt{n} } \right), \ \ \ \delta_{nl} (x_t) = - \zeta_x (x)^{-1} \frac{ \partial \zeta_x (x) }{ \partial x }    
\end{align}
such that $\zeta_x(x)$ is the marginal density of $x_t$. Define the empirical process
\begin{align}
\nu_n(h) = \frac{1}{ \sqrt{n} } \sum_{t=1}^n \bigg\{ m( \chi_t, \theta_0, \mathsf{g}_h ) + \gamma (\chi_t)  - \mathbb{E} \big[ m( \chi_t, \theta_0, \mathsf{g}_h ) \big]  \bigg\}.
\end{align}
Notice that the stochastic equicontinuity properties of the empirical process given above can be used to verify regularity conditions. Furthermore, the functional central limit theorem delivers a stochastic process representation of the limiting distribution of $\hat{\theta}_{\kappa}$ over the class of local alternatives. To obtain the limiting distribution of the above empirical process we consider the following auxiliary vector
\begin{align}
v_t = 
\begin{bmatrix}
u_t ( x_t - \mu_x )
\\
\frac{ \partial \mathsf{g}_h ( x_t ) }{ \partial x_t } - \theta_{n \ell} + \delta_{n \ell} (x_t) u_t
\end{bmatrix}
\end{align}
and the corresponding long-run covariance matrix given by 
\begin{align}
\Gamma (h) = \sum_{j = - \infty}^{ \infty } \mathbb{E} \big[ v_t v_{t-j}^{\prime}   \big].    
\end{align}
In summary, by expanding the following components separately, the convergence in probability of the estimator $\hat{\theta}_{\kappa}$ from its true parameter value can be expressed as
\begin{align}
\sqrt{n} \left( \hat{\theta}_{\kappa} - \theta_0 \right) = Q^{-1} \frac{1}{ \sqrt{n} }  \sum_{t=1}^n m_t \left( \chi_t, \theta_0, \hat{\mathsf{g}}_{\kappa}     \right) + o_p(1).   
\end{align}

%\begin{proof}
%Notice that the stochastic equicontinuity of the empirical process determining the limiting distribution is directly tied to the $\norm{ . }_{, \beta}$ norm. 

%Let $\hat{m} ( \chi_t, \theta, \hat{\mathsf{g} } ) = \hat{m}_t( \theta )$ and $m ( \chi_t, \theta, \mathsf{g}_h ) = m_t( \theta )$. Consider the expansion below 
%\begin{align*}
%\sqrt{n} m_n ( \theta_0 ) &= n^{-1/2} \sum_{t=1}^n \hat{m}_t ( \theta_0 ) 
%=  
%n^{-1/2} \sum_{t=1}^n \big( m_t ( \theta_0 ) + \gamma ( \chi_t ) \big)  
%\\
%&+ 
%n^{-1/2} \sum_{t=1}^n \big( m_t ( \theta_0 ) + \gamma ( \chi_t ) \big)
%\end{align*}

%\end{proof}

%%-------------------------------------------------------------------------%%
\newpage

\section{Nonstationary Panel Data Model Estimation}
\label{Section4}

The asymptotic orthogonality between the stationary and the integrated regressors that is, $x_{it}$ and $y_{i,t-1}$ can be examined in the context of nonstationary panel data model specification. In particular, related studies  to unit roots and cointegration in panels can be found in \cite{breitung2005parametric}, \cite{breitung2008unit} and \cite{han2010gmm}. Further applications include aspects of estimation and inference for panel VAR models (see, \cite{juodis2018first}, \cite{hayakawa2016improved} and \cite{camehl2023penalized}).

\subsection{A Simple AR(1) Panel Data Regression Model}

Consider the panel AR(1) model as below (see, \cite{juodis2021backward})
\begin{align}
y_{i,t} = \eta_i + \rho y_{i,t-1} + \varepsilon_{i,t}, \ \ \ \text{with} \ \ \mathbb{E} \big[ \varepsilon_{i,t} | y_{i,0}, \eta_i \big] = 0.    
\end{align}
where the data observed over $i = 1,...,N$ cross-sectional units in $t = 1,...,T$ time periods. 

As is well-known, the conventional Fixed Effects (FE) estimator suffers from a sizeable finite sample bias for small values of $T$. However, the bias is general more noticeable in case of persistent data which is a common pattern for most applications involving macroeconomic panels. Furthermore, we assume that idiosyncratic errors are $\varepsilon_{i,t}$ are independent over $i$, while the initial conditions $y_{i,0}$ are assumed to be observed. Therefore, an alternative estimator to mitigate the finite sample bias we consider the LS estimator of $\rho$ from the following augmented regression        
\begin{align}
y_{i,t} = \rho y_{i,t-1} +  \delta  \bar{y}_{i \bullet } + \tilde{\varepsilon}_{i,t}     
\end{align}
with the new composite error term given by 
\begin{align}
\tilde{\varepsilon}_{i,t} =  \varepsilon_{i,t} + \eta_i - \delta  \bar{y}_{i \bullet }    
\end{align}
where $\bar{y}_{i \bullet } = \frac{1}{T} \sum_{t=1}^T y_{i,t-1}$. The inclusion of $\bar{y}_{i \bullet }$ in the regression model, while at the same time ignoring the presence of $\eta_i$, ensures that the LS estimator, which is numerical equivalent to the FE estimator, is consistent as $T \to \infty$. On the other hand, using the full sample mean such that $\bar{y}_{i \bullet }$ creates other problems as it is correlated with all $\left\{ \varepsilon_{i,t} \right\}_{t=1}^{T-1}$. In particular, the sequence of combined error terms $\left\{ \tilde{\varepsilon}_{i,t}, \tilde{\varepsilon}_{i,t-1},...  \right\}$ is not a MDS even when $\eta_i = 0$. However, this issue can be fixed easily by considering in the estimation of $\bar{y}_{i \bullet }$ is replaced by the backward (recursive) mean of $y_{i,t-1}$ such that $\bar{y}_{i,t-1}^b = \frac{1}{t} \sum_{k=0}^{t - 1 } y_{i,k}$.   Therefore, unlike the full sample mean, the backward mean by construction is not correlated with the current and future values of $\varepsilon_{i,t}$. However, similar to the standard FE estimator, the LS estimator of this type is not consistent for any fixed $T$ but is consistent for $T$ large if the data is stationary (see, \cite{juodis2018first} and \cite{juodis2021backward}).   In particular, showed that in a model with the autoregressive parameter equal to unity both estimators have a substantially smaller asymptotic variance than the FE estimator.  These results are complementary to those provided by some other authors who study asymptotic and finite sample results under stationarity.

%%-------------------------------------------------------------------------%%
\newpage

\subsection{Panel Data Predictive Regression Model with Cross-Sectional Dependence}

In this section, we examine estimation and inference in panel data predictive regression systems with Cross-Sectional Dependence  (CSD) and heterogeneous degree of persistence. Specifically, we focus on the IVX instrumentation proposed by \cite{kostakis2015Robust}, but modifying the framework to account for panel data structure with Cross-Sectional Dependence. In terms of unobserved common factors in the panel structure, we assume that the proposed econometric specification identifies a set of common factors which impose the cross-sectional dependence structure in the panel data predictive regression system. We consider the identification and estimation of bias-corrected pooled IVX estimators, for which we investigate both the finite and asymptotic distributional properties in relation to the correct implementation of the instrumentation methodology using available information in the cross-sectional set of predictors which are observable for each cross-sectional observation. 

\subsubsection{Literature Review}

Panel data models are employed to account for unobserved heterogeneity and dependence. The standard fixed effects model captures only time-invariant heterogeneity, however in many time series applications heterogeneity and latent dynamics are time-varying effects which affect the parameter stability and robust econometric inference. Moreover, predictive regression systems are employed by econometricians to examine the joint predictability of a set of time series while accounting for the existence of local to unity dynamics in predictors, which allows to model the degree of persistence appeared in their time series. In its basic form, the degree of persistence is an unobserved random variable that describes the stochastic behaviour of the time series captured by the common localizing parameter $(c)$. 

When one is interested to model the dynamic persistence of a cross-section of observations the homogeneous persistence dynamics may not adequately reflect the corresponding stochastic processes.  For example, the framework proposed by \cite{katsouris2023statistical} can be extended to provide a unified framework for examining identification and estimation aspects related to panel predictive regression systems with network structure.  The proposed modelling approach provides a robust representation of the predictive-generating mechanism which allows to examine aspects of financial connectedness via the use of predictive regression systems in panel data structures. 

Our interest is the construction of a suitable IVX estimator (see, \cite{kostakis2015Robust}) to accommodate the cross-sectional time series structure of the data (see, \cite{hjalmarsson2006predictive}). We discuss the effectiveness of a pooled panel IVX estimator which can smooth the persistence levels across all cross-sectional observations $i$. Persistence homogeneity implies that the localizing coefficient of persistence remains fixed across panels although in a multivariate setting persistence levels is permitted to be different but of the same class. In the literature, various ways are presented regarding the analysis of cross section dependence. Specifically, we are interested for the case where both $N$ and $n$ are large and of the same order of magnitude, that is, $(N,n) \to \infty$ jointly which requires the notion of sequential asymptotic theory.

%%-------------------------------------------------------------------------%%
\newpage 

We follow the econometric specification proposed by \cite{hjalmarsson2006predictive}, which consists of a cross-sectional representation of panel predictive regression systems with common factors. This provides a natural way of examining the implementation of the IVX instrumentation when the econometric identification induces cross-sectional dependence structure. The IVX instrumentation can provide a suitable methodology of smoothing out persistence effects across panel data model which accommodate cross-sectional dependence, a reasonable assumption especially when considering common factors (such as macroeconomic conditions, volatility spillovers, industrial factors etc.) affecting the degree of persistence of the cross-sectional observations. The proposed framework has various applications in the examination of long-run economic relations driven by stochastic processes in which effects of shocks are propagated across the cross-sectional observations. 

\subsubsection{Conditional Mean Specification}

Suppose the information set up to time $t$, that is, $\mathcal{F}_t$ includes a large number of predictors $x_{it}$ for $i=1,...,N$ and $t=1,...,n$. We are interested to examine the predictive ability of a panel data predictive regression system using predictors of heterogeneous persistence. Specifically, consider a panel data structure with the pair of random variables $( y_{i,t}, x_{i,t} )$ where $x_{i,t}$ is an $m \times 1$ dimensional vector. Thus, we are constructing a predictive regression system which accounts for the predictability of a set of variables $y_{i,t}$, accounting this way for the idiosyncratic persistence of the other firms in the network. 

The particular specification is given by the following system of equations. 
\begin{align}
y_{i,t} &= \mu_i + \beta_i^{\prime} x_{i,t-1} + \gamma^{\prime}_i \boldsymbol{f}_t + u_{i,t} \\
x_{i,t} &= \boldsymbol{R}_i x_{i, t-1} + \boldsymbol{\Gamma}^{'}_i  \boldsymbol{f}_t  + \boldsymbol{v}_{i,t} 
\end{align}
where $\boldsymbol{R}_i = \displaystyle \left( \boldsymbol{I}  - \frac{ \boldsymbol{C}_i }{n^{\alpha_i}}   \right)$ is the $m \times m$ matrix of degree of persistence, $f_t$ is a $k \times 1$ vector capturing common factors in the error terms of the predictive regression system equations (see, \cite{kostakis2015Robust}). 

Firstly, the specification builds on the current frameworks of multivariate predictive regression systems, by considering instead of modelling simultaneously the joint predictability of stock returns with a panel data structure. Secondly, allowing for cross-sectional dependence which in terms of the econometric specification can be represented via the use of common factors affecting the panel of firms, can capture spillovers and network effects. Moreover, cross-sectional dependence which has an economic interpretation as well in terms of macroeconomic conditions and shocks provides a suitable framework for examining heterogeneous persistence in panel structures accounting for such cross-sectional effects across firms. Such common factors allow the inclusion of a baseline persistence across the set of predictands. Thirdly, panel estimators are derived using sequential limits (\cite{hjalmarsson2006predictive}), which usually implies first keeping the cross-sectional dimensions, $N$, fixed and letting the time-series dimension, $n$, go to infinity, and then letting $n$ go to infinity. We denote such sequential convergence denoted as $(N,n \to \infty)_{\text{seq}}$. We denote with $BM( \boldsymbol{\Omega} )$ the Brownian montion with covariance matrix $\boldsymbol{\Omega}$.

%%-------------------------------------------------------------------------%%
\newpage

In particular, \cite{moon2000estimation}, examine the case of inference in panel data autoregressive models with near to unity roots. In particular, the proposed framework allows for sequential asymptotic theory in order to examine the asymptotic properties of such panel data estimators which can accommodate for near to unity cases. We build on the theory of near to unity for panel data autoregressive estimators by focusing on the panel data predictive regression econometric specification. Assumption 1 below provides the necessary conditions for the identification of the proposed econometric specification. The innovation processes of the system indicates the stochastic behaviour of a system which describes a panel data structure with cross-sectional dependence.          
\begin{assumption}(\textbf{Innovation processes)} Let $\boldsymbol{\epsilon}_{i,t} = (v_{i,t}, u_{i,t}, f_t )$ denote the vector of innovation processes where $\boldsymbol{\epsilon}_{i,t} \in \mathbb{R}^{m + r}$  a real valued martingale difference sequence with respect to the natural filtration $\mathcal{F}_{i,t} =  \sigma(\boldsymbol{\epsilon}_{i,t}, \boldsymbol{\epsilon}_{i,t-1},...)$ for all $i=1,...,N$ and $t=1,...,n$.
\begin{align}
\mathbb{E}[\boldsymbol{\epsilon}_{i,t} | \mathcal{F}_{t-1}] = 0 \ \forall i \in \ \{ 1,...,N \}
\end{align}
be the expected value of the cross-sectional vector of innovation processes, and 
\begin{align}
\mathbb{E}[\boldsymbol{\epsilon}_{i,t} \boldsymbol{\epsilon}^{\prime}_{i,t} | \mathcal{F}_{t-1}] = \Sigma_i \ \text{almost surely and} \ \underset{ t \in \mathbb{Z} }{\text{sup}} \ \mathbb{E}||  \boldsymbol{\epsilon}_{i,t} ||^{\delta + 2} < \infty \ \text{for some} \ \delta > 0
\label{moment conditions}
\end{align}
where $\Sigma_i$ is a positive definite matrix. Let $u_{i,t}$ be a stationary linear process which describes the stochastic behaviour of the cross-sectional observation $i$ given by 
\begin{align}
u_{i,t} = \sum_{j = 0}^{\infty} \mathbf{C}_{i,j} e_{i, t-j}
\end{align}
where $( C_{i,j} )_{j \leq 0}$ is a sequence of constant matrices such that  $\sum_{j=0}^{\infty} \mathbf{C}_{i,j} \ \forall i \in \ \{1,...,N \}$ has full rank and $C_0 = I_r$. Furthermore, the following assumptions hold:

\begin{itemize}
\item[(i)] $\Sigma_{i,t} = \Sigma_{\boldsymbol{\epsilon}}$ for all $t$ and $\sum_{j=0}^{\infty} || \mathbf{C}_{i,j} || < \infty$.
\begin{align}
\Sigma_i = 
\begin{bmatrix}
\Sigma_{uv} & 0 \\
0 & \Sigma_{f} 
\end{bmatrix}, \ \text{where} \ 
\Sigma_{f} = 
\begin{bmatrix}
\sigma^2_{u} & \sigma_{uv} \\
\sigma_{uv} & \sigma^2_{v}
\end{bmatrix} 
\text{and} \
\Sigma = \underset{N \to \infty}{ \text{lim} } \sum_{i=1}^N \Sigma_i
\end{align} 
\item[(ii)] $({\boldsymbol{\epsilon}}_{i,t})_{t \in \mathbb{Z}}$ is strictly stationary and ergodic  satisfying moment conditions given by (\ref{moment conditions}) with $\delta = 2$ 
\begin{align}
\underset{ m \to \infty }{  \text{lim} } || \text{Cov} \big[ \text{vec} (\boldsymbol{\epsilon} \boldsymbol{\epsilon}^{'}), \text{vec} (\boldsymbol{\epsilon}_0 \boldsymbol{\epsilon}_0^{'}  ) ] || = 0
\end{align}
\end{itemize}

The sequence $({\boldsymbol{\epsilon}}_{i,t})_{t \in \mathbb{Z}}$ admits the following vec-GARCH(p,q) representation (see, \cite{kostakis2015Robust})   
\begin{align}
({\boldsymbol{\epsilon}}_{i,t}) = H^{1/2}_t \eta_t.
\end{align}
\end{assumption}

%%-------------------------------------------------------------------------%%
\newpage

\begin{assumption}(\textbf{Sequential Asympotics})
Consider that jointly $(N,T) \to \infty$,  which requires the notion of sequential asymptotics (see, \cite{moon2000estimation}). 
\end{assumption}

Our aim is to investigate whether the IVX instrumentation across the panel smooths out the abstract degree of persistence across the cross-sectional observation $i$ and under the sequential asymptotic framework. Specifically, assuming the existence of common effects for the cross-sectional observations $i$, then using the IVX estimator instead of a pooled estimator can achieve the mixed normality assumption even under abstract degree of persistence. Furthermore, we impose the following assumption which provides a necessary condition for the degree of cross-sectional dependence among the observations of the cross-section. In other words, we impose an assumption which ensures a maximum bound for the cross-sectional dependence which also ensures that there is a limited amount of cross-sectional interactions and financial interconnectedness. Such assumption provides also conditions for the asymptotic efficiency of the IVX estimator in the case of panel data predictive regression specification. Take for example, the case of $m-$dependence in time series. Then, imposing such a condition for the panel data models, induces an asymptotic independence condition across the cross-sectional observations. 

\begin{assumption}(\textbf{Strict Exogeneity})
$\{ z_{i,t} \}$ is strictly exogenous. 
\end{assumption}

Even though this is a strong assumption, we consider that the cross-sectional constructed instruments are strictly exogenous which ensures consistency and asymptotic normality and mixed normality results. Moreover, the assumption of weak exogeneity, provides conditions under which we can perform efficient inference on the conditional model without imposing any regulatory assumptions on the functional form of the system \cite{hatanaka1996time}. Therefore from the above econometric setting we can see that the main challenge with the robust estimation and identification is that usually the literature on large linear models focuses on the $\textit{i.i.d}$ case, while in our setting we have some type of cross-sectional dependence which could distort the asymptotic theory, especially with persistence and endogenous regressors.

\subsubsection{Model Estimation}

\paragraph{Cross-Sectional Independence}

\begin{theorem}
In the case that there is no common factors in the model, that is, $\gamma_i \equiv 0$ and $\Gamma_i \equiv 0$, then we have the pool estimator is given by 
\begin{align}
\hat{\beta}_{pool} = \left( \sum_{i=1}^N  \sum_{t=1}^n x_{i, t-1} x_{i, t-1}^{'} \right)^{-1} \left( \sum_{i=1}^N  \sum_{t=1}^n y_{i, t} x_{i, t-1}^{'}  \right)
\end{align}
Under Assumptions 1 and 2, with $\gamma_i = 0$, $\Gamma_i \equiv 0$, and $\alpha_i \equiv 0$ for all $i$, as $\left( N,n \to \infty \right)_{\text{seq}}$,
\begin{align}
\sqrt{N} n \left(  \hat{\beta}_{pool} - \beta \right) \implies  N \left( 0, \Omega^{-1}_{xx} \Phi_{ux} \Omega^{-1}_{xx} \right)   
\end{align}
\end{theorem}

%%-------------------------------------------------------------------------%%
\newpage 

\begin{theorem}
Let $\underline{y}_{i,t}$ and $\underline{x}_{i,t}$, denote the time-series demeaned data, that is, 
\begin{align}
\underline{y}_{i,t} = {y}_{i,t} - \frac{1}{N} \sum_{t=1}^N y_{i,t} \ \text{and} \ \underline{x}_{i,t} = {x}_{i,t} - \frac{1}{N} \sum_{t=1}^N x_{i,t-1}
\end{align}
The fixed effects pooled estimator, which allows for individual intercepts, is then 
\begin{align}
\hat{\beta}_{FE} = \left( \sum_{i=1}^N  \sum_{t=1}^n \underline{x}_{i, t-1} \underline{x}_{i, t-1}^{'} \right)^{-1} \left( \sum_{i=1}^N  \sum_{t=1}^n \underline{y}_{i, t} \underline{x}_{i, t-1}^{'}  \right)
\end{align}
The asymptotic distribution is affected by the demeaning of the $(y_{i,t}, x_{i,t} )$ observations. For fixed $N$, as $n \to \infty$,
\begin{align}
T \left(  \hat{\beta}_{FE} - \beta  \right) \implies \left( \frac{1}{N} \sum_{i=1}^N \underline{J}_i  \underline{J}^{'}_i  \right)^{-1} \left( \frac{1}{N} \sum_{i=1}^N \int_0^1 d B_{1,i} \ \underline{J}_i  \right) 
\end{align}
\end{theorem}

\paragraph{Cross-Sectional Dependence}

This is the main section of the paper. We aim to investigate whether the assumption of cross-sectional dependence as seen via the existence of common factors, along with the assumption of heterogeneous degree of persistence, induces a consistent IVX estimator. We argue that even though we include a set of common factors in the model, there is still the possibility of existence of heterogeneous degree of persistence across the predictors of the cross-sectional observations.

\subsubsection{IVX Instrumentation}

In this section, we examine the use of the IVX instrumentation methodology as in \cite{kostakis2015Robust}, to the framework of the panel data predictive regressions with cross-sectional dependence. We consider the first order difference of the corresponding state equation,  for the cross-sectional observation $i$ is
\begin{align}
\Delta x_{i,t} = \frac{ \mathbf{C}_i }{ n^{\alpha_i} } x_{i, t-1} + u_t
\end{align}
where $\alpha_i$ the cross-sectional rate of convergence of the first difference equation. Note that, the particular first difference is not an innovation process unless the regressor belongs to the persistence class of integrated processes. However, it behaves asymptotically as an innovation after linear filtering by a matrix consisting of near-stationary roots. The mildly integrated instrument for each cross-sectional observation $i$ is given by (see, \cite{kostakis2015Robust})
\begin{align}
z_{i,t} = \sum_{j=1}^t \widetilde{\mathbf{R}}_{i} z_{i, t-1} + u_{i, t}, \ \text{for} \ t \in \{1,...,n\}, z_{i,0} = 0 \ \forall \ i \in \{1,...,N\}.  
\end{align}

%%-------------------------------------------------------------------------%%
\newpage

The artificial matrix $\widetilde{\mathbf{R}}_{i}$ has the following form
\begin{align}
\widetilde{\mathbf{R}}_{i} =  \mathbf{I}_r + \frac{  \widetilde{ \mathbf{C}}_i  }{ n^{\beta_i} } , \ \beta_i \in (0,1), \mathbf{C}_i < 0 \ \forall i \in \{1,...,N\}
\end{align} 
The particular artificial matrix facilitates a way of smoothing the degree of persistence in the cross-section by constructing instruments derived from the existing information in the regressor of the data and therefore the induced instrumentation method has milder degree of persistence (see, \cite{kostakis2015Robust}). For the proposed panel data specification, the aim is to investigate whether using instruments corresponding to each of the cross-sectional observations $i$ induces a mixed Gaussian asymptotic distribution for the IVX estimator. 

\begin{remark}
Notice that when the moment process is allowed to exhibit autocorrelations of unknown forms, then the test statistics depend on a nonparametric estimator of the long-run variance (LRV) of the moment process. Therefore, in the regression case which is a special case of the GMM, this nonparametric LRV estimator is more commonly referred to as the heteroscedasticity and autocorrelation robust (HAR) variance estimator. Then, the asymptotic chi-square theory rely crucially on the assumption that the LRV estimator is consistent. In other words, the asymptotic chi-squared theory ignores the estimation uncertainty of the nonparametric LRV estimator. Thus, for this reason, the approximating chi-squared distributions can be far from the finite sample distributions. In other words, ignoring the estimation errors altogether will lead to unreliable inferences in finite samples. Therefore, one can develop fixed-smoothing asymptotics for the test statistics to account for the estimation uncertainty in the underlying LRV estimators. Specifically, unlike the conventional asymptotics where the amount of nonparametric smoothing increases with the sample size, the fixed-smoothing asymptotics holds the amount of nonparametric smoothing fixed. 
\end{remark}

\subsection{Differencing-based Transformation Approach}

Consider the model 
\begin{align}
y_t &= \alpha + \beta x_{t-1} + u_t, \ \ \ t = 1,...,n    
\\
x_t &= \rho x_{t-1} + v_t
\end{align}
Moreover, consider the differencing estimators proposed by  \cite{camponovo2015differencing}  for autoregressive models to the predictive regression models. Then, we consider the differenced observations
\begin{align}
\Delta x_{t - \ell} := \left( x_t - x_{t-\ell} \right), \ \ \ \text{with} \ \ t = \ell + 1,...,n   
\end{align}
Moreover, we consider also differenced response variables such that 
\begin{align}
\Delta y_{t - \ell} := \left( y_t - y_{t-\ell} \right), \ \ \ \text{with} \ \ t = \ell + 1,...,n   
\end{align}

%%-------------------------------------------------------------------------%%
\newpage 

Thus to determine the stationary instruments, $w_t$, $t = \ell + 1,..., n$, which are strongly correlated with stationary differenced predictors $\Delta x_{t-\ell-1}$, those should satisfy the following moment conditions
\begin{align}
\mathbb{E} \big[ \left( \Delta y_{t-\ell} - \beta \Delta x_{t-\ell-1} \right) w_t \big] = 0,
\end{align}
Using $w_t := \Delta x_{t-\ell-1}$, we get that $\mathbb{E} \big[ \left( \Delta y_{t-\ell} - \beta \Delta x_{t-\ell-1} \right)\Delta x_{t-\ell-1} \big]  = \mathbb{E} \big[ \left( u_t - u_{t-\ell} \right) \Delta x_{t-\ell-1} \big]$, so
\begin{align*}
\mathbb{E} \big[ \left( \Delta y_{t-\ell} - \beta \Delta x_{t-\ell-1} \right)\Delta x_{t-\ell-1} \big]  
=
\mathbb{E} \left[ \left( u_t - u_{t-\ell} \right) \left( \sum_{j=1}^{\ell} \rho^{j-1} v_{t-j} + \left( \rho^{\ell} - 1 \right) x_{t-\ell-1} \right) \right]
= 
\rho^{\ell-1} \sigma_{uv} \neq 0.
\end{align*}
Unless $\rho = 0$ or $\sigma_{uv} = 0$. Therefore, $\Delta x_{t-\ell-1}, t = \ell + 1,...,n$, are not valid instruments. However, with slight modifications and based on the following moment equalities
\begin{align}
\mathbb{E} \big[ \big( \Delta y_{t-\ell} - \beta \Delta x_{t-\ell-1} \big) \Delta x_{t-\ell-1}  \big] &= \sigma_{uv},   
\\
\mathbb{E} \big[ \big( \Delta y_{t-\ell} - \beta \Delta x_{t-\ell-1} \big) \Delta x_{t-\ell}  \big] &= \left( \rho^{\ell-1} - 1 \right) \sigma_{uv},
\end{align}
Then, we can prove that the instruments 
\begin{align}
w_t := \Delta x_{t-\ell} + \left( 1 - \rho^{\ell-1} \right) \Delta x_{t-\ell-1}, \ t = \ell + 1,...,n    
\end{align}
satisfy the moment conditions below
\begin{align}
\mathbb{E} \big[ \big( \Delta y_{t-\ell} - \beta \Delta x_{t-\ell-1} \big) \bigg(  \Delta x_{t-\ell} + \left( 1 - \rho^{\ell-1} \right)  \bigg) \Delta x_{t-\ell-1}  \big] = 0. 
\end{align}
Therefore, based on the above moment conditions, for a fixed value of $\ell \geq 2$, we define the new class of estimators $\beta_{n,\rho}^{(\rho)}$ of the parameter $\beta$ such that 
\begin{align}
 \beta_{n,\rho}^{(\ell)} = \frac{ \displaystyle \sum_{t=\ell+1}^n \Delta y_{t-\ell} \bigg( \Delta x_{t-\ell} + \left( 1 - \rho^{\ell-1} \right) \Delta x_{t-\ell-1} \bigg) \Delta x_{t-\ell-1} }{ \displaystyle \sum_{t=\ell+1}^n \Delta x_{t-\ell-1} \bigg( \Delta x_{t-\ell} + \left( 1 - \rho^{\ell-1} \right) \Delta x_{t-\ell-1} \bigg)  }   
\end{align}

\begin{remark}
A relevant framework within a panel data setting is presented by \cite{han2014x} who introduce a new estimation method for dynamic panel models with fixed effects and AR$(p)$ idiosyncratic errors. The proposed estimator uses a novel form of systematic differencing, called $X-$differencing, that eliminates fixed effects and retains information and signal strength in cases where there is a root at or near unity. The resulting "panel fully modified" estimator is obtained by pooled least squares on the system of $X-$differenced equations. The method is simple to implement, consistent for all parameter values, including unit root cases, and has strong asymptotic and finite sample performance characteristics, such as bias corrected least squares, GMM and system GMM methods.  The asymptotic theory holds as long as the cross section $(n)$ or time series $(T)$ sample size is large. 
\end{remark}

%%-------------------------------------------------------------------------%%
\newpage 

\subsection{Panel Cointegration}

The null hypothesis of exogeneity has also been reported in the cointegration literature before so this is an important aspect of consideration especially when considering cointegrating panel data regression models. Relevant studies on panel cointegration with respect to cross-sectional dependence and factor dynamics  include \cite{quintos1998analysis}, \cite{kao1999international}, \cite{bai2004estimating}, \cite{bai2009panel}, \cite{kapetanios2011panels} and \cite{westerlund2022factor} among many others. 

To study the distributional properties of such tests, we will describe the DGP in terms of the partitioned vector $z_{it}^{\prime} \equiv \big(  y_{it}, X_{it}^{\prime} \big)$ such that the true process $z_{it}$ is generated as 
\begin{align}
z_{it} = z_{it-1} + \xi_{it}, \ \ \ \text{for} \ \ \xi_{it} \equiv \big( \xi_{it}^y, \xi_{it}^X \big)    
\end{align}

\begin{assumption}[Invariance Principle]
The process $\xi_{it} \equiv \big( \xi_{it}^y, \xi_{it}^X \big)$ satisfies 
\begin{align}
\frac{1}{ \sqrt{T} } \sum_{t=1}^{ \floor{Tr} } \xi_{it} \Rightarrow B_i \big( \boldsymbol{\Omega}_i \big)  \ \ \ \text{as} \ \ \ T \to \infty.  
\end{align}
\end{assumption}

\begin{remark}
The above Assumption states that the standard functional central limit theorem is assumed to hold individually for each member series as $T$ grows large. Then, the $( m + 1 ) \times ( m + 1 )$ asymptotic covariance matrix is given by (e.g., see \cite{davidson1994stochastic})
\begin{align}
\Omega_i \equiv \underset{ T \to \infty }{ \mathsf{lim} } \mathbb{E} \left[ \frac{1}{T} \left( \sum_{t=1}^T \boldsymbol{\xi}_{it} \right) \left( \sum_{t=1}^T \boldsymbol{\xi}_{it}^{\prime} \right)  \right]  
\end{align}
\end{remark}

\subsubsection{Cointegrating Polynomial Regression}

\begin{assumption}
Suppose that the process $\left\{ \xi_t^0 \right\}_{ t \in \mathbb{Z} } = \left\{ [ \zeta_t, \varepsilon_t^{\prime} ]^{\prime} \right\}_{ t \in \mathbb{Z} }$ is a stationary and ergodic martingale difference sequence with natural filtration $\mathcal{F}_t = \sigma \left( \left\{ \xi_s^0 \right\}_{- \infty}^t \right)$ and conditional covariance matrix
\begin{align}
\Sigma^0 :=
\begin{pmatrix}
\Sigma_{\zeta \zeta} & \Sigma_{ \zeta \varepsilon }
\\
\Sigma_{\varepsilon \zeta} & \Sigma_{ \varepsilon \varepsilon } 
\end{pmatrix}
:= \mathbb{E} \left[  \xi_t^0  \xi_t^{0 \prime} \big| \mathcal{F}_{t-1} \right] > 0.
\end{align}
\end{assumption}

\begin{example}[Multicointegration in panel data, see \cite{berenguer2006testing}]
\

Let us consider a one-dimensional time series $\{ y_{i,t} \}_{0}^{\infty}$ and an $m-$dimensional time series $\{ x_{i,t} \}_{0}^{\infty}$, all being $I(1)$ non-stationary stochastic processes, for $t \in \{ 1,...,T \}$ and $i \in \{ 1,...,T \}$. These satisfy the following standard cointegration model
\begin{align}
y_{i,t} = c_t \alpha_i + x_{i,t} \beta_i + \epsilon_{i,t}
\end{align}
In typical applications, we have that $c_t = 0$, $c_t = 1$ or $c_t = (1, t)$ and $\epsilon_{i,t}$ is an $I(0)$ process.

%%-------------------------------------------------------------------------%%
\newpage 

Suppose that the cumulated cointegration residuals given by $S_{i,t} = \sum_{j=1}^t \epsilon_{i,t}$, cointegrate with either $\{ y_{i,t} \}_{0}^{\infty}$ and/or $\{ x_{i,t} \}_{0}^{\infty}$, then we obtain the standard multicointegration model, that is expressed as below
\begin{align}
S_{i,t} = m_t \delta_i + x_{i,t} \gamma_i + u_{i,t}
\end{align} 
where $u_{i,t}$ is an $I(0)$ series. Therefore, the multicointegration model can be written as below
\begin{align}
Y_{i,t} = C m_t \mu_i + X_{i,t} \beta_i + x_{i,t} \gamma_i + u_{i,t}
\end{align}
where $Y_{i,t} = \sum_{j=1}^t y_{i,j}$ and $X_{i,t} = \sum_{j=1}^t x_{i,j}$.
\end{example}

\begin{example}[Panel Data Cointegrating Polynomial Regression]
\cite{de2022panel} consider a panel data cointegrating polynomial regression analysis using a model of the form
\begin{align}
y_{it} &= \alpha_i + x_{it} \beta_1 + x_{it}^2 \beta_2 + u_{it},    
\\
x_{it} &= x_{it-1} + v_{it},
\end{align}
where $y_{it}$ denotes the $\mathsf{log} ( CO_2 )$ emissions per capita and $x_{it}$ $\mathsf{log} ( GDP )$ per capita. In particular, if the regressor is an integrated process (i.e., $\mathsf{log} ( GDP )$ per capita), then the above equation involves an integrated process and its square. Moreover, cointegration testing is known to be affected if the standard estimator in cointegerating linear regression with two integrated regressors is used while the underline stochastic process is likely to be driven by the presence of a CPR relationship.  Thus the particular example considers an extension of the FM OLS estimator to CPRs in a large $N$ and large $T$ panel setting allowing for individual and time fixed effects. In terms of assumptions, the authors follow \cite{phillips1999linear}, who introduced random linear processes to the panel cointegration literature.    

Regarding the limit results for the development of asymptotic theory those are taken with the time series dimension tending to infinitely first and the cross-sectional dimension tending to infinity thereafter. Furthermore, as it is also argued by \cite{de2022panel} since we consider the case where both $T$ and $N$, that is, the time series observations and the cross-sectional units tend to infinity, then the use of a cross-sectional modified OLS estimator (which we call SUR-OLS), allows to transform the individual specific random bias term into an expected value that can be consistently estimated. In particular, this estimator is based on subtracting a consistent estimator of a second-order bias term without the need to transform the dependent variable as in the case of the FM-OLS estimator or when leads and lags of the first difference of the integrated regressor are added in a Dynamic OLS setting (see, \cite{saikkonen1991asymptotically} and \cite{stock1993simple}).        

In a two-way effects model, the OLS (LSDV) estimator is given by 
\begin{align}
\boldsymbol{\beta} = \left( \sum_{i=1}^N \sum_{t=1}^T \boldsymbol{X}_{it}  \boldsymbol{X}_{it}^{\prime} \right)^{-1} \left(  \sum_{i=1}^N \sum_{t=1}^T \boldsymbol{X}_{it} \boldsymbol{y}_{it} \right).     
\end{align}

%%-------------------------------------------------------------------------%%
\newpage 

Then, paralleling the structure of the estimator for the one-way effects model, the modified OLS estimator of \cite{de2022panel}, which depends upon $\tilde{\boldsymbol{C}}_i$ just as in the one-way effects model, is given by the following expression 
\begin{align}
\hat{\boldsymbol{\beta} }_m  =  \left( \sum_{i=1}^N \sum_{t=1}^T \boldsymbol{X}_{it}  \boldsymbol{X}_{it}^{\prime} \right)^{-1} \sum_{i=1}^N \left( \sum_{t=1}^T \boldsymbol{X}_{it} \boldsymbol{y}_{it} - \tilde{\boldsymbol{C}}_i \right).   
\end{align}
Consequently, the limiting distribution of the modified OLS estimator is given by 
\begin{align}
N^{1/2} \boldsymbol{G}_T^{-1} \left( \hat{\boldsymbol{\beta} }_m - \boldsymbol{\beta} \right) \overset{d}{\to}  \mathcal{N} \left( \boldsymbol{0}, \boldsymbol{V}_2^{-1} \boldsymbol{\Sigma}_2 \boldsymbol{V}_2^{-1} \right).    
\end{align}
\end{example}

\subsubsection{Cointegrating Regression}

Consider the following data generating process
\begin{align}
y_t = \mu + x_t^{\prime} \beta + u_t, \ \ \ x_t = x_{t-1} + v_t.
\end{align}
Stacking the error process defines $\eta_t = \left[ u_t , \ v_t^{\prime} \right]^{\prime}$. Furthermore, it is assumed that $\eta_t$ is a vector of $I(0)$ processes in which case $x_t$ is a non-cointegrating vector of $I(1)$ processes and there exists a cointegrating relationship among $[ y_t, x_t^{\prime} ]^{\prime}$ with cointegrating vector $[ 1, - \beta^{\prime}  ]^{\prime}$. To review existing theory and to obtain the key theoretical results in the paper, assumptions about $\eta_t$ are required. It is sufficient to assume that $\eta_t$ satisfies a functional central limit theorem (FCLT) of the form given below
\begin{align}
\frac{1}{ \sqrt{n} } \sum_{t=1}^{ \floor{nr} } \eta_t \Rightarrow B(r) \Rightarrow \Omega^{1/2} W(r), \ \ \ r \in [0,1],    
\end{align}
Define the partial sum process such that $\widehat{S}_t = \sum_{j=1}^t \widehat{\eta}_t$. We start by establishing an invariance principle 
\begin{align*}
\frac{1}{T} \sum_{t=1}^{ \floor{nr}  } \widehat{u}_t 
&=
\frac{1}{\sqrt{T} } \sum_{t=1}^{ \floor{nr} } u_t - \frac{ \floor{nr} }{n} n^{1/2} \left(  \widehat{\mu} - \mu  \right) - \frac{1}{n \sqrt{n} } \sum_{t=1}^{ \floor{nr} } x_t^{\prime} n \left(  \widehat{\beta} - \beta \right)
\\
&\Rightarrow 
\int_0^r dB_u(s) - r \int_0^r B_v^{*} (s)^{\prime} ds \Theta.
\end{align*}
Using the definition of $\widehat{\eta} = \left[ \widehat{u}_t, v_t^{\prime}   \right]^{\prime}$ and stacking now leads to the following asymptotic theory result
\begin{align}
\frac{1}{ \sqrt{n} } \widehat{S}_{ \floor{nr} } = \frac{1}{ \sqrt{n} } \sum_{t=1}^{ \floor{nr} } \widehat{\eta}_t \Rightarrow 
\begin{bmatrix}
\displaystyle \int_0^r dB_u(s) - \int_0^r B_v^{*} (s)^{\prime} ds \Theta
\\
\displaystyle B_v(r)
\end{bmatrix}
\end{align}
Under the stated assumption it holds that
\begin{align}
\frac{1}{T} \sum_{t=2}^T S_{t-1}^{\eta} \eta_t^{\prime} \Rightarrow \int B(r) dB(r) + \Lambda.    
\end{align}

%%-------------------------------------------------------------------------%%
\newpage 

\begin{proposition}
As $b \to 0$, the fixed-b limiting distribution of $\widehat{\theta}_b^{+}$ converges in probability to the traditional limit distribution.    
\end{proposition}

\begin{remark}
In particular these results show that the performance of the FM-OLS estimator relies critically on the consistency approximation of the long-run variance estimators being accurate and that moving around the bandwidth and kernel impacts the sampling behaviour of the FM-OLS estimator. However, it is well-known that non-parametric kernel long run variance estimators suffer from bias and and sampling variability which as a result can affect the the accuracy of the traditional approximation.     
\end{remark}

We sketch the result for the Bartlett kernel only. The proposition is established by showing
\begin{align}
\underset{ b \to 0 }{ \mathsf{plim} }  \ \mathcal{F} \left( \widetilde{B}_{u.v}  \right) = \underset{ b \to 0 }{ \mathsf{plim} } \ \mathcal{F} \left( F \left( B_v^{*} \right) \right) = 0.
\end{align}
To begin with, showing that $\mathsf{plim}_{ b \to 0 } Q_b \left( B_v, B_{u.v}   \right) = 0$ is trivial. Furthermore, it is well-known in the fixed$-b$ literature that as $b \to 0$, fixed$-b$ limiting random variables converge to the long run variance being estimated. Specifically, the long-run covariance between $B_v$ and $B_{u.v}$ are independent and so it follows that $\mathsf{plim}_{ b \to 0 } Q_b \left( B_v, B_{u.v}   \right) = 0$. 

\begin{example}
\begin{align}
\tilde{X}_t = \big( X_t - X_0 \big) = \sum_{ k = 1 }^{ t-1 } \epsilon_k - \frac{ t - 1}{T} \sum_{ k = 1 }^T \epsilon_k  =  \sum_{ k = 1 }^{ t-1 } \left(  \epsilon_k - \frac{1}{T} \sum_{ s = 1 }^T \epsilon_s \right). 
\end{align}
\end{example}
To correctly define the limiting distributions we consider a sequence of rolling statistics. In particular, we first estimate the subsample statistics $t_{zx} ( \tau, \tau + \Delta \tau )$ for $t = \left\{ \floor{ \tau T} + 1,..., \floor{ \tau T} + \floor{ T \Delta \tau} \right\}$, where the window width is $\floor{ T \Delta \tau}$. Thus, the framework proposed by Magdalinos (2020) provides the first instance of standard Gaussian and chi-squared asymptotics applying respectively to the OLS estimator and the Wald statistic in a vector autoregression or predictive regression model with conditionally heteroscedastic innovations.Notice that because of the divergence of these partial sums we need to use an appropriate normalization factor which depends on the exponent rate of persistence in the LUR specification. Check also the paper: \cite{bruggemann2016inference}.

\begin{example}[Panel cointegration with global stochastic trends, see \cite{bai2009panel}]
\

Consider the following model 
\begin{align}
y_{it} = x_{it}^{\prime} \beta + e_{it}, \ \ \ x_{it} = x_{it-1} + \epsilon_{it}
\end{align}
\end{example}
where the model includes a $k-$dimensional vector of non-stationary regressors and the regression error $e_{it}$ is stationary and \textit{i.i.d} across $i$. Then, we can show that the pooled OLS estimator of $\beta$ is defined as
\begin{align}
\hat{\beta}_{ols} = \left( \sum_{i=1}^n \sum_{t=1}^T x_{it} x_{it}^{\prime}  \right)^{-1} \left(  \sum_{i=1}^n \sum_{t=1}^T  x_{it} y_{it}   \right)    
\end{align}

%%-------------------------------------------------------------------------%%
\newpage

The limiting distribution is shifted away from zero due to an asymptotic bias induced by the long-run correlation between $e_{it}$ and $\varepsilon_{it}$. The exception is when $x_{it}$ is strictly exogenous, in which case the estimator is $\sqrt{n} T$ consistent. The asymptotic bias can be estimated and a panel FM estimator can be developed along the lines of \cite{phillips1990statistical} to achieve $\sqrt{n} T$ consistency and asymptotic normality. Furthermore, the cross-section independence assumption is restrictive and difficult to justify when the data under investigation are economic time series. In view of co-movements of economic variables and shocks, we model the cross-section dependence by imposing a factor structure on $e_{it}$, \begin{align}
e_{it} = \lambda_i^{\prime} F_t + u_{it}    
\end{align}
where $F_t$ is an $r \times 1$ vector of latent common factors, $\lambda_i$ is an $r \times 1$ vector of factor loadings and $u_{it}$ is the idiosyncratic error. If $F_t$ and $u_{it}$ are both stationary, then $e_{it}$ is also stationary. In that case, a consistent estimator of the regression coefficients can still be obtained even when the cross-section dependence is ignored. In the first step, pooled OLS is used to obtain a consistent estimate of $\beta$. The residuals are then used to construct a FM estimator. In other words, nuisance parameters induced by cross-section correlation are dealt similar to the case of serial correlation by suitable estimation of the long-run covariance matrices. An alternative estimator can be developed by rewriting the equation as 
\begin{align}
y_{it} &= x_{it}^{\prime} \beta + \lambda_i^{\prime} F_t + u_{it}    
\\
x_{it} &= x_{it-1} + \epsilon_{it}
\\
F_t &= F_{t-1} + \eta_t
\end{align}
Moving $F_t$ from the error term to the regression function (i.e., treated as parameters) is desirable for the following reason. If some components of $x_{it}$ are actually $I(0)$, treating $F_t$ as part of the error process will yield an inconsistent estimate for $\beta$ when $F_t$ and $x_{it}$ are correlated. Under the presence of global stochastic trends, that is, $F_t$, which are shared by each cross-sectional unit, a new methodology needs to be developed. Denote with  $(n,T) \to \infty$ as the joint limit and with $( n, T )_{ \mathsf{sq} } \to \infty$ as the sequential limit which implies that $T \to \infty$ first and $n \to \infty$ later. Moreover, we denote with $\mathcal{MN} (0, V)$ the mixed normal distribution with variance $V$.
\begin{assumption}
Define with $w_{it} = \big(  u_{it}, \varepsilon_{it}^{\prime}, \eta_t^{\prime}  \big)^{\prime}, \ \ \ w_{it} = \Pi_i (L) v_{it} = \sum_{j=0}^{\infty } \Pi_{ij} v_{it-j}$.
\end{assumption}

\begin{assumption}
The previous assumptions imply that a multivariate invariance principle for $w_{it}$ holds, that is the partial sum process $\frac{1}{ \sqrt{T} } \sum_{t=1}^{ \floor{T \cdot} } w_{it}$ satisfies
\begin{align}
\frac{1}{ \sqrt{T} } \sum_{t=1}^{ \floor{T \cdot} } w_{it} \Rightarrow B_i ( \cdot ) \equiv B ( \Omega_i ), \ \ \text{as} \  T \to \infty \ \forall \ i,  
\end{align}
where $B_i = \big[ B_{ui} \ \ B_{\epsilon i}^{\prime} \ \ B_{\eta}^{\prime}   \big]$. Then, the long-run covariance matrix of $\left\{ w_{it} \right\}$ is given by 
\begin{align}
\Omega_i = \sum_{j = - \infty}^{\infty} \mathbb{E} \big[  w_{i0} w_{ij}^{\prime} \big| \Pi_{ij}^{u \eta} \big]    
\end{align}
\end{assumption}

%%-------------------------------------------------------------------------%%
\newpage 

\begin{lemma}
Suppose that $u_i$ is uncorrelated with $( x_i, F^0 )$, then as $( n, T )_{\mathsf{seq} } \to \infty$, it holds that
\begin{align}
\frac{1}{n} \sum_{i=1}^n \frac{1}{T^2} x_i^{\prime} M_{F^0} x_i &\overset{d}{\to}  
\underset{ n \to \infty }{  \mathsf{lim} }  \ \frac{1}{n} \sum_{i=1}^n \mathbb{E}  \left( \int Q_i Q_i^{\prime} | C \right)
\\
\frac{1}{  \sqrt{n} } \sum_{i=1}^n \frac{1}{T} x_i^{\prime} M_{F^0} x_i &\overset{d}{\to}  \mathcal{MN} \left( 0,  \underset{ n \to \infty }{  \mathsf{lim} }  \ \frac{1}{n} \sum_{i=1}^n  \Omega_{ui}  \mathbb{E}  \left( \int Q_i Q_i^{\prime} | C \right)   \right)
\end{align}
Notice that the convergence of the above two moment functions holds jointly.
\end{lemma}

\begin{lemma}
Let $Z_i = \displaystyle \left( M_{ F^0 x_i } - \frac{1}{n} \sum_{k=1}^n M_{ F^0 } x_k a_{ik} \right)$. Then, as $( n, T )_{\mathsf{seq} } \to \infty$  
\begin{itemize}
\item[(a)] 
\begin{align}
\frac{1}{nT^2} \sum_{i=1}^n Z_i^{\prime} Z_i \overset{d}{\to} \underset{ n \to \infty }{  \mathsf{lim} }  \ \frac{1}{n} \sum_{i=1}^n \mathbb{E}  \left( \int R_{ni} R_{ni}^{\prime} | C \right)   
\end{align}

\item[(b)] If $u_i$ is uncorrelated with $\left( x_i, F^0 \right)$ for all $i$ then
\begin{align}
\frac{1}{ \sqrt{n} T }  \sum_{i=1}^n Z_i^{\prime} u_i  \overset{d}{\to}  \mathcal{MN} \left( 0,  \underset{ n \to \infty }{  \mathsf{lim} }  \ \frac{1}{n} \sum_{i=1}^n  \Omega_{ui}  \mathbb{E}  \left( \int R_{ni} R_{ni}^{\prime}  | C \right)   \right) 
\end{align}

\item[(c)] If $u_i$ is possibly correlated with $\left( x_i, F^0 \right)$, then
\begin{align}
\frac{1}{ \sqrt{n} T }  \sum_{i=1}^n Z_i^{\prime} u_i  -\sqrt{n} \theta^n \overset{d}{\to}  \mathcal{MN} \left( 0,  \underset{ n \to \infty }{  \mathsf{lim} }  \ \frac{1}{n} \sum_{i=1}^n  \Omega_{u. bi }  \mathbb{E}  \left( \int R_{ni} R_{ni}^{\prime}  | C \right)   \right) 
\end{align}
\begin{align}
R_{ni} &= Q_i - \frac{1}{n} \sum_{k=1}^n Q_k a_{ik}, \ \ a_{ik} = \lambda_i^{\prime} \left( \Lambda^{\prime} \Lambda \big/ n \right)^{-1} \lambda_k    
\\
Q_i &= B_{ \epsilon i} - \left( \int  B_{ \epsilon i} B_{ \eta}^{\prime}  \right) \left( \int  B_{ \eta} B_{ \eta}^{\prime}  \right)^{-1} B_{\eta} 
\\
\theta^n &= \frac{1}{n} \sum_{i=1}^n \left[ \frac{1}{T} Z_i^{\prime} \big( \Delta \bar{x}_i \ \ \Delta F \big) \bar{\Omega}_{bi}^{-1}  \bar{\Omega}_{bui}  + \big(  I_k \ \ - \bar{\delta}_i^{\prime}  \big) \begin{pmatrix}
\bar{\Delta}^{+}_{ \varepsilon u i }
\\
\bar{\Delta}^{+}_{ \eta u }
\end{pmatrix}  \right]
\\
\bar{\delta}_i &= \left( F^{ \prime } F^0 \right)^{-1} F^{0 \prime} \bar{x}_i, \ \ \ \bar{x}_i = x_i - \frac{1}{n} \sum_{ k = 1 }^n x_k a_{ik}.
\end{align}
\end{itemize}
\end{lemma}
Relevant applications of panel cointegrating regressions is the framework proposed by \cite{wagner2020fully} who develop a SUR system with cointegrating dynamics. Although, the SUR cointegration literature differs in several respects from the SUR literature, where the former case implies that regressors are assumed to be strictly exogenous and stationary and the stationary errors serially uncorrelated. 

%%-------------------------------------------------------------------------%%
\newpage 

The econometric estimation of the SUR representation proposed by \cite{wagner2020fully} is based on the assumption of serially correlated errors which involves  the estimation of long-run variance matrices rather than estimates of contemporaneous variance matrices. Thus, the presence of regressor endogeneity in a cointegration setting necessitates the usage of modified least squares estimators to allow for asymptotically normal or chi-squared inference. The proposed framework is then employed for the analysis of the environmental Kuznets curve (EKC), specifically for carbon dioxide (CO$_2$) emissions\footnote{Specifically, the EKC hypothesis postulates an inverted U-shaped relationship between the level of economic development and pollution of emissions. Furthermore, there is significant use of unit root and cointegration techniques both in (single) time series and panel data settings.}. Specifically, \cite{wagner2020fully} consider the case of panel data with small cross-sectional dimension and develop estimation and inference techniques to combine cointegrating polynomial regressions to a system which allows to test general hypotheses concerning group-wise pooling\footnote{Notice that pooled estimation, when appropriate, leads to considerable efficiency gains, but can lead to misleading results when it is implemented incorrectly. In particular, \cite{wagner2020fully} conduct a simulation study to assess the finite sample performance of our estimators and tests based upon them. The FM-SUR estimator outperforms the FM-SOLS estimator in terms of bias and root mean squared error (RMSE). However, test statistics based on FM-SOLS exhibit in many configurations lower size distortions than tests based on FM-SUR. The latter have, however, higher size-corrected power than the former. Therefore, the evidence concerning hypothesis testing is mixed. Pooling, illustrated in the simulations by pooling the coefficients for the integrated regressors and its square over all cross-section members, leads to major performance improvements, in particular for the performance of tests. Moreover, the simulations also indicate the limitations of an unrestricted SUR approach in case of large $N$ and small $T$. Thus, it turns out that the estimation of large unrestricted long-run and half long-run covariance matrices is the key reason for the relatively poor performance in these constellations.}.   
\begin{assumption}
The process $\left\{ \boldsymbol{u}_t \right\}_{ t \in \mathbb{Z} }$ and $\left\{ \boldsymbol{v}_t \right\}_{ t \in \mathbb{Z} }$ are generated as below
\begin{align}
\boldsymbol{u}_t &:= \boldsymbol{C}_u ( L ) \zeta_t = \sum_{j=0}^{ \infty } \boldsymbol{C}_u \boldsymbol{\zeta}_{t-j},
\ \ \
\boldsymbol{v}_t := \boldsymbol{C}_v ( L ) \eta_t = \sum_{j=0}^{ \infty } \boldsymbol{C}_v \boldsymbol{\eta}_{t-j}
\\
\sum_{j=0}^{ \infty }  j \norm{ \boldsymbol{C}_{u,j} } &< \infty, \ \ \   \sum_{j=0}^{ \infty }  j \norm{ \boldsymbol{C}_{v,j} } < \infty, \ \ \  \mathsf{det} \left( \boldsymbol{C}_u(1) \right) \neq 0 , \ \ \ \mathsf{det} \left( \boldsymbol{C}_v(1) \right) \neq 0.
\end{align}
Moreover, the stacked process $\boldsymbol{\xi}^{\prime}_{ t \in \mathbb{Z} } := \left\{ [ \zeta_t^{\prime}, \eta_t^{\prime} ]^{\prime} \right\}_{ t \in \mathbb{Z} }$ is a strictly stationary and ergodic martingale difference sequence with respect to the natural filtration $\mathcal{F}_t := \sigma \left( \left\{ \boldsymbol{\xi}_t \right\}_{- \infty }^t \right)$ with positive definite conditional variance matrix such that $\boldsymbol{\Sigma} := \mathbb{E} \big[ \boldsymbol{\xi}_t \boldsymbol{\xi}_t^{\prime} | \mathcal{F}_{t-1} \big]$.
\end{assumption}

\subsubsection{FM-OLS Estimation and Inference for SUR Cointegrating regression}

Consider the system of equations with observations available for $i = 1,..., N$ and $t = 1,...,T$ such that 
\begin{align}
y_{i,t} &= D_{i,t}^{\prime} \theta_{D,i} + X_{i,t}^{\prime} \theta_{X,i} + u_{i,t} \equiv Z_{i,t}^{\prime} \theta_i + u_{i,t} 
\\
x_{i,t} &= x_{i,t-1} + v_{i,t}
\end{align}
Then the $N$ equations can be written in matrix form as a system of \textit{seemingly unrelated cointegrating regressions} defined as $y_t := Z_t^{\prime} \theta + u_t$.

%%-------------------------------------------------------------------------%%
\newpage

To establish the limiting distributions of model estimators we consider that the following functional central limit theorem (FCLT) to hold for $\left\{ \xi_t \right\}_{ t \in \mathbb{Z} } := \left\{ \big[ u_t^{\prime},  v_t^{\prime} \big]^{\prime} \right\}_{ t \in \mathbb{Z} }$ 
\begin{align}
\frac{1}{ \sqrt{T} } \sum_{ t = 1}^{ \floor{Tr} } \xi_t \Rightarrow 
\boldsymbol{B}(r) 
=
\begin{bmatrix}
\boldsymbol{B}_u(r)
\\
\boldsymbol{B}_v(r)
\end{bmatrix}
\end{align}
with the long-run covariance matrix $\boldsymbol{\Omega} := \sum_{ j = - \infty}^{ \infty } \mathbb{E} \left( \xi_t  \xi_{t-j}^{\prime} \right)$ and $W(r)$ is a $2N$ standard BM. Regarding the estimation methodology, \cite{wagner2020fully} examine the statistical properties of the OLS estimator  given by $\widehat{\boldsymbol{\theta}}_{OLS} = \left( Z^{\prime} Z \right)^{-1} Z^{\prime} y$ which is consistent but its limiting distribution is contaminated by second order bias terms. Moreover,  \cite{wagner2020fully}  consider a feasible GLS-type SUR estimator in which rather than the estimated error covariance matrix $\widehat{\boldsymbol{\Sigma}}_{uu}$, the estimated long-run covariance matrix $\widehat{\boldsymbol{\Omega}}_{uu}$ is used as a weighting matrix. The modified SUR estimator is defined as below
\begin{align}
\widetilde{ \boldsymbol{\theta} }_{MSUR}
= 
\bigg( \boldsymbol{Z}^{\prime} \left( \boldsymbol{I}_T \otimes \widehat{\boldsymbol{\Omega}}_{uu}^{-1} \right) \boldsymbol{Z} \bigg)^{-1} \bigg( \boldsymbol{Z}^{\prime} \left( \boldsymbol{I}_T \otimes \widehat{\boldsymbol{\Omega}}_{uu}^{-1} \right) \boldsymbol{y} \bigg).
\end{align}
Based on the above formulations, \cite{wagner2020fully} show that the MSUR estimator is consistent with a nuisance parameter dependent limiting distribution. Consequently, to construct estimators with a zero mean Gaussian mixture limiting distribution that allow for asymptotic chi-square inference, we propose fully modified type corrections to these two estimators. In particular, we consider the FM-OLS estimator of \cite{phillips1990statistical} which is based on a two-part transformation. The first transformation changes the dependent variable such that $y_t^{+} := y_t - \widehat{\boldsymbol{\Omega}}_{uv} \widehat{\boldsymbol{\Omega}}_{vv}^{-1} \Delta x_t$. In particular, the second transformation consists of subtracting an appropriately constructed correction term to remove bias terms otherwise present in the limiting distributions. This transformation depends upon the estimator considered as starting point OLS or MSUR and the specification of the equation system. 

\begin{proposition}[\cite{wagner2020fully}]
Suppose that based on the OLS residuals all required long-run covariances are estimated consistently. Then the fully modified systems OLS (FM-SOLS) and the fully modified SUR (FM-SUR) estimators are defined as below
\begin{align}
\widehat{\boldsymbol{\theta}}_{FM-SOLS} 
&:= 
\big( \boldsymbol{Z}^{\prime} \boldsymbol{Z} \big)^{-1} \left( \boldsymbol{Z}^{\prime} \boldsymbol{y}^{+} - \widehat{\boldsymbol{A}} \right),
\\
\widehat{\boldsymbol{\theta}}_{FM-SUR} 
&:= 
\bigg( \boldsymbol{Z}^{\prime} \left( \boldsymbol{I}_T \otimes \widehat{\boldsymbol{\Omega}}_{u.v}^{-1} \right) \boldsymbol{Z} \bigg)^{-1} \bigg( \boldsymbol{Z}^{\prime} \left( \boldsymbol{I}_T \otimes \widehat{\boldsymbol{\Omega}}_{u.v}^{-1} \right) \boldsymbol{y}^{+} - \widetilde{\boldsymbol{A}}^{*} \bigg),
\end{align}
with $\boldsymbol{y}^{+} = \big[ \boldsymbol{y}_1^{+ \prime},..., \boldsymbol{y}_T^{+ \prime} \big]^{\prime}$. As $T \to \infty$ it holds that 
\begin{align}
\boldsymbol{G}^{-1} \left( \widehat{\boldsymbol{\theta}}_{FM-SOLS}  - \boldsymbol{\theta} \right)
&\Rightarrow 
\left( \int_0^1 \boldsymbol{J}(r) \boldsymbol{J}(r)^{\prime} dr \right) \left( \int_0^1 \boldsymbol{J}(r) d B_{u.v} (r)  \right),
\\
\boldsymbol{G}^{-1} \left( \widetilde{\boldsymbol{\theta}}_{FM-SUR}  - \boldsymbol{\theta} \right)
&\Rightarrow 
\left( \int_0^1 \boldsymbol{J}(r) \widehat{\boldsymbol{\Omega}}_{u.v}^{-1} \boldsymbol{J}(r)^{\prime} dr \right) \left( \int_0^1 \boldsymbol{J}(r) \widehat{\boldsymbol{\Omega}}_{u.v}^{-1} d B_{u.v} (r)  \right),
\end{align}
where $\boldsymbol{B}_{u.v} := \boldsymbol{B}_u(r) - \boldsymbol{\Omega}_{uv} \boldsymbol{\Omega}_{vv}^{-1} \boldsymbol{B}_v(r)$ is a Brownian motion with covariance matrix $\boldsymbol{\Omega}_{u.v}$.
\end{proposition}

%%-------------------------------------------------------------------------%%
\newpage

\begin{proposition}[\cite{wagner2020fully}]
Consider $s$ linearly independent restrictions collected under the null such that $\mathbb{H}_0: R \theta = r$ with $R \in \mathbb{R}^{s \times d}$ of full row rank $s,r \in \mathbb{R}^s$. Suppose that there exists a sequence of full rank matrices $G_R = G_R (T)$ such that $\underset{ T \to \infty }{ \mathsf{lim} } G_R R G = R^{*}$ with $R^{*} \in \mathbb{R}^{s \times d}$ of full row rank $s$. Then, it holds that under the null hypothesis, the Wald-type statistics that correspond to the two estimators under consideration expressed as below
\begin{align}
\hat{\mathcal{W}}_T 
&:= 
\bigg( \boldsymbol{R} \hat{\boldsymbol{\theta}} - \boldsymbol{r} \bigg)^{\prime} \bigg[ \boldsymbol{R} \big( \boldsymbol{Z}^{\prime} \boldsymbol{Z} \big)^{-1} \boldsymbol{Z}^{\prime} \left( \boldsymbol{I}_T \otimes \widehat{\boldsymbol{\Omega}}_{u.v} \right) \boldsymbol{Z} \big( \boldsymbol{Z}^{\prime} \boldsymbol{Z} \big)^{-1} \boldsymbol{R}^{\prime}  \bigg] \bigg( \boldsymbol{R} \hat{\boldsymbol{\theta}} - \boldsymbol{r} \bigg)
\\
\tilde{\mathcal{W}}_T 
&:= 
\bigg( \boldsymbol{R} \tilde{\boldsymbol{\theta}} - \boldsymbol{r} \bigg)^{\prime} \bigg[ \boldsymbol{R}  \bigg( \boldsymbol{Z}^{\prime} \left( \boldsymbol{I}_T \otimes \widehat{\boldsymbol{\Omega}}_{u.v} \right) \boldsymbol{Z} \bigg)^{-1} \boldsymbol{R}^{\prime}  \bigg] \bigg( \boldsymbol{R} \tilde{\boldsymbol{\theta}} - \boldsymbol{r} \bigg)
\end{align}
are asymptotically chi-squared distributed with $s$ degrees of freedom. 
\end{proposition}

\begin{remark}
One of the main advantages of the SUR modelling approach is that it allows to test the poolability of the coefficients. Specifically, usually pooling is considered with respect to all cross-section members. In particular, if the null hypothesis corresponding to the variant of pooling considered is not rejected, then pooled estimation of a smaller number of parameters allows one to lift some efficiency gains when performing correspondingly pooled estimation. 
\end{remark}

\paragraph{Proof of Proposition 4}

\begin{proof}
Under the null hypothesis and the formulated constraints on the restriction matrix we have that 
\begin{align*}
\hat{\mathcal{W}}_T 
&:= 
\bigg( \boldsymbol{R} \hat{\boldsymbol{\theta}} - \boldsymbol{r} \bigg)^{\prime}  \bigg[ \boldsymbol{R} \big( \boldsymbol{Z}^{\prime} \boldsymbol{Z} \big)^{-1} \boldsymbol{Z}^{\prime} \left( \boldsymbol{I}_T \otimes \widehat{\boldsymbol{\Omega}}_{u.v} \right) \boldsymbol{Z} \big( \boldsymbol{Z}^{\prime} \boldsymbol{Z} \big)^{-1} \boldsymbol{R}^{\prime}  \bigg] \bigg( \boldsymbol{R} \hat{\boldsymbol{\theta}} - \boldsymbol{r} \bigg)
\\
&=
\bigg( \boldsymbol{R}\left( \hat{\boldsymbol{\theta}} -  \boldsymbol{\theta} \right)  \bigg)^{\prime} \bigg[ \boldsymbol{R} \big( \boldsymbol{Z}^{\prime} \boldsymbol{Z} \big)^{-1} \boldsymbol{Z}^{\prime} \left( \boldsymbol{I}_T \otimes \widehat{\boldsymbol{\Omega}}_{u.v} \right) \boldsymbol{Z} \big( \boldsymbol{Z}^{\prime} \boldsymbol{Z} \big)^{-1} \boldsymbol{R}^{\prime}  \bigg] \bigg( \boldsymbol{R}\left( \hat{\boldsymbol{\theta}} -  \boldsymbol{\theta} \right)  \bigg)
\\
&= 
\bigg( \boldsymbol{G}_R^{-1} \boldsymbol{R} \boldsymbol{G} \boldsymbol{G}^{-1} \left( \hat{\boldsymbol{\theta}} -  \boldsymbol{\theta} \right)  \bigg)^{\prime} \bigg[ \boldsymbol{G}_R^{-1} \boldsymbol{R} \big( \boldsymbol{Z}^{\prime} \boldsymbol{Z} \big)^{-1} \boldsymbol{Z}^{\prime} \left( \boldsymbol{I}_T \otimes \widehat{\boldsymbol{\Omega}}_{u.v} \right) \boldsymbol{Z} \big( \boldsymbol{Z}^{\prime} \boldsymbol{Z} \big)^{-1} \boldsymbol{R}^{\prime}  \boldsymbol{G}_R^{-1} \bigg] \bigg( \boldsymbol{G}_R^{-1} \boldsymbol{R} \boldsymbol{G} \boldsymbol{G}^{-1} \left( \hat{\boldsymbol{\theta}} - \boldsymbol{\theta} \right)  \bigg)
\end{align*}
Therefore, it follows that 
\begin{align*}
\hat{\mathcal{W}}_T  
&\Rightarrow 
\left[ \boldsymbol{R}^{*} \left( \int_0^1 \boldsymbol{J}(r) \boldsymbol{J}(r)^{\prime} dr  \right)^{-1}  \int_0^1 \boldsymbol{J}(r) d B_{u.v}(r) \right]^{\prime}
\\
&\ \ \ \ \ \ \ \times 
\left[  \boldsymbol{R}^{*} \left( \int_0^1 \boldsymbol{J}(r) \boldsymbol{J}(r)^{\prime} dr  \right)^{-1} \int_0^1 \boldsymbol{J}(r) \boldsymbol{\Omega}_{u.v} \boldsymbol{J}(r)^{\prime} dr \left( \int_0^1 \boldsymbol{J}(r) \boldsymbol{J}(r)^{\prime} dr  \right)^{-1}  \boldsymbol{R}^{* \prime}  \right] 
\\
&\ \ \ \ \ \ \ \times 
\left[ \boldsymbol{R}^{*} \left( \int_0^1 \boldsymbol{J}(r) \boldsymbol{J}(r)^{\prime} dr  \right)^{-1}  \int_0^1 \boldsymbol{J}(r) d B_{u.v}(r)  \right]
\end{align*}
which can be shown to be chi-squared distributed with $s$ degrees of freedom using standard arguments involving quadratic forms of functionals of zero mean Gaussian mixture distributions.   

\end{proof}

%%-------------------------------------------------------------------------%%
\newpage

\subsection{Estimating Cointegration from a Cross Section}

We follow the framework proposed by \cite{madsen2005estimating}. Consider the variables $Y_{it}$, $X_{1it}$ and $X_{2it}$ where $i = 1,...,N, t = 1,...,$ such that these variables are of dimensions $k_0 \times 1$, $k_1 \times 1$ and $k_2 \times 1$ respectively. For every cross-section unit we assume that the variables are generated by the following model
\begin{align}
Y_{it} &= \gamma_1^{\prime} X_{1it} + \gamma_2^{\prime} X_{2it} + \eta_{0it} 
\\
X_{1it} &= \eta_{1it}
\\
X_{2it} &= X_{2it-1} \eta_{2it}
\end{align}
where $\gamma_1$ and $\gamma_2$ are $k_1 \times k_0$ and $k_2 \times k_0$ matrices of parameters, respectively , and where the time-series processes  $\eta_{0it}$,  $\eta_{1it}$ and $\eta_{2it}$ are weakly stationary for every cross-section unit for $i = 1,...,N$. 

In terms of the identification and model specification, we employ the dynamic model for the variables at the individual level. Suppose that a cross section obtained at some point in time is available. The cross-section is sampled consisting of observations of $N$ cross-section units at time $t \in \mathbb{N}$. In particular, we consider the cross-section demeaned quantities
\begin{align}
Y_{it}^{*} = Y_{it} - \frac{1}{N} \sum_{i=1}^N Y_{it}, \ X_{1it}^{*} = X_{1it} - \frac{1}{N} \sum_{i=1}^N X_{1it}, \ \ X_{2it}^{*} = X_{2it} - \frac{1}{N} \sum_{i=1}^N X_{2it} 
\end{align}
For notational convenience the stacked $( k_1 + k_2 )-$dimensional stochastic variable $X_{it}^{*}$ is defined such that $X_{it}^{*} = \left( X_{1it}^{*\prime}, X_{2it}^{*\prime} \right)^{\prime}$ and the corresponding $( k_1 + k_2 ) \times k_0$ parameter matrix as $\gamma = \left( \gamma_1^{\prime}, \gamma_2^{\prime} \right)^{\prime}$. We can now specify the regression equation describing the cointegrating relations can be expressed in terms of the demeaned variables below
\begin{align}
Y_{it}^{*} = \gamma^{\prime} X_{it}^{*} + \eta_{0it}^{*}, \ \ i = 1,...,N \ \text{and} \ t \in \mathbb{N}, 
\end{align}
where $\eta_{0it}^{*} = \eta_{0it} - \frac{1}{N} \sum_{i=1}^N \eta_{0it}$. The corresponding cross-section OLS estimator is defined as below
\begin{align}
\hat{ \gamma }_{N,t} = \left( \sum_{i=1}^N  X_{it}^{*} X_{it}^{* \prime}  \right)^{-1} \left( \sum_{i=1}^N  X_{it}^{*} Y_{it}^{* \prime}  \right).
\end{align}
Moreover, by relevant assumptions the regressor $X_{it}^{*}$ is independent of the regression error $\eta_{0it}^{*}$ since the aggregate shocks have been removed from the variables. This immediately implies that $\hat{ \gamma }_{N,t}$ is an unbiased estimator of $\gamma$, that is, $\mathbb{E} \left( \hat{ \gamma }_{N,t} \right) = \gamma$. 
Then, the asymptotic behaviour as $N \to \infty$ of the cross-section estimator $\hat{ \gamma }_{N,t}$ is given by the following proposition.

%%-------------------------------------------------------------------------%%
\newpage

\begin{proposition}[\cite{madsen2005estimating}] 
Under regularity conditions the following hold: $\hat{ \gamma }_{N,t}$ is a consistent estimator of $\gamma$, that is, 
\begin{align}
\hat{ \gamma }_{N,t} \to_p \gamma, \ \text{as} \ N \to \infty. 
\end{align}
Then, the limiting distribution of $\hat{ \gamma }_{N,t}$ is given by 
\begin{align}
\sqrt{N} \left( \hat{ \gamma }_{N,t} - \gamma \right) \to_w \mathcal{N} \left( 0, \Omega \otimes \Sigma_t^{-1} \right) \ \text{as} \ N \to \infty. 
\end{align}
The asymptotic variance can be estimated consistently by using the following results:
\begin{align}
\frac{1}{N} \sum_{i=1}^N X_{it}^{*} X_{it}^{*\prime} &\to_p \Sigma_t \ \ \text{as} \ N \to \infty, 
\\
\frac{1}{N} \sum_{i=1}^N \left( Y_{it}^{*} - \hat{ \gamma }_{N,t}^{\prime} X_{it}^{*} \right) \left( Y_{it}^{*} - \hat{ \gamma }_{N,t}^{\prime} X_{it}^{*} \right)^{\prime} &\to_p \Omega \ \ \text{as} \ N \to \infty, 
\end{align}
\end{proposition}
Notice that since the regressor $X_{2it}$ is nonstationary when viewed as a time series, the asymptotic variance of $\sqrt{N} \left( \hat{ \gamma }_{N,t} - \gamma \right)$ depends on the point in time where the cross-section is obtained. We consider the cross-section estimator of $\gamma_2$ defined as the submatrix of $\hat{ \gamma }_{N,t}$ corresponding to the regressor $X_{2it}$. To be more specific let this estimator denoted by $\hat{ \gamma }_{2N,t}$ be the last $k_2$ rows in $\hat{ \gamma }_{N,t}$ and let $\Sigma^{22t}$ be the lower $k_2 \times k_2$ diagonal block matrix of $ \Sigma_t^{-1}$, that is, 
\begin{align}
\Sigma^{22t} = \left( \Sigma_{22t} - \Sigma_{21,t} \Sigma_{11,t}^{-1}  \Sigma_{12,t}^{-1} \right)^{-1}
\end{align}  
where $\Sigma_{t}$ is decomposed according to $X_{1it}$ and  $X_{2it}$ as below
\begin{align}
\Sigma_{t} 
= 
\begin{pmatrix}
\Sigma_{11,t} & \Sigma_{12,t} \\
\Sigma_{21,t} & \Sigma_{22,t} 
\end{pmatrix}
\end{align}
Then according to Proposition 1 the limiting distribution of $\hat{ \gamma }_{2,N,t}$ is given by 
\begin{align}
\sqrt{N} \left( \hat{ \gamma }_{2,N,t} - \gamma_2 \right) \to_w \mathcal{N} \left( 0, \Omega \otimes \Sigma^{22t} \right) \ \text{as} \ N \to \infty. 
\end{align}

\begin{assumption}
For $a \in \mathbb{R}$ the diagonal matrix $F_t$ is defined in the following way: 
\begin{align}
F_t 
=
\begin{pmatrix}
I_{k_1}  &  0 \\
0        &  t^{a} I_{k_2}
\end{pmatrix}
\end{align}
and the following condition is satisfied
\begin{align}
\underset{ t \to \infty }{ \text{lim} } \left( F_t \Sigma_t F_t \right), \ \ \text{is positive definite}. 
\end{align}
\end{assumption}

%%-------------------------------------------------------------------------%%
\newpage 

\subsection{Functional Coefficient Panel Modeling with smoothing covariates}

Consider the fixed effects functional coefficient panel data model (see, \cite{phillips2022functional}) 
\begin{align}
y_{it} = \alpha_i + \beta ( z_{it} )^{\prime} x_{it} + u_{it}, \ \ \ i = 1,...,N, \ \ t = 1,...,R    
\end{align}
where $x_{it}$ is a $p-$vector of regressors, $z_{it}$ is a $q-$vector of covariates that determine the (random) coefficients $\beta( z_{it} ) = \big( \beta_1 ( z_{it} ),...,  \beta_p ( z_{it} ) \big)$, the $\alpha_i$ are individual fixed effects, and the error $u_{it}$ has zero mean and finite variance $\sigma_u^2$. Thus, we focus on the case where both $x_{it}$ and $z_{it}$ are exogenous. 

Moreover, \cite{phillips2022functional} developed asymptotic theory for the estimator $\hat{\beta}$ in various settings depending on whether $N$ is fixed or $N \to \infty$ and whether $\eta =0$ or $\eta \neq 0$. In all cases, it is presumed that $T \to \infty$. Results with $N \to \infty$ include both sequential limit $( N, T )_{ \mathsf{seq} } \to \infty$ theory, where $T \to \infty$ followed by $N \to \infty$, and joint limit $( N, T ) \to \infty$ theory, where $T, N$ pass to infinity together. Joint limit theory is obtained by following the double indexed limit theory but our results provide an important extension that covers cases of multiple convergence rates and possibly degenerate limit distributions. In consequence, the divergence rate of the cross-section sample size $N$ needs to be controlled in order to control the random bias contributed by estimation of $\beta_0 ( z_t )$ (see, \cite{phillips2022functional}). 

\subsubsection{Testing constancy of the functional coefficients}

Statistical inference for the econometric specification of \cite{phillips2022functional} corresponds to  testing specific parametric forms of functional coefficients. Thus, the relevant hypothesis concerning the functional coefficient $\beta(z)$ is whether this vector of coefficient functions can be treated as a constant vector. In particular, tests of such hypotheses can be constructed by examining the discrepancy between the nonparametric estimate of $\beta_0$ and parametric estimate of $\beta_0$.  To distinguish the alternative from the null we further require some conditions, so that the function $g(z)$ is not a constant function. This kind of local alternative is commonly used in the study of nonparametric and semiparametric inference involving stationary and nonstationary data.  Moreover, to establish joint asymptotics for $\beta$ as $( N, T ) \to  \infty$ when $\eta \neq 0$, \cite{phillips2022functional}, take into account of the singularity that arises in the limiting signal matrix in the passage to joint asymptotics.  

\begin{remark}
Notice that the \textit{i.i.d} assumption on $\epsilon_t$ can be relaxed to allow for martingale difference innovations and to allow for some mild heterogeneity in the innovations without disturbing the limit theory in a material way. Denote the long-run variance of $\left\{ u_t \right\}_{ t \geq 1 }$ as $\boldsymbol{\Omega}_u = \sum_{ h = - \infty }^{ + \infty } \boldsymbol{\Sigma}_{uu}(h)$. Then, from the Wold decomposition, we have that $\boldsymbol{\Omega}_u = \boldsymbol{D}(1) \boldsymbol{\Sigma}_{ \epsilon \epsilon}    \boldsymbol{D}(1)^{\prime}$, which is positive definite because $\boldsymbol{D}(1)$ has full rank and $\boldsymbol{\Sigma}_{ \epsilon \epsilon}$ is positive definite. The fourth moment assumption is needed for the limit distribution of sample autocovariances in the case of misspecified transient dynamics (see, \cite{phillips2022functional}).   
\end{remark}
As expected, under general weak dependence assumptions on $u_t$, the simple reduced rank regression models are susceptible to the effects of potential misspecification in the transient dynamics. These effects bear on the stationary components in the system. In particular, due to the centering term, both the OLS estimator and the shrinkage estimator are asymptotically biased (see, \cite{phillips2022functional}).

%%-------------------------------------------------------------------------%%
\newpage

\subsection{A Panel Clustering Approach to Analyzing Bubble Behaviour}

Using an empirical illustration we can demonstrate the need to implementing a different cluster-based structure when regressors are near-unit root; in the case when there are persistent data and there is no correction due to endogeneity and high persistence. Assuming that there is a known cluster structure a common approach in the literature is to use a clustering algorithm. Current methodologies presented in the literature can be compared with new panel tests that make use of clustering individual time series into common groups reveal the additional discriminatory power obtained by grouping. Specifically, the clustered panel $t-$test introduced by \cite{liu2022panel} help to diagnose mildly explosive price behaviour in US city housing markets where individual time series tests reveal no evidence of such behaviour, confirming this way the discriminatory power gains that arise from cross section aggregation.

\subsubsection{Model Structure}

In order to capture explosive and mildly explosive behaviour in panels we used the following data generating process based on the time series model of PM (2007). 
\begin{align}
y_{it} &= \mu_i + \rho_{ \mathsf{g}_i } y_{i,t-1} + u_{it}, \ \ i = 1,.., n, \ t = 1,...,T,
\\
\rho_{\mathsf{g}_i} &= \left(  1 + \frac{ c_{ \mathsf{g}_i } }{ n^{\gamma}  }  \right)
\end{align}
Notice that the exponent rate $\gamma \in (0,1)$ and the scale coefficients $c_{ \mathsf{g}_i }$, both influence the extend of departure of the autoregressive coefficients $\rho_{ \mathsf{g}_i }$ from unity, and $\mathsf{g}_i$ denotes the group membership of individual $i$, for which the group structure is defined late. For nonstationary data of each cluster, it holds that the innovations $u_{it}$, follow a stationary linear process for each $i$ and the various variance estimates
\begin{itemize}
    \item[(i).] Long-run variances: $\bar{\omega}_i^2 = \sum_{h = - \infty}^{+\infty} \mathbb{E} \big( u_{it} u_{i,t-h} \big)$.

    \item[(ii).] One-sided long-run covariances: $\bar{\lambda}_i^2 = \sum_{h = 1}^{+\infty} \mathbb{E} \big( u_{it} u_{i,t-h} \big)$.

    \item[(iii).] Variances: $\bar{\sigma}_{iu}^2 = \mathbb{E} \big( u_{it}^2 \big)$, such that $ \bar{\omega}_i^2 = 2  \bar{\lambda}_i + \bar{\sigma}_{iu}^2$ for each individual unit $i$.
    
\end{itemize}

Latent group membership of the $\rho_{ \mathsf{g}_i }$ arises through the localizing scale parameters $\left\{ c_{ \mathsf{g}_i }  \right\}_{ i = 1}^n$. The framework we adopt lies between a homogeneous panel; in which case $c_{ \mathsf{g}_i }  \equiv c$ for all $i$, and a fully heterogeneous panel; in which case $c_{ \mathsf{g}_i } \neq c_{ \mathsf{g}_{\ell} }$ for  all $i \neq \ell$. In the paper of PCB et al, the authors assume a group structure involving a fixed number $G < n$ of unknown separate groups that are classified according to the scale parameter $c_{ \mathsf{g}_i }$. The group membership variables are given by the $\left\{  \mathsf{g}_i \right\}_{i=1}^n$ which maps individual units such that $i \in \left\{ 1,..., n \right\}$ into specific groups for which $j \in \left\{ 1,..., G \right\}$ with $G < n$ and allows for several possible midly explosive and mildly integrated groups together with a unit root group.

%%-------------------------------------------------------------------------%%
\newpage 

\subsection{SUR Representation of VAR Models with Explosive Roots}

Consider the following VAR-type model with explosive roots (see, \cite{chen2023seemingly})
\begin{align}
\boldsymbol{X}_t = \boldsymbol{R} \boldsymbol{X}_{t-1} + \boldsymbol{u}_t, \ \ \ \text{for} \ \  t = 1,...,n,   
\end{align}
where $\boldsymbol{X}_t$ is a $d-$dimensional vector with $\boldsymbol{X}_t = \left[ x_{1,t},..., x_{d,t}  \right]^{\top}$, while the initial value is set to $x_{i,0} = 0$ for $i \in \left\{ 1,..., k \right\}$ for simplicity.  The residual sequence $\boldsymbol{u}_t = \left[ u_{1,t},..., u_{d,t}  \right]^{\top}$ is assumed to be a martingale difference sequence with resepect to $\mathcal{F}_t = \sigma \big( \boldsymbol{u}_t, \boldsymbol{u}_{t-1},...  \big)$ satisfying
\begin{align}
\mathbb{E} \big[ \boldsymbol{u}_t \boldsymbol{u}_t^{\top}  \big] = \boldsymbol{\Sigma}_u      
\end{align}
with $\mathsf{Cov} ( u_{i,t}, u_{j,t} ) = \sigma_{i,j}$ for $i,j \in \left\{ 1,..., d \right\}$. Then, the autoregressive coefficient matrix is defined as $\boldsymbol{R} = \mathsf{diag} ( \rho_1,..., \rho_d )$. Moreover, we consider two cases: 

\begin{itemize}

\item[1.] \textbf{distinct explosive roots:} $\rho_i > 1$ for $i \in \left\{ 1,..., d \right\}$ and $\rho_i \neq \rho_j$, $i,j \in \left\{ 1,..., d \right\}$.

\item[2.] \textbf{common explosive root:} $\rho_i = \rho > 1$ for $i \in \left\{ 1,.., d \right\}$.
    
\end{itemize}

Notice that the OLS estimator of the above model with a common explosive root is inconsistent. In particular, the standardized sample variance matrix $\sum_{t=1}^n X_t X_t^{\top}$ is asymptotically singular.

\subsubsection{SUR regression estimate and asymptotics}

Let the $i-$th regression model be 
\begin{align}
x_{i,t} = \rho_i x_{i,t-1} + u_{i,t}    
\end{align}

Furthermore, define with $X_i= \big[ x_{i,1},...,  x_{i,n} \big]^{\top}$ which is an $( n \times 1 )$ vector. Let $A = [ \rho_1,..., \rho_d ]^{\top}$ and $U = [ U_1,..., U_d ]^{\top}$ such that 
\begin{align}
X_{-} = 
\begin{bmatrix}
X_{1 n-1}  &  0_{n \times 1}   &  \hdots   &  0_{n \times 1}  
\\  
0_{n \times 1} &  X_{2 n-1}  &   \hdots    &  0_{n \times 1} 
\\
\vdots & \vdots & \ddots & \vdots
\\
0_{n \times 1} & 0_{n \times 1} \hdots  &  \hdots    & X_{d n-1} 
\end{bmatrix}_{ ( n d \times d )  } 
\end{align}
The dependence structure of the model is given by $\mathsf{Var} ( \boldsymbol{U} ) = \boldsymbol{\Sigma}_u \otimes \boldsymbol{I}_n$ is a $\left( nd \times nd \right)$ matrix. 
\begin{align}
\widehat{\boldsymbol{A}}_{sur} =  \big[ \boldsymbol{X}_{n-1} \big( \boldsymbol{\Sigma}_u \otimes \boldsymbol{I}_n \big) \boldsymbol{X}_{n-1} \big]^{-1} \big[ \boldsymbol{X}_{n-1} \big( \boldsymbol{\Sigma}_u \otimes \boldsymbol{I}_n \big) \boldsymbol{X}_{n-1} \big]    
\end{align}
Therefore, in order to facilitate the development of the asymptotic theory, we begin by considering the relevant assumptions for the regressors and the error terms. 

%%-------------------------------------------------------------------------%%
\newpage

\subsection{Panel VAR Models}

The aim of this section is to present an asymptotic theory analysis of the performance of fixed $T$ consistent estimation techniques for PVARX$(1)$ model-based on observations in first differences which can be found in the framework proposed by \cite{juodis2018first}.  

Consider the following PVAR$(1)$ model specification defined as below: 
\begin{align}
\boldsymbol{y}_{i,t} = \boldsymbol{\eta}_{i,t} + \boldsymbol{\Phi} \boldsymbol{y}_{i,t-1} + \boldsymbol{\epsilon}_{i,t}, \ \ \ t = 1,...,N, \ \ \ t = 1,..., T,   
\end{align}
where $\boldsymbol{y}_{i,t}$ is an $( m \times 1 )$ vector and $\boldsymbol{\Phi}_{ m \times m}$ matrix of parameters to be estimated, where $\boldsymbol{\eta}_i$ is an $( m \times 1 )$ vector of fixed effects and $\boldsymbol{\epsilon}_{i,t}$ is an $(m \times 1)$ vector of innovations independent across $i$, with zero mean and constant covariance matrix $\boldsymbol{\Sigma}$. For various empirical applications the PVAR$(1)$ model specification might be two restrictive and incomplete. Therefore, in that case the original model specification can be extended by including strictly exogenous variables, the so-called PVARX$(1)$ model 
\begin{align}
\boldsymbol{y}_{i,t} = \boldsymbol{\eta}_{i} + \boldsymbol{\Phi} \boldsymbol{y}_{i,t-1} +    \boldsymbol{B} \boldsymbol{x}_{i,t}  +  \boldsymbol{\epsilon}_{i,t}, \ \ \ t = 1,...,N, \ \ \ t = 1,..., T,    
\end{align}
where $\boldsymbol{x}_{i,t}$ is a $( k \times 1 )$ vector of strictly exogenous regressors and $\boldsymbol{B}$ is an $( m \times k )$ unknown parameter matrix. Furthermore, econometric models with group specific spatial dependence can be formulated as a reduced form PVARX$(1)$ (see,  \cite{kripfganz2014unconditional} and \cite{verdier2016estimation}).   

\begin{assumption}[Effect stationary initial condition, see \cite{juodis2018first}]
\label{Assumption1}
The initial condition $\boldsymbol{y}_{i,0}$ is said to be effect stationary \textit{iff}
\begin{align}
\mathbb{E} \big[ \boldsymbol{y}_{i,0} | \boldsymbol{\eta}_i \big] = \big( \boldsymbol{I}_m - \boldsymbol{\Phi}_0 \big)^{-1} \boldsymbol{\eta}_i,    
\end{align}
implying that the process $\left\{ \boldsymbol{y}_{i,t} \right\}_{t=0}^T$ based on the data generating process that corresponds to a PVAR$(1)$ model is \textit{effect stationary} such that $\mathbb{E} \big[ \boldsymbol{y}_{i,t} | \boldsymbol{\eta}_i \big] = \mathbb{E} \big[ \boldsymbol{y}_{i,0} | \boldsymbol{\eta}_i \big]$ for all $\rho ( \boldsymbol{\Phi}_0 ) < 1$.  
\end{assumption}

\begin{remark}
Note that Assumption \ref{Assumption1} indicates a \textit{conditional moment stationarity} condition. On the other hand, \textit{effect nonstationarity} does not imply that the process $\left\{ \boldsymbol{y}_{i,t} \right\}_{t=0}^T$ is \textit{mean nonstationary}, such that $\mathbb{E} [ \boldsymbol{y}_{i,t} ] \neq \mathbb{E} [ \boldsymbol{y}_{i,0} ]$. In particular, mean nonstationarity is property of the underline stochastic process that crucially depends on $\mathbb{E} [  \boldsymbol{\eta}_i ]$ (see,  \cite{juodis2018first}).      
\end{remark}

\begin{definition}[Covariance stationary initial condition, see \cite{juodis2018first}]
The initial condition $\boldsymbol{y}_{i,0}$ is said to be covariance stationar $\textit{iff}$ the following two conditions hold: 
\begin{align}
\mathbb{E} \big[ \boldsymbol{y}_{i,0} | \boldsymbol{\eta}_i \big] = \big( \boldsymbol{I}_m - \boldsymbol{\Phi}_0 \big)^{-1} \boldsymbol{\eta}_i, \ \ \ \ \ \mathsf{Var} \big[ \boldsymbol{y}_{i,0} | \boldsymbol{\eta}_i \big]  = \sum_{t=0}^{\infty} \left( \boldsymbol{\Phi}_0^t \right) \boldsymbol{\Sigma}_0 \left( \boldsymbol{\Phi}_0^t \right)^{\top},    
\end{align}
implying that the process is \textit{covariance stationary}, such that the autocovariance function of $\left\{ \boldsymbol{y}_{i,t} \right\}_{t=0}^T$ is not time dependent.    
\end{definition}

%%-------------------------------------------------------------------------%%
\newpage

\subsubsection{OLS in First Differences}

The FD transformation can be employed in order to remove any individual effects that are present in the original model in levels. Therefore, the econometric specification is as below: 
\begin{align}
\Delta \boldsymbol{y}_{i,t} = \boldsymbol{\Phi} \Delta \boldsymbol{y}_{i,t-1} + \boldsymbol{B} \Delta \boldsymbol{x}_{i,t} + \Delta \boldsymbol{\varepsilon}_{i,t}    
\end{align}
Define the following variables as below:
\begin{align}
\Delta \boldsymbol{w}_{i,t} \equiv 
\begin{pmatrix}
\Delta \boldsymbol{y}_{i,t-1}
\\
\Delta \boldsymbol{x}_{i,t}
\end{pmatrix},
\ \ \ 
\boldsymbol{S}_N \equiv \left( \frac{1}{N} \sum_{i=1}^T \sum_{t=1}^T \Delta \boldsymbol{w}_{i,t} \boldsymbol{w}_{i,t}^{\top} \right), 
\end{align}
and denote with 
\begin{align}
\boldsymbol{\Sigma}_W = \underset{ N \to \infty  }{ \mathsf{plim} }  \boldsymbol{S}_N, \ \ \  \boldsymbol{\mathcal{Y}} \equiv \big( \boldsymbol{\Phi}, \boldsymbol{B} \big).   
\end{align}
Therefore, after pooling observations for all $t$ and $i$, we define the pooled panel FD estimator (FD-OLS) 
\begin{align}
\widehat{\boldsymbol{\mathcal{Y}}} = \boldsymbol{S}_N^{-1} \left( \frac{1}{N} \sum_{i=1}^T \sum_{t=1}^T \Delta \boldsymbol{w}_{i,t} \boldsymbol{w}_{i,t}^{\top} \right)     
\end{align}

\begin{remark}
According to \cite{juodis2018first}, similarly to the conventional FE transformation, the FD transformation introduces correlation between the explanatory variable $\Delta \boldsymbol{y}_{i,t-1}$ and the modified error term $\Delta \boldsymbol{\epsilon}_{i,t}$. Therefore, this estimator is considered to be inconsistent and an analytic form of the asymptotic bias exists. Specifically, the asymptotic bias is given as below: 
\begin{align}
\underset{ N \to \infty  }{ \mathsf{plim} }  \left( \widehat{\boldsymbol{\mathcal{Y}}} - \widehat{\boldsymbol{\mathcal{Y}}}_0 \right)^{\prime} = - ( T - 1 ) \boldsymbol{\Sigma}_W^{-1} 
\begin{bmatrix}
\boldsymbol{\Sigma}_0
\\
\boldsymbol{0}_{ k \times m }
\end{bmatrix}
\end{align}
\end{remark}

\subsubsection{No Exogenous Regressors}

In the econometric model without exogenous regressors the FD-OLS estimator is given as below: 
\begin{align}
\hat{\boldsymbol{\Phi}}_{\Delta} 
= 
\left( \frac{1}{N} \sum_{i=1}^N \sum_{t=1}^T \Delta \boldsymbol{y}_{i,t}  \Delta \boldsymbol{y}_{i,t}^{\prime} \right) \left( \frac{1}{N} \sum_{i=1}^N \sum_{t=1}^T \Delta \boldsymbol{y}_{i,t-1}  \Delta \boldsymbol{y}_{i,t-1}^{\prime} \right)^{-1}.  
\end{align}
Moreover, assuming that $\boldsymbol{y}_{i,0}$ is covariance stationary and as a consequence it holds that
\begin{align}
\boldsymbol{\Sigma}_{W} \equiv (T-1) \left\{  \boldsymbol{\Sigma}_0 + ( \boldsymbol{I}_m - \boldsymbol{\Phi}_0 ) \left[ \sum_{t=0}^{\infty} \left( \boldsymbol{\Phi}_0^t \right) \boldsymbol{\Sigma}_0  \left( \boldsymbol{\Phi}_0^t \right)^{\top} \right] ( \boldsymbol{I}_m - \boldsymbol{\Phi}_0 )^{\top} \right\}.    
\end{align}

%%-------------------------------------------------------------------------%%
\newpage

\begin{proposition}[Asymptotic Normality FDLS, see \cite{juodis2018first}]
Let DGP for covariance stationary $\boldsymbol{y}_{i,t}$ satisfy extensibility condition together with conditions of the abive Proposition. Then, it holds that 
\begin{align}
\sqrt{N} \left(  \hat{\boldsymbol{\phi}}_{fdls} - \boldsymbol{\phi}_0 \right) \overset{d}{\to} \mathcal{N}_m \big(  \boldsymbol{0}_{m^2}, \mathcal{S}  \big),   
\end{align}
where 
\begin{align}
\mathcal{S} &\equiv \big( \boldsymbol{\Sigma}_W^{-1} \otimes \boldsymbol{I}_m \big) \boldsymbol{\Xi}  \big( \boldsymbol{\Sigma}_W^{-1} \otimes \boldsymbol{I}_m \big), \ \ \ \boldsymbol{\Xi} \equiv \underset{ N \to \infty  }{ \mathsf{plim} }  \frac{1}{N} \sum_{i=1}^N \mathsf{vec} (  \mathcal{A}_i ) \mathsf{vec} (  \mathcal{A}_i )^{\top},
\\
\mathcal{A}_i  &\equiv   \sum_{t=1}^T  \left[ 2 \Delta \boldsymbol{y}_{i,t} + ( \boldsymbol{I}_m - \boldsymbol{\Phi}_0 ) \Delta \boldsymbol{y}_{i,t-1} \right] \Delta \boldsymbol{y}_{ i,t-1}^{\top}.   
\end{align}
\end{proposition}

\begin{proof}
Consider the following quantities:
\begin{align}
\bar{\boldsymbol{y}}_{i T-1} = \frac{1}{T} \sum_{t=1}^T \boldsymbol{y}_{i,t-1}  \ \ \ \text{and} \ \ \   \bar{\boldsymbol{y}}_{i T} = \frac{1}{T} \sum_{t=1}^T \boldsymbol{y}_{i,t} 
\end{align}
Moreover, we consider within group transformations, and thus the variables denoted by $\tilde{x}$ correspond to variables after within group transformation, such that, $\tilde{\boldsymbol{y}}_{i,t} = \boldsymbol{y}_{i,t} - \bar{\boldsymbol{y}}_i$, while $\ddot{\boldsymbol{x}}$ is used for variables after a quasi-averaging transformation. 

Define the following concentrated variables: 
\begin{align}
\dot{\boldsymbol{y}}_i 
&\equiv
\ddot{\boldsymbol{y}}_i - \left( \sum_{i=1}^N \ddot{\boldsymbol{y}}_i \Delta \boldsymbol{X}_i^{\top} \right) \left( \sum_{i=1}^N  \Delta \boldsymbol{X}_i \Delta \boldsymbol{X}_i^{\top} \right)^{-1}  \Delta \boldsymbol{X}_i,
\\
\dot{\boldsymbol{y}}_{iT-1}  
&\equiv
\ddot{\boldsymbol{y}}_{iT-1} - \left( \sum_{i=1}^N \ddot{\boldsymbol{y}}_{i,T-1} \Delta \boldsymbol{X}_i^{\top}  \right) \left( \sum_{i=1}^N  \Delta \boldsymbol{X}_i\Delta \boldsymbol{X}_i^{\top} \right)^{-1}  \Delta \boldsymbol{X}_i,
\\
\boldsymbol{y}_{i,t}^{\star}  
&\equiv
\tilde{\boldsymbol{y}}_{i,t} - \left( \sum_{i=1}^N \sum_{t=1}^T \tilde{\boldsymbol{y}}_{i,t} \tilde{\boldsymbol{x}}^{\top}_{i,t}  \right) \left( \sum_{i=1}^N \sum_{t=1}^T \tilde{\boldsymbol{x}}_{i,t} \tilde{\boldsymbol{x}}^{\top}_{i,t} \right)^{-1}  \tilde{\boldsymbol{x}}_{i,t},
\\
\boldsymbol{y}_{i,t-1}^{\star}  
&\equiv
\tilde{\boldsymbol{y}}_{i,t-1} - \left( \sum_{i=1}^N \sum_{t=1}^T \tilde{\boldsymbol{y}}_{i,t-1} \tilde{\boldsymbol{x}}^{\top}_{i,t}  \right) \left( \sum_{i=1}^N \sum_{t=1}^T \tilde{\boldsymbol{x}}_{i,t} \tilde{\boldsymbol{x}}^{\top}_{i,t} \right)^{-1}  \tilde{\boldsymbol{x}}_{i,t},
\end{align}
Therefore, using the the above concentrated variables, the concentrated log-likelihood function for $\boldsymbol{\vartheta}_{co} = \big( \boldsymbol{\varphi}^{\top}, \boldsymbol{\sigma}^{\top}, \boldsymbol{\theta}^{\top} \big)^{\top}$ is given by the following expression 
\begin{align*}
\ell_{co} \left( \boldsymbol{\vartheta}_{co} \right) 
&= 
- \frac{N}{2} \left\{  (T-1) \mathsf{log} | \boldsymbol{\Sigma} |  + \mathsf{trace} \left[ \boldsymbol{\Sigma}^{-1} \frac{1}{N} \sum_{i=1}^N \sum_{t=1}^T \left( \boldsymbol{y}_{i,t}^{\star}  - \boldsymbol{\Phi} \boldsymbol{y}_{i,t-1}^{\star} \right) \left( \boldsymbol{y}_{i,t}^{\star}  - \boldsymbol{\Phi} \boldsymbol{y}_{i,t-1}^{\star} \right)^{\top} \right] \right\}  
\\
&\ \ \
- \frac{N}{2} \left\{  \mathsf{log} | \boldsymbol{\Theta} |  + \mathsf{trace} \left[ \boldsymbol{\Theta}^{-1} \frac{T}{N} \sum_{i=1}^N \left( \dot{\boldsymbol{y}}_{i}  - \boldsymbol{\Phi} \dot{\boldsymbol{y}}_{i,T-1} \right) \left( \dot{\boldsymbol{y}}_{i}  - \boldsymbol{\Phi} \dot{\boldsymbol{y}}_{i,T-1} \right)^{\top} \right] \right\}.
\end{align*}
\end{proof}

%%-------------------------------------------------------------------------%%
\newpage

\section{Network Panel Data Model Estimation}
\label{Section5}

A large stream of literature has proposed econometric methodologies for capturing cross sectional dependence and heterogeneity via the use of dynamic panel models. Firstly, the particular literature has been significantly developed with the seminal paper of \cite{pesaran2006estimation}. Moreover, \cite{kapetanios2014nonlinear} present a framework for nonlinear panel models with cross-sectional dependence. Recently, in the spatial econometrics literature various methodologies have been proposed to model both spatial dependence and cross-sectional effects. \cite{Olmo2023} propose a network regression model with an estimated interaction matrix which incorporates both the cross-sectional dependence as well as the network dependence in the form of a metric distance between the set of regressors.

\subsection{Nonlinear Panel Data Model  with Cross-Sectional Dependence}

\subsubsection{Econometric Model}

We assume a sample of $T$ observations for $N$ agents. Then, we specify the following model
\begin{align}
x_{i,t} = \frac{ \rho }{ m_{i,t} } \sum_{j=1}^N \ell \big( \left| x_{i,t-1} - x_{j,t-1} \right|  \leq r  \big) x_{j,t-1} + \epsilon_{i,t}
\end{align}
for $t = 2,...,T$ and $i = 1,...,N$, where
\begin{align}
m_{i,t} = \sum_{j=1}^N \ell \big( \left| x_{i,t-1} - x_{j,t-1} \right|  \leq r  \big)  
\end{align}
This specification implies that $x_{i,t}$ is influenced by the cross-sectional average of a selection of $x_{j,t-1}$ and that in particular that the relevant $x_{j,t-1}$ are those that lie closest to $x_{i,t-1}$. The model involves a $K$ nearest neighbour mechanism, however all neighbours $x_{j,t-1}$ within a given threshold $r$ contribute equally. The formulation aims to capture the intuition that people are affected by those with whom they share common views or behaviour, reflecting the fact that similar agents are affected by similar effects (\cite{kapetanios2014nonlinear}, \cite{moon2017dynamic}).

\subsubsection{Estimation Methodology}

In this section, we discuss the estimation methodology of the aformentioned nonlinear model proposed by \cite{kapetanios2014nonlinear}. We consider the standard estimation procedure for a threshold model, whereby a grid of values for $r$ is constructed. Then, for all values on that grid the model is estimated by least squares to obtain estimates of the autoregression parameter, $\rho$.  Specifically, denoting with
\begin{align}
\widetilde{x}_{i,t} = \frac{1}{ m_{i,t} } \sum_{j = 1}^N \ell \big( \left| x_{i,t-1} - x_{j,t-1} \right|  \leq r  \big) \widetilde{x}_{j, t-1}  
\end{align}

%%-------------------------------------------------------------------------%%
\newpage 

with $\widetilde{x}_{i} = \left( \widetilde{x}_{i,1}, ....,     \widetilde{x}_{i,T - 1} \right)^{ \prime }$. The value of $r$ that minimizes the sum of of squared residuals is 
\begin{align}
\frac{1}{ N T } \sum_{ i=1 }^N \sum_{ t=1 }^T \hat{ \epsilon }_{i,t}^2 \left( \rho, r \right),
\ \ \hat{ \epsilon }_{i,t} \left( \rho, r \right) x_{i,t} - \frac{\rho}{ m_{i,t} }  \sum_{ j=1 }^N  \ell \big( \left| x_{i,t-1} - x_{j,t-1} \right|  \leq r  \big) x_{j,t-1}, 
\end{align}
\begin{assumption}[\cite{kapetanios2014nonlinear}]
$\epsilon_t$ is an $\textit{i.i.d}$ across $t$ and independent across $i$. Then, $\mathbb{E} \left( \epsilon_{i,t}   \right) = \sigma_{ \epsilon_i }^2$ and $\mathbb{E} \left( \epsilon_{i,t}^4 \right) < \infty$. For all $i$, the density of $\epsilon_{i,t}$ is bounded and positive over all compact subsets of $\mathbb{R}$. 
\end{assumption}

\begin{theorem}[\cite{kapetanios2014nonlinear}]
\label{theorem1111}
Let Assumption 1 hold, for $\epsilon_{i,t}$ in (2). Then, as long as $| \rho | < 1$, the least squares estimator of $\left( \rho, r \right)$ is consistent as $N$, $T \to \infty$. 
\end{theorem}
\begin{theorem}[\cite{kapetanios2014nonlinear}]
Let Assumption 1 hold, for $\epsilon_{i,t}$ in (2). Let $\left( \rho^0, r^0 \right)$ denote the true value of $\left( \rho, r \right)$. Then, as long as $| \rho | < 1$, $NT \left( \hat{r} - r^0  \right) = \mathcal{O}_p(1)$. Further, as long as $| \rho | < 1$, $\left( \hat{\rho} - \rho^0 \right)$ has the same asymptotic distribution as if $r^0$ was unknown. 
\end{theorem}
A panel data model with intercept is given by 
\begin{align}
x_{i,t} = \nu_i + \frac{ \rho }{ m_{i,t} } \sum_{j = 1}^N \ell \big( \left| x_{i,t-1} - x_{j,t-1} \right|  \leq r  \big) x_{j, t-1} + \epsilon_{i,t}, 
\end{align}  
where $\nu_i \sim \textit{i.i.d} \left( 0, \sigma_v^2 \right)$. A more general version is given by 
\begin{align}
x_{i,t} = \nu_i \zeta_t + \frac{ \rho }{ m_{i,t} } \sum_{j = 1}^N \ell \big( \left| x_{i,t-1} - x_{j,t-1} \right|  \leq r  \big) x_{j, t-1} + \epsilon_{i,t}, 
\end{align}  
for $r \times 1$ vectors of observable variables, $\zeta_t$, and coefficients $\kappa_i = \left( \kappa_{i,1}, ...., \kappa_{i,r} \right)$, where $\kappa_{i,j} \sim \textit{i.i.d} \left( 0 , \sigma^2_{ \kappa_j } \right)$, for $j = 1,...,r$. 
It is worth noticing, given our interest in the persistence properties of our class of models, that it is known from the literature on linear dynamic panel data models that high persistence can be generated by moderate values of $\rho$ in combination with a large variance for the individual specific effects. We now examine the properties of the least squares estimator above. The presence of $\nu_i$ induces endogeneity in standard panel AR models, leading to biased estimation of the autoregressive parameter for finite $T$, when standard panel least squares estimators, such as the within group estimator, are used. Endogeneity arises because consistency for least squares estimators requires that 
\begin{align}
\mathbb{E} \left( x_{i,t-1} \left( \epsilon_{i,t} - \frac{1}{T} \sum_{t=1}^T \epsilon_{i,t} \right) \right) = 0. 
\end{align}
Moreover, it is straightforward to allow for higher order, $p$, lags such that 
\begin{align}
x_{i,t} &= \rho_1 \widetilde{x}_{i,t-1} +  \rho_2 x^c_{i,t-1} + \epsilon_{i,t}, 
\end{align}

%%-------------------------------------------------------------------------%%
\newpage 

Similarly, we define with 
\begin{align}
\widetilde{x}_{i,t-1} &= \frac{1}{ m_{i,t} } \sum_{ j = 1 }^N I \left( \left| x_{i,t-1} - x_{j,t-1} \right| \leq r \right) x_{j,t-1}, 
\\
\widetilde{x}^c_{i,t-1} &= \frac{1}{ N - m_{i,t} } \sum_{ j = 1 }^N I \left( \left| x_{i,t-1} - x_{j,t-1} \right| > r \right) x_{j,t-1}, 
\end{align}
are the cross-section averages associated with the group of neighbours and non-neighbours respectively. The particular model is more relevant in the case where we are interested to model heterogenous interactions. Furthermore, another important issue is how best to modify the basic model to decompose the slope parameter, $\rho$, into an own effect and a neighbour effect. This case, can be captured with
\begin{align}
\widetilde{x}_{i,t-1} &= \rho_0 x_{ i, t - 1 } + \rho_1 x^{*}_{ i, t-1 } + \epsilon_{i,t}  
\\
x^{*}_{ i, t-1 } &=  \frac{1}{ m_{i,t} - 1 } \sum_{j=1, j \neq i}^N I \left( \left| x_{ i, t - 1 } - x_{ j, t - 1 }  \right| \leq r \right) x_{ j, t - 1 },    
\end{align}
Similarly a time-space recursive model is formulated as below 
\begin{align}
x_{i,t} = \rho_0 x_{i,t-1} + \rho_1 \sum_{ j=1, j \neq i}^N w_{ij} x_{j, t-1} + \epsilon_{i,t}, 
\end{align}
where the weights are given by 
\begin{align}
w_{ ij } = \frac{ d_{ij}^{-2} }{ \sum_{j=1}^N  d_{ij}^{-2} }, \ \ \ d_{ij} = \left| x_{i,t-1} - x_{j,t-1} \right|, \ \ \ \text{with} \ w_{ii} = 1. 
\end{align}
The estimation of the last model can be conducted consists of a two step estimation procedure. First, the consistent estimate of $r$ is obtained from (2), then we construct the weights and estimate the model by least squares. Notice that the above modelling approaches involved threshold mechanisms for constructing the unit-specific cross-sectional averages. But as discussed in Section 1, the class of models we wish to propose is much more general. In particular, we consider for example models of the form
\begin{align}
\label{model111}
x_{i,t} = \rho \sum_{ j = 1}^N   \frac{ \displaystyle w \left( \left| x_{i, t-1} -   x_{j, t-1} \right| ; \gamma \right) x_{j,t-1} }{ \displaystyle \sum_{ j = 1}^N w \left( \left| x_{i, t-1} -   x_{j, t-1} \right| ; \gamma \right)} + \epsilon_{i,t}, 
\end{align}
where $\gamma$ is a finite-dimensional vector of parameters and $w \left( x ; \gamma \right)$ is a positive twice differentiable integrable function such as the exponential function $\text{exp} \left( - \gamma x^2 \right)$.

%%-------------------------------------------------------------------------%%
\newpage
 
\begin{theorem}[\cite{kapetanios2014nonlinear}]
\label{theorem456}
Let Assumption 1 hold for $\epsilon_{i,t}$ where $w \left( x ; \gamma \right)$ is a positive twice differentiable integrable function. Then, as long as $| \rho | < 1$, the nonlinear least squares estimator of $\left( \rho, \gamma \right)$ is $( NT )^{1 / 2}$ is consistent and asymptotically normal as $N, T \to \infty$
\end{theorem}

Therefore, it can be also shown that 
\begin{align}
\mathbb{E} \left[ \left( \sum_{ j = 1}^N   \frac{ \displaystyle w \left( \left| x_{i, t-1} -   x_{j, t-1} \right| ; \gamma \right) x_{j,t-1} }{ \displaystyle \sum_{ j = 1}^N w \left( \left| x_{i, t-1} -   x_{j, t-1} \right| ; \gamma \right)} \right) \left( \epsilon_{i,t} - \frac{1}{T} \sum_{t=1}^T \epsilon_{i,t} \right) \right] = \mathcal{O} \left( \frac{1}{NT} \right),
\end{align}
which implies that the within estimator is valid for estimating \eqref{model111} when fixed effects are incorporating in \eqref{model111}. For example, the model relates to a univariate process $y_t$, $t = 1,...., T$ and the associated finite-dimensional covariates, denoted by $p_t$. Moreover, if the data are ordered, as in the case of time series, then their model is given by
\begin{align}
y_t = \frac{ \displaystyle \sum_{j=1}^{t-1} w \left( p_j, p_t \right) y_j }{ \displaystyle \sum_{j=1}^{t-1} w \left( p_j, p_t \right)} + \epsilon_t, 
\end{align}
The above model has the property that it can incorporate the similarity/distance $w$. Another set of extensions to the above models, arises by introducing other variables or lags to the model, either linearly 
\begin{align}
x_{i,t} =  \frac{ \rho }{ m_{i,t} } \sum_{j=1}^N \ell \left( \left| x_{i,t-1} -  x_{j,t-1} \right| \leq r \right) x_{j,t-1}  + \beta^{\prime} z_{i,t} + \epsilon_{i,t}, 
\end{align}
where $z_{i,t} = \left( z_{1,i,t},..., z_{k,i,t} \right)^{\prime}$ is a set of exogenous stationary variables, that is, $\beta = \left(   \beta_1,...., \beta_k \right)^{\prime}$, or nonlinearly as below
\begin{align}
x_{i,t}
=
\frac{ \rho }{ m_{i,t} } &\sum_{ j = 1}^N \ell \left( \left| x_{i,t-1} -  x_{j,t-1} \right| \leq r \right) x_{j,t-1}
+  
\sum_{ s = 1}^{ k } \left[ \frac{ \beta_s }{ m_{i,t} } \sum_{ j = 1}^{ k } \ell \left( \left| x_{i,t-1} -  x_{j,t-1} \right| \leq r \right) z_{s,j,t} \right] + \epsilon_{i,t} 
\end{align} 

\subsubsection{Main Limit Results}

\begin{lemma}[\cite{kapetanios2014nonlinear}]
Let $\left\{ \left\{ x_{i,t} \right\}_{i=1}^N \right\}_{ t = 1}^N$ follow (2). Then, for all $N_0 \leq N$, there exists $T_0$ such that for all $T > T_0$. Then, $\left\{ \left\{ x_{i,t} \right\}_{i=1}^{ N_0 } \right\}_{ t = T_0 }^T$ is geometrically ergodic and asymptotically stationary, as long as $| \rho | < 1$. Further, if $\text{sup}_{ i \leq N_0 } \mathbb{E} \left( \epsilon_{i,t}^4  \right) < \infty$, and $\text{sup}_{ i \leq N_0 } \mathbb{E} \left( x_{i,t}^4 \right) < \infty$. 
\end{lemma}

%%-------------------------------------------------------------------------%%
\newpage

\begin{proof}
We can write the part of (2), relevant for $\left\{ x_{i,t} \right\}_{ i = 1 }^{ N_0 }$, as below
\begin{align}
x_t^{(N_0)} = \Phi_t^{(N_0)} x_{t-1}^{(N_0)} + \epsilon_t^{(N_0)}, 
\end{align}
where $x_t^{(N_0)} = \left( x_{1,t},...., x_{N_0 ,t} \right)^{\prime}$, $\epsilon_t^{(N_0)} = \left( \epsilon_{1,t},...,   \epsilon_{ N_0 ,t} \right)^{\prime}$ and $\Phi_t^{ (N_0 ) } = \left[   \Phi_{i,j,t} \right]$, where 
\begin{align}
\Phi_{i,j,t} = \frac{ \rho }{ m_{i,t} } \ell \left( \left| x_{i,t-1} -  x_{j,t-1} \right| \leq r \right)
\end{align}
Then, we have that $\text{sup}_t \lambda_{ \text{max} } \left( \Phi_t^{(N_0)} \right) < 1$, where $\lambda_{ \text{max} } \left( \Phi_t^{(N_0)} \right)$  denotes the maximum eigenvalue of $\Phi_t^{(N_0)}$ in absolute value. We also have that the the supremum over $t$ for this set of maximum eigenvalues is bounded from above by the supremum over $t$ of the row sum norm of $\left( \Phi_t^{(N_0)} \right)$.  
\end{proof}
\begin{lemma}[\cite{kapetanios2014nonlinear}]
Let $\left\{ \left\{ x_{i,t} \right\}_{i=1}^N \right\}_{t=1}^T$ is given by $x_{i,t} = q_{i, t-1 } + \epsilon_{i,t}$, such that the column sum norm of the variance-covariance matrix of $\epsilon_t^{(N)}$ is $O(1)$ as $N \to \infty$. Moreover, the column sum norm of the variance-covariance matrix of $x_t^{(N)}$ is $\mathcal{O}(N)$ if 
\begin{enumerate}
\item[(i)] $q_{i,t-1}$ is stationary,

\item[(ii)] there is $\delta > 0$ for all $N$, there exist units $i,j = 1,..., \delta N$ such that
\begin{align}
0 < \text{lim}_{ N \to \infty} \ \text{sup}_{ i = 1,..., \delta N} \ \text{Var} \left( q_{i, t-1} \right) < \infty
\end{align}

\item[(iii)] There is a $\delta > 0$, for all $N$, there exist units $i,j = 1,..., \delta N$, such that $\text{Cov} \left( q_{i,t-1}, q_{j, t-1} \right) \neq 0$.  
\end{enumerate}
 
\end{lemma}

\begin{lemma}[\cite{kapetanios2014nonlinear}]
Let $\displaystyle \left\{ \left\{ x_{i,t} \right\}_{ i = 1 }^{ N } \right\}_{ t=1 }^T$ be given by $x_{i,t} = \displaystyle \frac{ \rho }{ N } \sum_{ j=1 }^N x_{ j, t-1 } + \epsilon_{i,t}$.  
\end{lemma}

\begin{proof}
To prove this theorem, we use the second part of Lemma 2, which can be written as $x_t = \rho \tilde{x}_{t-1} + \epsilon_t = \nu + \rho \Phi x_{t-1} + \epsilon_t$, where $\displaystyle \tilde{x}_{t-1} = \frac{1}{N} \sum_{ j=1 }^N x_{j, t-1}$, $\displaystyle \Phi = \frac{1}{N} \mathbf{1} \mathbf{1}^{ \prime }$ and $\mathbf{1} = \left( 1,...,1 \right)^{ \prime }$. Since $\Phi$ is idempotent, hence we have  
\begin{align*}
x_t = 
\rho^t \Phi x_0 + \epsilon_t + \Phi \sum_{i = 1}^{t-1} \rho^i \epsilon_{t-i} 
= 
\rho^t \Phi x_0 + \epsilon_t + \mathbf{1} \left[ \frac{1}{N} \sum_{j=1}^N \xi_{j,t} \right], 
\end{align*}  
where $\xi_{j,t} = \sum_{i=1}^{ t-1 } \rho^i \epsilon_{j, t-i}$. But  it is straightforward to show that 
\begin{align}
\underset{ N \to \infty }{ \text{lim} } \text{Var} \left( \frac{1}{N} \sum_{ j = 1}^N \xi_{j,t} \right) = 0,
\end{align}

\end{proof}

%%-------------------------------------------------------------------------%%
\newpage

\begin{lemma}[\cite{kapetanios2014nonlinear}]
Let $\displaystyle \left\{ \left\{ x_{i,t} \right\}_{ i = 1 }^{ N } \right\}_{ t=1 }^T$ follow the model below
\begin{align}
x_{i,t} = \rho \sum_{ j = 1}^N   \frac{ \displaystyle w \left( \left| x_{i, t-1} -   x_{j, t-1} \right| ; \gamma \right) x_{j,t-1} }{ \displaystyle \sum_{ j = 1}^N w \left( \left| x_{i, t-1} -   x_{j, t-1} \right| ; \gamma \right)} + \epsilon_{i,t}, 
\end{align}
then, for every $N_0 \leq N$, there exists $T_0$ such that for all $T > T_0$, $\displaystyle \left\{ \left\{ x_{i,t} \right\}_{ i = 1 }^{ N_0 } \right\}_{ t= T_0 }^T$ is geometrically ergodic and asymptotically stationary, as long as $| \rho | < 1$.  
\end{lemma}

\begin{proof}
Proceeding as in the proof of Lemma 1, we can write part of (22) relevant for $\displaystyle \left\{ x_{i,t} \right\}_{ i = 1 }^{ N_0 }$ as below
\begin{align}
x_t = \Phi_t^{ w, (N_0)} x_{t-1} + \epsilon_t,   
\end{align}
where $\Phi_t^{ w, (N_0)} = \left[ \Phi_{i,j,t}^w \right]$ and
\begin{align}
\Phi_{i,j,t}^w = \frac{ \rho w \left( \left| x_{i, t-1} - x_{j, t-1} \right| ; \gamma \right) }{ \sum_{j=1}^N w \left( x_{i, t-1} - x_{j, t-1} \right) ; \gamma }
\end{align}
\end{proof}

\begin{lemma}[\cite{kapetanios2014nonlinear}]
Let $\displaystyle \left\{ \left\{ x_{i,t} \right\}_{ i = 1 }^{ N } \right\}_{ t=1 }^T$ follow the model below
\begin{align}
x_{i,t} = \nu_i + \frac{ \rho }{ m_{i,t} } \sum_{j = 1}^N \ell \big( \left| x_{i,t-1} - x_{j,t-1} \right|  \leq r  \big) x_{j, t-1} + \epsilon_{i,t}, 
\end{align}  
where $v_i \sim \textit{i.i.d} \left( 0, \sigma_v^2 \right)$. Then, there exits $T_0$ such that for all $T > T_0$, 
\begin{align}
\mathbb{E} \left( \left[  \frac{\rho}{ m_{i,t} } \sum_{j=1}^N \ell \left( \left| x_{i, t-1} - x_{j, t-1} \right| \leq r \right) x_{j, t-1} \right] \left( \epsilon_{i,t} - \bar{\epsilon}_i  \right) \right) = \mathcal{O} \left( \frac{1}{NT} \right). 
\end{align}  
\end{lemma}

\begin{proof}
We establish the result for $r = \infty$. Then, the result follows by Lemma 1 and the assumption that the stationary density of $\left\{ x_{i,t} \right\}_{ i = 1}^{ N_0 }$ is positively uniformly over $N_0$, since this assumption implies that there exists $T_0$ such that for all $T > T_0$, and uniformly over $i$, the expected number of $j$ such that $\ell \left( \left| x_{i, t-1} - x_{j, t-1}    \right| \leq r \right) = 1$ for any $t$, is a non-zero proportion of $N_0$, for all $N_0$. To see this, note that the last statement is equivalent to the statement that, uniformly over $i$ and $j$, $\mathbb{P} \left( \left| x_{i, t-1} -  x_{j, t-1} \right| \leq r \right) \geq c > 0$ for some constant $c$. 
\end{proof}

%%-------------------------------------------------------------------------%%
\newpage

\begin{lemma}[\cite{kapetanios2014nonlinear}]
Let $\displaystyle \left\{ \left\{ x_{i,t} \right\}_{ i = 1 }^{ N } \right\}_{ t=1 }^T$ follow the model below
\begin{align}
x_{i,t} =  \sum_{s=1}^p \left[ \frac{ \rho_s }{ m_{i,t,s} } \sum_{j=1}^N \ell \left( \left| x_{i,t-s} -  x_{j,t-s} \right| \leq r \right) x_{j,t-s} \right] + \epsilon_{i,t}, 
\end{align}
where 
$m_{i,t,s} = \sum_{j=1}^N I \left( \left| x_{i,t-s} - x_{j,t-s}   \right| \leq r \right)$. 
Then, for all $N_0 \leq N$, there exists $T_0$ such that for all $T > T_0$, $\displaystyle \left\{ \left\{ x_{i,t} \right\}_{ i = 1 }^{ N } \right\}_{ t= T_0 }^T$ is \textit{geometrically ergodic} and \textit{asymptotically stationary} as long as $p \sum_{ i=1 }^p | \rho_s | < 1$. 
\end{lemma}

\begin{proof}
For the higher lag order autoregressive models, we write the model in a companion form. Therefore, we write the $\left\{ x_{i,t} \right\}_{ i = 1 }^{ N_0 }$ as below
\begin{align}
x_{t}^{\left( p, N_0 \right)} &= \Phi_t^{\left( p, N_0 \right)} x_{t-1}^{\left( p, N_0 \right)} + \epsilon_t^{\left( p, N_0 \right)},  
\\
x_{t}^{\left( p, N_0 \right)} &= \left( x_{1,t},..., x_{N_0, t },..., x_{ 1 , t - p },...., x_{ N_0 , t - p } \right)^{ \prime } , \ \ \epsilon_t^{ (N_0) } = \left( \epsilon_{1,t},...., \epsilon_{ N_0 ,t}, 0 ,...., 0 \right)^{ \prime } 
\end{align}
and
\begin{align}
\Phi_t^{\left( p, N_0 \right)} 
= 
\begin{bmatrix}
\tilde{ \Phi }_t^{\left( 1, N_0 \right)}  & \tilde{ \Phi }_t^{\left( 2, N_0 \right)} & \ldots & \tilde{ \Phi }_t^{\left( p, N_0 \right)} \\
I & 0 & \ldots & 0 
\\
0 & \ldots  & I & 0
\end{bmatrix}
\end{align}
such that $\tilde{ \Phi }_t^{\left( s, N_0 \right)} = \left[ \tilde{ \Phi }_{ i, j, t }^{\left( s \right)} \right]$, for $s = 1,...,p$ and $\tilde{ \Phi }_{ i, j, t }^{\left( s, \right)} = \frac{ \rho_s }{ m_{i,t,s} } \ell \left( \left| x_{i, t-s} - x_{j, t-s} \right| \leq r \right) x_{j, t-s}$. Therefore, it is sufficient to show that the row sum norm of $\left(    \tilde{ \Phi }_t^{\left( 1, N_0 \right)}, ... , \tilde{ \Phi }_t^{\left( p, N_0 \right)}\right)$ is bounded from above by one. This requires that $ p \sum_{ s = 1}^p | \rho_s | < 1$, proving the result.  
\end{proof}

\paragraph{ Proof of Theorem \ref{theorem1111}}

Consider the model 
\begin{align}
x_{i,t} = \frac{\rho}{ m_{i,t} } \sum_{j=1}^N \ell \big( \big| x_{i, t-1} - x_{j,t-1} \big| \leq r \big) x_{j,t-1} + \epsilon_{i,t}
\end{align}
for $t= 2,...,T$, $i = 1,...,N$ where
\begin{align}
m_{i,t} = \sum_{j=1}^N \ell \big( \big| x_{i, t-1} - x_{j,t-1} \big| \leq r \big)
\end{align}
and $\left\{ \epsilon_{i,t} \right\}_{ t = 1}^T$ is an error process. 

\begin{proof}
We prove consistency of the least squares estimator of $\rho$ and $r$ which are the parameters of interest given the above econometric specification. We define with $x_{ ij, t-s } = \left| x_{i, t-s} -  x_{j, t-s} \right|$ and $\mathcal{F}_{t-1} = \sigma \left(  x_{1,t-1},...., x_{N,t-1}, x_{1, t-2}, ...., x_{N, t-2}, .... \right)$. Recall that $\rho^0$ and $r^0$ denote the true values of $\rho$ and $r$ and let $\mathbb{E}_{ \rho, r} \left( . | t-1 \right)$ denote the corresponding expectation conditional on $\mathcal{F}_{t-1}$.   
\end{proof}

%%-------------------------------------------------------------------------%%
\newpage

\paragraph{Proof of Theorem \ref{theorem456}}

We prove that the NLS estimator of $\left( \rho^0, \gamma^0 \right)$ denoted by $\left( \rho^0, \gamma^0 \right)$ is consistent and asymptotically normal. Notice that for consistency we need to establish the following two conditions: 

\begin{enumerate}
\item[(\textbf{C1})] We need to show that the data $x_{i,t}$, are geometrically ergodic and hence asymptotically covariance stationary. 

\item[(\textbf{C2})] We need to show that the limiting objective function is minimised at the true parameter values. 
\end{enumerate}

\underline{Condition C2:} 

\begin{align}
\mathbb{E} \big( x_{i,t} - \mathbb{E}_{ \rho^0, r^0 } \left( x_{i,t} | t - 1 \right) \big)^2 < \mathbb{E} \big( x_{i,t} - \mathbb{E}_{ \rho, r } \left( x_{i,t} | t - 1 \right) \big)^2 , \ \forall \ ( \rho , r ) \neq ( \rho^0, r^0 )
\end{align}
\begin{align}
\underset{ \delta \to 0 }{ \text{lim} } \ \mathbb{E} \left( \underset{ ( \rho, r ) \in B \left( ( \rho^0, r^0 ), \delta  \right) }{ \text{sup} } \ \big| \mathbb{E}_{ \rho^0, r^0 } \left( x_{i,t} | t - 1 \right) -  \mathbb{E}_{ \rho, r } \left( x_{i,t} | t - 1 \right) \big| \right) = 0,
\end{align}

where $B( \alpha, \beta )$ is an open ball of radius $b$ centered around $a$, is satisfied. These three conditions together imply the uniform convergence of the objective function given by 
\begin{align}
S( \rho, r ) = \frac{1}{NT} \sum_{i=1}^N \sum_{t=1}^T \left( x_{i,t} - \frac{\rho}{ m_{i,t} } \sum_{j=1}^N  \ell \big( \left| x_{i,t-1} - x_{j,t-1} \right|  \leq r  \big) x_{j, t-1}  \right)^2,
\end{align}
to the limit objective function which is the key to establishing consistency. 

\begin{proof}
Let
\begin{align*}
Q( \rho, \gamma ) := \frac{1}{NT} \sum_{ i=1 }^N  \sum_{ t=2 }^T \left( x_{i,t} - \sum_{j=1}^N  \frac{ \displaystyle w \big( \big| x_{i, t-1} - x_{j, t-1} \big| ; \gamma \big) x_{j, t-1} }{ \displaystyle \sum_{ j=1 }^N w \big( \big| x_{i, t-1} - x_{j, t-1} \big| ; \gamma \big) } \right)^2 
\end{align*}
For asymptotic normality under the assumption that $\left( \rho^0, \gamma^0 \right)$ lies in the interior of the parameter space and $w(.,.)$ is twice differentiable and integrable, it is sufficient to show that 
\begin{align*}
\frac{1}{ \sqrt{NT} } \sum_{ i=1 }^N  \sum_{ t=2 }^T \left( \sum_{ j=1 }^N \frac{ \displaystyle \rho^0 \frac{\partial w}{ \partial \gamma } \big( \big| x_{i, t-1} - x_{j, t-1} \big| ; \gamma^0 \big) x_{j, t-1} \epsilon_{i,t} }{ \displaystyle \sum_{ j=1 }^N w \big( \big| x_{i, t-1} - x_{j, t-1} \big| ; \gamma^0 \big) } \right) \to \mathcal{N} \left( 0, \mathbb{W}_1 \right),
\end{align*}

%%-------------------------------------------------------------------------%%
\newpage

and that, 
\begin{align*}
\frac{1}{ \sqrt{NT} } \sum_{ i=1 }^N  \sum_{ t=2 }^T \left( \sum_{ j=1 }^N \frac{ \displaystyle w \big( \big| x_{i, t-1} - x_{j, t-1} \big| ; \gamma^0 \big) x_{j, t-1} \epsilon_{i,t} }{ \displaystyle \sum_{ j=1 }^N w \big( \big| x_{i, t-1} - x_{j, t-1} \big| ; \gamma^0 \big) } \right) \to \mathcal{N} \left( 0, W_2 \right),
\end{align*}
where
\begin{align}
\mathbb{W}_1 &= \underset{ N \to \infty }{ \text{lim} } \mathbb{E} \left\{ \left[ \frac{1}{\sqrt{N} } \sum_{i=1}^N \left( \sum_{ j=1 }^N \frac{ \displaystyle \rho^0 \frac{\partial w}{ \partial \gamma } \big( \big| x_{i, t-1} - x_{j, t-1} \big| ; \gamma^0 \big) x_{j, t-1} \epsilon_{i,t} }{ \displaystyle \sum_{ j=1 }^N w \big( \big| x_{i, t-1} - x_{j, t-1} \big| ; \gamma^0 \big) }  \right)  \right]^2  \right\}  
\\
\mathbb{W}_2 &= \underset{ N \to \infty }{ \text{lim} } \mathbb{E} \left\{ \left[ \frac{1}{\sqrt{N} } \sum_{i=1}^N \left( \sum_{ j=1 }^N \frac{ \displaystyle w \big( \big| x_{i, t-1} - x_{j, t-1} \big| ; \gamma^0 \big) x_{j, t-1} \epsilon_{i,t} }{ \displaystyle \sum_{ j=1 }^N w \big( \big| x_{i, t-1} - x_{j, t-1} \big| ; \gamma^0 \big) }  \right)  \right]^2  \right\}  
\end{align}

and that $plim_{ N, T \to \infty} \left( \nabla^2 Q \left( \rho, \gamma \right) \right)^{- 1}$, exists where
\begin{align}
\nabla^2 Q \left( \rho, \gamma \right)
=
\begin{pmatrix}
\displaystyle \frac{ \partial^2 Q}{ \partial \rho^2 } \ \ & \ \ \displaystyle \frac{ \partial^2 Q}{ \partial \rho \partial \gamma     } 
\\
\\
\displaystyle \left( \frac{ \partial^2 Q}{ \partial \rho \partial \gamma} \right)^{\prime}  \ \ & \ \ \displaystyle \frac{ \partial^2 Q}{ \partial \gamma^{\prime}  \partial \gamma}
\end{pmatrix}.  
\end{align}
We focus on the term 
\begin{align}
w_{i,t} = \sum_{ j = 1}^N \frac{ \displaystyle \rho^0 \frac{\partial w}{ \partial \gamma } \big( \big| x_{i, t-1} - x_{j, t-1} \big| ; \gamma^0 \big) x_{j, t-1} }{ \displaystyle \sum_{ j=1 }^N w \big( \big| x_{i, t-1} - x_{j, t-1} \big| ; \gamma^0 \big) }  
\end{align}
By Lemma 8 which implies that $w_{i,t}$ has finite variance, uniformly over $i$, the fact that $w_{i,t}$ and $\epsilon_{i,t}$ are independent, and the fact that  $\epsilon_{i,t}$ has finite variance, uniformly over $i$, by assumption, it follows that $\left\{ w_{i,t} \epsilon_{i,t} \right\}_{ i = 1}^N$ is a martingale difference with finite second moments. Therefore, $w_t = \frac{1}{\sqrt{N} } \sum_{ i=1 }^N w_{i,t} \epsilon_{i,t}$ has zero mean and finite second moments for all $N$. Moreover, by the independence of $\epsilon_{i,t}$ across $t$, it follows that $\left\{ w_t \right\}_{ t = 1}^T$ is a martingale difference sequence. Hence, a martingale difference CLT holds for $w_t$ proving the above result.   
\end{proof}

%%-------------------------------------------------------------------------%%
\newpage

\subsection{A Network Generated Regression Model}

In this section, we discuss in details the framework proposed by \cite{Olmo2023}. Consider the cross-sectional regression model with $N$ the number of units
\begin{align}
Y = \mathbb{X} \gamma + \epsilon, 
\end{align}
where $Y$ is the demeaned outcome variable, and $\mathbb{X} = \left[ X_1,..., X_L \right]$ is a vector of exogenous demeaned covariates, and    $\gamma$ the vector of slope coefficients associated to $\mathbb{X}$. More precisely, \cite{Olmo2023} develop a novel econometric framework for network dependence that allows for exogenous spillover effects between the cross-sectional units. To do this,  \cite{Olmo2023} define a distance measure $d: \mathbb{Z} \times \mathbb{Z} \to \mathbb{R}^{+}$, with $\mathbb{Z}$ a set of elements that reflect the network features of the model. For example, let $\mathbf{z}_i = \left( x_{1i},...,  x_{Li}   \right)$ and $\mathbf{z}_j = \left( x_{1j},...,  x_{Lj} \right)$ be elements of this set. Then, the distance between these two elements satisfy that $d \left( \mathbf{z}_i, \mathbf{z}_j \right) \geq 0$ and 
$d \left( \mathbf{z}_i, \mathbf{z}_j \right) = 0$ if and only if $ \mathbf{z}_i = \mathbf{z}_j$. For example, a suitable metric to capture this is given by the Euclidean distance, 
\begin{align}
d \left( \mathbf{z}_i, \mathbf{z}_j \right) = \left( \sum_{l=1}^L \left( x_{l,i} - x_{l,j} \right)^2 \right)^{1 / 2}. 
\end{align}
Another choice is the distance characterized by the $l_1$ norm: $\displaystyle d \left( \mathbf{z}_i, \mathbf{z}_j \right) = \sum_{l=1}^L \left| x_{l,i} - x_{l,j} \right|$.

\subsubsection{Asymptotic Properties and Parameter Estimation}

Consider $\delta_K = \left( \gamma, \Gamma_K \right)^{\prime}$ be the OLS regression of $Y$ on $\mathbb{U}_K = \left( \mathbb{X}, \mathbb{X}_K \right)$, which is as 
\begin{align}
\widehat{\delta}_K = 
\begin{pmatrix}
\widehat{\gamma}
\\
\widehat{\Gamma}_K 
\end{pmatrix}
= 
\left( \mathbb{U}_K^{\prime} \mathbb{U}_K \right)^{-1} \mathbb{U}_K^{\prime} Y.  
\end{align}
Using the partitioned inverse one can show that 
\begin{align}
\widehat{\gamma} = \left( \widehat{\mathbb{X}}_u^{\prime} \widehat{\mathbb{X}}_u \right)^{-1} \widehat{\mathbb{X}}_u^{\prime} \left( Y - \widehat{Y} \right), 
\end{align}
with $\widehat{\mathbb{X}}_u = \mathbb{X} - \widehat{\mathbb{X}}$, where $\widehat{\mathbb{X}} = \mathbb{X}_K \left( \mathbb{X}_K  ^{\prime} \mathbb{X}_K  \right)^{-1} \mathbb{X}_K^{\prime} \mathbb{X}$. Similarly, we have that $\widehat{Y} = \mathbb{X}_K \left( \mathbb{X}_K^{\prime} \mathbb{X}_K  \right)^{-1} \mathbb{X}_K^{\prime} Y$ is the projection of Y on $\mathbb{X}_K$. Therefore, the network parameters are estimated from the partitioned regression proposed by \cite{Olmo2023} 
\begin{align}
\widehat{\Gamma }_K = \big( \mathbb{X}_K^{\prime} \mathbb{X}_K \big)^{-1} \mathbb{X}_K^{\prime} \left( Y - \mathbb{X} \widehat{\gamma} \right). 
\end{align} 
Another important quantity to make statistical inference about the network parameters is the variance of $\widehat{\delta}_K$. Let $\mathbb{X}_u = \mathbb{X} - \mathbf{E} \left[ \mathbb{X} | \mathbb{X}_K \right]$, $\Phi = \mathbf{E} \left[ \mathbb{X}^{\prime} \mathbb{X}    \right]$, and $\Psi = \mathbf{E} \left[ \left( \mathbb{X}_u \epsilon   \right) \left( \mathbb{X}_u \epsilon   \right)^{\prime} \right]$, with $\epsilon = Y - \mathbb{X} \gamma - \mathbb{X}_K \Gamma_K$.

%%-------------------------------------------------------------------------%%
\newpage

The asymptotic variance of the standardized estimator of $\gamma$ is $\Phi^{-1} \Psi \Phi^{-1}$, which can be estimated as 
\begin{align}
\widehat{V} \left( \widehat{\gamma} \right) = \frac{1}{N} \widehat{\Phi}^{-1} \widehat{\Psi} \widehat{\Phi}^{-1}, 
\end{align}
with $\widehat{\Phi} = \displaystyle \frac{1}{N} \sum_{i=1}^N \widehat{ \mathbb{X} }_{i,u}^{\prime} \widehat{ \mathbb{X} }_{i,u}$ and $\widehat{\Psi} = \displaystyle \frac{1}{N}  \sum_{i=1}^N e_i^2 \widehat{ \mathbb{X} }_{i,u}^{\prime} \widehat{ \mathbb{X} }_{i,u}$, where $e = Y - \mathbb{X} \widehat{\gamma} - \mathbb{X}_K \widehat{\Gamma}_K$. Under homoscedasticity of the error term, the asymptotic variance of the standardized estimator is $V \left( \widehat{\gamma} \right) = \Phi^{-1} \sigma_{\epsilon}^2 = \mathbf{E} \left[ \epsilon^2 \right]$, and the corresponding estimator is $\widehat{V} \left( \widehat{\gamma} \right) = \widehat{\Phi}^{-1} \widehat{\sigma}_{\epsilon}^2 / N$, with $ \widehat{\sigma}_{\epsilon}^2 = \frac{1}{N} \sum_{i=1}^N \displaystyle  \epsilon_i^2$. \cite{Olmo2023} derive the asymptotic variance of $\widehat{\Gamma}_{K}$ is formally derived in the proof of Proposition 2 below. 
\begin{align}
\sqrt{N} \left( \widehat{\Gamma} - \Gamma \right) = \left( \frac{ \mathbb{X}_K^{\prime} \mathbb{X}_K }{N} \right) \left( \frac{ \mathbb{X}_K^{\prime} \mathbb{X}_K }{N} \sqrt{N} \left( \gamma - \widehat{\gamma} \right) + \frac{1}{ \sqrt{N} } \mathbb{X}_K^{\prime} \epsilon  \right)
\end{align}
and let $Q = \mathbf{E} \left[ \mathbb{X}_K^{\prime} \mathbb{X}_K \right]$. Then, it follows that,
\begin{align*}
V \big( \sqrt{N} \left( \widehat{\Gamma}_K - \Gamma_K \right) \big) = Q^{-1} \big\{  \mathbf{E} \left[ \mathbb{X}_K^{\prime} \mathbb{X} \right]  \Phi^{-1} \Psi \Phi^{-1} \mathbf{E} \left[ \mathbb{X}^{\prime} \mathbb{X}_K \right] + \mathbf{E} \left[ \mathbb{X}_K^{\prime} \epsilon \epsilon^{\prime} \mathbb{X}_K \right]  \big\} Q^{-1} + 0_p(1),
\end{align*}
since $\mathbf{E} \left[ \mathbb{X}_K^{\prime} \mathbb{X} \left( \gamma - \widehat{\gamma} \right) \epsilon^{\prime} \mathbb{X}_K^{\prime} \right] = 0$.

Under homoscedasticity of the error term $\epsilon$ the variance of the standardized estimator satisfies 
\begin{align}
V \big( \sqrt{N} \left( \widehat{\Gamma}_K - \Gamma_K \right) \big) = Q^{-1} \big\{  \mathbf{E} \left[ \mathbb{X}_K^{\prime} \mathbb{X} \right]  \Phi^{-1} \sigma^2_{\epsilon} \mathbf{E} \left[ \mathbb{X}^{\prime} \mathbb{X}_K \right] \big\} Q^{-1} + 0_p(1).
\end{align}
Therefore, a suitable estimator of the variance $\widehat{\Gamma}_K$ is 
\begin{align}
\widehat{V} \left( \widehat{Q}_K^{-1} \right) \left[ \left(  \frac{1}{N} \sum_{ i=1 }^N \mathbb{X}_{iK}^{\prime} \mathbb{X}_{iK} \right)   \widehat{\Phi}^{-1}  \widehat{\sigma}^2_{\epsilon} \left(  \frac{1}{N} \sum_{ i=1 }^N \mathbb{X}_{iK}^{\prime} \mathbb{X}_{iK} \right)    \right] / N + \widehat{Q}_K^{-1}  \widehat{\sigma}_{\epsilon} / N,
\end{align}
where $\widehat{Q}_K = \displaystyle \frac{1}{N} \sum_{ i=1 }^N \mathbb{X}_{iK}^{\prime} \mathbb{X}_{iK}$ and $\mathbb{X}_{iK}$ are rows of the matrix $\mathbb{X}_{K}$. The random quantities of interest for estimating the presence of network effects are the parameters, $\beta_l (d)$, for $\ell = 1,...,L$ which are interpreted as realizations of the continuous and differentiable functional coefficients $\beta_{\ell} (d)$, with $d \in (0, C] \subset \mathbb{R}^{+}$. Then, the estimator of $\beta_{\ell} ( d)$ is defined as
\begin{align}
\widehat{\beta}_{\ell} (d) =  v( d )^{\prime} \widehat{ \Gamma }_{ L \ell },  
\end{align}
with $v( d ) = \big[ v_1( d )^{\prime},..., v_K( d )^{\prime} \big]^{\prime}$, where 
\begin{align}
v_k ( d_{ij} ) = d_{h,ijk} \otimes \big[ 1, \left( d_{ij} - z_k \right), \left( d_{ij} - z_k \right)^2,...  \left( d_{ij} - z_k \right)^1 \big]^{\prime}
\end{align}

%%-------------------------------------------------------------------------%%
\newpage

\subsubsection{Main Asymptotic Results}

\paragraph{Proof of Theorem 1}[\cite{Olmo2023}] The proof of this result follows from noting 
\begin{align*}
\widehat{ \gamma } &= \big( \widehat{X}_u^{\prime} \widehat{X}_u \big)^{-1} \widehat{X}_u^{\prime} \left( Y - \widehat{Y} \right) 
= 
\left( X^{\prime} M_{\mathbb{X } } X \right)^{-1} \widehat{X}^{\prime} M_{ \mathbb{X} } Y
\\
&= 
\gamma + \left( X^{\prime} M_{ \mathbb{X } } X \right)^{-1} X^{\prime} M_{\mathbb{X }} \mathbb{X}_K \Gamma_K + \left( X^{\prime} M_{ \mathbb{X } } X \right)^{-1} X^{\prime} M_{ \mathbb{X } } \epsilon,  
\end{align*}
with $\widehat{X}_u = M_{ \mathbb{X} } X$, $M_{ \mathbb{X} } = \left[ I_N - \mathbb{X}_K \left( \mathbb{X}_K^{\prime} \mathbb{X}_K \right)^{-1}  \mathbb{X}_K \right]$, and $I_N$ is the identity matrix, $Y - \widehat{Y}= M_{ \mathbb{X} } Y$. Then, given that $M_{ \mathbb{X} }$ and $\mathbb{X}_K$ are orthogonal, we obtain 
\begin{align}
\widehat{\gamma} - \gamma = \left( \frac{ X^{\prime} M_{\mathbb{X} } X }{N} \right) \frac{ X^{\prime} M_{\mathbb{X} } \epsilon  }{N}. 
\end{align}
Now, under the assumption that $\mathbb{E} \left[ \epsilon | X \right] = 0$, applying the law of large numbers, it follows that $\frac{ X^{\prime} M_{\mathbb{X} } \epsilon  }{N} \overset{ p }{ \to } \Phi$, with $\Phi$ a positive definite matrix. Therefore, $\widehat{\gamma} - \gamma = o_p(1)$, as $N \to \infty$. Then, the asymptotic convergence of the standardized estimator immediately follows by slightly adapting the proof of nonlinear additive partial models for series estimators. 

\paragraph{Proof of Proposition 1}[\cite{Olmo2023}]  In order to prove the result in Proposition 1, it is sufficient to show that $\mathbf{E} \left[ \norm{  \widehat{Q}_K - I_{\widetilde{K} } }^2 \right] = \mathcal{O} \left(  \frac{ \xi_0(h) }{ Nh} \right)$, as $N \to \infty$, with $I_{\widetilde{K} }$ the $\widetilde{K} \times \widetilde{K}$ identity matrix for $\widehat{K} = K \left( q + 1 \right)$. Thus, for a square symmetric square matrix $Q^{-1/2}$ of $Q^{-1}$, with $Q = I_{ \widehat{K} }$, the vector $\mathbb{X}_K (x) Q^{-1 / 2}$ is a nonsingular transformation of $\mathbb{X}_K (x)$, and thus it can be shown that 
\begin{align}
\tilde{ \xi }_0 (h) = \underset{ x \in \mathcal{X}_X }{ \text{sup} } \norm{ \mathbb{X}_X (x) Q^{-1 / 2 } } \leq \bar{C} \xi_0 (h), 
\end{align}
with $\bar{C}$ some positive constant. Next, we assume that $\mathbb{X}_K$ is a standardized version of our regression matrix. Let $\mathbb{X}_{ij, K}$ denote the element $( i, j)$ of the matrix $\mathbb{X}_{K}$, and $\delta_{ij}$ denote the element $( j, \ell)$ of the matrix $I_{ \widetilde{K} }$. Then, the assumption $Q = I_{ \widetilde{K} }$ implies that $\mathbf{E} \left[ \mathbb{X}_{ij, K}^{\prime} \mathbb{X}_{ij, K} \right] = \delta_{ij}$, and we have that 
\begin{align*}
\mathbf{E} \big[ \norm{ \widehat{Q}_K - I_{ \widetilde{K}} } \big] 
&= 
\mathbb{E} \left[ \sum_{j=1}^{ \widetilde{K} } \sum_{j=1}^{ \widetilde{K} } \left( N^{-1} \sum_{j=1}^{ N } \mathbb{X}_{ij, K} \mathbb{X}_{ij, K} - \delta_{j \ell} \right)^2 \right]
= 
N^{-1} \sum_{j=1}^{ \widetilde{K} } \sum_{j=1}^{ \widetilde{K} } \mathbb{E} \big[ \big( \mathbb{X}_{ij, K} \mathbb{X}_{ij, K} - \delta_{j \ell} \big)^2 \big]
\\
&\leq N^{-1} \sum_{j=1}^{ \widetilde{K} } \sum_{j=1}^{ \widetilde{K} }  \mathbb{E} \big[ \mathbb{X}^2_{ij, K} \mathbb{X}^2_{ij, K} \big] =  N^{-1} \mathbb{E} \left[ \norm{ \mathbb{X}^2_{iK} }^2 \sum_{j=1}^{ \widetilde{K}} \mathbb{X}^2_{ij,K} \right]
\\
&\leq N^{-1} \xi_0 \left( h \right)^2 \mathbb{E} \left[ \sum_{j=1}^{ \widetilde{K}} \mathbb{X}^2_{ij,K} \right] = N^{-1} \xi_0 \left( h \right)^2 \mathbb{E} \left[ \text{trace} \big( \mathbb{X}_{K}^{\prime} \mathbb{X}_{K} \big) \right]
\\
&= N^{-1} \xi_0 \left( h \right)^2 \text{trace} \left(  I_{ \widetilde{K} }  \right) = \xi_0 \left( h \right)^2 \widetilde{K} /N = \xi_0 \left( h \right)^2 \widetilde{K} /N. 
\end{align*}

%%-------------------------------------------------------------------------%%
\newpage

We have that, under Assumption A.5, $\underset{ x \in \mathcal{X}_X }{ \text{sup} } \ \leq \xi_0 \left( h \right)$, with $h \in 0$. Furthermore, $\mathbb{E} \left[ \text{trace} \big( \mathbb{X}_{K}^{\prime} \mathbb{X}_{K} \big) \right] = \text{trace}  \left(  I_{ \widetilde{K} } \right)$, with $\widetilde{K} = (q+1)K$. Then, $\mathbf{E} \big[ \norm{ \widehat{Q}_K - I_{ \widetilde{K}} } \big] \leq  \xi_0 \left( h \right)^2 C ( q + 1)/ ( 2hN )$. Therefore, for $q$ and $C$ fixed, $\norm{ \widehat{Q}_K - I_{ \widetilde{K}} } = \mathcal{O} \left( \xi_0 \left( h \right) \left( N h \right)^{-1/2} \right)$.   

Furthermore, since the smallest eigenvalue of $\widehat{Q}_K - I_{ \widetilde{K}}$  is bounded by $\norm{ \widehat{Q}_K - I_{ \widetilde{K}} }$, this implies that the smallest eigenvalue of  $\widehat{Q}_K$ converges to one in probability. Letting $1_N$ be the indicator function for the smallest eigenvalue of $\widehat{Q}_K$ being greater than $1 / 2$, then $\mathbb{P} \left( 1_N = 1 \right) = 1$.

\paragraph{Proof of Lemma 1}[\cite{Olmo2023}] 

Let $X = \mathbf{E} \left[ X | \mathbb{X}_K   \right] + X_u$, with $X_u$ the error term of the projection of $X$ on $\mathbb{X}_K$, and let $\widehat{X} = \mathbb{X}_K \big( \mathbb{X}_K^{\prime} \mathbb{X}_K \big) \mathbb{X}_K^{\prime} X$ be the linear projection of $X$ on $\mathbb{X}_K$. Then, $\widehat{X} = \mathbb{X}_K^{\prime} \widehat{\beta}_X$, with $\widehat{\beta}_X = \big( \mathbb{X}_K^{\prime} \mathbb{X}_K \big)^{-1} \mathbb{X}_K^{\prime} X$ the slope coefficient of the approximating regression given by $\widetilde{K}$ regressors. Then, $\widehat{X}_u = X - \widehat{X}_u$,  
\begin{align}
\widehat{X}_u = \mathbb{E} \left[ X | \mathbb{X}_K \right] + X_u -  \mathbb{X}_K \widehat{\beta}_X.
\end{align} 
Simple algebra shows that 
\begin{align}
\widehat{X}_u = X_u + \mathbb{E} \left[ X | \mathbb{X}_K \right] - \mathbb{X}_K \beta_X + \mathbb{X}_K \left( \beta_X - \widehat{\beta}_X   \right). 
\end{align}
 
\begin{assumption}[\cite{Olmo2023}]
For every $K$, there exists a constant $\kappa > 0$, such that
\begin{align}
\mathbf{E} \left[ X | \mathbb{X}_K \right] - \mathbb{X}_K \beta_X = \mathcal{O} \left( K^{- \kappa} \right)
\end{align} 
\end{assumption}
with $\beta_X$ the vector of coefficients associated to the linear prediction model $\mathbb{X}_K \beta_X$.  Therefore, it follows 
\begin{align}
\widehat{\Phi}^* &= \frac{1}{N} \sum_{ i = 1}^N X_{i,u}^{\prime} X_{i,u} + \mathcal{S}_f + 2 \mathcal{S}_{X_u , f} ,
\\
\widehat{\Phi} &= \frac{1}{N} \sum_{ i = 1}^N X_{i,u}^{\prime} X_{i,u}, \ \mathcal{S}_f = \frac{1}{N} \sum_{ i = 1}^N \big( \mathbb{E} \left[ X_i | \mathbb{X}_K \right] - \mathbb{X}_{K,i} \widehat{\beta}_X \big)^{\prime} \big( \mathbb{E} \left[ X_i | \mathbb{X}_K \right] - \mathbb{X}_{K,i} \widehat{\beta}_X \big)
\\
\mathcal{S}_{X_u , f}  &= \frac{1}{N} \sum_{ i = 1}^N X_{i,u}^{\prime}  \big( \mathbb{E} \left[ X_i | \mathbb{X}_K \right] - \mathbb{X}_{K,i} \widehat{\beta}_X \big) 
\end{align}
where $\mathbb{X}_{K,i}$ represents the $i$th row of matrix $\mathbb{X}_{K}$. The convergence in probability of the matrix $\widehat{\Phi}$ implies $\frac{1}{N} \sum_{ i = 1}^N X_{i,u}^{\prime} X_{i,u} \overset{ p }{ \to } \Phi$, where $\Phi = \mathbb{E} \left[ X_{u}^{\prime} X_{u} \right]$ by applying the law of large numbers. Next we show that 
\begin{align*}
\mathcal{S}_{f} = \frac{1}{N} \norm{ \mathbb{E} \left[ X | \mathbb{X}_K  \right] - \mathbb{X}_{K}  \widehat{\beta}_X  }^2 \leq \frac{1}{N} \norm{  \mathbb{E} \left[ X | \mathbb{X}_K  \right] - \mathbb{X}_{K} \beta_f }^2 + \frac{1}{N} \norm{ \mathbb{X}_{K} \left( \beta_f - \widehat{\beta}_f  \right) }^2.
\end{align*}

%%-------------------------------------------------------------------------%%
\newpage

From Assumption A.6, it follows that $\norm{ \mathbb{E} \left[ X | \mathbb{X}_K  \right] - \mathbb{X}_{K} \beta_X }^2 = \mathcal{O} \left( K^{-2 \kappa } \right)$. Now, it holds that
\begin{align}
\frac{1}{N} \norm{ \mathbb{X}_K \left( \beta_X - \widehat{\beta}_X \right) }^2 
=  \left( \beta_X - \widehat{\beta}_X \right)^{\prime} \widehat{Q}_K \left( \beta_X - \widehat{\beta}_X \right) = \mathcal{O}_p(1) \norm{  \beta_X - \widehat{\beta}_X }^2
\end{align}
by Proposition 1. We now prove that $\norm{  \beta_X - \widehat{\beta}_X }^2 = \mathcal{O} \left( K^{- \kappa } \right)$. Notice that, 
\begin{align*}
\norm{ \beta_X - \widehat{\beta}_X } 
&= 
\norm{ \big( \mathbb{X}_K^{\prime}   \mathbb{X}_K \big)^{-1} \mathbb{X}_K^{\prime} \big( \mathbb{E} \left[ X | \mathbb{X}_K \right] + X_u - \mathbb{X}_K \beta_X \big) }
\\
&\leq   \norm{ \big( \mathbb{X}_K^{\prime}   \mathbb{X}_K \big)^{-1} \mathbb{X}_K^{\prime} \big( \mathbb{E} \left[ X | \mathbb{X}_K \right] - \mathbb{X}_K \beta_X \big) } +  \norm{ \big( \mathbb{X}_K^{\prime}   \mathbb{X}_K \big)^{-1} \mathbb{X}_K^{\prime} X_u }
\end{align*}

\paragraph{Proof of Theorem 1}[\cite{Olmo2023}]

We prove first the consistency of the estimator $\gamma$. The estimator is defined as 
\begin{align}
\widetilde{ \gamma } = \big( \widehat{X}_u^{\prime} \widehat{X}_u \big)^{-1} \widehat{X}_u^{\prime}  \big( Y - \widehat{Y} \big), 
\end{align}

where $\widehat{Y} = \mathbb{X}_K \big( \mathbb{X}_K^{\prime} \mathbb{X}_K \big)^{-1} \mathbb{X}_K^{\prime}$, the linear projection of Y on $\mathbb{X}_K$. The infeasible estimator is
\begin{align}
\widetilde{ \gamma } = \big( \widehat{X}_u^{\prime} \widehat{X}_u \big)^{-1} \widehat{X}_u^{\prime} \left( Y - \mathbb{E} \left[ Y | \mathbb{X}_K \right] \right),  
\end{align}

\subsection{Network Cluster Robust Inference}

Recently, the literature is considering cluster-robust methods are widely used to account for cross-sectional dependence. The standard model of cluster dependence partitions the set of observations into may independent clusters. Usually researchers use HAC variance estimators, which account for spatial or temporal dependence. According to \cite{leung2023network} such simulation evidence show that for spatially or temporally dependent data, tests using HAC estimators can exhibit size distortion in smaller samples, unlike cluster-robust inference methods. Specifically, \cite{leung2023network} develop a novel framework with simulated and theoretical evidence for applying cluster-robust methods to network-dependent data. Within the proposed setting of \cite{leung2023network} a main econometric challenge to ensure robust estimation and inference is to obtain obtain HAC robust estimation techniques which have some special characteristics especially in the choice of the bandwidth for data is network, rather than spatially, dependent. Moreover, the simulation results of  \cite{leung2023network} show that there are advantages to using cluster-robust methods for network data, in comparison to to classical HAC estimators. In particular, \cite{leung2023network} find that the randomization test, a leading method for cluster-robust inference with a small number of clusters, better controls size in smaller samples, provided clusters have low conductance. However, when no such clusters exist, the test, when naively applied to the output of spectral clustering, can exhibit substantial size distortion even in large samples, unlike the HAC estimator. This is because clusters in this case cannot generally satisfy the requirement of asymptotic independence, so we expect all existing cluster-robust methods to exhibit similar size distortion.

%%-------------------------------------------------------------------------%%
\newpage

\subsubsection{Setup}

Following the framework proposed by \cite{leung2023network}, observe a set of units $\mathcal{N}_n = \left\{ 1,..., n \right\}$, data $W_i \in \mathbb{R}^{d_w}$ associated with each unit $i \in \mathcal{N}_n$, and an undirected network or graph $\boldsymbol{A}$ on $\mathcal{N}_n$. We represent $\boldsymbol{A}$ as a binary, symmetric adjacency matrix with $ij-$th entry $A_{ij}$, where $A_{ij} = 1$ signifies a link between $i$ and $j$. There are no self-links, meaning $A_{ii} = 0$ for all $i$. Moreover, the settings of \cite{leung2023network} treats $\boldsymbol{A}$ as fixed (conditional upon), whereas $\left\{ W_i \right\}$ is random and not necessarily indetically distributed. Let $\theta_0 \in \mathbb{R}^{d_{\theta}}$ be the estimand of interest and $g : \mathbb{R}^{d_{\theta}} \to \mathbb{R}^{d_{\theta}}$ a moment function such that 
\begin{align}
\mathbb{E} \left[ g( W_i, \theta_0 ) \right] = \boldsymbol{0} \ \ \ \forall i \in \mathcal{N}_n.    
\end{align}
Denote with $\displaystyle  G(\theta) = \frac{1}{n}  \sum_{i=1}^n g( W_i, \theta )$, and $\Psi_n$ be a weighted matrix. Define the generalized method of moments (GMM) estimator
\begin{align}
\hat{\theta} = \underset{ \theta }{  \mathsf{argmin} } \ \hat{G}(\theta)^{\prime} \Psi_n \hat{G}(\theta).    
\end{align}
Various studies in the literature develop cluster-robust methods for GMM when $\left\{ W_i \right\}$ satisfies weak temporal or spatial dependence. However, we instead employ a notion of weak network dependence, which is analogous to mixing conditions used in time series and spatial econometrics.

\subsubsection{Cluster}

Cluster-robust methods take as input a partition of $\mathcal{N}_n$ into $L$ clusters, which we denote by $\left\{  \mathcal{C}_{\ell}  \right\}_{\ell = 1}^L$. Being a partition, the clusters satisfy $\cup_{\ell = 1}^L \mathcal{C}_{\ell} = \mathcal{N}_n$ and $\mathcal{C}_{\ell} \cap \mathcal{C}_m = \varnothing$ for all $\ell \neq m$. Notice that $\mathcal{C}_{\ell}$ depends on $\boldsymbol{A}$ since different networks may be partitioned differently. Furthermore, the number of "quality" clusters in a network is small, so we develop inference procedures robust to a small number of clusters. We refer to these procedures as ``small-L cluster robust methods''. In particular such methods have been shown to exhibit substantially improved size control relative to conventional cluster-robust procedures. Moreover, under weak network dependence, observations in different components are independent, so components may therefore be treated as separate clusters. This implies that, if a network consists of many components, standard many-cluster asymptotics are applicable, and one can simply use standard errors on the components (see,  \cite{leung2023network}). 

\subsubsection{Variance-Covariance Matrix Identifications}

Let $\mathcal{N}_k = \left( \mathcal{C}_K, \boldsymbol{X}_K  \right)$ be a $\mathcal{C}-$stationary network and $\mathcal{C}_K^* \in \mathbb{G}_k$ be its realization. Then a natural estimator of $\mathsf{Var} \left( \boldsymbol{X}_K | \mathcal{C}_K = \mathcal{C}_K^{*} \right)$ is given by
\begin{align}
\widehat{\mathsf{Var} } \left( \boldsymbol{X}_K | \mathcal{C}_K = \mathcal{C}_K^{*} \right) = \big( \hat{\gamma}_{\mathcal{N}_K, \mathcal{C} } \left( \mathcal{C} (i,j; \mathcal{C}_K^{*} \right) \big)_{i,j \in \mathcal{V}_K}
\end{align}

%%-------------------------------------------------------------------------%%
\newpage

\subsubsection{Conductance}

Consider formal conditions under which a given set of clusters can be used for asymptotically valid cluster-robust inference. In particular, we consider a sequence of networks with associated clusters indexed by the network size $n$, taking $n$ to infinity while keeping the number of clusters $L$ fixed. Usually the correct way to consider asymptotics is to employ a sequence of networks. Under weak network dependence and standard regularity conditions, we can show that 
\begin{align*}
\frac{1}{ \sqrt{n} } 
\begin{pmatrix}
n_1  \hat{G}_1(\theta_0)
\\
\vdots
\\
n_L  \hat{G}_L(\theta_0)
\end{pmatrix}
\overset{d}{\to}
\mathcal{N} \left( \boldsymbol{0}, \boldsymbol{\Sigma}^* \right), \ \ \ \boldsymbol{\Sigma}^* =
\begin{pmatrix}
\rho_1 \boldsymbol{\Sigma}_{11} \ \ & \ \ \sqrt{\rho_1 \rho_2 } \boldsymbol{\Sigma}_{12} \ \ & \ \ \hdots \ \ \sqrt{\rho_1 \rho_L } \boldsymbol{\Sigma}_{11}
\\
\sqrt{\rho_2 \rho_1 } \boldsymbol{\Sigma}_{21} \ \ & \ \ \rho_2 \boldsymbol{\Sigma}_{12} \ \ & \ \ \hdots \ \ \sqrt{\rho_2 \rho_L } \boldsymbol{\Sigma}_{2L}
\\
\vdots \ \ & \ \ \vdots \ \ & \ \  \ \ \vdots
\\
\sqrt{\rho_2 \rho_1 } \boldsymbol{\Sigma}_{L1} \ \ & \ \ \sqrt{\rho_L \rho_2 } \boldsymbol{\Sigma}_{12} \ \ & \ \ \hdots \ \ \rho_L \boldsymbol{\Sigma}_{LL}
\end{pmatrix}
\end{align*}
where $\rho_{\ell} = \mathsf{lim}_{n \to \infty} n_{\ell} / n$ and
\begin{align}
\boldsymbol{\Sigma}_{\ell m} = \mathsf{lim}_{n \to \infty} \mathsf{Cov} \left( \sqrt{n_{\ell}} \hat{G}_{\ell}(\theta_0), \sqrt{n_{m}} \hat{G}_{m}(\theta_0)  \right)     
\end{align}
The above asymptotic result can ensure that the vector of GMM estimates $\left\{ \sqrt{n} \left( \hat{\theta}_{\ell} - \theta_0  \right) \right\}_{\ell = 1}^L$ is asympotically normal. Conventional cluster-robust methods require independent clusters. Then small-L cluster robust methods exploit the weaker requirement of \textit{asymptotic independence}, that $\boldsymbol{\Sigma}_{\ell m} = \boldsymbol{0}$ for all $\ell \neq m$. Furthermore, the symmetry of the limit distribution which, under the group of transformations corresponds to having off-diagonal blocks equal to zero. In other words, we interpret the zero off-diagonal blocks $\boldsymbol{\Sigma}_{\ell m}$ as the key requirement for the validity of cluster-robust methods.  

\begin{definition}[\cite{leung2023network}]

\begin{itemize}
    
    \item[(a)] The edge boundary size of $S \subset \mathcal{N}_n$ with respect to $\boldsymbol{A}$ is given by 
    \begin{align}
        \left| \partial_{\boldsymbol{A}} (S) \right| = \sum_{i \in S} \sum_{ j \in \mathcal{N}_n \ S} A_{ij}
    \end{align}
    the number of links involving a unit in $S$ and a unit not in $S$.
    
    \item[(b)] The volume of $S$ is $\mathsf{vol}_{\boldsymbol{A}}(S) = \sum_{i \in S} \sum_{j=1}^n A_{ij}$, the sum of the degree $\sum_{j=1}^n A_{ij}$ of unit $i$ in $S$. 
    
    \item[(c)] The conductance of $S$ (assuming it has at least one link) is given by 
    \begin{align}
      \phi_{\boldsymbol{A}}(S) = \frac{ \left| \partial_{\boldsymbol{A}} (S) \right| }{ \mathsf{vol}_{\boldsymbol{A}}(S)  }    
    \end{align}
    
\end{itemize}

\end{definition}

\begin{remark}
Conductance is a $[0,1]$ measure of how integrated S is within $\boldsymbol{A}$. Therefore, our main condition for guaranteeing that $\boldsymbol{\Sigma}_{\ell m} = \boldsymbol{0}$ for all $\ell \neq m$ is $\underset{ 1 \leq \ell \leq L }{ \mathsf{max} } \phi_{\boldsymbol{A}}(\mathcal{C}_{\ell}) \to 0, \ \ \ \text{as} \ \ n \to \infty$. In particular, the above condition says that \textit{maximal conductance of the clusters is small}, which means that each cluster's boundary size is of smaller order than its volume. Furthermore, the average degree is asymptotically bounded which rules out a completely connected network in which all units are linked, which is the denset possible topology (see, \cite{leung2023network}).     
\end{remark}

%%-------------------------------------------------------------------------%%
\newpage 

In other words, the main assumption imposed by \cite{leung2023network} regarding dependence between observations is that they are weakly dependent, meaning that random variables approach independence as the distance between their locations grows. 
Moreover, our methods involve partitioning the data into groups defined by the researcher. We define $G_N$ to be the total number of groups and index them by $g = 1,..., G_N$. The OLS estimator can be written as below
\begin{align}
\hat{\beta}_N = \left(  \sum_{g=1}^G x_g^{\prime} x_g \right)^{-1} \left( \sum_{g=1}^G x_g^{\prime} y_g \right)    
\end{align}
using group-level notation. Therefore, the most common approach to inference with weakly dependent data is to use a plug-in estimator, call it $\widetilde{V}_N$, of the variance matrix of $x_{s_i} \varepsilon_{s_i}$, along with the usual large-sample approximation for the distribution of $\widehat{\beta}_N$ such that 
\begin{align}
\sqrt{N} \left( \hat{\beta}_N - \beta \right) \overset{d}{\to} \mathcal{N} \left( 0, Q^{-1} V Q^{-1} \right)  
\end{align}
where the long-run covariance matrix is defined as below
\begin{align}
\boldsymbol{V} = \underset{ N \to \infty }{\mathsf{lim}} \ \mathsf{Var} \left( \frac{1}{\sqrt{N}} \sum_{i=1}^N x_{s_i} \varepsilon_{s_i} \right)    
\end{align}
where $Q$ is the limit of the second moment matrix for $x$. The typical method uses the sample average of $x_{s_i} x_{s_i}^{\prime}$ to estimate $Q$ and plugs-in a consistent estimators, $\widetilde{V}_N$, of $V$ to arrive at the approximation
\begin{align}
\hat{\beta}_N \sim \mathcal{N} \left( \beta, \frac{1}{N} \left[ \frac{1}{N} \sum_{i=1}^N x_{s_i} x_{s_i}^{\prime} \right]^{-1} \widetilde{V}_N \left[ \frac{1}{N} \sum_{i=1}^N x_{s_i} x_{s_i}^{\prime} \right]^{-1}  \right).
\end{align}

\subsubsection{Asymptotic Theory}

We consider a sequence of networks and associated clusters, both implicitly indexed by the network size $n$. Recall that $n_{\ell} = \left| \mathcal{C}_{\ell}  \right|$, the size of cluster $\ell$.

\begin{assumption}[Limit Sequence, see \cite{leung2023network}] (a) The number of clusters $L$ is fixed as $n \to \infty$. (b) For any $\ell = 1,..., L, n_{\ell} / n \rho_{\ell} \in [0,1]$.
\end{assumption}

Furthermore, based on the notions proposed by \cite{kojevnikov2021bootstrap}, \cite{leung2023network} consider a formal notion of \textit{weak network dependence} called $\psi-$dependence. The particular notion of dependence is analogous to familiar notions of temporal or spatial weak dependence, except distance between observations is measured using path distance. In other words, weak dependence simply means that the correlation between two sets of observations decays as the network distance between the sets grows.

%%-------------------------------------------------------------------------%%
\newpage 

Therefore, for any $H, H^{\prime} \subset \mathcal{N}_n$, define with 
\begin{align}
G_H &= \left( g(W_i, \theta_0 \right)_{ i \in H}  
\\
\ell_{\boldsymbol{A}} (H, H^{\prime} ) &= \mathsf{min} \left\{ \ell_{\boldsymbol{A}}(i,j): i \in H, j \in H^{\prime} \right\}
\\
\ell_{\boldsymbol{A}} (H, H^{\prime} ) &= \mathsf{min} \left\{ \ell_{\boldsymbol{A}}(i,j): i \in H, j \in H^{\prime} \right\}    
\end{align}
the distance between the two sets. Let $\mathcal{L}_d$ be the set of bounded $\mathbb{R}-$valued Lipschitz functions on $\mathbb{R}^d$, on $\norm{f}_{\infty} = \mathsf{sup}_x | f(x) |$, $\mathsf{Lip} (f)$, the Lipschitz constant of $f \in \mathcal{L}_d$, and  
\begin{align}
\mathcal{P}_n ( h, h^{\prime}; s ) = \big\{ (H,H^{\prime}): H, H^{\prime} \subset \mathcal{N}_n, |H| = h, |H^{\prime}| = h^{\prime},  \ell_{\boldsymbol{A}} (H, H^{\prime} ) \geq s \big\},     
\end{align}
the set of pairs of sets $H, H^{\prime}$ with respective sizes $h, h^{\prime}$ that are at least distances $s$ apart in the network. 

Moreover,  \cite{leung2023network} define the $i'$s $s-$neighbourhood boundary $\mathcal{N}^{\partial}_{\boldsymbol{A}} (i,s) = \left\{ j \in \mathcal{N}_n : \ell_{\boldsymbol{A}} (i,j) = s \right\}$, and its $k-$th moment is given by 
\begin{align}
\delta_n^{\partial} (s;k) = n^{-1} \sum_{i=1}^n \left| \mathcal{N}^{\partial}_{\boldsymbol{A}} (i,s)   \right|^k.    
\end{align}

\medskip

\paragraph{Open Problems}  The last few decades the time series econometrics literature has considered applications of moderate deviations from unit root in univariate autoregressive models and multivariate predictive regression models as well as in settings such as panel data . In particular, predictive regression models with regressors generated as stable or unstable autoregressive processes has recently seen a growing attention in the econometrics and statistics literature. A related open problem include the development of a formal econometric framework for cluster-based predictive regression models. In this direction, our objective is to develop an identification and estimation method that allows us to develop the asymptotic theory for a suitable system estimator under the presence of both nonstationarity and network dependence. To establish a robust estimation and inference procedure, developing central limit theory that explicitly accounts for the dependence between the cross-sectional and time series data will be essential, such as the notation of \textit{stable dependence} (see,  \cite{anatolyev2021limit}). On the other hand, in our framework we assume that the network induced dependence corresponds to the nodes of a network and the corresponding time series observations. We remain agnostic regarding the structure of the network as well as the form of the distance between nodes which will require to impose further metric space assumptions. As a result, the proposed approach allows to further generalize to a data structure that allows for both a cross-sectional network type dependence in the time series dimension.

%%-------------------------------------------------------------------------%%
\newpage

\section{Interval Estimation and Forecasting Methods}
\label{Section6}

\begin{example}[Threshold Factor Model] 

Suppose that $\boldsymbol{y}_t$ be an observed $( p \times 1 )$ time series. Then, the general form of a factor model for time series data is given by 
\begin{align}
\boldsymbol{y}_t = \boldsymbol{A} \boldsymbol{x}_t + \boldsymbol{\varepsilon}_t , \ \ \   t = 1,...,n,
\end{align}
where $\boldsymbol{x}_t = \left( x_{t,1}, x_{t,2},..., x_{t,k} \right)^{\prime}$  is a set of unobserved factor time series with dimension $k$ that is much larger than $p$ the dimension of the $\boldsymbol{y}_t$ vector. Notice that in order to differentiate the signal component from the error process, strong cross-sectional dependence is not allowed for $\left\{ \boldsymbol{\varepsilon}_t \right\}$. As a result, the noise process $\left\{ \boldsymbol{\varepsilon}_t \right\}$ may have weak serial dependence such that 
\begin{align}
\frac{1}{n} \sum_{t=1}^n  \sum_{s=1}^n \left| \mathbb{E} \left( \boldsymbol{\varepsilon}_t^{\prime} \boldsymbol{\varepsilon}_t \right) \right| < C   
\end{align}
where $C$ is a positive constant. 

Notice that one disadvantage of these assumptions is that the dynamic component and error process are not separable when the dimension is finite, since both of them have serial dependence. Alternatively, a different setting for time series data it assumes that the error process is white noise without serial dependence, such that $\mathbb{E} \left( \boldsymbol{\varepsilon}_t^{\prime} \boldsymbol{\varepsilon}_t \right) = 0$, for $t \neq s$. In other words, the particular specification implies that the observed process $\boldsymbol{y}_t$ is completely driven by the common factors. In other words, this ensures that the signal component is identifiable when the dimension of the panel time series is finite. Thus, the error process is allowed to have strong cross-sectional correlation.  

We consider the following two-regime threshold factor model for high-dimensinal time series. Let $\boldsymbol{y}_t$ be an observed $\left( p \times 1 \right)$ and $\boldsymbol{x}_t$ be an $\left( k \times 1 \right)$ latent factor process with time series observations such that 
\begin{align}
\boldsymbol{y}_t = 
\begin{cases}
\boldsymbol{A}_1 \boldsymbol{x}_t + \boldsymbol{\varepsilon}_{t,1}, & z_t \leq \gamma_0 
\\
\boldsymbol{A}_2 \boldsymbol{x}_t + \boldsymbol{\varepsilon}_{t,2}, & z_t > \gamma_0 
\end{cases}
\ \ \text{and} \ \ \ \boldsymbol{\varepsilon}_{t,i} \sim \mathcal{N} \left( \boldsymbol{0}, \boldsymbol{\Sigma}_{t,i} \right)
\end{align}

The particular methodology used in the paper of  , estimates the unknown threshold variable by partitioning the eigenspace and checking the rank properties of the corresponding matrices depending on which regime is "switched on". The rationale behind this approach is the following. However, the main condition that should hold is that at the partition step of the procedure, the corresponding moment matrices from the two partitions do not loose their rank properties, otherwise the identification of the model will lead to a singularity, making it impossible to detect the presence of the true threshold effect.  

Therefore, when we denote the true threshold variable with $\gamma_0$ we can split the data into two subsets such that $\left\{ z_t < \gamma_0 \right\}$ and $\left\{ z_t > \gamma_0 \right\}$ and so we define the following objective function
\begin{align}
G( \gamma_0 ) = \sum_{i=1}^2 \norm{ \boldsymbol{B}_i^{\prime} \boldsymbol{M}_i \boldsymbol{B}_i  }_2 = \sum_{i=1}^2  \norm{ \sum_{h=1}^{h_0} \sum_{j=1}^2 \boldsymbol{B}_i^{\prime} \boldsymbol{\Sigma}_{y,i,j}(h,r) \boldsymbol{\Sigma}_{y,i,j}(h,r)^{\prime} \boldsymbol{B}_i }_2 
\end{align}
\end{example}

%%-------------------------------------------------------------------------%%
\newpage

\begin{remark}
Notice that the objective function $G( \gamma_0 )$ measures the sum of the squared norm of the projections of the cross moment matrices $\boldsymbol{\Sigma}_{y,i,j}(h,r)$ onto the space of the matrices given by $\mathcal{M} \left(   \boldsymbol{B}_i  \right)$ for $h = 1,..., h_0$. An alternative framework is presented by \cite{seo2016dynamic} which implies the use of a GMM approach for identification and estimation. Notice that this approach is quite useful especially when modelling nonlinear asymmetric dynamics and unobserved individual heterogeneity. Moreover, this allows for both the thresholds and regressors to be endogenous. 
\end{remark}

\subsection{Interval Estimation}   

Consider a random sample $\mathcal{X}_n = \left\{ \boldsymbol{X}_1,...,  \boldsymbol{X}_n \right\}$ from an unknown $p-$dimensional distribution depending on a scalar parameter $\theta$. In particular, the aim is to construct a $( 1 - \alpha )-$level upper confidence bound for $\theta$, based on some appropriate point estimator $\hat{\theta}_n$ for $\theta$. Denote by $\mathcal{X}^{*}_n = \left\{ \boldsymbol{X}^{*}_1,...,  \boldsymbol{X}^{*}_n \right\}$ a random sample of size $n$ taken with replacement from $\mathcal{X}_n$ and let $\hat{\theta}_n^{*}$ be the equivalent function of $\mathcal{X}^{*}_n$ as $\hat{\theta}_n$ is to  $\mathcal{X}_n$. Let $\hat{\sigma}_n^{*}$ be the bootstrap version of $\hat{\sigma}_n$. Then, the standard percentile $( 1 - \alpha )-$level bootstrap confidence bounds for $\theta$ can be written as below:
\begin{align}
\hat{I}_H (\alpha) &:= \left( - \infty,  \hat{\theta}_n - n^{-1/2} \hat{\sigma}_n \hat{\xi}_{n, \alpha}  \right]
\\
\hat{I}_B (\alpha) &:= \left( - \infty,  \hat{\theta}_n + n^{-1/2} \hat{\sigma}_n \hat{\xi}_{n, \alpha}  \right]
\end{align}
where $\xi_{n, \alpha}$ is the $\alpha-$quantile of the bootstrap distribution of the standardized $\hat{\theta}_n$ such that 
\begin{align}
\mathbb{P}^{*} \left( n^{1/2} \left( \hat{\theta}^{*}_n - \hat{\theta}_n \right) \big/ \hat{\sigma}_n \leq \hat{\xi}_{n, \alpha} \right) = \alpha  
\end{align}
where $\mathbb{P}^{*}$ refers to the conditional probability law of $\mathcal{X}_n^{*}$ given $\mathcal{X}_n$. Generally, it holds that 
\begin{align}
\mathbb{P} \left( \theta \in \hat{I}_B ( \alpha ) \right) = 1 - \alpha + \mathcal{O} \left( n^{-1/2} \right) = \mathbb{P} \left( \theta \in \hat{I}_H ( \alpha ) \right).   
\end{align}

\subsubsection{Prediction Intervals under state-varying predictability}

Following the framework proposed by \cite{yan2022factor}, the forecast error is given by $\left( \hat{y}_{T+h|T} - y_{T+h|T} \right)$. Notice that when forecasting, one would be more interested in the distribution of the forecast error. Since $Y_{T+h} = y_{T+h|T} + \varepsilon_{T+h}$, it follows that the forecasting error is given by 
\begin{align}
\hat{\varepsilon}_{T+h} =  \hat{y}_{T+h|T} - y_{T+h} = \left( \hat{y}_{T+h|T} - y_{T+h} \right) - \varepsilon_{T+h}   
\end{align}
Thus, if $\hat{\varepsilon}_{T+h}$ is asymptotically normal with 
\begin{align}
\mathsf{var} \left( \hat{\varepsilon}_{T+h}  \right) = \mathsf{var} \left( \hat{y}_{T+h|T} - y_{T+h} \right) = \sigma^2 + \mathsf{var} \left( \hat{y}_{T+h|T} \right)    
\end{align}
with $\mathsf{var} \left( \hat{y}_{T+h|T} \right) = B_T^2$.

%%-------------------------------------------------------------------------%%
\newpage 

\begin{corollary}[\cite{yan2022factor}]
Under the assumptions of the above theorem and assuming that $\varepsilon_t$ is normally distributed then the forecasting error $\hat{\varepsilon}_{T+h}$ is given by 
\begin{align}
\hat{\varepsilon}_{T+h} \sim \mathcal{N} \left( 0, \sigma^2 + \hat{\varepsilon}_{T+h} \right).    
\end{align}
\end{corollary}

Therefore, a confidence interval can be obtained by replacing $\sigma^2$ by its consistent estimate, $\frac{1}{T-h} \sum_{t=1}^{T-h} \hat{\varepsilon}^2_{T+h}$. Therefore, we can get the $95\%$ confidence interval for the conditional mean $y_{T+h|T}$ such that
\begin{align}
\left( \hat{y}_{T+h|T} - 1.96 \sqrt{ \hat{\mathsf{var}} \left( \hat{y}_{T+h|T} \right)},  \hat{y}_{T+h|T} + 1.96 \sqrt{ \hat{\mathsf{var}} \left( \hat{y}_{T+h|T} \right)} \right)    
\end{align}

Therefore, the $95\%$ confidence interval for the forecasting variable $y_{T+h}$ is given by 
\begin{align}
\left( \hat{y}_{T+h|T} - 1.96 \sqrt{  \hat{\sigma}^2 +  \hat{\mathsf{var}} \left( \hat{y}_{T+h|T} \right)},  \hat{y}_{T+h|T} + 1.96 \sqrt{ \hat{\sigma}^2 + \hat{\mathsf{var}} \left( \hat{y}_{T+h|T} \right)} \right)    
\end{align}

\begin{remark}
Notice that in small samples, the above confidence intervals of conditional and unconditional forecasts is likely to underrepresent the true sampling uncertainty due to the uncertainty of $\widehat{\gamma}$. 
\end{remark}

\subsection{Testing for threshold effects}

Moreover, the framework of  \cite{yan2022factor} allows to test for the presence of threshold effects of our model. The question of interest is whether the nonlinear term 
\begin{align}
z_t ( \gamma ) = I( q_t \leq \gamma ) \left( F_t^{\prime}, x_t^{\prime}     \right)^{\prime}    
\end{align}
enters the regression model, that is, whether $\delta_T = 0$. If $\gamma_0$ were known, the traditional Lagrange multiplier statistic and Wald statistic would be good choices to solve this issue. However, $\gamma_0$ is usually unknown and not identified under the null hypothesis. Therefore, we propose a sup-Wald statistic which extends the seminal work of Hansen to test the linearity of the model, as it does not require prior knowledge of $\gamma_0$. Therefore, to facilitate the establishment of distributional theory, we consider a local-to-null reparametrisation: $\delta_T = \frac{c}{\sqrt{T}}$. Thus, based on this model specification the null hypothesis is $H_0: c = 0$ with alternative $H_1: c \neq 0$. Therefore, for each $\gamma \in \Gamma = \left[ \underline{\gamma}, \bar{\gamma} \right]$, we obtain the estimator $\hat{\beta}(\gamma)$ and $\hat{\delta}(\gamma)$. Then, we build a sup-Wald statistic to test the presence of threshold effects such that
\begin{align}
\mathsf{sup} \mathcal{W}_T = \underset{ \gamma \in \Gamma }{ \mathsf{sup} }  \bigg\{ T. \widehat{\underline{\beta}}(\gamma)^{\prime} R \left( R^{\prime} \widehat{M}^{*}_T (\gamma, \gamma )^{-1} \widetilde{\Omega}_T \left( \gamma, \gamma \right) \widehat{M}^{*}_T (\gamma, \gamma )^{-1} R \right)^{-1} R^{\prime} \widehat{\underline{\beta}}(\gamma)^{\prime}  \bigg\}   
\end{align}
where $R = \big[ 0, I_q \big]$ and $q$ denotes the dimension of $z_t$. 
Thus, we obtain that 
\begin{align}
\hat{M}^{*}_T (\gamma, \gamma ) = \frac{1}{T} \sum_{t=1}^{T-h} \widehat{z}_t^{*} (\gamma) \widehat{z}_t^{*} (\gamma)^{\prime} 
\ \ \
\widetilde{\Omega}_T \left( \gamma, \gamma \right) = \frac{1}{T} \sum_{t=1}^{T-h} \widehat{\varepsilon}_{t+h}^2 (\gamma) \widehat{z}_t^{*} (\gamma) \widehat{z}_t^{*} (\gamma)^{\prime}
\end{align}

%%-------------------------------------------------------------------------%%
\newpage

\begin{theorem}[\cite{yan2022factor}]
Suppose that Assumptions hold and $\sqrt{T} / N \to 0$. Then, under the local alternative $H_1: \delta_T = \frac{c}{ \sqrt{T} }$, we have that 
\begin{align}
\mathsf{sup} \ \mathcal{W}_T &\overset{ d }{ \to } \underset{ \gamma \in \Gamma  }{\mathsf{sup} } \ \mathcal{W}^c(\gamma)  
\\
\mathcal{W}^c(\gamma) 
&= \left[ \bar{J}^{*}(\gamma) + \bar{Q}(\gamma) c \right]^{\prime} \bar{K}(\gamma, \gamma ) \left[ \bar{J}^{*}(\gamma) + \bar{Q}(\gamma) c \right] \\
\bar{Q}(\gamma) &= R^{\prime} \Phi^{*-1 \prime} M^* (\gamma, \gamma)^{-1}
M^{*}(\gamma, \gamma_0) \Phi^{*,-1 \prime} R 
\end{align}
\end{theorem}
\begin{remark}
Notice that Theorem 3.2 in \cite{yan2022factor} gives the asymptotic distribution of the sup-Wald test under the alternative such that $H_1: \delta_T = \frac{c}{ \sqrt{T} }$. Under $H_0: c = 0$ and
\begin{align}
\underset{ \gamma \in \Gamma  }{\mathsf{sup} } \ \mathcal{W}^0(\gamma) = \underset{ \gamma \in \Gamma  }{\mathsf{sup} } \ \bar{J}(\gamma)^{\prime} \bar{K}(\gamma, \gamma)^{-1} \bar{J}(\gamma)  
\end{align}
\end{remark}

\begin{remark}
Notice also that the limiting distribution of the sup$_T \ \mathcal{W}_T$ depends on the Gaussian process $\bar{J}^*(\gamma)$, which is not pivotal, and we cannot tabulate the asymptotic critical values for the sup-Wald statistic. Thus, we can compute the p-value based on the procedure proposed by Hansen (1996) such 
\begin{itemize}
    
    \item[(i)] Generate $v_t$, $t = 1,..., T-h$ independently from the standard normal distribution. 
    
    \item[(ii)] Calculate 
    \begin{align}
        \tilde{J}_T^* (\gamma) = \frac{1}{\sqrt{T}} \sum_{t=1}^{T-h} \widehat{z}_t^*(\gamma) \widehat{\varepsilon}_{t+h}(\gamma) v_t.  
    \end{align}
    
    \item[(iii)] Compute the statistic such that 
    \begin{align*}
         \mathsf{sup} \mathcal{W}_T^* \equiv \underset{ \gamma \in \Gamma }{ \mathsf{sup} } \ \bigg\{  \widetilde{J}_T^* (\gamma)^{\prime} \widehat{M}_T^* (\gamma, \gamma)^{-1} R  \bigg( R^{\prime} \widehat{M}_T^* (\gamma, \gamma)^{-1}  \widetilde{\Omega}_T(\gamma, \gamma) \widehat{M}_T^* (\gamma, \gamma)^{-1} R \bigg)^{-1} R^{\prime} \widehat{M}_T^* (\gamma, \gamma)^{-1} \widetilde{J}_T^* (\gamma) \bigg\}
    \end{align*}
    
    \item[(iv)] Repeat steps 1-3 $B$ times and denote the resulting $\mathsf{sup} \ \mathcal{W}_T^{*}$ test statistic as $\mathsf{sup} \ \mathcal{W}_{T,j}^*$ for $j = 1,...,B$. 
    
    \item[(v)] Calculate the simulated $p-$value for the $\mathsf{sup} \ \mathcal{W}_{T}$ as below 
    \begin{align}
        \widehat{p}_T = \frac{1}{J} \sum_{j=1}^J \mathbf{1} \left\{ \mathsf{sup} \ \mathcal{W}_{T,j}^* \geq \mathsf{sup} \ \mathcal{W}_T   \right\}
    \end{align}
    and reject the null hypothesis when $\widehat{p}_T$ is smaller than $\alpha \in (0,1)$, the nominal level.  
    
\end{itemize}
\end{remark}

\begin{theorem}[\cite{yan2022factor}]
Suppose that Assumptions hold and $\sqrt{T} / N \to 0$.Then, under the null hypothesis $H_0: c = 0$, we have that $\mathsf{sup} \ \mathcal{W}_T^* \overset{d}{\to} \underset{ \gamma \in \Gamma }{ \mathsf{sup} } \ \mathcal{W}^0 (\gamma)$. The asymptotic distribution implies that the empirical distribution of $\left\{ \mathsf{sup} \ \mathcal{W}_{T,j}^* \right\}_{ j=1}^{ \bar{J} }$ approximates the asymptotic distribution of $\mathsf{sup} \ \mathcal{W}_T$ under the null hypothesis quite well.
\end{theorem}

%%-------------------------------------------------------------------------%%
\newpage

\begin{example}
A baseline linear threshold regression model is given by
\begin{align}
y_t = x_t^{\prime} \beta_0 + z_t^{\prime} \delta_0 \cdot \mathbf{1} \left\{ q > \gamma_0 \right\} + u_t    
\end{align}
Therefore, to apply any statistical estimation method, it is important to determine whether the threshold effect is statistically significant. We consider a test of no threshold effect against the presence of threshold effects. The null and alternative hypotheses are such that
\begin{align}
\mathcal{H}_0: \delta_0 = 0 \ \ \ \text{for any} \ \gamma_0 \in \Gamma \ \ \ \text{against} \ \ \ \mathcal{H}_1: \delta_0 \neq 0 \ \ \ \text{for some} \ \gamma_0 \in \Gamma.   
\end{align}
All the unknown parameters are identifiable under the alternative hypothesis while the threshold parameter $\gamma_0$ is not identified under the null. Thus, a general method for testing the presence of threshold effects in various regression settings, is to use the sup-likelihood-ratio statistics. A key ingredient of their testing framework is that there exist an objective function and a corresponding extreme estimator for the model with no threshold (under the null) and for the model with threshold effect (under the alternative). Then, the criterion function is expressed: 
\begin{align}
Q_n^{*} (\gamma) &\equiv \underset{ \theta \in \Theta }{ \mathsf{arg \ max} } \ Q_n^{*} (\gamma; \theta)     
\\
\widetilde{Q}_n^{*} &\equiv \underset{ \beta \in \mathcal{B}, \delta \in 0 }{ \mathsf{arg \ max} } \  Q_n^{*} (\gamma; \theta),     
\end{align}
where $\mathcal{B}$ is a compact set containing $\beta_0$ as the interior. The above criterion function is well-defined since $Q_n^{*} (\gamma; \theta)$ does not depend on $\gamma$ when $\delta = 0$. The limiting distribution of the sup-LR statistic under the null hypothesis is highly non-standard and non-normal and thus cannot be directly tabulated. 
\end{example}

\begin{example}[Dynamic Panel Regression with a threshold]
The structural equation of interest is
\begin{align}
y_{it} = \alpha_i + \beta_1 y_{it-1} \boldsymbol{1} \left\{ q_{it} < \gamma \right\} + \beta_2 y_{it-1} \boldsymbol{1} \left\{ q_{it} > \gamma \right\} + u_{it},    
\end{align}    
where the threshold parameter $\gamma \in \Gamma$, such that $\Gamma$ is a strict subset of the support of $q_{it}$. Notice that this threshold parameter is unknown and needs to be estimated. Moreover, the slope parameters $\beta = \left( \beta_1, \beta_2 \right)^{\prime}$ are the slope parameters of interest assumed to be different from each other and $\alpha_i$ is the individual specific effect assumed to be fixed. Furthermore, for econometric identification purposes we allow for a "small threshold effect" which allows statistical inference for the threshold parameter. Relevant studies on threshold estimation and inference include among others \cite{liu2020threshold}, \cite{armillotta2022testing}, \cite{chiou2018nonparametric}, \cite{barigozzi2018simultaneous}, \cite{yu2021threshold}.
\end{example}

Moreover, the endogenous threshold regression model (ETR) has attracted much attention in recent econometric practice. This is due to the fact that economic relationships may shift over time.

%%-------------------------------------------------------------------------%%
\newpage

Suppose the first-stage regression is given by 
\begin{align}
\boldsymbol{x} = \boldsymbol{\Pi}^{\prime} \boldsymbol{z} + \boldsymbol{v},    
\end{align}
where the instruments $\boldsymbol{z}$ contain both exogenous regressors such as 1 and $q$, and excluded exogenous regressors, $\mathbb{E} \left(  \boldsymbol{v} | \boldsymbol{z} \right) = 0$ and $\mathbb{E} \left( u | \boldsymbol{z} \right) = 0$. Then, by taking the conditional expectation we obtain the following expression  
\begin{align}
\mathbb{E} \left[ y | \boldsymbol{z} \right] 
=
\left( \Pi_0^{\prime} \boldsymbol{z} \right)^{\prime} \beta_{10} \boldsymbol{1} \left\{ q \leq \gamma_0 \right\}
+   
\left( \Pi_0^{\prime} \boldsymbol{z} \right)^{\prime} \beta_{20} \boldsymbol{1} \left\{ q > \gamma_0 \right\} =: \mathsf{g}_{CH} ( \boldsymbol{z}; \theta_0 ), 
\end{align}
where $\theta = \left( \theta^{\prime}, \Pi^{\prime} \right)^{\prime}$. Then, the estimator proposed by \cite{caner2004instrumental} for the endogenous threshold variable $\gamma$ minimizes the sample analogue of the following unconditional condition 
\begin{align}
\mathbb{E} \left[ \big( y - \mathsf{g}_{CH} ( \boldsymbol{z}; \theta ) \big)^2 \right]     
\end{align}
Consider the following GMM estimators which use the moment conditions given below
\begin{align}
\mathbb{E} \left[ \boldsymbol{z} \left( y - \boldsymbol{x}^{\prime} \beta_2 - \boldsymbol{x}^{\prime} \delta_{\beta} \boldsymbol{1} \left\{ q \leq \gamma \right\} \right) \right] = \boldsymbol{0}.    
\end{align}
to identify $\gamma$. Although GMM estimators are essential in handling endogeneity, compared with $M-$estimators, they suffer from at least three drawbacks. First GMM changes the nature of $\gamma$ from a threshold point (which is nonregular) to a quantile of $q$ (which is regular), which implies that the convergence rate of $\widehat{\gamma}$ is $n^{1/2}$, much slower than the convergence rate $n$ of $M-$estimators. Second, $\gamma$ is not always identiable by GMM, for example, when $q$ is independent of the rest of the system such as the time index in structural change models, $\gamma$ cannot be identified by GMM. Third, GMM requires more instruments than our CF estimators for identification, which implies that GMM may have less applicability since good instruments are hard to find in practice. Specifically, the model becomes nonlinear threshold regression, and $\gamma$ can be estimated by minimizing the objective function below
\begin{align*}
S_n ( \theta, k ) = \frac{1}{n} \sum_{i=1}^n \left\{ \big[ y_i - \beta_1^{\prime} \widehat{\mathsf{g}}_{xi} - \kappa . \lambda_1 \left( \gamma - \widehat{\mathsf{g}}_{qi} \right) \big]^2 \boldsymbol{1} \left\{ q_i \leq \gamma \right\} + \big[ y_i - \beta_2^{\prime} \widehat{\mathsf{g}}_{xi} - \kappa . \lambda_2 \left( \gamma - \widehat{\mathsf{g}}_{qi} \right) \big]^2 \boldsymbol{1} \left\{ q_i > \gamma \right\}  \right\}.
\end{align*}
where $\widehat{\mathsf{g}}_{xi} = \widehat{\Pi}_x^{\prime} \boldsymbol{z}_i$  and $\widehat{\mathsf{g}}_{qi} = \widehat{\pi}^{\prime} \boldsymbol{z}_i$, such that the estimators $\widehat{\Pi}_x$ and $\widehat{\pi}$ are obtained from a first-stage regression. The estimation procedure of the parameter $\gamma$ requires to regress $y_i$ on $\mathsf{g}_{xi} \boldsymbol{1} \left\{ q_i \leq \gamma \right\}$ and $\mathsf{g}_{xi} \boldsymbol{1} \left\{ q_i > \gamma \right\}$ where
\begin{align}
\widehat{\Lambda}_i (\gamma) 
:= 
\lambda_1 \left( \gamma - \widehat{\mathsf{g}}_{qi} \right) \boldsymbol{1} \left\{ q_i \leq \gamma \right\} +  \lambda_2 \left( \gamma - \widehat{\mathsf{g}}_{qi} \right) \boldsymbol{1} \left\{ q_i > \gamma \right\}  
\end{align}
to obtain $\widehat{\beta}_1 (\gamma), \widehat{\beta}_2 (\gamma)$ and $\widehat{\kappa} (\gamma)$. Then, $\gamma$ can be estimated by the extremum problem as below
\begin{align}
\widehat{\gamma} = \underset{ \gamma \in \Gamma  }{ \mathsf{arg min}  }  \ S_n(\gamma),  \ \ \text{where} \ \ \ S_n ( \gamma ) = S_n \big( \gamma, \widehat{\gamma}_1 (\gamma), \widehat{\gamma}_2 (\gamma), \widehat{\kappa} ( \gamma ) \big),  \ \ \Gamma = [ \underline{\gamma}, \bar{\gamma} ].   
\end{align}

%%-------------------------------------------------------------------------%%
\newpage

Given $\widehat{\gamma}$, the parameter $\beta$ can be estimated by 2SLS/GMMs as in CH. In particular, in the small-threshold-effect framework of \cite{hansen2000sample}, KST show that $\widehat{\gamma}$ is $n^{1 - 2 \alpha}-$consistent and its asymptotic distribution is based on a functional of two-sided Brownian motion, under the assumption that both $\delta_{\beta}$ and $\kappa$ are $\mathcal{O}\left( n^{ 1 - 2 \alpha} \right)$ with $\alpha \in (0, 1/2)$. Now assume that the unknown parameters $( \gamma, \delta, \kappa )^{\prime}$ lie a compact set with their true value in the interior. Moreover, we define centered versions of $S_n$ and $S$ as 
\begin{align}
\mathcal{Q}_n ( \gamma, \delta, \kappa ) \equiv \mathcal{S}_n ( \gamma, \delta, \kappa ) -  \mathcal{S}_n ( \gamma_0, \delta_0, \kappa_0 )   
\end{align}

\subsubsection{Bootstrap Prediction Intervals for Factor Models}

Assume that $y_{t+h}$ follows a factor-augmented regression model (see, \cite{bai2006confidence}) given by
\begin{align}
y_{t+h} = \alpha^{\prime} F_t + \beta^{\prime} W_t + \varepsilon_{t+h}, \ \ \ t = 1,..., T-h,     
\end{align}
where $W_t$ is a vector of observed regressors (including for instance lags of $y_t$), which jointly with $F_t$ help forecast $y_{t+h}$. Then, the $k-$dimensional vector $F_t$, describes the common latent factors in the panel factor model, $X_{it} = \lambda_i^{\prime} F_t + e_{it}, \ \ \ i = 1,...,N, \ t = 1,...,T$, where the $r \times 1$ vector $\lambda_i$ contains the factor loadings and $e_{it}$ is an idiosyncratic error term. Thus, we can forecast $y_{t+h}$ or its conditional mean $y_{T+h|T} = \alpha^{\prime} F_T + \beta^{\prime} W_T$ using the pair $\big\{ ( y_t, X_t, W_t  ) : t = 1,...,T \big\}$, the available data at time $T$. Since factors are not observed, the diffusion index forecast approach typically involves a two-step procedure (see, \cite{gonccalves2017bootstrap}).

\begin{example}[Estimation for Threshold Models with Integrated Regressors]
Suppose that $\delta_n = n^{ - \frac{1}{2} - \tau } \delta_0$, then the following limiting results hold. If $\tau = \frac{1}{2}$, then $\widehat{\gamma}_n \Rightarrow \gamma ( \gamma_0, \delta_0 )$ and $\gamma ( \gamma_0, \delta_0 )$ is a random variable that maximizes $Q( \gamma, \gamma_0, \delta_0 )$, where 
\begin{align}
\mathcal{Q}( \gamma, \gamma_0, \delta_0 ) = \frac{1}{ F(\gamma) \big( 1 - F(\gamma) \big)} \Gamma_1^{\prime} (\gamma) \left( \int \boldsymbol{B}_v(s) \boldsymbol{B}_v(s)^{\prime} ds   \right)^{-1} \Gamma_1(\gamma),    
\end{align}
To generate the confidence interval of $\gamma$, we consider the following likelihood ratio statistic for the null hypothesis $\gamma = \gamma_0$, given by
\begin{align}
LR_n (\gamma_0) = n \frac{ SSR_n (\gamma_0) - SSR_n ( \widehat{\gamma}_n ) }{ SSR_n ( \widehat{\gamma}_n )  }    
\end{align}
where $\widehat{\gamma}_n$ is the profiled LS estimator. In empirical studies, usually $\tau$ is unknown. Thus, we consider the construction of a robust CI which has approximately correct coverage probability irrespective of the value of $\tau$. For a fixed $\gamma \in [ \underline{\gamma}, \bar{\gamma} ]$, let $X (\gamma) = \left( x_1(\gamma), x_2(\gamma),..., x_n(\gamma) \right)^{\prime}$. Then, the Wald test statistic for testing $H_0: \delta_n = 0$ can be defined as 
\begin{align}
T_n(\gamma) = \widehat{\delta}(\gamma)^{\prime} \big( X^{\prime}(\gamma) ( I - P_n ) X(\gamma)   \big) \widehat{\delta}(\gamma) \big/ \widehat{\sigma}_u^2,   
\end{align}
where $P_n$ is the projection matrix of $X$, given by $P_n = X \left( X^{\prime} X \right)^{-1} X^{\prime}$.
\end{example}

%%-------------------------------------------------------------------------%%
\newpage

\subsection{Simultaneous Confidence Bands}

\begin{example}[Statistical Theory of Simultaneous Confidence Bands]
Consider the nonparametric time series regression model studied by \cite{liu2010simultaneous} as below  
\begin{align}
Y_i = \mu( X_i ) dt + \sigma (X_i) \eta_i    
\end{align}
where $\mu(.)$ is an unknown regression function to be estimated and $( X_i, Y_i )$ is a stationary process and $\eta_i$ are unobserved independent and identically distributed \textit{i.i.d} errors with $\mathbb{E} \eta_i = 0$ and $\mathbb{E} \eta_i^2 = 1$. Moreover consider the Nadaraya-Watson estimator given by 
\begin{align}
\mu_n (x) = \frac{1}{ n b f_n(x) } \sum_{k=1}^n K \left( \frac{ X_k - x }{b} \right) Y_k,    
\end{align}
where $K$ is a kernel function with $K(.) \geq 0$ and $\int_{ \mathbb{R} } K(u) du = 1$, the bandwidths $b = b_n \to 0$ and $n b_n \to \infty$  
\begin{align}
f_n(x) = \frac{1}{n b} \sum_{k=1}^n K \left( \frac{ X_k - x }{b} \right)     
\end{align}
is the kernel density estimate of $f$, the marginal density of $X_i$. Then, under appropriate dependence conditions, in the case of stationary time series processes, the following central limit theorem holds
\begin{align}
\sqrt{nb} \big[ f_n(x) - \mathbb{E} f_n(x) \big] \Rightarrow  \mathcal{N} \big( 0, \lambda_K f(x) \big), \ \ \ \text{where} \ \ \lambda_K = \int_{ \mathbb{R} } K^2 (u) du.       
\end{align}
\end{example}
Notice that the above result can be employed to construct point-wise confidence intervals of $f(x)$ at a fixed $x$. To assess shapes of density functions so that one can perform goodness-of-fit tests, however, one needs to construct \textit{uniform} or \textit{simultaneous confidence bands} (SCB). Therefore, to do this we need to deal with the maximum absolute deviation over some interval $[ \ell, u ]$:
\begin{align}
\Delta_n := \underset{ \ell \leq x \leq u }{ \mathsf{sup} }  \frac{ \sqrt{nb}  }{ \sqrt{\lambda_K f(x) } } \left|  f_n(x) - \mathbb{E} f_n(x)  \right|.      
\end{align}
In practice we first need to study the asymptotic uniform distributional theory for the NW estimator $\mu_n(x)$> Specifically, one needs to find the asymptotic distribution for $\underset{ x \in T }{ \mathsf{sup} } \left| \mu_n(x) - \mu(x)  \right|, \ \ \text{where} \ T = \left[ \ell, u \right] $. Then, building on the aforementioned result one can construct an asymptotic $( 1 - \alpha )$ SCB, where $0 < \alpha < 1$, by finding two functions $\mu_n^{ \mathsf{lower} } (x)$ and $\mu_n^{ \mathsf{upper} } (x)$, such that the following holds
\begin{align}
\underset{ n \to \infty }{ \mathsf{lim} } \ \mathbb{P} \left( \mu_n^{ \mathsf{lower} } (x) \leq \mu(x) \leq \mu_n^{ \mathsf{upper} } (x) \ \ \text{for all} \ x \in T \right) = 1 - \alpha.     
\end{align}
In the standard case the SCB can be used for model validation: one can test whether $\mu(.)$ is of certain parametric functional form by checking whether the fitted parametric form lies in the SCB. 

%However, the main difficulty in the proposed framework is that we have a two-stage estimation procedure due to the well-known statistical problem of elicitability especially in the case when estimating risk measures. 

%%-------------------------------------------------------------------------%%
\newpage 

Let $m = \floor{ n^{\tau} }$, where $\delta_1 / \gamma < \tau < 1 - \delta_1$, and 
\begin{align}
Z_k(t) = Z_k^{\star} \left\{ K \left( \frac{X_k}{b} - t \right) - \mathbb{E} \left[ K \left( \frac{X_k}{b} - t \right) \right] \bigg|  \xi_{k-m, k} \right\}, \ \ \ 1 \leq k \leq n.    
\end{align}
Consequently, we can show that 
\begin{align}
\mathbb{P} \left(  \underset{ 0 \leq t \leq b^{-1} }{ \mathsf{sup} } \left| \sum_{k=1}^{n/2} Z_{2k-1} (t) \right| \geq \sqrt{nb} \left( \mathsf{log} n \right)^{-2} \right) &= o(1).      
\\
\mathbb{P} \left(  \underset{ 0 \leq t \leq b^{-1} }{ \mathsf{sup} } \left| \sum_{k=1}^{n/2} Z_{2k} (t) \right| \geq \sqrt{nb} \left( \mathsf{log} n \right)^{-2} \right) &= o(1).      
\end{align}
Set the following random quantity
\begin{align}
\widetilde{N}_n(t) = \frac{1}{ \sqrt{ n b \lambda_K f(bt) } } \sum_{k=1}^n \mathbb{E} \left[ K \left( \frac{X_k}{b} - t \right) \bigg| \xi_{k-m, k-1} \right] Z^{\star}_k.   
\end{align}

Constructing robust confidence intervals and simultaneous confidence bands has important applications. In particular, it is related to the statistical problem of understanding the trends of extremes and variability of climate variables. By interpreting climate extremes as upper and lower quantiles and climate variability as interpercentile ranges, the nonparametric quantile estimation provides a simple and effective means to address the latter problem (e.g., see \cite{li2022simultaneous}). The simultaneous confidence band (SCB) is a classical tool for nonparametric inference. 

To construct a 100 $(1 - \beta) \%$ SCB for $Q_{\alpha}(.)$, one finds two functions $\ell$ and $u$ depending on $\left( X_i \right)_{i=1}^n$, such that the following condition holds
\begin{align}
\underset{ n \to \infty }{ \mathsf{lim} } \ \mathbb{P} \big( \ell (t) \leq Q_{\alpha}(t) \leq u(t), \ \forall \ t \in (0,1) \big) = 1 - \beta.   
\end{align}

\begin{theorem}
Assume that $\sqrt{n} b_n / \mathsf{log}^5 n \to \infty$, then we have that 
\begin{align*}
\underset{ n \to \infty }{ \mathsf{lim} } \  \mathbb{P} \left[ \underset{ t \in \mathcal{T}_n }{ \mathsf{sup} } \left\{  \frac{ \sqrt{n b_n} f \big( t, Q_{\alpha}(t) \big) }{  \sqrt{\phi} \sigma(t) } \times \bigg| \hat{Q}_{\alpha}(t) - Q_{\alpha} (t)  - \mu_2 b_n^2 Q_{\alpha}^{ \prime \prime } (t) / 2 \bigg|  \right\} - B( m^{*} ) \leq \frac{x}{ \sqrt{2 \mathsf{log} m^{*} } } \right] = e^{ - 2 e^{-x} }   
\end{align*}
where $\mathcal{T}_n = [ b_n, 1 - b_n ], m^{*} = 1 / b_n$.
\end{theorem}

%%-------------------------------------------------------------------------%%
\newpage 

\subsection{Factor Driven Two-Regime Regression}

In this section we briefly discuss the framework of \cite{lee2021factor} who proposed a novel two-regime regression model where regime switching is driven by a vector of possibly unobservable factors. When the factors are latent, these are estimated by the principle component analysis of a panel data set. Then, \cite{lee2021factor}  show that the optimization problem can be reformulated as mixed integer optimization, and we present two alternative computational algorithms. Moreover, they derive the asymptotic distribution of the resulting estimator under the scheme that the threshold effect shrinks to zero. 

Suppose that $y_t$ is generated from 
\begin{align}
y_t &= x_t^{\prime} \beta_0 + x_t^{\prime} \delta_0 \boldsymbol{1} \big\{ f_t^{\prime} \gamma_0 > 0 \big\} + \varepsilon_t,     
\mathbb{E} \left( \varepsilon_t | \mathcal{F}_{t-1} \right) = 0, \ \ \ t = 1,...,T,
\end{align}
where $x_t$ and $f_t$ are adapted to the filtration $\mathcal{F}_{t-1}$, $( \beta_0, \delta_0, \gamma_0 )$ is a vector of unknown parameters and the unobserved random variable $\varepsilon_t$ satisfies the conditional mean restriction. Thus the above model specification is related to the literature on thereshold models with unknown change points. 

In particular, the regression function can be written as:
\begin{align}
y_t =
\begin{cases}
x_t^{\prime} \big( \beta_0 + \delta_0 \big) + \varepsilon_t, & \text{if} \ f_t^{\prime} \gamma_0 > 0,
\\
x_t^{\prime} \beta_0 + \varepsilon_t, & \text{if} \ f_t^{\prime} \gamma_0 \leq 0,
\end{cases}
\end{align}
When the factor $f_t$ is latent, we estimate it using PCA from a potentially much larger dataset, whose dimension is $N$. We illustrate that the asymptotic distribution for the estimator $\alpha_0 \equiv \left( \beta_0^{\prime}, \delta_0^{\prime} \right)^{\prime}$ is identical to that when $\gamma_0$ were known regardless of whether factors are directly observable or not; therefore, the estimator of $\alpha_0$ enjoys an oracle property (see,  \cite{lee2021factor}). 

Furthermore, the asymptotic properties of the proposed estimator are established by adopting a diminishing threshold effect. In other words, we assume that $\delta_0 = T^{ - \varphi } d_0$ for some unknown $\phi \in (0, 1/2)$ and unknown nondiminishing vector $d_0$. Specifically, the unknown parameter $\varphi$ reflects the difficulty of estimating $\gamma_0$ and affects the identification and estimation of the change-point $\gamma_0$. Both the rate of convergence and the asymptotic distribution depends on $\varphi$. In particular, when factors are directly observable, then the distribution of the estimator of $\gamma_0$ is given by
\begin{align}
T^{1 - 2 \varphi }  \left( \widehat{\gamma} - \gamma_0 \right) \overset{d}{\to} \underset{ \mathsf{g} \in \mathcal{G} }{ \mathsf{arg \ min} }  \ B( \mathsf{g}) + 2 W(\mathsf{g}), 
\end{align}

%%-------------------------------------------------------------------------%%
\newpage

\subsection{Forecasting with Dynamic Panel Data Models}

Consider the linear dynamic panel data model given by 
\begin{align}
Y_{it} = \lambda_i + \rho Y_{it-1} + U_{it}    
\end{align}
as a special case of the econometric specification studied by \cite{liu2020forecasting}. Then, a suitable estimator for $\hat{\rho}$ is given by the truncated instrumental variable (IV) estimator such that 
\begin{align}
\hat{\rho}_{IV} =  \left[ \left( \sum_{i=1}^N \sum_{t=2}^T Y_{it-2} \Delta Y_{it-1} \right)^{-1} \right]^{M_N}  \times  \left( \sum_{i=1}^N \sum_{t=2}^T Y_{it-2} \Delta Y_{it} \right)
\end{align}
where $M_N$ is a sequence that slowly diverges to infinity. Define the residuals $\hat{U}_{it} = Y_{it} - \hat{\lambda} \hat{\rho}_{IV} - \hat{\rho}_{IV} \hat{Y}_{it-1}$

\paragraph{Pooled-OLS Predictor}

Ignoring the heterogeneity in the $\lambda_i$'s and imposing that $\lambda_i = \lambda$ for all $i$, we can define that 
\begin{align}
\left( \hat{\rho}_{p}, \hat{\lambda}_{p}  \right) = \underset{ \rho, \lambda }{ \mathsf{arg min} }  \sum_{i=1}^N \sum_{t=1}^T \left( Y_{it} - \rho Y_{it-1} - \lambda \right)^2.  
\end{align}

A key point here is to ensure that the empirical model is able to accurately predict bank revenues and balance sheet characteristics under observed macroeconomic conditions. 

\paragraph{Compound Risk}

\begin{align}
L_N ( \widehat{Y}^{N}_{T+1}, Y_{T+1}^N ) = \sum_{i=1}^N \left( \widehat{Y}_{iT+1} - Y_{iT+1} \right)^2,    
\end{align}

\paragraph{Ratio Optimality}

The convergence is uniform with respect to the correlated random effects distributions $\pi$ in some set $\Pi$. The uniformity holds in the neighborhood of point masses for which the prior and posterior variances of $\lambda_i$ are zero. Thus, the convergence statement covers the case of $\lambda_i$ being homogenous across $i$. Then, the autoregressive coefficient in the basic dynamic panel model can be $\sqrt{N}-$consistently estimated, which suggests that $\sum_{i=1}^N \left( \hat{\rho} - \rho \right)^2 Y_{iT}^2 = \mathcal{O}_p(1)$. Then, the discrepancy between the predictor $\widehat{Y}_{iT+1}$ and $\lambda_i + \rho Y_{it}$ can be decomposed into three terms is given by
\begin{align}
\widehat{Y}_{iT+1} - \lambda_i - \rho Y_{iT} 
= 
\hat{\rho}_i ( \hat{\rho} ) + \left( \frac{\hat{\sigma}^2 }{T} + B_N^2  \right) \frac{ \partial }{ \partial \hat{ \lambda }_i ( \hat{\rho} ) }     
\end{align}

%%-------------------------------------------------------------------------%%
\newpage 

\section{High Dimensional Panel Data Regression Models}
\label{Section7}

Additional discussion on some aspects to identification, estimation and inference for high dimensional models is presented in \cite{katsouris2023high}. Further related studies from the perspective of shrinkage-based estimation using GMM methods include among others \cite{cheng2015select}. In this section we focus on some specific applications to panel data regression models such as the framework proposed by  \cite{lu2016shrinkage}. The literature on discovering latent structures in panel data models include \cite{su2016identifying} as well as the study of \cite{bing2022inference}.

\subsection{Shrinkage Estimation of Dynamic Panel Regression with interactive FE}

The particular study considers the problem of determining the number of factors and selecting the proper regressors in linear dynamic panel data models with interactive fixed effects (see, \cite{lu2016shrinkage}). The authors propose a methodology for simultaneous selection of regressors and factors and estimation through the method of adaptive group Lasso. In particular, we show that with probability approaching one, the proposed method correctly select all relevant regressors and factors and shrink the coefficients of irrelevant regressors and redundant factors to zero. 

Therefore, the model in matrix form can be written as below: 
\begin{align}
\boldsymbol{Y} = \sum_{j=1}^K \beta_j^{0} \boldsymbol{X}_j + \lambda^0 F^{0 \prime} + \boldsymbol{\varepsilon}.    
\end{align}
Without loss of generality we assume that only the first $K_0$ elements of $X_{it}$ have nonzero slope coefficients, and express the vector of regressors $X_{it} = \big( X_{(1)it}^{\prime},  X_{(2)it}^{\prime} \big)^{\prime}$, where $X_{(1)it}^{\prime}$ and $X_{(2)it}^{\prime}$ are $( K_0 \times 1 )$ and $( K - K_0)  \times 1$ vectors respectively, and the true coefficients of $X_{(1)it}^{\prime}$ are nonzero while those of $X_{(2)it}^{\prime}$ are assumed to be all zero. Acoordingly, we decompose the parameter vector $\beta^0$ as $\beta^0 = \left( \beta^{0 \prime}_{ (1) },  \beta^{0 \prime}_{ (2) }  \right)^{\prime} = \left( \beta^{0 \prime}_{ (1) }, 0    \right)$. Consider the Gaussian QMLE $\left(  \tilde{\beta}, \tilde{\lambda}, \tilde{F} \right)$ of $\left(  \beta^0, \lambda^0, F^0 \right)$ which is given by 
\begin{align}
\left(  \tilde{\beta}, \tilde{\lambda}, \tilde{F} \right) = \underset{ ( \beta, \lambda, F )  }{ \mathsf{arg \ min}  } \ \mathcal{L}_{NT}^0  ( \beta, \lambda, F ),  
\end{align}
where
\begin{align}
\mathcal{L}_{NT}^0  ( \beta, \lambda, F ) \equiv \frac{1}{NT} \mathsf{trace} \left[ \left(  \boldsymbol{Y} -  \sum_{j=1}^K \beta_j^{0} \boldsymbol{X}_j - \lambda F^{\prime} \right)^{\prime} \left(  \boldsymbol{Y} -  \sum_{j=1}^K \beta_j^{0} \boldsymbol{X}_j - \lambda F^{\prime} \right) \right] 
\end{align}
where $\beta \equiv \left(  \beta_1,..., \beta_K \right)^{\prime}$ is a $(K \times 1)$ vector and $F \equiv \left( F_1,..., F_T \right)$ is a $(T \times R)$ and $\lambda \equiv \left( \lambda_1,..., \lambda_N \right)^{\prime}$ is an $(N \times R)$ matrix. Therefore, further aspects of the statistical optimization methodology will need to tackle the estimation  of the unknown factors in the high dimensional regression model.

%%-------------------------------------------------------------------------%%
\newpage

In particular, a possible solution is the use of group Lasso for selection of the number of factors. Thus, considering the simultaneous variable and factor selection in a dynamic panel data model along with relevant asymptotic theory analysis is the aim of this section. 

Consider the following expression 
\begin{align*}
\sqrt{NT} \mathbb{C}_{ K_0 } \left(  \hat{\beta}_{(1)} - \beta_{(1)}    \right) 
&=
\mathbb{C}_{ K_0 } \hat{D}^{-1}_{ \hat{F}(1)  } \frac{1}{ \sqrt{NT} } \sum_{i=1}^N X_{(1)i}^{\prime} M_{ F^0 } \varepsilon_i
\\
&\ \ \ +
\mathbb{C}_{ K_0 } \hat{D}^{-1}_{ \hat{F}(1)  } \frac{1}{ \sqrt{NT} } \sum_{i=1}^N X_{(1)i}^{\prime} M_{ \hat{F}(1)  }  \left( F_{(1)}^{*} - \hat{F}_{(1)} \right) \lambda^{*}_{(1) i}
\\
&\ \ \ +
\mathbb{C}_{ K_0 } \hat{D}^{-1}_{ \hat{F}(1)  } \frac{1}{ \sqrt{NT} } \sum_{i=1}^N X_{(1)i}^{\prime} \left(  M_{ \hat{F}(1) }  - M_{ F^{*}(1)}     \right) \varepsilon_i + o_p(1).
\end{align*}

\medskip

\begin{remark}
Notice that it is important to determine whether or not these moment functions are centered around 0 asymptotically. For example, a strategy to demonstrate whether this property holds is to decompose a moment function into an asymptotic bias term and an asymptotic variance term.  Usually the former term converges to a zero mean normal distribution, while the conditional expectation of the latter term  contributes to the asymptotic bias which can be corrected and so the corresponding term after substracting its mean is asymptotically negligible.    
\end{remark}

Notice that we define with $\norm{A}^2_{\mathsf{sp}} \equiv \lambda_1 \left( A^{\prime} A \right)$, which implies that $\norm{A}_{\mathsf{sp}} \equiv \sqrt{ \lambda_1 \left( A^{\prime} A \right) }$, where $\lambda_1$ denotes the largest eigenvalue of a real symmetric matrix $A$. Moreover, we use $\lambda_{ \mathsf{min} } (A)$ and  $\lambda_{ \mathsf{max} } (A)$ to denote the smallest and largest eigenvalues of symmetric matrix $A$, respectively.

Denote with 
\begin{align}
\hat{\boldsymbol{Y}} &= \left( \boldsymbol{Y} - \sum_{j=1}^K \tilde{\beta}_j^{c} \boldsymbol{X}_j  \right) 
\\
\hat{F} &= \frac{1}{NT} \hat{\boldsymbol{Y}} \hat{\boldsymbol{Y}}^{\prime} \tilde{F}
\\
\hat{D}_{ \hat{F}(1) } &= \frac{1}{NT} \sum_{i=1}^N X^{\prime}_{(1)u} M_{ \hat{F}(1) } X_{(1)i} 
\\
\hat{\Sigma}_{ \hat{F}(1) }  &= \frac{1}{T} \sum_{t=1}^T \hat{F}_{(1)t} \hat{F}_{(1)t}^{\prime}
\end{align}
We also partition the variance matrix such that $V_{NT} \equiv \mathsf{diag} \big( V_{(11), NT}, V_{(22), NT}  \big)$.

%%-------------------------------------------------------------------------%%
\newpage

\subsection{Robust IV Estimation and Variable Selection}

In this section we study estimation and variable selection methodologies for high dimensional linear regression models under the presence of endogeneity. In particular, this setting requires to apply an IV estimation approach to tackle the aspect of endogeneity. Several studies in the literature consider the construction of confidence intervals for estimators of interest, under the assumption that all the IVs are valid after controlling for the said covariates.  Moreover, in invalid IV settings different statistical frameworks were developed to provide robust inferential methods. In other words, a propose statistical procedure under the presence of endogeneity especially in a high dimensional estimation setting should be valid even under the presence of invalid instruments.

\subsubsection{Asymptotic Theory}

In terms of the asymptotic theory intially we need to examine the stability conditions of the system such that $\mathsf{max}_{ i = 1,..., v } \left| \lambda_i \left( \boldsymbol{C} \right) \right| < 1$, then for any $\boldsymbol{d}$ and all $t = 1,2,3,...$ the norm 
\begin{align}
\norm{ \left( - \boldsymbol{C}^t \boldsymbol{d} \right) } < C \lambda^t, \ \ \ \text{for some} \ \  C < + \infty \ \ \ \text{and} \ \ \ \lambda < 1.      
\end{align}
%Next, based on the norm we can bound the matrix $\norm{ \boldsymbol{R}_T } \leq \frac{1}{T} \sum_{t=1}^T C^2 \lambda^{2t} \frac{1}{T} \frac{ ( C \lambda )^2 }{ \left( 1 - \lambda^2 \right) }$.  Therefore, we can prove that $\boldsymbol{R}_T \to 0$ as $T \to \infty$, and thus the large sample behaviour of the partitioned matrix $\lef(  \bar{\boldsymbol{A}}_T \ \vdot \ \bar{\boldsymbol{M}}_T \right)$ will reflect that $\boldsymbol{y}_t$ is asymptotically stationary. If $\boldsymbol{C}$ has at least one eigenvalue on the unit circle.

\begin{theorem}
Under the assumption that $\norm{ \boldsymbol{N}_T \left( \tilde{\boldsymbol{\beta}}^{(i)}_T - \hat{\boldsymbol{\beta}}_T^{(i)} \right) } = o_p(1)$ for all $i \geq 1$ as $T \to \infty$. Moreover, both estimators converge in distribution to $\boldsymbol{\beta}_T = \left( \boldsymbol{\theta}_T^{\prime}, \boldsymbol{\gamma}_T^{\prime}  \right)^{\prime}$ where 
\begin{align}
\sqrt{T} \left( \boldsymbol{\theta}_T - \boldsymbol{\theta} \right) \overset{ d }{ \to }  \mathcal{N} \left( \boldsymbol{0}, \boldsymbol{V}_{\theta}^{-1} \right),   
\end{align}
\begin{align}
\boldsymbol{V}_{\theta} := \underset{ T \to \infty }{ \mathsf{lim} } \ \mathbb{E} \left[ \sum_{t=1}^T \boldsymbol{W}_{\theta t}^{\prime} \boldsymbol{\Sigma}_{\varepsilon}^{-1} \boldsymbol{W}_{\theta t} \right]
\end{align}
which is the asymptotic information matrix for $\boldsymbol{\theta}$, the vectors $\sqrt{T} \left( \boldsymbol{\theta}_T - \boldsymbol{\theta}  \right)$ and $T \left( \boldsymbol{\gamma}_T - \boldsymbol{\gamma} \right)$ are asymptotically mutually uncorrelated and it holds that 
\begin{align}
T  \left( \boldsymbol{\gamma}_T - \boldsymbol{\gamma} \right) = T \mathsf{vec} \left( \boldsymbol{\Gamma}_{\rho, T} - \boldsymbol{\Gamma}_{\rho} \right)   
\end{align}
where the components of $\left( \boldsymbol{\Gamma}_{\rho, T} - \boldsymbol{\Gamma}_{\rho} \right)  $ satisfy the asymptotic mixed-normality result such that 
\begin{align}
\mathsf{vec} \left( \left[ \sum_{t=1}^T \boldsymbol{H}^{\prime} \boldsymbol{z}_{t-1} \boldsymbol{z}_{t-1}^{\prime} \boldsymbol{H} \right]^{1/2} \left[ \boldsymbol{\Gamma}_{ - \rho, T } \ \ \  - \boldsymbol{\Gamma}_{\rho} \right]  \right) \overset{ d }{ \to } \mathcal{N} \left(  \boldsymbol{0}, \boldsymbol{V}_{\gamma} \right), 
\end{align}
\begin{align}
\boldsymbol{V}_{\gamma} = \bigg( \left( \boldsymbol{\mathcal{Y}}^{\prime} \big[ \boldsymbol{M}(1) \boldsymbol{\Sigma}_{\varepsilon} \boldsymbol{M}(1)^{\prime} \big]^{-1} \boldsymbol{\mathcal{Y}} \right) \otimes \boldsymbol{I}_{( v + u - \rho )} \bigg).   
\end{align}
\end{theorem}

%%-------------------------------------------------------------------------%%
\newpage

Therefore, the above theorem can be applied such that $\boldsymbol{V}_{\theta}$ and $\boldsymbol{V}_{\gamma}$ can be consistently estimated by replacing the unknown parameters of the estimator $\boldsymbol{\beta}$ by their Gaussian maximum likelihood values and substituting a consistent estimate for $\boldsymbol{\Sigma}_{\varepsilon}$, while dropping the expectation from $\boldsymbol{V}_{\theta}$. In practice we are interested to obtain simple estimates of the covariance matrix $\boldsymbol{\Sigma}_{\varepsilon}$.

\begin{example}
A partially linear IV regression (PLIV) model take the following form: 
\begin{align}
Y &= D \theta_0 + g_0 (X) + \zeta
\\
Z &= m_0 (X) + V
\end{align}
where $\mathbb{E} [ \zeta | Z, X ] = 0$ and $\mathbb{E} [ V | X ] = 0$. Let $Y$ be the outcome variable of interest, $D$ is the policy variable of interest and $Z$ denotes a scalar or a vector of instrumental variables. Moreover, the high dimensional vector $X = ( X_1,..., X_p )$ consists of other confounding covariates and $\zeta$ and $V$ are stochastic errors. The R implementation has as inputs the following functions:

\begin{itemize}

\item $\texttt{dml procedure}$: A character(.) ("dml1" or "dml2") specifying the double machine learning algorithm.

\item $\texttt{draw sample splitting}$: Indicates whether the sample splitting should be drawn during initialization of the object.

\item $\texttt{learner}$: The machine learners for the nuisance functions.

\item $\texttt{n folds}$: The number of folds (with default being equal to 5). 

\item $\texttt{n rep}$: The number of repetitions for the sample splitting.  

\item $\texttt{psi}$: The value of the score function component given by $\psi_a ( W; \theta, \eta ) = \psi_{a} ( W; \eta ) \theta + \psi_b ( W; \eta )$.
    
\end{itemize}

\end{example}

%%-------------------------------------------------------------------------%%
\newpage 

\appendix
\numberwithin{equation}{section}
\makeatletter

%%-------------------------------------------------------------------------%%
\newpage 

\section{Elements on Weak Instrumentation and Differential Geometry}

\subsection{Elements on Weak Instrumentation}

In this section we present key properties in relation to the asymptotic optimality of the LIML estimator, which is commonly used as an optimal estimation strategy when modeling panel data (see,  \cite{anderson2010asymptotic} and \cite{akashi2012some}). More precisely, when the instruments are weak, in general $\hat{\beta} (k)$ is not consistent and has a nonstandard asymptotic distribution. Moreover, $T \left( \hat{k}_{LIML} - 1 \right)$ has a nondegenerate asymptotic distribution such that $\hat{\beta}_{TSLS}$ and $\hat{\beta}_{LIML}$ are not equivalent under weak instrument asymptotics. The asymptotic distribution of the test statistics is nonstandard.  Note that the decorrelated quasi-score function $\bar{ \boldsymbol{S} }_n ( \boldsymbol{\theta} )$ is of dimension $d_0 K$ instead of dimension $d K$. In particular, given our initial estimator $\widehat{ \boldsymbol{ \beta } } = \big( \widehat{ \boldsymbol{ \theta } }^{\top} , \widehat{ \boldsymbol{ \gamma } }^{\top} \big)^{\top}$, we define our QDIF estimator given as below
\begin{align}
\tilde{ \boldsymbol{ \theta } } = \underset{ \boldsymbol{ \theta } \in \Theta_n }{ \mathsf{arg min} } \tilde{Q}_n \left( \boldsymbol{ \theta } \right), \ \ \ \text{where} \ \ \ \tilde{Q}_n \left( \boldsymbol{ \theta } \right) = n \bar{\boldsymbol{S}}_n \left( \boldsymbol{ \theta } \right)^{\top} \boldsymbol{C}^{-1} \bar{\boldsymbol{S}}_n \left( \boldsymbol{ \theta } \right). 
\end{align} 
In particular, $\Theta_n := \left\{ \boldsymbol{\theta} \in \mathbb{R}^{d_0} : \norm{ \boldsymbol{\theta} - \widehat{ \boldsymbol{\theta} } }_2  \leq c d_0^{-1 / 2} \right\}$ is a neighbourhood around the initial estimator  $\widehat{ \boldsymbol{\theta} }$ for some small constant $c > 0$ and
\begin{align}
\boldsymbol{C} := \frac{1}{n} \sum_{i=1}^n \bar{\boldsymbol{S}}_i \left( \boldsymbol{ \theta } \right) \bar{\boldsymbol{S}}_i \left( \boldsymbol{ \theta } \right)^{\top} \in \mathbb{R}^{ d_0 K \times d_0 K }. 
\end{align}
To alleviate these issues, a suitable approach proposed in the literature is the above localized estimator obtained by minimizing $\tilde{Q}_n ( \boldsymbol{\theta} )$ in a small neighbourhood around the initial estimator $\boldsymbol{\theta}$. Therefore, it can be shown that $\tilde{Q}_n ( \boldsymbol{\theta} )$ is strongly convex for $\boldsymbol{\theta}  \in \Theta_n$ with probability tending to one. Thus, any off-the-shelf convex optimization algorithm is applicable to solving the above problem. 

Consider the following hypothesis testing problem
\begin{align}
H_0: \boldsymbol{ \theta }^{*} = 0 \ \ \text{versus} \ \ H_1: \boldsymbol{ \theta }^{*} \neq 0
\end{align}
Based on the above result, we define the Wald-type test statistic as follows
\begin{align}
\widehat{T}_n = n \tilde{ \boldsymbol{\theta} }^{\top} \widehat{ \boldsymbol{ \Sigma } }^{-1}_{ \boldsymbol{\theta} } \tilde{ \boldsymbol{\theta} }.
\end{align}
Lemma above implies that the distribution of the test statistic $\widehat{T}_n$ can be approximated by a chi-square distribution with $d_0$ 
\begin{align}
\underset{ 1 \leq k \leq K }{ \mathsf{max} } \norm{ \frac{1}{n} \sum_{i=1}^n \boldsymbol{Z}_i \boldsymbol{\Psi}_i  \boldsymbol{Z}_i - \mathbb{E} \big( \boldsymbol{Z}_i \boldsymbol{\Psi}_i  \boldsymbol{Z}_i \big) }_2 = \mathcal{O}_{ \mathbb{P} } \left( \sqrt{ d_0 \mathsf{log} d_0 / n } \right),
\end{align}
$\sqrt{n} \tilde{\beta}_j$ converges to $\frac{1}{ \sqrt{n} } \sum_{ i = 1 }^n A_{ij}$ such that as $n \to \infty$, $\underset{  j \in \mathcal{H}_0 }{ \mathsf{max} } \left| \sqrt{n} \tilde{\beta}_j - \frac{1}{ \sqrt{n} } \sum_{ i = 1 }^n A_{ij}  \right| 
= 
\mathcal{O}_{ \mathbb{P} } \left( \sqrt{ d_0 \mathsf{log} d_0  n } \right)$.

%%-------------------------------------------------------------------------%%
\newpage

\subsection{Elements of Differential Geometry}

\subsubsection{A geometric Approach to Nonlinear Models}

The connection of econometric models with differential geometry is discussed by some authors such as in several studies of Professor Grant Hiller (Professor of Econometrics, University of Southampton) as in  \cite{hillier1987classes} ,  \cite{hillier2000exact}, \cite{armstrong1999density} and \cite{van1997curved} among others. Moreover, such perspectives tailored in the case of nonstationary time series models exist in the literature (some references are discussed by \cite{katsouris2023bootstrapping}). 

Following \cite{andrews2016geometric}, consider the region 
\begin{align*}
B_R( x_0 ) = \left\{ x \in \mathbb{R}^k : \norm{ x - x_0 } \leq \left( 1 + \sqrt{2} \right) R \right\}
\end{align*}
is a $k-$dimensional ball of radius $(1 + \sqrt{2}) R$ with center $x_0$, and $|A|$ is the cardinality of a set $A$.

In particular,  \cite{andrews2016geometric} consider the problem of testing a potentially nonlinear hypothesis on the mean of a multivariate Gaussian vector. Assume that we observe a $k-$dimensional Gaussian vector $\widehat{\theta}$ with known covariance matrix $\Sigma$ and unknown mean $\theta_0$. We wish to test a $p-$dimensional restriction which may be formulated either as $g( \theta_0 ) = 0$ for some $(k-p)-$dimensional smooth function $g$ or as $\theta_0 = \theta ( \beta_0 )$ for a known link function $\theta (.)$ and some unknown $p-$dimensional parameter $\beta_0$lying in a parameter space $U \subset \mathbb{R}^p$. 

Such testing problems arise in many contexts, for example in testing hypotheses with nuisance parameters or in testing model specification. In particular, this is the limiting testing problem in many weakly identified minimum-distance models, as well as cases when one fits a highly nonlinear structural model based on reduced-form parameter estimates. For example, let $\widehat{\theta}$ be a preliminary or reduced-form estimator, which is approximately normal with a well-estimable covariance matrix $\Sigma$. Assume that the relationship between the structural and reduced-form parameters is described by the link function $\theta ( \beta)$ for structural parameter $\beta$. Then, testing the correct model specification (asymptotically) is equivalent to testing that $\theta_0 = \theta ( \beta_0 )$ for some $\beta_0 \in U$. Alternatively, if there are two structural parameters $\lambda$ and $\beta$ with link function $\theta ( \lambda, \beta  )$, then testing a hypothesis about $\lambda$ alone, $H_0: \lambda = \lambda_0$, is equivalent to testing $H_0: \theta_0 = \theta( \lambda_0, \beta)$ for some $\beta \in U$. The statistical inference is based on the minimum-distance (or for the exact Gaussian case, likelihood ratio) statistic, which may be formulated as   
\begin{align*}
MD &= \underset{ \theta: g(\theta) = 0}{ \text{min} } \left( \widehat{\theta} - \theta \right)^{\prime} \Sigma^{-1} \left( \widehat{\theta} - \theta \right)
\\
\nonumber
\\
MD &= \underset{ \theta: g(\theta) = 0}{ \text{min} } \left( \widehat{\theta} - \theta(\beta) \right)^{\prime} \Sigma^{-1} \left( \widehat{\theta} - \theta(\beta) \right)
\end{align*} 
depending on the formulation of the null hypothesis. 

%%-------------------------------------------------------------------------%%
\newpage 

\cite{andrews2016geometric} introduce the normalized random vector $\xi = \Sigma^{- 1/ 2} \left( \widehat{\theta} - \theta_0 \right) \sim \mathcal{N} (0, I_k)$ and the $p-$dimensional manifold 
\begin{align}
\mathcal{S} = \left\{ x : x = \Sigma^{- 1/ 2} \left( \theta(\beta) - \theta_0 \right), \beta \in \mathbb{R}^p  \right\} \ \text{or} \ \mathcal{S} = \left\{ x : x = \Sigma^{- 1/ 2} \left( \theta(\beta) - \theta_0 \right), g(\theta) = 0 \right\}
\end{align}
Notice that the manifold $\mathcal{S}$ is known-up to a location shift determined by the true value $\theta_0$. Therefore, under the null hypothesis we know the shape of the manifold $\mathcal{S}$ and that is passes through the origin. Then, the minimum distance statistics defined above are simply the squared distance between $\xi$ and $\mathcal{S}:$ 
\begin{align}
MD = \underset{ x \in \mathcal{S} }{ \text{min} } \ \left( \xi - x \right)^{\prime} \left( \xi - x \right) = \rho^2 \left( \xi, \mathcal{S}\right) 
\end{align}
where $\rho$ is the Euclidean distance from a point to a set. The distribution of $\rho^2 \left( \xi, \mathcal{S}\right)$ is, in general, nonstandard and depends on the unknown parameter $\theta_0$.  The central issue of this paper is how to find computationally tractable critical values such that tests based on $\rho^2 \left( \xi, \mathcal{S}\right)$ control size. According to \cite{andrews2016geometric} the main challenges that one faces is that the location of the true value of $\theta_0$ on the null manifold are unknown which implies that the asymptotic distribution of the particular test statistic will depend on this unknown parameter (\textit{nuisance parameter problem}). 

Then, \cite{andrews2016geometric} distinguish between the linear and nonlinear cases. If $\mathcal{S}$ is a $p-$dimensional linear subspace in $\mathbb{R}^k$, then the squared distance $\rho^2 \left( \xi, \mathcal{S} \right)$ has a $\chi^2_{k - p}$ distribution. Most of the classical statistics literature deals with testing hypotheses that are either linear or asymptotically linear, in the sense that $\mathcal{S}$ is either a linear subspace or is well-approximated  by one in large samples. In particular, classical delta-method arguments assume that the reduced-form parameter is precisely estimated relative to the nonlinearity of the null hypothesis manifold, and thus we can linearize the null hypothesis manifold around the true parameter value (e.g., see \cite{gafarov2018delta}). Moreover, the authors also allow for cases where the nonlinearity of the model is important relative to the sampling error of the reduced-form parameter estimates, rending linear approximations unreliable. Another potential source of nonlinearity in $\mathcal{S}$ is \textit{weak identification}.         

One bound can be placed on $\rho^2 \left( \xi, \mathcal{S} \right)$ without any assumptions, namely that $\rho^2 \left( \xi, \mathcal{S} \right)$ is dominated by a $\chi^2_k$ distribution. Indeed, since $0 \in \mathcal{S}$, 
\begin{align}
\rho \left( \xi, \mathcal{S} \right)^2 = \underset{ x \in \mathcal{S} }{ \text{min} } \left( \xi - x \right)^{\prime} \left( \xi - x \right) \leq \left( \xi - 0 \right)^{\prime} \left( \xi - 0 \right) \sim \chi^2_k.
\end{align}
Using the particular bound, it gives the projection method, which is currently the main approach available for testing with weakly identified nuisance parameters. This paper, proposes new critical values based on a stochastic bound on the distribution of the MD statistic. There critical values are smaller that those used by the projection method and coincide with $\chi^2_{k-p}$ critical values for linear hypotheses. More specifically, this bound is based on measuring the curvature of the null hypothesis relative to the variance $\Sigma$ of the reduced-form parameter estimates (\cite{andrews2016geometric}).

%%-------------------------------------------------------------------------%%
\newpage

\subsubsection{Geometric Concepts}

\cite{andrews2016geometric} focus on regular manifolds embedded in $k-$dimensional Euclidean space. A subset $\mathcal{S} \subset \mathbb{R}^k$ is called a $p-$\textit{dimensional regular manifold} if for each point $q in \mathcal{S}$, there exists a neighborhood $V \in \mathbb{R}^k$ and a twice continuously-differentiable map $\mathbf{x}: \widetilde{U} \to V \cap \mathcal{S}$ from an open set $\widetilde{U} \subset \mathbb{R}^p$ onto $V \cap \mathcal{S} \subset \mathbb{R}^k$ such that (i) $\mathbf{x}$ is a homeomorphism, which is to say it has a continuous inverse and (ii) the Jacobian $d\mathbf{x}_q$ has full rank. A mapping that satisfies these conditions is called a parametrization or a system of local coordinates. For $\mathbf{x}$ a system of local coordinates at $q$, the set of all tangent vectors to $\mathcal{S}$ at $q$ coincides with the linear space spanned by the Jacobian $d\mathbf{x}_q$ and is called the \textit{tanget space} to $\mathcal{S}$ at $q$ (denoted $T_q(\mathcal{S})$). Denote by $\gamma: ( - \epsilon, \epsilon ) \to \mathcal{S}$ a curve which lies in $\mathcal{S}$ and passes through $q = \gamma(0)$. The measure of curvature we consider is as below
\begin{align}
\kappa_q ( \mathcal{S} ) 
= \underset{ X \in T_q(.), \dot{ \gamma }(0) = X  }{ \text{sup}} \ \kappa_q \left( \gamma , \mathcal{S} \right) \equiv \underset{ X \in T_q(.), \dot{ \gamma }(0) = X  }{ \text{sup} } \ \frac{ \left( \ddot{ \gamma}(0) \right)^{\perp} }{ \norm{ \dot{\gamma}(0)}^2 }, 
\end{align} 
where $( W )^{\perp}$ stands for the projection of $W$ onto the space orthogonal to $T_q(S)$. This measure of curvature is equal to the maximal curvature over all geodesics passing through the point $q$ and is invariant to the parametrization. If $\mathcal{S}$ is $p-$dimensional sphere of radius $C$, then, for each $q \in \mathcal{S}$, we have that $\kappa_{ \mathcal{S} } = 1 / C$. If, on the other hand, $\mathcal{S}$ is a linear subspace, its curvature is zero at all points.  

\paragraph{Calculating the Curvature}

Let $\mathcal{S}$ be a $p-$dimensional manifold in $\mathbb{R}^k$, and $\mathbb{X}$ a local parametrization at a point $q$, $q = \mathbf{x} (y^{*})$. Denote the derivatives of $\mathbf{x}$ at $q$ by $v_i = \frac{\partial \mathbf{x} }{ \partial y_i } \left( y^{*} \right)$, and let $Z = \left( v_1,..., v_p  \right)$. For any vector $W \in \mathbb{R}^k$, let $W^{\perp} = N_{Z} W := \left( I - Z \left( Z^{\prime} Z \right)^{-1} \right) W$. Finally, denote the second derivatives of $V_{ij} = \frac{ \partial^2 }{ \partial y_i \partial y_j } \mathbf{x} \left( y^{*} \right)$. The curvature can then be written as below (see, \cite{andrews2016geometric})
\begin{align}
\kappa_q \left( \mathcal{S} \right) = \underset{ \norm{ \sum_{i=1}^p u_i v_i } = 1 }{ \text{sup} } \norm{ \sum_{i,j = 1}^p  u_i v_i V_{ij}^{\perp} } 
\equiv 
\underset{ \left( w_1,..., w_p \right) \in \mathbb{R}^p }{ \text{sup} } \frac{ \displaystyle  \sum_{i,j = 1}^p w_i w_j V_{ij}^{\perp}  }{ \displaystyle  \norm{ \sum_{i=1}^p w_i v_i }^2  }, \ \ u = \left( u_1,..., u_p \right) \in \mathbb{R}^p.
\end{align}
 
\paragraph{Geometric Bounds}

Moreover, the distance in $\mathbb{R}^k$ from a random vector $\xi \sim \mathcal{N} \left( 0, I_k \right)$ can be bounded such that corresponds to a $p-$dimensional nonrandom manifold $\mathcal{S}$ that contains zero. Then, the curvature measure depends on the maximal curvature $\kappa_q \left( \mathcal{S} \right)$ over all relevant points in the manifold $\mathcal{S}$. The bound depends on global properties of the manifold, in the sense of properties that hold on a fixed bounded set, but the behaviour of the manifold at infinity is irrelevant (see, \cite{andrews2016geometric}). Recently, estimation and inference techniques in the case of nonlinear models with general forms of potential identification failure which implies the presence of a nearly singular Jacobian matrix are discussed by \cite{han2019estimation} who propose a general linear rotation-based reparmetrization approach to correct for identification failures and singularities.

%%-------------------------------------------------------------------------%%
\newpage

\begin{theorem}[\cite{andrews2016geometric}]
Let $\mathcal{S}$ be a regular $p-$dimensional manifold in $\mathbb{R}^k$ passing through zero. Assume that the tangent space $T_0 ( \mathcal{S} )$ is spanned by the first $p$ basis vectors. Assume that for some constants, $C > 0$, we have that $\kappa_q \left(  \mathcal{S} \right) \leq \frac{1}{C}$ for all points $q \in \mathcal{S}_C$. Then, 
\begin{itemize}
\item[(a)] Manifold $\mathcal{S}_C$ lies inside the set $\mathcal{S} \cap D_C$, where 
\begin{align}
\mathcal{M} = \left\{ \norm{ x^{(1)} }^2 + \left( C - \norm{ x^{(2)}} \right)^2 \geq C^2 \right\}. 
\end{align} 

\item[(b)] If Assumption 1 is satisfied, then, for any point $\xi \in \mathbb{R}^k$, we have almost surely that 
\begin{align}
\rho \left( \xi, \mathcal{S} \right) \leq \underset{ u \in \mathbb{R}^{p-k}, \norm{u} = 1 }{ \text{max} } \rho \left( \xi, \mathcal{S} \right), 
\end{align} 
where $N_u = \left\{ x \in \mathbb{R}^k : x = \left( x^{(1)}, zu \right), x^{(1)} \in \mathbb{R}^p, z \in \mathbb{R}_{+} , \norm{  x^{(1)} }^2 + \left( C - z \right)^2 = C^2 \right\}$

\item[(c)] Almost surely, max$_{ u \in \mathbb{R}^{p-k}, \norm{u} = 1 } \rho \left( \xi, N_u \right) = \rho \left( \xi, N_{ \widetilde{u} } \right)$
\begin{align}
\rho \left( \xi, \mathcal{S} \right) \leq \underset{ u \in \mathbb{R}^{p-k}, \norm{u} = 1 }{ \text{max} } \rho \left( \xi, \mathcal{S} \right), \ \ \text{where} \ \widetilde{u} = - \frac{1}{ \norm{ \xi^{(2)} } } \xi^{(2)}. 
\end{align} 

\item[(d)] If $\xi \sim (0, I_k )$, we have for all $x,$ $y$: 
\begin{align}
\mathbb{P} \left(  \underset{ u \in \mathbb{R}^{p - k}, \norm{u} = 1  }{ \text{max} }  \rho^2 \left( \xi, N_u \right) \leq x, \norm{ \xi } \leq y  \right) = \mathbb{P} \left( \rho_2^2 \left( \eta, N_2^C \right) \leq x, \norm{\eta} \leq y \right), 
\end{align}
where the coordinates of the two-dimensional random vector $\eta = \left( \sqrt{\chi^2_p},  \sqrt{\chi^2_{k-p}} \right) \in \mathbb{R}^2$ are independently distributed, $N_2^{C} = \left\{ \left( z_1, z_2 \right) \in \mathbb{R}^2 : z_1^2 + (C + z_2 )^2 = C^2 \right\}$ is a circle of radius $C$ with the center at $( 0, - C)$, and $\rho_2$ is Euclidean distance in $\mathbb{R}^2$.  
\end{itemize}
\end{theorem}

Condition (a), establishes that the manifold $\mathcal{S}_C$ lies inside the set $\mathcal{M}$ bounded by an envelope we construct from a collection of $p-$dimensional spheres $N_u$. Statement (b) assets that the distance from a point $\xi$ to the manifold $\mathcal{S}$ is bounded by the distance from $\xi$ to the furthest sphere in this collection, while (c) picks out exactly which sphere $N_{ \widetilde{\xi} }$ is the furthest away for a given $\xi$. Finally, (d) shows that the distribution of the distance from $\xi \sim \mathcal{N} \left( 0, I_k \right)$ to $N_{ \widetilde{\xi} }$ is the same as the distribution of the distance from a random variable $\eta$ to a particular circle in $\mathbb{R}^2$ (see, \cite{andrews2016geometric}).

\paragraph{Stochastic Bound}

Theorem 1 implies a bound on the distribution of the distance from $\xi \sim \mathcal{N} \left( 0, I_k \right)$ to a $p-$dimensional manifold $\mathcal{S}$. Assume that for some $C > 0$, $\mathcal{S}$ satisfies all the assumptions of Theorem 1 including Assumption 1. Then almost surely, 
\begin{align}
\rho^2 \left( \xi, \mathcal{S} \right) \leq \rho^2 \left( \xi, N_{\widetilde{u}} \right), 
\end{align}

%%-------------------------------------------------------------------------%%
\newpage 

By Theorem 1 (d), the distribution of the right-hand side of (6) is the same as the distribution of the random variable $\psi_C$, 
\begin{align}
\psi_{C} = \rho_2^2 \left( \eta, N_2^2 \right), 
\end{align}
 where the coordinates of the two-dimensional random vector $\eta = \left( \sqrt{\chi^2_{p}} , \sqrt{\chi^2_{k - p}} \right) \in \mathbb{R}^2$ are independently distributed,  $N_2^{C} = \left\{ (z_1, z_2) \in \mathbb{R}^2 : z_1^2 + ( C + z_2^2)^2 = C^2 \right\}$ is a circle of radius $C$ with the center at $(0,-C)$, and $\rho_2$ is Euclidean distance in $\mathbb{R}^2$. Combining these results, we establish the bound
\begin{align}
\mathbb{P} \left( \rho^2 \left( \xi, \mathcal{S} \right) \right) \leq \mathbb{P} \left( \psi_C \geq x \right), \ \ \text{for all} \ x > 0, 
\end{align}
so the distribution of $\psi_C$ is an upper bound on the distribution of $\rho^2 \left( \xi, \mathcal{S} \right)$.  

Notice that the distribution of $\psi_C$ depends only on the distribution of the space $k$, the dimension $p$ of the manifold, and the maximal curvature $\frac{1}{C}$. Then, the distribution of $\psi_{C}$ is stochastically increasing in the maximal curvature and hence stochastically decreasing in $ C$, so if $C_1 < C_2$, then $\psi_{C_1}$ first-order stochastically dominates $\psi_{C)2}$. As $C \to \infty$, $\psi_C \rightarrow \chi^2_{ k - p }$, so if the curvature converges to zero at all relevant points, then our bounding distribution converges to the distribution of the distance from $\xi \sim N( 0, I_k )$ to a $p-$dimensional linear subspace. Moreover, $\psi_C \rightarrow \chi^2_{k}$ as $C \to 0$, so if the curvature of the manifold becomes arbitrarily large, our bound coincides with the naive bound (2) that can be imposed without any assumptions on the manifold.  The authors here emphasize that the proposed stochastic bound holds under general assumptions. For example, if the model of interest has additional structure, this can potentially be exploited to obtain tighter bounds (\cite{andrews2016geometric}). 

\subsubsection{Asymptotic Properties}

If the manifold $\mathcal{S}$ satisfies the assumptions of Theorem 1, then the MD statistic is stochastically dominated by $\psi_{C}$ under the null hypothesis. We examine the following asymptotic properties. 

\paragraph{Uniformity}

Define a model to be a set consisting of a true value of the $k-$dimensional reduced-form parameter $\theta_0$, a data generating process $F_n$ consistent with $\theta_0$, and a link function connecting the structural and reduced-form parameters, or more generally a manifold $\widetilde{S}_n$ describing the null hypothesis, $H_0: \theta_0 : \widetilde{S}_n$. We assume that the null holds. We allow the data-generating process $F_n$ and the structural model $\widetilde{S}_n$ to change with the sample size $n$, this accommodates sequences of link functions such as those which arise under drifting asymptotic embeddings, for example, the weak identification embeddings of $\mathcal{D}$. This also allows to consider the case where the aim of the practitioner is  to fit a more complicated or nonlinear model. Suppose that we have an estimator $\widehat{\theta}_n$, which will be asymptotically normal with asymptotic covariance matrix $\Sigma = \Sigma (F_n)$. Let $\widehat{\Sigma}_n$ be an estimator for $\Sigma$. We consider the set possible models $\mathcal{M} = \left\{ M : M = \left( \theta_0, \left\{ F_n \right\}_{n=1}^{\infty}, \{ \widetilde{S}_n \}_{n=1}^{\infty} \right) \right\}$ and impose the following assumption.

%%-------------------------------------------------------------------------%%
\newpage 

\begin{assumption}[\cite{andrews2016geometric}]
We impose the following assumptions: 
\begin{itemize}
\item[(i)] $\sqrt{n} \ \Sigma^{ - 1/2 } \left( \widehat{\theta}_n - \theta_0   \right) \rightarrow \mathcal{N} \left( 0, I_k \right)$ uniformly over $\mathcal{M}$; 

\item[(ii)] $\widetilde{ \Sigma } - \Sigma \to_p 0$ uniformly over $\mathcal{M}$; 

\item[(iii)] the maximal and minimal eigenvalues of $\Sigma$ are bounded above and away from zero uniformly over $\mathcal{M}$; 

\item[(iv)] for each $n$ and manifold $\mathcal{S}_n = \left\{ x = \sqrt{n} \Sigma^{ - 1/2 } \left( \widehat{\theta}_n - \theta_0 \right) , y \in \widetilde{S}_n \right\}$, the manifold $S_n$ satisfies Assumption 1 for $C = C_n = 1 / \text{sup}_{ q \in S_n } \kappa_q \left( S_n \right)$. 
\end{itemize}
\end{assumption}
Notice that Assumption 2(i) and (ii)  of \cite{andrews2016geometric} require that the reduced-form parameter estimates are uniformly asymptotically normal with a uniformly consistently estimable covariance matrix. This assumption holds generally for many standard reduced-form estimators, such as OLS estimates and sample covariances, over large classes of models. Assumption (iii) uniformly bounds the eigenvalues of the asymptotic covariance matrix above and below, and will generally follow from a uniform bound on the moments of the data-generating process (\cite{andrews2016geometric}).

\paragraph{Description of the procedure} Consider a manifold of the form (see, \cite{andrews2016geometric})
\begin{align}
\widetilde{S}_n = \left\{ \sqrt{n} \widetilde{\Sigma}_n^{- 1 / 2} \left( x - \theta_0 \right) : x \in \widetilde{S}_n \right\}, 
\end{align}
which differs from $S_n$ in using an estimator $\widetilde{S}_n$ in place of $\Sigma$. Define with $\widetilde{C}_n := \left(\text{sup}_{q \in \widetilde{S}_n } \kappa_1 \left( \widetilde{S}_n \right) \right)^{-1}$. Our main test used the statistic $n$ min$_{ \theta \in \widetilde{S}_n } \left( \widehat{\theta}_n - \theta \right)^{\prime} \widetilde{ \Sigma }_n^{-1} \left( \widehat{\theta}_n - \theta \right)$, along with critical value $F_{1 - \alpha} \left( \widehat{C}_n , k , p \right)$. Notice that the critical values are obtained based on the limiting distribution of the test where we denote with $F_{ 1 - \alpha} \left( C, k , p \right)$ the $(1 - \alpha)-$quantile of the random variable $\psi_C$.    

\begin{theorem}[\cite{andrews2016geometric}]
If Assumption 2 holds, then the testing procedure described above has uniform asymptotic size $\alpha$: 
\begin{align}
\underset{ n \to \infty }{ \text{lim sup} } \ \underset{  M \in \mathcal{M} }{ \text{sup} } \mathbb{P} \left( n \ \underset{ \theta \in \widetilde{S}_n }{ \text{min} } \left( \widehat{\theta}_n - \theta \right)^{\prime} \widetilde{ \Sigma }_n^{-1} \left( \widehat{\theta}_n - \theta \right) \right) \leq \alpha. 
\end{align}
\end{theorem}
The theorems establishes the uniform asymptotic validity of the proposed test statistic allowing arbitrarily nonlinear (or linear) behaviour in the sequence of null hypothesis manifolds $\widetilde{S}_n$. The main intuition here is that critical values reflect the curvature of the null hypothesis manifold measured relative to the uncertainty about the reduced-form parameters for each sample size (see, \cite{andrews2016geometric}) due to the projection properties of the curvature. Recent applications include inference on manifolds with Fr\'echet regressions (see,  \cite{petersen2019frechet}) as well as other global testing methodologies such as \cite{arias2011global}, \cite{mansmann2021package} and \cite{vesely2023permutation}). Moreover, the literature on inference with many instruments include \cite{crudu2021inference} and others.

%%-------------------------------------------------------------------------%%
\newpage

\begin{example}
When the instruments are weak, in general $\hat{\beta} (k)$ is not consistent and has a nonstandard asymptotic distribution. Moreover, $T \left( \hat{k}_{LIML} - 1 \right)$ has a nondegenerate asymptotic distribution such that $\hat{\beta}_{TSLS}$ and $\hat{\beta}_{LIML}$ are not equivalent under weak instrument asymptotics. The asymptotic distribution of the test statistics is nonstandard. Notice that the decorrelated quasi-score function $\bar{ \boldsymbol{S} }_n ( \boldsymbol{\theta} )$ is of dimension $d_0 K$ instead of dimension $d K$. In particular, given our initial estimator $\widehat{ \boldsymbol{ \beta } } = \big( \widehat{ \boldsymbol{ \theta } }^{\top} , \widehat{ \boldsymbol{ \gamma } }^{\top} \big)^{\top}$, we define our QDIF estimator given as below
\begin{align}
\tilde{ \boldsymbol{ \theta } } = \underset{ \boldsymbol{ \theta } \in \Theta_n }{ \mathsf{arg min} } \tilde{Q}_n \left( \boldsymbol{ \theta } \right), \ \ \ \text{where} \ \ \ \tilde{Q}_n \left( \boldsymbol{ \theta } \right) = n \bar{\boldsymbol{S}}_n \left( \boldsymbol{ \theta } \right)^{\top} \boldsymbol{C}^{-1} \bar{\boldsymbol{S}}_n \left( \boldsymbol{ \theta } \right). 
\end{align} 
In particular, $\Theta_n := \left\{ \boldsymbol{\theta} \in \mathbb{R}^{d_0} : \norm{ \boldsymbol{\theta} - \widehat{ \boldsymbol{\theta} } }_2  \leq c d_0^{-1 / 2} \right\}$ is a neighbourhood around the initial estimator  $\widehat{ \boldsymbol{\theta} }$ for some small constant $c > 0$ and
\begin{align}
\boldsymbol{C} := \frac{1}{n} \sum_{i=1}^n \bar{\boldsymbol{S}}_i \left( \boldsymbol{ \theta } \right) \bar{\boldsymbol{S}}_i \left( \boldsymbol{ \theta } \right)^{\top} \in \mathbb{R}^{ d_0 K \times d_0 K }. 
\end{align}
To alleviate these issues, we propose the above localized estimator by minimizing $\tilde{Q}_n ( \boldsymbol{\theta} )$ in a small neighbourhood around the initial estimator $\boldsymbol{\theta}$. Therefore, in this theoretical analysis we show that $\tilde{Q}_n ( \boldsymbol{\theta} )$ is strongly convex for $\boldsymbol{\theta}  \in \Theta_n$ with probability tending to one. Thus, any off-the-shelf convex optimization algorithm is applicable to solving the above problem. 

Consider the following hypothesis testing problem
\begin{align}
H_0: \boldsymbol{ \theta }^{*} = 0 \ \ \text{versus} \ \ H_1: \boldsymbol{ \theta }^{*} \neq 0
\end{align}
Based on the above result, we define the Wald-type test statistic as follows
\begin{align}
\widehat{T}_n = n \tilde{ \boldsymbol{\theta} }^{\top} \widehat{ \boldsymbol{ \Sigma } }^{-1}_{ \boldsymbol{\theta} } \tilde{ \boldsymbol{\theta} }.
\end{align}
In other words, the above result implies that the distribution of the test statistic $\widehat{T}_n$ can be approximated by a chi-square distribution with $d_0$ such that 
\begin{align}
\underset{ 1 \leq k \leq K }{ \mathsf{max} } \norm{ \frac{1}{n} \sum_{i=1}^n \boldsymbol{Z}_i \boldsymbol{\Psi}_i  \boldsymbol{Z}_i - \mathbb{E} \big( \boldsymbol{Z}_i \boldsymbol{\Psi}_i  \boldsymbol{Z}_i \big) }_2 = \mathcal{O}_{ \mathbb{P} } \left( \sqrt{ d_0 \mathsf{log} d_0 / n } \right).
\end{align}
We have that $\sqrt{n} \tilde{\beta}_j$ converges to $\frac{1}{ \sqrt{n} } \sum_{ i = 1 }^n A_{ij}$ such that as $n \to \infty$, 
\begin{align}
\underset{  j \in \mathcal{H}_0 }{ \mathsf{max} } \left| \sqrt{n} \tilde{\beta}_j - \frac{1}{ \sqrt{n} } \sum_{ i = 1 }^n A_{ij}  \right|
= 
\mathcal{O}_{ \mathbb{P} } \left( \sqrt{ d_0 \mathsf{log} d_0  n } \right).
\end{align}

\end{example}

%%-------------------------------------------------------------------------%%
\newpage 

\section{Elements on Convex and Non-Convex Optimization Problems}

Tools for convex optimization problems are well-known in the literature and discussed in several studies. However, when the statistical problem of interest is formulated using a non-convex objective function various challenges arises to ensure robust estimation and inference in structural econometric models. From statistical decision theory literature (see, seminal study of \cite{wald1950statistical}) as well as the literature on optimization techniques several studies consider the non-convex optimization scenario. In particular, \cite{bednarczuk2020lagrangian} consider conditions for determining the zero duality gap as well as strong duality. These results can be applied to specific classes of functions, such as prox-bounded functions, DC functions, weakly convex functions and paraconvex functions. Furthermore, \cite{yalcin2020weak} consider cases in which the Legendre-Fenchel transform doesn't apply anymore. In other words, for an optimization problem, without convexity conditions, the conventional Lagrangian function may not always guarantee the zero duality gap which may occur between the given (primal) problem and the dual one (\cite{spini2021robustness}). The particular case implies that to construct dual problems and establish zero duality gap relations for nonconvex optimization problems, more a general class of conjugate Lagrangian functions are required to ensure an admissible solution space. In this direction, related results developed in the literature include the weak conjugate functions approach proposed by \cite{azimov1999weak} and \cite{azimov2002stability} (see, also \cite{aubin1976estimates},  \cite{li1995zero}, \cite{kuccuk2012generalized}, \cite{cheraghi2017some} and \cite{bagirov2019sharp}). The zero duality gap property is characterized in terms of the lower semicontinuity of a certain perturbation function by \cite{rubinov2002zero}. Lastly, the problem of nonlinear constraints is discussed in the seminal paper of \cite{powell1978algorithms}. 

A second issue of consideration is that regularity conditions are needed to ensure that the use of inverse mappings (commonly used in structural econometric models) preserves convexity (see, \cite{kelly1957inversive}). For example, a vast literature considers estimators of conditional and unconditional treatment effects (see, \cite{li2017estimation}). Thus in cases with nonlinearities in the statistical problem, one needs to verify that the constraint sets are convex which implies verifying that the inversion map on the space of CDF is convex (possibly under additional restrictions on the CDFs). This is a crucial property commonly used for establishing risk measures with desirable properties (see, \cite{ruszczynski2006conditional}, \cite{frittelli2014conditionally}, \cite{leskela2017conditional} and \cite{de2023conditional}). Moreover, notice that a curve is convex if it lies in the boundary of its convex hull.  

\subsection{Duality in nonconvex optimization}

In this section we present the duality scheme and strong duality theorems for nonconvex optimization problems, which are based on the weak conjugate functions and the weak subdifferential concepts. This allows us to formulate conditions guaranteeing zero duality gap relations and existence of optimal solutions to primal and dual problems, are formulated in terms of objective and constrained functions defining the given primal problem. 

%%-------------------------------------------------------------------------%%
\newpage

\begin{remark}
On the other hand, an open problem in the literature is in the case of inequality constrained problem, to formulate conditions guaranteeing the existence of optimal solutions to the dual problem, with nonconvex constraint function. Therefore, in relation to the proposed approach above for the formulation of a dual problem, and the zero duality gap conditions, related research questions include:
\begin{itemize}

\item How the given general scheme can be used to obtain a dual problem for a particular optimization problem, for example for the problem with equality and/or inequality constraints? 

\item How the zero duality gap conditions of the main theorem, formulated in terms of dualizing parametrization or perturbation functions, can be utilized to formulate conditions in terms of objective and constraint functions defining the problem under consideration? 

\item Under which conditions the zero duality gap conditions formulated in terms of dualizing parameterization or perturbation functions, are guaranteed? 
    
\end{itemize}

\end{remark}

Therefore, following the approach proposed by \cite{azimov1999weak, azimov2002stability}, we apply the conjugacy scheme which uses suplinear functions of the form $\mathsf{g} = \langle x^{*}, x \rangle - \alpha \norm{x}$, instead of linear functions of the form $\ell (x) = \langle x^{*}, x \rangle$, which are commonly used in convex analysis. The purpose of using the particular approach is to prove the convexity of the nonlinear map due to the inversion of the cumulative distribution function when considering the conditional set. 

\subsubsection{The general setting}

Let $X$ and $Y$ be normed spaces and let $X^{*}$ and $Y^{*}$ be their dual spaces, respectively. Taking a function $f$ of $X$ into $\bar{\mathbb{R}} = \mathbb{R} \cup \left\{ \pm \infty \right\}$, we consider the minimization problem as below $\underset{ x \in X }{ \mathsf{inf} } \ f(x)$. The optimal solution of the above problem will be denoted with $f(x) = \mathsf{inf} (P)$. Moreover, consider a dualizing parametrization function $\Phi: x \times Y \to \bar{\mathbb{R}}$ such that $\Phi (x,0) = f(x)$. Then, it follows that $\mathsf{inf}_{ x \in X } \Phi (x,0) = \mathsf{inf} (P)$. Therefore, to construct the dual problem with respect to the function $\Phi$, we need to establish the weak conjugate function $\Phi^w$. 

\medskip

The value of $\Phi^w$ will simply be denoted by $\Phi^w \big( 0, y^{*}, \beta \big):$
\begin{align}
\Phi^w \big( 0, y^{*}, \beta \big) = \underset{ (x,y) \in X \times Y }{ \mathsf{sup} } \ \big\{ \langle y^{*}, y \rangle - \beta \norm{ y } - \Phi (x,y) \big\}.    
\end{align}
Then, the dual problem is defined as below
\begin{align}
(P^w): \ \ \underset{ \big( y^{*}, \beta \big) \in Y^{*} \times \mathbb{R}_{+} }{ \mathsf{sup} }  \left\{ - \Phi^w \big( 0, y^{*}, \beta \big) \right\}. 
\end{align}

Therefore, the supremum of the problem $( P^w )$, denoted by $\mathsf{sup} ( P^w )$ and any element $\big( y^{*}, \beta \big) \in Y^{*} \times \mathbb{R}_{+}$ such that $\Phi^w \big( 0, y^{*}, \beta \big) = \mathsf{sup} ( P^w )$, is called an optimal solution to $( P^w )$.

%%-------------------------------------------------------------------------%%
\newpage

\begin{remark}
An important theoretical contribution to this stream of literature that generalizes related statistical problems and associated estimation techniques (see, \cite{anatolyev2011methods}) is proposed by \cite{komunjer2016existence} who consider the conditional density projections, especially in the case in which the projection set is defined by moment inequality constraints. A related discussion is given in the study of \cite{tabri2021information}. The particular study presents new existence, dual representation and approximation results for the information projection in the infinite-dimensional setting for moment inequality models. Furthermore, these results are established under general conditions nesting both unconditional and conditional models, and allowing for an infinite number of such inequalities (see, \cite{bhattacharya2006iterative}, \cite{andrews2010inference},  \cite{andrews2012inference} and \cite{chen2003estimation}). Thus, the $I-$projection problem becomes a semi-infinite program because the choice variable is finite-dimensional with dimension being equal to the sample size, and there are infinitely many moment inequality constraints. Therefore, it can be viewed as a statistical procedure that tilts the empirical measure by an amount that minimizes the Kullback-Leibler divergence subject to the moment inequality constraints (see, also \cite{kumar2004information} for information inequality measures). 
\end{remark}

\subsection{Generalized Moment Problems}

Generalized moment problems optimize functional expectation over a class of distributions with generalized moment constraints. These problems have recently attracted growing interest due to their great flexibility in representing nonstandard moment constraints, such as entropy constraints and exponential-type moment constraints. In the paper of \cite{guo2022unified}, the authors propose a novel primal-dual optimality condition. In particular, this optimality condition enables us to reduce the original infinite dimensional problem to a nonlinear equation system with a finite number of variable. Their framework demonstrates a clear path for identifying the analytical solution if one is available, otherwise, it produces semi-analytical solutions that lead to efficient numerical algorithms.   

Let $Y$ and $X$ denote dependent and independent variables supported on $\mathcal{Y} \subset \mathbb{R}$ and $\mathcal{X} \subset \mathbb{R}$, respectively. Let $\mathcal{B} ( \mathcal{X} )$ denote the Borel sigma algebra on $\mathcal{X}$. 

\begin{theorem}(Inverse mapping theorem)
Let $E$ and $F$ be coordinate spaces, let $X$ be an open subset of $E$, and let $f: X \to F$ be locally compact morphism. If $D f(a)$ is a topological isomorphism, then there exist open sets $M$ and $N$, of $E$ and $F$ respectively, such that $a \in M \subset X, f : M \to N$ is a homomorphism, and the inverse $f^{-1} : N \to M$ is a compact holomorphic perturbation of $\big[ D f(a) \big]^{-1}$ with $ D f^{-1} \big[ f(x) \big] = \big[ D f(x) \big]^{-1}$ for every $x \in M$ (see,  \cite{ma2001inverse}). 
\end{theorem}

\subsubsection{The Generalized Moment Problem and the Optimality Condition}

Consider the generalized moment problem in the following form: 
\begin{align}
Z_P = \underset{ F(.) }{ \mathsf{max} } \ \int_{ \Omega } g(x) . dF(x) \ \ \text{s.t} \ \ \int_{\Omega} h_i (x) . dF(x) = m_i, \ \ i = 0,...,n,    
\end{align}

%%-------------------------------------------------------------------------%%
\newpage

\subsection{Topological Convergence}

Let $C(X,Y)$ denote the set of continuous functions from a metric space $X$ to a metric space $Y$. In particular, considering elements of $C(X,Y)$ as closed subsets of $X \times Y$, we say that $\left\{ f_n \right\}$ converges topologically to $f$ if $\mathsf{Li} f_n = \mathsf{Ls} f_n = f$. 

If $X$ is connected, then topological convergence in $C (X,Y)$ does not imply pointwise convergence, but if $X$ is locally connected and $Y$ is locally compact, then topological convergence in $C(X,Y)$ is equivalent to uniform convergence on compact subsets of $X$ (see, \cite{beer1985more}).  
 
\begin{itemize}

\item We consider in detail topological convergence versus pointwise convergence. In general both pointwise convergence and topological convergence in $C(X,Y)$ are weaker than Hausdorff metric convergence of graphs (induced by a metric compatible with the product uniformity).   

\item If $\left\{ f_n \right\}$ converges to a uniformly continuous function $f$ in the Hausdorff metric, then $\left\{ f_n \right\}$ actually converges uniformly to $f$. In particular, if $X$ is compact, then the Hausdorff metric on $C(X,Y)$ is topologically equivalent to the usual metric of uniform convergence (related definitions are given in \cite{Di2008statistical}).

\end{itemize}

\begin{proposition}
If $\Theta = \mathcal{U} \times S_p$, then the mapping $\psi: \Theta \to M_{ p \times p } \times S_p$ defined by 
\begin{align}
\theta := ( \Pi, \Omega ) \overset{ \psi}{\mapsto} \psi(\theta) = (A, \Sigma)\footnote{reflecting that $( A, \Sigma ) = \big( \mathsf{exp} ( \delta \Pi ), f_{\Pi} (\Omega) \big)$ characterizes the law $\mathbb{P}_{\theta}^{\top}$. The main idea is to first address the identification problem with $\Gamma = M_{p \times p}$ and then restrict $M_{ p \times p }$ appropriately to obtain identification results. },    
\end{align}
is injective, and therefore the parameter $\theta = ( \Pi, \Omega )$ is identifiable. 
\end{proposition}

\begin{proof}
First notice that applying the $\mathsf{vec}-$operator on the matrix mapping $\Sigma = f_{\Pi} (\Omega)$ gives
\begin{align}
\mathsf{vec} (\Sigma) := \left[ \int_0^{\delta} A(u) \otimes A(u) du    \right] \mathsf{vec} (\Omega).    
\end{align}
Assume that $\Theta = \mathcal{U} \times S_p$ and take two arbitrary parameter vectors such that $\theta = ( \Pi, \Omega )$ and $\tilde{\theta} = ( \tilde{\Pi}, \tilde{\Omega} )$, with $\theta \neq \tilde{\theta}$. Then, there are two possibilities:
\begin{itemize}
\item[\textit{(i)}] $\Pi \neq \tilde{\Pi}$, which implies that $A \neq \tilde{A}$ and 

\item[\textit{(ii)}] Let $\theta = ( \Pi, \Omega )$ and $\tilde{\theta} = ( \Pi, \tilde{\Omega} )$, with $\Omega \neq \tilde{\Omega}$. Since for $\Sigma \neq \tilde{\Sigma}$, it holds that:
\begin{align}
\mathsf{vec} (\Sigma) &= \left(  \int_0^{\delta} A(u) \otimes A(u) du \right) \mathsf{vec} (\Omega)
\\
\mathsf{vec} (\tilde{\Sigma}) &= \left(  \int_0^{\delta} A(u) \otimes A(u) du \right) \mathsf{vec} (\tilde{\Omega})
\end{align}
\end{itemize}
    
\end{proof}

%%-------------------------------------------------------------------------%%
\newpage

\section{Elements on Unit Root Testing}

We briefly discuss some relevant issues to unit root testing since in this lecture series we discuss nonstationary panel data models with unit roots. More detailed derivations, examples and related results for estimation and inference in nonstationary regression models are presented by \cite{katsouris2023limit}.

Under regularity conditions the pair of random variables 
\begin{align}
A_g \equiv - \frac{1}{n} \sum_{t=1}^n y_{t-1} g^{\prime} ( \Delta y_t  ) \ \ \text{and} \ \ \ B_g \equiv - \frac{1}{n^2} \sum_{t=1}^n y^2_{t-1} g^{\prime \prime} ( \Delta y_t  ).
\end{align}
has nondegenerated limiting distributions and the remainder terms $r_n(c)$ is $o_p(1)$ uniformly on compact sets of $c$ values. Although there is only one nuisance parameter (the localizing coefficient of persistence) in our model, the asymptotic sufficient statistic is two dimensional. Therefore, there no uniformly best estimate or uniformly most powerful test exists asymptotically. We can see that any test which rejects for small values of a linear combination of $A_g$ and $B_g$ will be asymptotically admissible in the sense that, in large samples, it has highest possible power for some alternative in a neighborhood of unity. Furthermore, inference based on $A_g$ and $B_g$ may have good properties even if $g( \epsilon )$ is not the correct log density. For example, if normality is assumed so that $g( \epsilon ) \equiv - \frac{1}{2} \epsilon^2$, then $A_g$ and $B_g$ become the least-squares statistics as below
\begin{align}
A_n \equiv - \frac{1}{n} \sum_{t=1}^n y_{t-1} \Delta y_t \ \ \text{and} \ \ \ B_n \equiv - \frac{1}{n^2} \sum_{t=1}^n y^2_{t-1}.   
\end{align}
Therefore, if the errors are actually normal, admissible tests can be constructed from the sufficient statistics $A_n$ and  $B_n$. But inference based on $A_n$ and  $B_n$ may be good regardless of the Gaussianity assumption. Since our interest is in inference when the parameter $\rho$ is close to one, we shall employ local-to-unity asymptotics where the parameter space is assumed to be a shrinking neighborhood of unity as the sample size grows. In other words, we reparametrize the model by writing it in the form $c = n ( \rho - 1 )$ and take $c$ to be a constant when making limiting arguments.      

\medskip

Then as $n \to \infty$ the sequence of random functions $y_n(s) \equiv n^{-1/2} y_{ \floor{sn} }$ converges weakly on $[0,1]$ to the Ornstein-Uhlenbeck process
\begin{align}
J_c( r) = \int_0^r e^{c (r-s) } dW(s)    
\end{align}
which satisfies the stochastic differential equation
\begin{align}
d J_c(r) = c J_c(r) dr + dW(r)     
\end{align}
with initial condition $J_c(0) = 0$. 

%%-------------------------------------------------------------------------%%
\newpage 

Thus, since our statistics $A_g$ and $B_g$ are well-behaved functions of $y_t$, they can be approximated by functionals of the process $J_c(r)$.  Thus, it is known that, under local alternatives where $c$ is fixed as the sample size $n \to \infty$,
\begin{align}
\big(  A_n, B_n \big) \Rightarrow \left(  \frac{1}{2} \big[ J_c^2 (1) - 1 \big], \int_0^1 J_c^2 (t)  \right)    
\end{align}
Furthermore, notice that although $A_n$ and $B_n$ are sufficient statistics under normality, the convergence holds for arbitrary error distributions satisfying our moment assumptions. Therefore, to compute asymptotically valid tests based on $A_n$ and $B_n$, we need the null distributions of the test statistics. Thus, for approximate Neyman-Pearson tests of the unit-root hypothesis that $c = 0$, the distributions of linear combinations of the following random variables are required
\begin{align}
A_n^* =  \frac{1}{2} \big[ J_c^2 (1) - 1 \big] \ \ \ \text{and} \ \ \ B_n^* = \int_0^1 J_c^2 (t)
\end{align}
Under correct specification, the limiting distribution of $\big( A_g, B_g \big)$ depends on $c$ and $\omega^2$. On the other hand, if $g = - \frac{1}{2} \epsilon^2$, then $\rho = \delta = \omega = 1$ regardless of the distribution of the error term.   Nevertheless, we can take advantage of the fact that the joint distribution of $A_g^*$ and $B_g^*$. We show that considerable simplification occurs if inference is based on particular functions of  $A_g$ and $B_g$, namely, the maximum likelihood estimator and the $t-$ratio.

\subsection{Inference on $c$ based on the MLE}

We can construct a test statistic for testing the null hypothesis $c = c_0$ against the alternative $c < c_0$ and then reject for small values of $\hat{c} - c_0$ where $\hat{c}$ is the maximum likelihood estimate of $c$. Since the reciprocal of the standardized Hessian $B_N$ is a common estimate of the asymptotic variance of the least squares estimate, an alternative is to reject the hypothesis that $c = c_0$ if the $t-$ratio. Paradoxically, the asymptotic null distribution of the ML $t-$statistic has moments closer to a standard normal the farther the population distribution is from normality (see, \cite{rothenberg1997inference}).

\subsubsection{Asymptotic power of unit-root tests under correct specification}

Following \cite{rothenberg1997inference}, when $e^g$ is the true likelihood function of the data, the Neyman-Pearson test of the null hypothesis $c = c_0$ against the point alternative $c = k_0 < c_0$ for some $k_0$, has an asymptotic local power function. The envelope power function $\pi^* (c) = \pi (c, c)$ is an upper bound for the local asymptotic power of any one-sided test. 

\begin{itemize}

\item Since the asymptotic distribution of $A_g$ and $B_g$ under correct specification of the likelihood depends only on $c$ and $\omega$, the envelope $\pi^* (c)$ depends on the single parameter $\omega$.

%%-------------------------------------------------------------------------%%
\newpage

\item  The large increases in power as $\omega$ rises is not surprising. For the AR(1) model $y_t = \rho y_{t-1} + \varepsilon_t$, where $\varepsilon_t$ are \textit{i.i.d} with density $e^g$ and $\left| \rho  \right|$ is considerably less than one, efficient tests of $\rho = \rho_0$ against $\rho < \rho_0$ have local asymptotic power functions of the form $\Phi \left( d + \omega \sigma \beta_n \right)$ where $\Phi(.)$ is the standard normal distribution function, $\omega^2 = \mathsf{Var} \left[ g^{\prime} ( \varepsilon_t ) \right]$ and $\beta_n = n^{1/2} \left(  \rho = \rho_0 \right)$. In other words, asymptotically the stable AR model behaves like the location model with power measured in terms of normal quantiles in linear in $\beta_n$ with slope proportional to $\omega$.  When $\omega = 1$ (which is the case for normally distributed $\varepsilon$), transformed power envelope becomes more curved, especially when $c$ is small, and its slope rises. However, when testing for a unit root, power is higher the further the errors are from normal (i.e., deviations from the Gaussianity assumption)  and the effect is somewhat greater than in the case where standard asymptotic theory applies.   

\item  Then, one can consider the asymptotic power (in terms of normal quantiles) of tests based on the maximum likelihood estimator and $t-$statistic. When the errors are normal $( \omega = 1 )$, the power functions for the two tests are essentially identical to the power envelope.  For all practical purposes, they are asymptotically equivalent and efficient. However, this is no longer true when the errors are nonnormal.  The power functions for the estimator and $t-$tests are both tangent to the envelope and hence are asymptotically admissible. But the power functions show considerably greater curvature than the envelope particularly when $\omega$ is large. In each case, the tangency for the test based on the $t-$statistic occur when power is approximately one-half. In practise, this exactly the prediction that second-order asymptotic theory makes for standard (non-unit root problems).     
\color{black}

\end{itemize}

\begin{example}
As an illustrative example, we consider parameter estimation in predictive regression models with autoregressive error structure. An extension to a corresponding panel data structure is possible although more challenging and left as a future exercise. 

Consider the first-order autoregressive model with AR error structure given by
\begin{eqnarray}
\label{AR1}
    X_t = \mu + \rho_n X_{t-1} + \varepsilon_t,~~~~ \varepsilon_t = \sum_{j=1}^p \psi_j \varepsilon_{t-j} + e_t
\end{eqnarray}

\begin{itemize}
    \item Moderate deviation Case (I): $\rho_n = \left( 1 + \frac{c}{n^a} \right)$ for some $c < 0$ and $a \in (0, 1)$;
    \item Moderate deviation Case (II): $\rho_n = \left( 1 + \frac{c}{n^a} \right)$ for some $c > 0$ and $a \in (0, 1)$;
\end{itemize}
such that $\varepsilon_t = \psi_1 \varepsilon_{t-1} + ... + \psi_p \varepsilon_{t-p} + e_t$. In particular, the OLS-based parameter vector $(\hat\mu,\hat\rho)$ can be obtained by minimizing the following expression 
\begin{eqnarray*}
    \sum_{t=1}^n \left(X_t - \mu - \rho X_{t-1}\right)^2.
\end{eqnarray*}
We estimate $\theta = (\mu, \rho, \psi_1, \cdots, \psi_p)^\top$ by minimizing
\begin{eqnarray*}
\widehat{ \theta } := \underset{ \theta \in \mathbb{R}^{p+2}  }{ \mathsf{arg max} }  \ \left\{  \sum_{t=1}^n \left( X_t - \mu - \rho X_{t-1} - \sum_{j=1}^p \psi_j \big( X_{t-j} - \mu - \rho X_{t-j-1} \big) \right)^2 \right\}.
\end{eqnarray*}

%%-------------------------------------------------------------------------%%
\newpage

Denote the resulting estimators as $\hat\mu, \hat \rho$, respectively.
\color{black}
\begin{align*}
\frac{ \partial \widehat{\theta} }{ \partial \mu }  &= - 2n  \left\{  \sum_{t=1}^n \left(X_t - \mu - \rho X_{t-1} - \sum_{j=1}^p \psi_j \big( X_{t-j} - \mu - \rho X_{t-j-1} \big) \right) \right\} \times \left\{ 1 - \sum_{j=1}^p \psi_j \right\} \equiv 0
\\
\\
\frac{ \partial \widehat{\theta} }{ \partial \rho }  &= -2  \left\{  \sum_{t=1}^n \left(X_t - \mu - \rho X_{t-1} - \sum_{j=1}^p \psi_j \big( X_{t-j} - \mu - \rho X_{t-j-1} \big) \right) \right\} \times \left\{  \sum_{t=1}^n X_{t-1} - \sum_{t=1}^n \sum_{j=1}^p  \psi_j X_{t-j-1} \right\} \equiv 0    
\\
\\
\frac{ \partial \widehat{\theta} }{ \partial \psi_j }  &= -2  \left\{  \sum_{t=1}^n \left(X_t - \mu - \rho X_{t-1} - \sum_{j=1}^p \psi_j \big( X_{t-j} - \mu - \rho X_{t-j-1} \big) \right) \right\} \times \left\{  \sum_{t=1}^n \sum_{j=1}^p \big( X_{t-j} - \mu - \rho X_{t-j-1} \big) \right\} \equiv 0   
\end{align*}
Then, it follows that
\begin{align}
\sum_{t=1}^n \sum_{j=1}^p X_{t-j} - np \widehat{\mu} - \widehat{\rho}  \sum_{t=1}^n \sum_{j=1}^p X_{t-j-1} &= 0.    
\\
\sum_{t=1}^n \left(X_t - \widehat{\mu}  - \widehat{\rho} X_{t-1} - \sum_{j=1}^p \widehat{\psi}_j \big( X_{t-j} - \widehat{\mu}  - \widehat{\rho} X_{t-j-1} \big) \right) &= 0
\\
\sum_{t=1}^n X_t  - n \widehat{\mu}  - \widehat{\rho} \sum_{t=1}^n X_{t-1} - \sum_{t=1}^n \sum_{j=1}^p \widehat{\psi}_j X_{t-j} + n \widehat{\mu} \sum_{j=1}^p \widehat{\psi}_j + \widehat{\rho} \sum_{t=1}^n \sum_{j=1}^p \widehat{\psi}_j X_{t-j-1} &= 0.
\end{align}
which implies that $\widehat{\mu} = \frac{1}{n} \sum_{t=1}^n \sum_{j=1}^p \widehat{\psi}_j X_{t-j}$.  Denote with 
\begin{align}
\hat{\mu}_x := \frac{1}{n - 1}  \sum_{t=1}^n X_{t-1} \ \ \ \text{and} \ \ \ \tilde{\mu}_x := \frac{1}{n} \sum_{t=1}^n X_t
\end{align}
Related conditions can be imposed to ensure identification of nonstationary time series with moderate deviations from the unit root boundary. In the special case where $a = 1$, and $\rho_n = ( 1 + c/n )$, then $X_t$ follows a local-to-unity $(I1)$ process and with $c = 0$ corresponding to an exact unit root and for values of $c \neq 0$ generating data that are less $( c < 0 )$ or more $( c > 0 )$ persistent than the exact unit root process\footnote{Furthermore note that when $\rho_n = \rho$, with $| \rho | < 1$, then $X_t$ is $I(0)$, a stable and stationary process, when $\rho_n = 1$, then $X_t$ is an $I(1)$ process while when $\rho_n = \rho$ with $\rho > 1$, then $X_t$ is an unstable or explosive process. }. 
\end{example}

\begin{example}
Consider the multivariate predictive regression model which is estimated using a likelihood approach. Assume that the data $\big( Y_1,..., Y_n, \boldsymbol{X}_0,...,  \boldsymbol{X}_n^{\prime} \big)$ follow
\begin{align}
Y_t &= \mu + \boldsymbol{\beta}^{\prime} \boldsymbol{X}_{t-1} + u_t,
\\
\boldsymbol{X}_{t} &= \boldsymbol{\eta} + \boldsymbol{R} \boldsymbol{X}_{t-1} + \boldsymbol{v}_t,
\end{align}
where $u_t = \boldsymbol{\phi}^{\prime} \boldsymbol{v}_t + e_t$, such that $\big( e_t,  \boldsymbol{v}_t^{\prime} \big)^{\prime} \sim \mathcal{N} \big( \boldsymbol{0}, \mathsf{diag} \left( \sigma_e^2, \boldsymbol{\Sigma}_v   \right) \big)$ is an \textit{i.i.d} series and $\boldsymbol{R}$ is a $(k \times k)$ matrix.

%%-------------------------------------------------------------------------%%
\newpage 

Moreover, suppose that all eigenvalues of the autoregressive matrix is  less than unity in absolute value, which is a stability condition of the system.  Define with $\boldsymbol{\Sigma}_v \equiv \mathsf{Var} ( \boldsymbol{v}_t )$ and $\boldsymbol{\Sigma}_{ \boldsymbol{X} } \equiv \mathsf{Var} \left( \boldsymbol{X}_{t} \right)$ as 
\begin{align}
\mathsf{vec} \left( \boldsymbol{\Sigma}_{ \boldsymbol{X} } \right) = \big( I_{K^2} - \boldsymbol{R} \otimes \boldsymbol{R} \big)^{-1} \mathsf{vec} \left( \boldsymbol{\Sigma}_{v} \right)
\end{align}
Denote with 
\begin{align}
\hat{\boldsymbol{K}} 
= 
\left[ \boldsymbol{\Sigma}_{ \boldsymbol{X} }^{-1} + n ( \boldsymbol{I} - \boldsymbol{R} )^{\prime} \boldsymbol{\Sigma}_{v}^{-1} ( \boldsymbol{I} - \boldsymbol{R} ) \right]^{-1} \left[ \boldsymbol{\Sigma}_{ \boldsymbol{X} }^{-1} \boldsymbol{X}_0 +  ( \boldsymbol{I} - \boldsymbol{R} )^{\prime} \boldsymbol{\Sigma}_{v}^{-1} ( \boldsymbol{I} - \boldsymbol{R} ) \sum_{t=1}^n \boldsymbol{X}_t  \right]
\end{align}
Then, the REML log-likelihood up to an additive constant for the model is given by 
\begin{align*}
L_M 
&= 
- \left( \frac{n-1}{2} \right) \mathsf{log} \sigma_e^2 - \frac{1}{ 2 \sigma_e^2 } S \left( \boldsymbol{\phi}, \boldsymbol{\beta}, \boldsymbol{R}  \right) - \frac{1}{2} \mathsf{log} | \boldsymbol{\Sigma }_{ \boldsymbol{X} } | - \frac{n}{2} \mathsf{log} | \boldsymbol{\Sigma }_{ v } |
\nonumber
\\
&\ \ \ -
\frac{1}{2} \mathsf{log} \left| \boldsymbol{\Sigma }_{ \boldsymbol{X} } + n ( \boldsymbol{I} - \boldsymbol{R} )^{\prime} \boldsymbol{\Sigma}_{v}^{-1} ( \boldsymbol{I} - \boldsymbol{R} )  \right|
\nonumber
\\
&\ \ \ -
\frac{1}{2} \left\{ \left( \boldsymbol{X}_0 - \hat{\boldsymbol{K}} \right)^{\prime} \boldsymbol{\Sigma }_{ \boldsymbol{X} }^{-1} \left( \boldsymbol{X}_0 - \hat{\boldsymbol{K}} \right) + \sum_{t=1}^n \left[  \boldsymbol{X}_t - \hat{\boldsymbol{K}} - \boldsymbol{R} \left( \boldsymbol{X}_{t-1} - \hat{\boldsymbol{K}} \right) \right]^{\prime} \boldsymbol{\Sigma}_v^{-1} \left[  \boldsymbol{X}_t - \hat{\boldsymbol{K}} - \boldsymbol{R} \left( \boldsymbol{X}_{t-1} - \hat{\boldsymbol{K}} \right) \right] \right\}
\end{align*}
where
\begin{align}
S \left( \boldsymbol{\phi}, \boldsymbol{\beta}, \boldsymbol{R}  \right)  
&= 
\sum_{t=1}^n \bigg[ Y_t^{\mu} - \boldsymbol{\phi}^{\prime} \boldsymbol{X}_t^{\mu} - \big( \boldsymbol{\beta}^{\prime} - \boldsymbol{\phi}^{\prime} \boldsymbol{R} \big) \boldsymbol{X}_{t-1}^{\mu} \bigg]^2.
\\
\boldsymbol{X}_t^{\mu}
&= 
\boldsymbol{X}_t - \frac{1}{n} \sum_{t=1}^n \boldsymbol{X}_t \ \ \ \ \text{and} \ \ \ \ \boldsymbol{X}_{t-1}^{\mu} = \boldsymbol{X}_{t-1} - \frac{1}{n} \sum_{t=1}^n \boldsymbol{X}_{t-1}. 
\end{align}
\begin{remark}
Notice that in the case where $\boldsymbol{R}$ is assumed to be diagonal matrix, the predictive regression model is no longer a seemingly unrelated regression (SUR) system and hence OLS will no longer be efficient. However, REML will clearly retain efficiency, no matter what the form of $\boldsymbol{R}$ is, thus giving it an advantage in terms of both asymptotic efficiency and power over any OLS-based procedure. Furthermore, since the dimension of the parameter space is very large in the vector case,  it is not feasible to obtain a result such as Theorem 3 in the most general case. However, in the case where $\boldsymbol{R}$ is a diagonal matrix and where $\big( \sigma_e^2, \boldsymbol{\phi}, \boldsymbol{\Sigma}_v \big)$ are assumed known with $\boldsymbol{\Sigma}_v$ diagonal, we are able to obtain the following result on the finite-sample behaviour of the RLRT for testing $\mathbb{H}_0: \boldsymbol{\beta} = \boldsymbol{0}$.  
\end{remark}

\end{example}

%%-------------------------------------------------------------------------%%
\newpage

\bibliographystyle{apalike}
\bibliography{myreferences1}

\addcontentsline{toc}{section}{References}

%%-------------------------------------------------------------------------%%
\newpage

\end{document}